\documentclass[a4paper,10pt,numbers=noenddot,headsepline=true,DIV=8,BCOR=10mm]{scrbook}

\usepackage[utf8]{inputenc}
\usepackage[T1]{fontenc}
\usepackage{textcomp}
\usepackage[english,ngerman]{babel}
\usepackage{color}
\usepackage{amsmath}
\usepackage{amsopn}
\usepackage{amsbsy}
\usepackage[cal=boondoxo,bb=boondox]{mathalfa}
\usepackage{enumerate}
\usepackage[plainpages=false,pdfpagelabels,breaklinks=true,pdfborderstyle={/S/U/W 1.2}]{hyperref}
\usepackage{graphics}
\usepackage{units}
\usepackage{tabularx}
\usepackage[shortlabels]{enumitem}

\numberwithin{equation}{section} 
\numberwithin{table}{section} 
\numberwithin{figure}{section} 
\setlist[enumerate,1]{itemsep=1mm,listparindent=\parindent,label={(\arabic*)}}

\usepackage[amsmath,thmmarks,thref,hyperref]{ntheorem}

\theoremstyle{plain}
\newtheorem{theorem}{Theorem}[section]
\newtheorem{definition}[theorem]{Definition}
\newtheorem{lemma}[theorem]{Lemma}
\newtheorem{corollary}[theorem]{Corollary}
\newtheorem{proposition}[theorem]{Proposition}
\newtheorem{consequence}[theorem]{Consequence}
\newtheorem{assumption}[theorem]{Assumption}
\newtheorem{hypothesis}{Hypothesis}

\theorembodyfont{\upshape}
\theoremsymbol{\ensuremath{\star}}
\newtheorem{remark}[theorem]{Remark}
\newtheorem{example}[theorem]{Example}

\theoremstyle{nonumberplain}
\theoremsymbol{\ensuremath{_\Box}}
\newtheorem{proof}{Proof}

\usepackage[style=alphabetic,
	citestyle=alphabetic,
	firstinits=true,
	backend=biber,
	hyperref=auto,
	sorting=nyt,
	maxcitenames=3,
	maxbibnames=99,
	minnames=3,
	arxiv=pdf,
	url=false,
	isbn=false,
	dateabbrev=true,
	datezeros=true,
	bibencoding=utf8]{biblatex}
\newbibmacro{string+doi}[1]{%
	\iffieldundef{doi}{#1}{\href{http://dx.doi.org/\thefield{doi}}{#1}}}
\DeclareFieldFormat
	{title}{\usebibmacro{string+doi}{\mkbibemph{#1}}\isdot}
\DeclareFieldFormat
	[article,inbook,incollection,inproceedings,patent,thesis,unpublished]
	{title}{\usebibmacro{string+doi}{\mkbibemph{#1}}\isdot}
\DeclareFieldFormat
	[article,inbook,incollection,inproceedings,patent,thesis,unpublished]
	{pages}{#1}
\DeclareFieldFormat
	[article]
	{journaltitle}{#1\isdot}
\DeclareFieldFormat
	[article]
	{volume}{\textbf{#1}}
\DefineBibliographyStrings{english}{pages={}}
\DeclareBibliographyDriver{article}{%
	\usebibmacro{bibindex}%
	\usebibmacro{begentry}%
	\usebibmacro{author/translator+others}%
	\setunit{\labelnamepunct}\newblock
	\usebibmacro{title}%
	\usebibmacro{journal+issuetitle}%
	\setunit*{\addcomma\space}
	\printfield{pages}
	\setunit*{\addcomma\space}
	\printfield{year}
	\usebibmacro{finentry}%
}

\newbibmacro*{journal+issuetitle}{%
	\usebibmacro{journal}%
	\setunit*{\addspace}%
	\iffieldundef{series}
	{}
	{\newunit
		\printfield{series}%
		\setunit{\addspace}}%
	\iffieldundef{volume}
		{}
		{\printfield{volume}%
		\setunit{\addspace}}%
	\setunit*{\addcolon\space}%
	\usebibmacro{issue}%
	\newunit}

\bibliography{./bibliography}

\usepackage[scaled]{beramono}
\usepackage{helvet}
\usepackage{mathrsfs}
\usepackage{eucal}
\usepackage{dsfont}
\usepackage[charter]{mathdesign} 
\usepackage{makeidx}
\usepackage{xcolor}

\hypersetup{
	bookmarks=true,
	pdftitle={Linear Response Theory — A Modern Analytic-Algebraic Approach},
	pdfauthor={Giuseppe De Nittis & Max Lein}
}

\colorlet{chapter}{black!75}
\addtokomafont{chapter}{\color{chapter}}

\makeatletter
\renewcommand*{\chapterformat}{%
\begingroup
\setlength{\unitlength}{1mm}%
\begin{picture}(20,40)(0,5)%
\setlength{\fboxsep}{0pt}%
\put(20,15){\line(1,0){\dimexpr
\textwidth-20\unitlength\relax\@gobble}}%
\put(0,0){\makebox(20,20)[r]{%
\fontsize{28\unitlength}{28\unitlength}\selectfont\thechapter
\kern-.04em
}}%
\put(20,15){\makebox(\dimexpr
\textwidth-20\unitlength\relax\@gobble,\ht\strutbox\@gobble)[l]{%
\ \normalsize\color{black}\chapapp~\thechapter\autodot
}}%
\end{picture} 
\endgroup
}

\newcolumntype{Y}{>{\small\raggedright}X}

\providecommand{\beq}{\begin{align}}
\providecommand{\bea}{\begin{eqnarray}}
\providecommand{\eea}{\end{eqnarray}}
\providecommand{\beas}{\begin{eqnarray*}}
\providecommand{\eeas}{\end{eqnarray*}}
\providecommand{\beql}{\begin{align} \label}
\providecommand{\eeq}{\end{align}}

\providecommand{\R}{\mathbb R}
\providecommand{\N}{\mathbb N}
\providecommand{\C}{\mathbb C} 

\providecommand{\Z}{\mathbb Z}
\providecommand{\T}{\mathbb T}

\providecommand{\bb}[1]{\mathscr{#1}}
\providecommand{\rr}[1]{\mathfrak{#1}}
\providecommand{\n}[1]{\mathds {#1}}

\providecommand{\ketbra}[2]{|#1\rangle\langle#2|}

\providecommand{\dd}{\,{\rm d}}
\providecommand{\ii}{\,{\rm i}\,}
\providecommand{\ncint}{\mathrel{{\ooalign{$ \int$\cr\kern+.07em\raise.15ex\hbox{$\pmb{\scriptstyle-}$}\cr}}}} \providecommand{\ncpartial}{\mathrel{{\ooalign{$\partial$\cr\kern+.29em\raise.79ex\hbox{$\pmb{\scriptstyle-}$}\cr}}}}

\providecommand{\ie}{i.~e.~}
\providecommand{\eg}{e.~g.~}
\providecommand{\cf}{cf.~}

\providecommand{\R}{\mathbb{R}}

\providecommand{\C}{\mathbb{C}}
\renewcommand{\C}{\mathbb{C}}

\providecommand{\T}{\mathbb{T}}
\renewcommand{\T}{\mathbb{T}}
\providecommand{\N}{\mathbb{N}}
\providecommand{\Z}{\mathbb{Z}}
\providecommand{\ii}{\mathrm{i}}
\providecommand{\e}{\mathrm{e}}
\renewcommand{\Re}{\mathrm{Re} \,}

\providecommand{\Hil}{\mathcal{H}}

\providecommand{\eps}{\varepsilon}

\providecommand{\ker}{\mathrm{Ker} \, }

\providecommand{\trace}{\mathrm{Tr} \,}

\providecommand{\dd}{\mathrm{d}}
\providecommand{\id}{\mathds{1}}

\providecommand{\order}{\mathcal{O}}

\providecommand{\trace}{\mathrm{Tr}}

\providecommand{\abs}[1]{\left \lvert #1 \right \rvert}
\providecommand{\sabs}[1]{\lvert #1 \vert}
\providecommand{\babs}[1]{\bigl \lvert #1 \bigr \rvert}
\providecommand{\Babs}[1]{\Bigl \lvert #1 \Bigr \rvert}
\providecommand{\norm}[1]{\left \lVert #1 \right \rVert}
\providecommand{\snorm}[1]{\lVert #1 \rVert}
\providecommand{\bnorm}[1]{\bigl \lVert #1 \bigr \rVert}
\providecommand{\Bnorm}[1]{\Bigl \lVert #1 \Bigr \rVert}
\providecommand{\scpro}[2]{\left \langle #1 , #2 \right \rangle}
\providecommand{\sscpro}[2]{\langle #1 , #2 \rangle}
\providecommand{\bscpro}[2]{\bigl \langle #1 , #2 \bigr \rangle}

\providecommand{\sexpval}[1]{\langle #1 \rangle}

\providecommand{\ketbra}[2]{\sopro{#1}{#2}}

\providecommand{\ncint}{\mathrel{{\ooalign{$ \int$\cr\kern+.07em\raise.15ex\hbox{$\pmb{\scriptstyle-}$}\cr}}}} \providecommand{\ncpartial}{\mathrel{{\ooalign{$\partial$\cr\kern+.29em\raise.79ex\hbox{$\pmb{\scriptstyle-}$}\cr}}}}

\makeatletter
\newsavebox{\@brx}
\newcommand{\llangle}[1][]{\savebox{\@brx}{\(\m@th{#1\langle}\)}%
 \mathopen{\copy\@brx\mkern2mu\kern-0.9\wd\@brx\usebox{\@brx}}}
\newcommand{\rrangle}[1][]{\savebox{\@brx}{\(\m@th{#1\rangle}\)}%
 \mathclose{\copy\@brx\mkern2mu\kern-0.9\wd\@brx\usebox{\@brx}}}
\makeatother
\providecommand{\sdscpro}[2]{\llangle {#1 , #2} \rrangle}

\providecommand{\Bdscpro}[2]{{\llangle[\Big]} {#1 , #2} {\rrangle[\Big]}}

\providecommand{\switch}{s_{\eps}}
\providecommand{\Alg}{\mathscr{A}}
\providecommand{\Hil}{\Hil}
\providecommand{\domain}{\CMcal{D}}
\providecommand{\spec}{\mathrm{Spec}}
\providecommand{\res}{\mathrm{Res}}
\providecommand{\proj}{\mathrm{Proj}}
\providecommand{\affil}{\mathrm{Aff}}
\providecommand{\tnorm}[1]{{\left\vert\kern-0.25ex\left\vert\kern-0.25ex\left\vert #1 
    \right\vert\kern-0.25ex\right\vert\kern-0.25ex\right\vert}}
\providecommand{\stnorm}[1]{{\vert\kern-0.25ex\vert\kern-0.25ex\vert #1 
    \vert\kern-0.25ex\vert\kern-0.25ex\vert}}
\providecommand{\btnorm}[1]{{\bigl\vert\kern-0.25ex\bigl\vert\kern-0.25ex\bigl\vert #1 
    \bigr\vert\kern-0.25ex\bigr\vert\kern-0.25ex\bigr\vert}}
\providecommand{\Btnorm}[1]{{\Bigl\vert\kern-0.25ex\Bigl\vert\kern-0.25ex\Bigl\vert #1 
    \Bigr\vert\kern-0.25ex\Bigr\vert\kern-0.25ex\Bigr\vert}}

\providecommand{\slim}{\mbox{s} \negthinspace - \negthinspace \lim}
\providecommand{\wlim}{\mbox{w} \negthinspace - \negthinspace \lim}
\providecommand{\ad}{\mathrm{ad}}

\begin{document}

\selectlanguage{english}
\author{Giuseppe De Nittis \& Max Lein}
\title{Linear Response Theory}
\subtitle{A Modern Analytic-Algebraic Approach}
\dedication{For Raffa \& Tenko}
\begin{titlepage}
	\makebox{}
	\vfill
	\begin{center}
		\textsf{{\Huge Linear Response Theory} 
		\\[1mm]
		{\LARGE\bfseries\color{black!40} A Modern Analytic-Algebraic Approach}
		\\
		\vspace{6mm}
		{\LARGE Giuseppe De Nittis \& Max Lein}
		}
	\end{center}
	\vfill
\end{titlepage}
\begin{titlepage}
	\makebox{}
	\vfill
	\begin{center}
		\emph{For Raffa \& Tenko}
	\end{center}
	\vfill
\end{titlepage}
%

\frontmatter

%
\chapter*{Acknowledgements}
We are indebted to François Germinet who introduced G.~D.\ to \cite{Bouclet_Germinet_Klein_Schenker:linear_response_theory_magnetic_Schroedinger_operators_disorder:2005,Dombrowski_Germinet:linear_response_theory_non_commutative_integration:2008}, the works which have inspired this whole endeavor. In addition to François, many other colleagues were not just kind enough to share their insights with us, but also encouraged us to see this book project to an end. We would particularly like to express our gratitude to Jean Bellissard, Hermann Schulz-Baldes, Daniel Lenz, Georgi Raikov and Massimo Moscolari. Thanks to their help we were able to overcome many of the obstacles more quickly and more elegantly. 

Lastly, G.~D.\ thanks FONDECYT which supported this book project through the grant “\emph{Iniciaci\'{o}n en Investigaci\'{o}n 2015} - $\text{N}^{\text{o}}$~11150143”. 

\tableofcontents

\mainmatter
\chapter{Introduction} 
\label{intro}
Linear response theory is a tool with which one can study systems that are driven out of equilibrium by external perturbations. The prototypical example is a first-principles justification of Ohm's empirical law $J = \sum_{j = 1}^d \sigma_j \, E_j$ \cite{Ohm:galvanische_Kette:1827}, which states that the current is linearly proportional to the applied external electric field: These ideas have been pioneered by Green \cite{Green:linear_response_theory:1954} and Kubo \cite{Kubo:linear_response_theory:1957} in the context of statistical mechanics, and later used by Str\v {e}da \cite{Streda:linear_response_theory_quantum_Hall_effect:1982} to link the transverse conductivity in a 2d electron gas to the number of Landau levels below the Fermi energy. The aim of this book is to provide a modern tool for mathematical physicists, allowing them to make linear response theory (LRT) rigorous for a wide range of systems — including some that are beyond the scope of existing theory. We will explain all the moving pieces of this \emph{analytic-algebraic} framework below and put it into context with the literature. But first let us give a flavor of the physics. 

Initially, the unperturbed system, governed by a selfadjoint operator $H$, is at \emph{equilibrium}, meaning that it is an a state described by a density operator $\rho$ commuting with $H$. Then, in the distant past we adiabatically switch on a perturbation which drives the system out of equilibrium. Here we distinguish between a set of perturbation parameters $\pmb{\Phi} = \bigl ( \Phi_1 , \ldots , \Phi_d \bigr )$ (\eg components of the electric field) and the adiabatic parameter $\eps$ which quantifies how quickly the perturbation is ramped up. Thus, both enter as parameters in the \emph{perturbed Hamiltonian} $H_{\Phi,\eps}(t) = H_{\Phi,\eps}(t)^*$. The perturbation is switched on at $t_0 < 0$ (which in principle could be $-\infty$) and at $t = 0$ the Hamiltonian has reached the perturbed state. The adiabatic switching allows us to start with the \emph{same} initial state $\rho$ as in the unperturbed case, and evolve it according to 
\begin{align}
	\rho(t) = U_{\Phi,\eps}(t,t_0) \, \rho \, U_{\Phi,\eps}(t,t_0)^*
	, 
	\label{intro:eqn:definition_rho_full}
\end{align}
where $U_{\Phi,\eps}(t,t_0)$ is the unitary propagator associated to the time-dependent Hamiltonian $H_{\Phi,\eps}(t)$. 

In the simplest case we want to see the effects of the perturbation by studying the \emph{net current} 
\begin{align}
	\mathscr{J}^{\Phi,\eps}(t) &= \mathcal{T} \bigl ( J_{\Phi,\eps}(t) \, \rho(t) \bigr ) - \mathcal{T} \bigl ( J \, \rho \bigr ) 
	\label{intro:eqn:net_current}
\end{align}
which is the difference of the expectation values of the current operators $J_{\Phi,\eps}(t)$ and $J$ with respect to $\rho(t)$ and $\rho$, computed with the trace-per-volume $\mathcal{T}$. Typically, $J_{\Phi,\eps}(t) = - \ii \bigl [ X , H_{\Phi,\eps}(t) \bigr ]$ and $J = - \ii [X , H]$ are given in terms of commutators with the appropriate Hamiltonians, thereby explaining why the two current operators are different and one of them depends on $\Phi$, $\eps$ and $t$. 

The crucial step in making LRT rigorous is to justify the “Taylor expansion” of the net current to first order in $\pmb{\Phi} = 0$, 
\begin{align}
	\mathscr{J}^{\Phi,\eps}(t) &= \mathscr{J}^{\Phi,\eps}(t) \big \vert_{\pmb{\Phi} = 0} + \sum_{j = 1}^d \Phi_j \; \partial_{\Phi_j} \mathscr{J}^{\Phi,\eps}(t) \big \vert_{\pmb{\Phi} = 0} + o(\pmb{\Phi})
	\notag \\
	&= \sum_{j = 1}^d \Phi_j \, \sigma_j^{\eps}(t) + o(\pmb{\Phi})
	, 
	\label{intro:eqn:Taylor_expansion_net_current}
\end{align}
where the $0$th order terms vanish — no perturbation, no net current — and the \emph{conductivity coefficients} $\sigma_j^{\eps}(t)$ quantify the linear response. To wash out some of the details of the interpolation between the perturbed and the unperturbed system, typically one takes the \emph{adiabatic limit} $\eps \to 0$ of the conductivity coefficients. 
\medskip

\noindent
Kubo's contribution \cite{Kubo:linear_response_theory:1957} was the derivation of an explicit formula for the $\sigma_j^{\eps}(t)$ (\cf equation~\eqref{main_results:eqn:Kubo_formula}). Str\v {e}da has a second expression in case $\rho$ is a spectral projection; This Kubo-Str\v {e}da formula (\cf equation~\eqref{main_results:eqn:Kubo_Streda_formula}) has helped give a topological interpretation to the Quantum Hall Effect \cite{Thouless_Kohmoto_Nightingale_Den_Nijs:quantized_hall_conductance:1982,Hatsugai:Chern_number_edge_states:1993}, giving birth to the field of Topological Insulators in the process. 

That is why a significant share of the mathematically rigorous literature concerns LRT for various models of the Quantum Hall Effect (\eg \cite{Bellissard_van_Elst_Schulz_Baldes:noncommutative_geometry_quantum_Hall_effect:1994,Bellissard_Schulz_Baldes:quantum_transport_aperiodic_media:1998,Bouclet_Germinet_Klein_Schenker:linear_response_theory_magnetic_Schroedinger_operators_disorder:2005,Dombrowski_Germinet:linear_response_theory_non_commutative_integration:2008,Elgart_Schlein:Kubo_for_Landau:2004}). Roughly speaking, there are two approaches, those that attack LRT from the functional analytic side (\eg \cite{Bouclet_Germinet_Klein_Schenker:linear_response_theory_magnetic_Schroedinger_operators_disorder:2005,klein-lenoble-muller-07}) and those which formulate the problem in algebraic terms (\eg \cite{Bellissard_van_Elst_Schulz_Baldes:noncommutative_geometry_quantum_Hall_effect:1994,Bellissard_Schulz_Baldes:quantum_transport_aperiodic_media:1998,Jaksic_Pillet:mathematical_theory_non_equilibrium_quantum_statistical_mechanics:2002,Jaksic_Ogata_Pillet:Green_Kubo_formula_Onsager_relations:2006}). Typically, the main advantage of algebraic approaches is that they \emph{give a scheme} for how to make LRT rigorous, which applies to a whole class of systems, at the expense of rather strong assumptions on $H$, $\rho$ and $\mathcal{T}$. Very often these approaches require $H$ to lie in a $C^*$- or von Neumann algebra $\Alg$, and therefore $H$ is necessarily bounded, or that $\mathcal{T}$ is finite. This excludes many physically interesting and relevant models, most notably continuum (as opposed to discrete) models. Conversely, analytic approaches typically focus on one particular system, including those described by unbounded Hamiltonians. However, the details are usually specific to the Hamiltonian of interest, and these techniques do not readily transfer from one physical system to another. 

Therefore we have developed a \emph{unified and thoroughly modern framework} which \emph{combines the advantages} of both approaches: we give an \emph{explicit scheme} for LRT, based on von Neumann algebras and associated non-commutative $L^p$-spaces, that applies to discrete and continuous models alike, that can deal with disorder and is \emph{not tailored to one specific model}. It not only subsumes many previous results, notably \cite{Bellissard_Schulz_Baldes:quantum_transport_aperiodic_media:1998,Schulz-Baldes_Bellissard:kinetic_theory_quantum_transport:1998,Bouclet_Germinet_Klein_Schenker:linear_response_theory_magnetic_Schroedinger_operators_disorder:2005,Dombrowski_Germinet:linear_response_theory_non_commutative_integration:2008}, but also applies to systems that have not yet been considered in the literature. We will detail the precise setting, all hypotheses and our main results in Chapter~\ref{main_results}. Nevertheless, let us anticipate the most important aspects in order to contrast and compare with the literature. Our book is inspired by the works of Bouclet, Germinet, Klein and Schenker \cite{Bouclet_Germinet_Klein_Schenker:linear_response_theory_magnetic_Schroedinger_operators_disorder:2005} as well as Dombrowski and Germinet \cite{Dombrowski_Germinet:linear_response_theory_non_commutative_integration:2008}, who make LRT rigorous for Schrödinger operators $H_{\omega} = (- \ii \nabla - A_{\omega})^2 + V_{\omega}$ for a non-relativistic particle on the continuum subjected to a random electric and magnetic field (\cf \cite[Assumption~4.1]{Bouclet_Germinet_Klein_Schenker:linear_response_theory_magnetic_Schroedinger_operators_disorder:2005}). What singles these works out is that they bridge the worlds of functional analysis and von Neumann algebras in a very elegant fashion via the use of non-commutative $L^p$-spaces \cite{Nelson:noncommutative_integration:1974,Terp:noncommutative_Lp_spaces:1981}, so as to be able to combine techniques from both worlds to their advantage. While these authors recognize that they in fact propose a new scheme for making LRT rigorous (\cf \cite[Section~3.5]{Bouclet_Germinet_Klein_Schenker:linear_response_theory_magnetic_Schroedinger_operators_disorder:2005} and \cite{Dombrowski_Germinet:linear_response_theory_non_commutative_integration:2008}), they implement it only for a single Hamiltonian. Even though their specific model contains all the key technical obstacles, (1)~the trace is only \emph{semi}-finite rather than finite, (2)~$H_{\omega}$ is \emph{unbounded} and (3)~$H_{\omega}$ is \emph{not} $\mathcal{T}$-measurable in the sense of Definition~\ref{framework:defn:measurable_operators}, these obstacles are overcome using specifics of the model instead of tackling them in the abstract. To give but one example, Bouclet et al forwent having to deal with issues of $\mathcal{T}$-measurability by making a common core assumption and “hands-on” functional analytic arguments (\cf \cite[Section~3.1]{Bouclet_Germinet_Klein_Schenker:linear_response_theory_magnetic_Schroedinger_operators_disorder:2005}); this also applies to \cite{Dombrowski_Germinet:linear_response_theory_non_commutative_integration:2008}, which makes the connection of \cite{Bouclet_Germinet_Klein_Schenker:linear_response_theory_magnetic_Schroedinger_operators_disorder:2005} to non-commuative $L^p$-spaces more explicit, but it still relies on Bouclet et al's work in their proofs. Our contribution with this book is to extract and abstract the main strategy of \cite{Bouclet_Germinet_Klein_Schenker:linear_response_theory_magnetic_Schroedinger_operators_disorder:2005,Dombrowski_Germinet:linear_response_theory_non_commutative_integration:2008}, consistently frame it in the language of non-commutative $L^p$-spaces, and finish the proofs solely on the basis of model-independent, structural arguments. 
\medskip

\noindent
The main technical challenges are to ascribe meaning to \emph{products} and (generalized) \emph{commutators} between unbounded, $\mathcal{T}$-non-measurable operators such as $H$, which are only affiliated to a von Neumann algebra $\Alg$, and elements from $\mathfrak{L}^p(\Alg)$ (\cf Section~\ref{framework:commutators} dedicated to this subject). One way to deal with that is to define the derivation $\partial_X$ associated to some selfadjoint operator $X$ as the generator of an $\R$-flow $\eta^X_t(A) := \e^{+ \ii t X} \, A \, \e^{- \ii t X}$ on $\mathfrak{L}^p(\Alg)$ rather than the potentially ill-defined commutator $\ii \, [X,A]$ (\cf \cite{Pagter_Sukochev:commutator_estimates_R_flows:2007}). A lot of these difficulties disappear if we impose stronger assumptions such as $H \in \Alg$ (as is the case for many tight-binding models) or that $\mathcal{T}$ is a \emph{finite} trace (so that $\Alg \subseteq \mathfrak{L}^p(\Alg)$). Unfortunately, continuum models satisfy neither, and extending the range of validity to include many physically interesting models means we need to grapple with these issues. 

Apart from non-commutative integration, the second axis along which we want to compare our framework with existing results is the form that perturbations and evolution laws take. The seminal works by Bellissard, van Elst and Schulz-Baldes \cite{Bellissard_van_Elst_Schulz_Baldes:noncommutative_geometry_quantum_Hall_effect:1994} as well as  Bellissard and Schulz-Baldes  \cite{Schulz-Baldes_Bellissard:kinetic_theory_quantum_transport:1998,Bellissard_Schulz_Baldes:quantum_transport_aperiodic_media:1998} study the Quantum Hall Effect via LRT, which pioneered the use of ($C^*$-)algebraic techniques to include effects of disorder. More specifically, they start by considering a microscopic Hamiltonian 
\begin{align}
	H_{\Phi,\mathrm{coll}}(t) = H + \pmb{\Phi} \cdot \mathbf{X} + W_{\mathrm{coll}}(t) 
	\label{intro:eqn:additive_perturbation}
\end{align}
where $H$ is a covariant random tight-binding (unperturbed) operator that is an element of a $C^*$-algebra, the term $\pmb{\Phi} \cdot \mathbf{X}$ is the potential due to a constant electric field $\pmb{\Phi}$, and $W_{\mathrm{coll}}(t)$ is a random collision term, governed by Poisson statistics, that introduces a phenomenological thermalization mechanism into the model. Since measurements in experiments happen on a time scale that is much longer than the average collision time, the relevant macroscopic states are obtained by taking the Poisson average of the microscopic ones. Due to the random collisions these \emph{averaged} states obey \emph{Lindblad} dynamics \cite[Section~7.1]{Prodan_Schulz-Baldes:topological_insulators:2016}, 
\begin{align}
	\frac{\dd \rho}{\dd t}(t) &= - \ii \bigl [ H_{\Phi} \, , \, \rho(t) \bigr ] - \Gamma \bigl ( \rho(t) \bigr )
	, 
	\label{intro:eqn:lindblad_dynamics}
\end{align}
where $H_{\Phi}= H + \pmb{\Phi} \cdot \mathbf{X}$ and $\Gamma$, the so-called collision operator, contains the macroscopic aspects of the diffusion process \cite[Proposition~4]{Schulz-Baldes_Bellissard:kinetic_theory_quantum_transport:1998}. It is precisely the presence of a non-zero $\Gamma$ in \eqref{intro:eqn:lindblad_dynamics} which allows a non-zero current average \cite[Proposition~4]{Bellissard_van_Elst_Schulz_Baldes:noncommutative_geometry_quantum_Hall_effect:1994}. This ($C^*$-)algebraic approach to LRT and the Quantum Hall Effect explains the topological origin of the quantization of the transverse conductivity $\sigma_{\perp}$, and crucially applies also to systems with spatial disorder which can be encoded in the definition of $H$. From a physical perspective disorder is a necessary ingredient, because it leads to Anderson localization which in turn is responsible for the presence of the plauteaux in between the jumps of $\sigma_{\perp}$. 

In contrast to \cite{Bellissard_van_Elst_Schulz_Baldes:noncommutative_geometry_quantum_Hall_effect:1994} we assume perturbations to be \emph{multiplicative} on as opposed to additive, 
\begin{align}
	H_{\Phi,\eps}(t) = G_{\Phi,\eps}(t) \, H \, G_{\Phi,\eps}(t)^*
	, 
	\label{intro:eqn:multiplicative_perturbation}
\end{align}
that are adiabatically switched on over time. They are facilitated by a unitary $G_{\Phi,\eps}(t)$ that parametrically depends on $\Phi$, $\eps$ and $t$, and is compatible with the algebraic structure of $(\Alg,\mathcal{T})$. The states evolve according to the Liouville equation \cite[Theorem~5.3]{Bouclet_Germinet_Klein_Schenker:linear_response_theory_magnetic_Schroedinger_operators_disorder:2005}
\begin{align}
	\frac{\dd \rho}{\dd t}(t) &= - \ii \bigl [ H_{\Phi,\eps}(t) \, , \, \rho(t) \bigr ]
	,
	\label{intro:eqn:liouville_dynamics}
\end{align}
with time-dependent generator. Compared to the Lindblad dynamics~\eqref{intro:eqn:lindblad_dynamics} the collision term $\Gamma(t)$ is absent. Perturbations of the type \eqref{intro:eqn:multiplicative_perturbation} have already been studied in the context of LRT (see \eg \cite{Elgart_Schlein:Kubo_for_Landau:2004,Bouclet_Germinet_Klein_Schenker:linear_response_theory_magnetic_Schroedinger_operators_disorder:2005,Dombrowski_Germinet:linear_response_theory_non_commutative_integration:2008}). While this seems very restrictive at first, among other things it covers models for the ac and dc Stark effects \cite{Nenciu_Nenciu:Wannier_Stark_ladder_1:1981,Nenciu_Nenciu:Wannier_Stark_ladder_2:1982,Graffi_Yajima:1983}. In fact, $G_{\Phi,\eps}(t)$ can be interpreted as a time-dependent change of representation, and in this \emph{interaction representation} $H_{\Phi,\eps}(t)$ can be connected to a Hamiltonian of the form $\tilde{H}_{\Phi,\eps}(t) = H + \dot{F}_{\Phi,\eps}(t)$ (see Section~\ref{dynamics:perturbed:additive_vs_multiplictaive} for more details). Note that due to the time-dependence of $G_{\Phi,\eps}(t)$, the operators $H_{\Phi,\eps}(t)$ and $\tilde{H}_{\Phi,\eps}(t)$ are in general \emph{not} isospectral (such is the case in models for the dc Stark effect where $\sigma(H)$ is bounded from below and $\sigma \bigl ( H + \pmb{\Phi} \cdot \mathbf{X} \bigr ) = \R$), nor is it \emph{a priori} clear whether $\tilde{H}_{\Phi,\eps}(t)$ is mathematically well-defined of if it can generate a genuine unitary evolution which solves the related Liouville equation. Therefore, we may regard \eqref{intro:eqn:multiplicative_perturbation} as a mathematically sensible way to rigorously define the evolution equations describing certain physical systems. 

What is more, unlike in condensed matter physics, there \emph{are} other physical systems where transport is non-diffusive ($W_{\mathrm{coll}} = 0$) and perturbations are naturally multiplicative rather than additive: the equations governing the propagation of electromagnetic \cite{DeNittis_Lein:ray_optics_photonic_crystals:2014} and other waves (\eg \cite{Jin_Lu_et_al:topological_magnetoplasmon:2016,Suesstrunk_Huber:classification_mechanical_metamaterials:2016}) in artificially structured media can be recast in the form of a Schrödinger equation, $\ii \partial_t \psi = H \psi$. Here, the analog of the quantum Hamiltonian $H = W \, D$ is the product of a bounded weight operator $W$ and a potentially unbounded operator $D$. Perturbations of the background where the waves travel enter by modifying $W$, although equivalently, they can be recast in the form~\eqref{intro:eqn:multiplicative_perturbation} \cite[Section~2.2]{DeNittis_Lein:adiabatic_periodic_Maxwell_PsiDO:2013}. Unlike in condensed matter physics, transport in many situations in non-diffusive, and therefore the ideas of \cite{Bellissard_van_Elst_Schulz_Baldes:noncommutative_geometry_quantum_Hall_effect:1994,Schulz-Baldes_Bellissard:kinetic_theory_quantum_transport:1998,Bellissard_Schulz_Baldes:quantum_transport_aperiodic_media:1998} do not apply. 

Therefore, it is possible to adapt our LRT framework to systems that are not necessarily quantum and which have not previously been considered in the literature, rigorously or not. We envision other new applications: more algebraically minded researchers may be able to establish \eg Onsager relations \cite{Jaksic_Ogata_Pillet:Green_Kubo_formula_Onsager_relations:2006} or to define current-current correlation measures \cite{Combes_Germinet_Hislop:current_current_correlation_measure:2010,Prodan_Bellissard:current_current_correlation_function:2016} in this broader setting. On the other hand, verifying the Hypotheses from Chapter~\ref{main_results} for specific models requires functional analytic tools.

\paragraph{Structure of this book} 
\label{par:structure_of_this_book}
Let us close this chapter by giving an outline of this book. All of the moving parts of our LRT framework, all assumptions, labelled as Hypothesis~\ref{hypothesis:trace}–\ref{hypothesis:state}, and all of our main results are collected and explained in Chapter~\ref{main_results}. To make this book as self-contained as possible, we have written up all the mathematical background on non-commutative $\mathcal{T}$-measurable operators, $L^p$-spaces, different commutators and $C_0$ groups on these various Banach spaces. We show in Chapter~\ref{unified} that the most common physical systems fit our framework: we construct the von Neumann algebra of covariant operators in the abstract, starting from an ergodic topological dynamical system, construct its associated trace-per-volume, and discuss generators that are compatible with it. The main content is found in Chapters~\ref{dynamics} and \ref{Kubo_formula}: first, we prove the existence of the unperturbed and perturbed dynamics, and investigate important facts such as the dependence of the dynamics on the perturbation parameter. Then the main results, the expansion of the dynamics of the states, the Kubo and the Kubo-Str\v{e}da formula, will be proven in Chapter~\ref{Kubo_formula}. We close this book by showing how our framework applies to existing results and giving an outlook to new applications in Chapter~\ref{applications}. 
%
%
%
\chapter{Setting, Hypotheses and Main Results} 
\label{main_results}
The overarching goal of this book is to propose an analytic-algebraic framework to make Linear Response Theory (LRT) rigorous for a wide class of systems. The purpose of this chapter is to explain the setting, enumerate the mathematical hypotheses and state the main results.

\section{Description of the abstract setting} 
\label{main_results:setting}
The first step for the construction of a general, and quite abstract setting for LRT is the introduction of the main ingredients.
\medskip

\noindent
\textbf{Basic elements of the theory:}
%
\begin{enumerate}[leftmargin=*,label=(H\arabic*)]
	\item A \emph{von Neumann algebra} $\Alg\subseteq \mathscr{B}(\Hil)$ of bounded operators on the {(not necessary separable)} Hilbert space $\Hil$ which contains the relevant information about the system of interest. This von Neumann algebra is endowed with a \emph{trace} $\mathcal{T}$ which allows us to compute expectation values of observables related to the system of interest.
	\item A (possibly unbounded) selfadjoint Hamiltonian $H$ \emph{affiliated} to $\Alg$ which prescribes the \emph{unperturbed} dynamics of the system.
	\item A set $\{ X_1 , \ldots , X_d \}$ of (possibly unbounded) selfadjoint operators, a vector $\pmb{\Phi} \in \R^d$ of length $\Phi := \sabs{\pmb{\Phi}}$ and a positive parameter $\eps > 0$ which enter in the definition of a unitary-valued map $\R \ni t \mapsto G_{\Phi,\eps}(t)$. 
	\item The latter has the role to define a time-dependent \emph{adiabatic isospectral perturbation} by conjugating the Hamiltonian $H$ with $G_{\Phi,\eps}(t)$, 
	\begin{align}
		H_{\Phi,\eps}(t) := G_{\Phi,\eps}(t) \, H \, G_{\Phi,\eps}(t)^*
		.
		\label{main_results:eqn:definition_isospectrally_perturbed_hamiltonian}
	\end{align}
	\item An \emph{instantaneous observable} described by a time-dependent (possibly unbounded) operator $\R \ni t \mapsto J(t)$, which models some relevant physical property of the system at time $t$.
	\item A positive operator $\rho \in \Alg^+$ called (initial) \emph{equilibrium state} which encodes the status of the system at $t = -\infty$.
\end{enumerate}
The “compatibility” and the “interplay” between the elements listed above is guaranteed by a set of six hypotheses, one for each of the items above. 
\begin{remark}[Hypotheses of this work]
	Throughout this work \emph{whenever we write “under the Hypotheses” we mean to impose Hypotheses~\ref{hypothesis:trace}–\ref{hypothesis:state} below.} Each of them enumerates the technical assumptions associated the to corresponding item in the list of basic elements, \eg Hypothesis~\ref{hypothesis:trace} stipulates the precise setting of (H1). 
	
	We reckon that some results may be proven under weaker or plainly different hypotheses, but since the overarching goal of this work is to provide a widely applicable and robust \emph{framework} for LRT, we do not strive for utmost generality and deliberately aim to avoid unnecessary technical complications. 
\end{remark}
The first of these describes the relation between the von Neumann algebra and the trace in (H1): 
\begin{hypothesis}[Von Neumann algebra and trace]\label{hypothesis:trace}
	The von Neumann algebra $\Alg$ is \emph{semi-finite} and the trace $\mathcal{T}$ is \emph{faithful}, \emph{normal} and \emph{semi-finite} (f.n.s.).
\end{hypothesis}
The basic facts about von Neumann algebras (including the notion of affiliation) and f.n.s.\ traces will be recalled in Section~\ref{framework:algebra_observables} while some concrete examples will be anticipated in Chapter~\ref{unified} and Chapter~\ref{applications}. Hypothesis~\ref{hypothesis:trace} ensures the possibility of developing a “coherent” (non-commuta\-tive) integration theory over $\Alg$ (see Section~\ref{framework:nc_Lp_Sobolev_spaces}). Such a trace allows for the construction of the Banach spaces $\mathfrak{L}^p(\Alg)$ for all $1 \leqslant p < \infty$ as the “$p$-Schatten classes” associated to $\Alg$, and these Banach spaces are the analogue of the classical $L^p$-spaces. In particular, one has that $\mathfrak{L}^1(\Alg) \cap \Alg = \Alg_{\mathcal{T}}$ is the maximal domain (indeed an ideal) in $\Alg$ where the trace $\mathcal{T}$ is well-defined. The set of elements which are Hilbert-Schmidt with respect to $\mathcal{T}$ naturally form a Hilbert space $\mathfrak{L}^2(\Alg)$ with scalar product $\sdscpro{A}{B}_{\mathfrak{L}^2} := \mathcal{T}(A^* \, B)$, $A , B \in \mathfrak{L}^2(\Alg)$. We point out that the spaces $\mathfrak{L}^p(\Alg)$ \emph{may contain also unbounded operators}, and here we see why the framework of von Neumann algebras is necessary if we want to go beyond bounded tight-binding operators (\cf Section~\ref{main_results:tight_binding_simplification} and Chapter~\ref{applications}). Indeed, the von Neumann algebra itself can be identified with $\mathfrak{L}^{\infty}(\Alg) = \Alg$, and we will often use this identification to simplify notation.

Because we admit \emph{unbounded} Hamiltonians, $H$ need not be an element of $\Alg$. Instead, it suffices if we make the following 
\begin{hypothesis}[Unperturbed Hamiltonian]\label{hypothesis:Hamiltonian}
	$H \in \affil(\Alg)$ is a selfadjoint operator on $\Hil$ that is \emph{affiliated with $\Alg$} (\cf Definition~\ref{framework:defn:affiliation}), where $\affil(\Alg)$ is the set of closed and densely defined operators affiliated to $\Alg$. 
\end{hypothesis}
Unfortunately, in all physically relevant situations in which $H$ is unbounded the Hamiltonian $H$ does not belong to any of the spaces $\mathfrak{L}^p(\Alg)$ (see {Remark} \ref{framework:example:non_T_measurable_operators}), even though the $\mathfrak{L}^p(\Alg)$ have unbounded elements. Nevertheless, thanks to the affiliation of $H$ to $\Alg$ the time-evolution 
\begin{align}
	\alpha^0_t(A) := \e^{- \ii t H} A \, \e^{+ \ii t H} 
	, 
	&&
	t \in \R
	, \; 
	A \in \Alg
	, 
	\label{main_results:eqn:unperturbed_dynanmics}
\end{align}
generated by $H$ often extends naturally to a one-parameter group of isometries ($\R$-flow) $\R \ni t \mapsto \alpha^0_t \in \mathrm{Iso} \bigl ( \mathfrak{L}^p(\Alg) \bigr )$ on each of the $\mathfrak{L}^p(\Alg)$ spaces. As a consequence of Proposition \ref{framework:prop:extension_isometry_Lp} the flow turns out to be \emph{strongly continuous} on $\mathfrak{L}^p(\Alg)$ in the sense that 
\begin{align*}
	\lim_{t \to t_0} \, \Bnorm{\alpha^0_t(A) - \alpha^0_{t_0}(A)}_p = 0
	&& 
	\forall A \in \mathfrak{L}^p(\Alg)
	. 
\end{align*}
Standard arguments of the theory of $C_0$-groups on Banach spaces ensure that the dynamics $\alpha^0_t$ on $\mathfrak{L}^p(\Alg)$ is induced by a densely defined infinitesimal generator $\mathscr{L}_H^{(p)}$ called \emph{$p$-Liouvillian} (see Section~\ref{dynamics:unperturbed:generator}).\\
\begin{consequence}[Unperturbed dynamics]\label{main_results:conseq:unperturbed_evolution}
	Suppose Hypotheses~\ref{hypothesis:trace} and \ref{hypothesis:Hamiltonian} hold true. Then the Hamiltonian $H$ induces strongly continuous one-parameter group of isometries (an $\R$-flow)
	\begin{align*}
		\R \ni t \mapsto \alpha^0_t \in \mathrm{Iso} \bigl ( \mathfrak{L}^p(\Alg) \bigr )
	\end{align*}
	for each, $1 \leqslant p < \infty$, which is called \emph{unperturbed dynamics}. Its generator $\mathscr{L}_H^{(p)}$, the \emph{$p$-Liouvillian}, has a core $\rr{D}^{00}_{H,p}$ (see equation~\eqref{framework:eqn:domain_maximal_generalized_commutator}) where it acts as a generalized commutator,
	\begin{align*}
		\mathscr{L}_H^{(p)}(A) = - \ii [H,A]_{\ddagger}
		= - \ii \bigl ( H \, A - \bigl ( H \, A^* \bigr )^* \bigr )
		,
		&&
		A \in \rr{D}^{00}_{H,p}
		.
	\end{align*}
\end{consequence}
We will study the unperturbed dynamics on $\mathfrak{L}^p(\Alg)$ and its generator $\mathscr{L}_H^{(p)}$ in Section~\ref{dynamics:unperturbed}. The notion of \emph{generalized commutator}
\begin{align}
	[A , B]_{\ddagger} := A \, B - \bigl ( A^* \, B^* \bigr )^*
	\label{main_results:eqn:generalized_commutator}
\end{align}
and of the domain $\rr{D}^{00}_{H,p} $ will be described in Section~\ref{framework:commutators:hamiltonian_t_measurable_operators}. 

The set of operators $\{ X_1 , \ldots , X_d \}$ described in (H3) is required to be compatible with the notion of integration induced by the trace $\mathcal{T}$ on the algebra $\Alg$. This fact is covered by the following 
\begin{hypothesis}[$\mathcal{T}$-compatible generators]\label{hypothesis:generators}
	The selfadjoint operators $\{ X_1 , \ldots , X_d \}$ are \emph{$\mathcal{T}$-compatible generators} in the sense that for all $k = 1 , \ldots , d$ and for all $s \in \R$ they satisfy 
	\begin{enumerate}[leftmargin=*,label=(\roman*)]
		\item $\e^{+ \ii s X_k}\, A \, \e^{- \ii s X_k} \in \Alg$ for all $A \in \Alg$, 
		\item $\mathcal{T} \bigl ( \e^{+ \ii s X_k} \, A \, \e^{- \ii s X_k} \bigr ) = \mathcal{T}(A)$ for all $A$ in the trace-ideal $\Alg_{\mathcal{T}} \subset \Alg$, and 
		\item the $\{ X_1 , \ldots , X_d \}$ are \emph{strongly commuting}, \ie 
		\begin{align*}
			\e^{+ \ii s X_j} \, \e^{+ \ii s X_k} = \e^{+ \ii s X_k} \, \e^{+ \ii s X_j}
			,
			&&
			\forall j , k = 1 , \ldots , d
			. 
		\end{align*}
	\end{enumerate}
	We refer to the integer $d$ as the \emph{dimension}, and to a common invariant core $\domain_{\mathrm{c}}$ as \emph{localizing domain}.
\end{hypothesis}
This assumption allows to introduce ($\mathcal{T}$-compatible) \emph{spatial derivations} on $\mathfrak{L}^p(\Alg)$ via 
\begin{align*}
	\partial_{X_k}(A) := \lim_{s \to 0} \frac{\e^{+ \ii s X_k} \, A \, \e^{- \ii s X_k} - A}{s}
	.
\end{align*}
Formally at least, the $\partial_{X_k}(A)$ can be seen as commutators $\ii [X_k , A]$. Evidently, these are densely defined on each $\mathfrak{L}^p(\Alg)$ (see Section~\ref{framework:Sobolev_spaces}), and the associated gradient $\nabla := \bigl ( \partial_{X_1} , \ldots , \partial_{X_d} \bigr )$ gives rise to (non-commutative) \emph{Sobolev spaces}
\begin{align*}
	\mathfrak{W}^{1,p}(\Alg) := \Bigl \{ A \in \mathfrak{L}^p(\Alg) \; \; \big \vert \; \; \nabla(A) \in \mathfrak{L}^p(\Alg) \times \ldots \times \mathfrak{L}^p(\Alg) \Bigr \} 
	.
\end{align*}
The fact that the generators strongly commute (Hypothesis~\ref{hypothesis:generators}~(iii)) has several implications: First of all, it ensures the commutativity between the derivations $\partial_{X_j}$. Secondly, it guarantees the existence of a localizing domain $\domain_{\mathrm{c}} \subset \Hil$ \cite[Corollary~5.28]{schmudgen-12}. Consequently, linear combinations $\lambda_1 \, X_1 + \ldots + \lambda_d \, X_d$ with $\lambda_1 , \ldots , \lambda_d \in \R$ are essentially selfadjoint on $\domain_{\mathrm{c}}$ and therefore extend to uniquely defined selfadjoint operators. 

Hypothesis~\ref{hypothesis:generators} also plays a key role in the definition of the family of unitary operators $G_{\Phi,\eps}(t)$ which implement the adiabatic isospectral perturbation of the Hamiltonian via equation~\eqref{main_results:eqn:definition_isospectrally_perturbed_hamiltonian}.
\begin{hypothesis}[Adiabatic isospectral perturbations]\label{hypothesis:perturbation}
	Let $\pmb{X} := (X_1 , \ldots , X_d)$ be the operator-valued vector made up of $\mathcal{T}$-compatible generators with localizing domain $\domain_{\mathrm{c}}$ (\cf Hypothesis~\ref{hypothesis:generators}). We define the \emph{switch function}
	\begin{align*}
		\switch(t) := 
		\begin{cases}
			\e^{\eps t} & t \leqslant 0 \\
			1 & t > 0 \\
		\end{cases}
	\end{align*}
	with \emph{adiabatic rate} $\eps > 0$. Consider a system of real valued and continuous \emph{modulation functions} $f_k \in C(\R)$, $k = 1 , \ldots , d$, which fulfill the following integrability condition
	\begin{align}
		\int_{-\infty}^t \dd \tau \; \switch(\tau) \, \babs{f_k(\tau)} = \int_{-\infty}^0 \dd \tau \; \e^{\eps t} \, \babs{f_k(\tau)} + \int_{0}^t \dd \tau \; \babs{f_k(\tau)} < +\infty
		\label{main_results:eqn:integrability_condition}
	\end{align}
	for all $t \in \R$ and $\eps>0$. Given a \emph{field} $\pmb{\Phi} := \bigl ( \Phi_1 , \ldots , \Phi_d \bigr ) \in \R^d$ define the vector valued functions $\pmb{f}^{\Phi}(t) := \bigl (\Phi_1 \, f_1(t) , \ldots , \Phi_d \, f_d(t) \bigr )$ and $\pmb{\Phi}^{\eps}(t) := \bigl ( \Phi_1^{\eps}(t) , \ldots , \Phi_d^{\eps}(t) \bigr )$ where
	\begin{align}
		\Phi_k^{\eps}(t) := \int_{-\infty}^t \dd \tau \; \switch(\tau) \, f_k^\Phi(\tau) = \int_{-\infty}^t \dd \tau \; \switch(\tau) \, \Phi_k \, f_k(\tau)
		.
		\label{main_results:eqn:definition_Phi_t}
	\end{align}
	The \eqref{main_results:eqn:definition_Phi_t} implies $\Phi_k^{\eps} \in C^1(\R)$ and $\lim_{t \to -\infty} \Phi_k^{\eps}(t) = 0$. The modulus $\Phi := \sabs{\pmb{\Phi}}$ will be called the \emph{field strength}. For each $t \in \R$ the operators
	\begin{subequations}
		\begin{align}
			F_{\Phi,\eps}(t) :& \negmedspace= \pmb{\Phi}^{\eps}(t) \cdot \pmb{X} 
			 = \sum_{k = 1}^d \Phi_k^{\eps}(t) \, X_k\,
			,
			\\
			\dot{F}_{\Phi,\eps}(t) :& \negmedspace= \switch(t) \, \pmb{f}^{\Phi}(t) \cdot \pmb{X} = \switch(t) \,\sum_{k = 1}^d \Phi_k \, f_k(t) \, X_k
			, 
			\label{main_results:eqn:F_dot}
		\end{align}
	\end{subequations}
	are essentially selfadjoint on the core $\domain_{\mathrm{c}}$, and so extend to uniquely defined selfadjoint operators. The \emph{adiabatic isospectral perturbations} associated to the generators $X_k$ the field components $\Phi_k$ and the modulations $f_k$ is given by exponentiating $F_{\Phi,\eps}(t)$, 
	\begin{align}
		G_{\Phi,\eps}(t) := \e^{+ \ii F_{\Phi,\eps}(t)} 
		 = \prod_{k = 1}^d \e^{+ \ii \Phi_k^{\eps}(t) \, X_k}
		,
		&&
		t \in \R
		. 
		\label{main_results:eqn:definition_G_Phi_eps}
	\end{align}
	We refer to $\eps \to 0$ as the \emph{adiabatic limit}. 
\end{hypothesis}
Notice that the second equality in \eqref{main_results:eqn:definition_G_Phi_eps} is just a consequence of the strong commutativity of the generators and the Trotter product formula \cite[Theorem~VII.31]{Reed_Simon:M_cap_Phi_1:1972}. Hypothesis~\ref{hypothesis:perturbation} is the “abstract” version of the perturbations used in \cite{Bouclet_Germinet_Klein_Schenker:linear_response_theory_magnetic_Schroedinger_operators_disorder:2005} and \cite{klein-lenoble-muller-07}. In case all modulations $f_k$ are supported in $[t_0 , +\infty)$ with $t_0 > -\infty$ one can fix $t_0$ as the \emph{finite} initial time. Then one has $\Phi_k^{\eps}(t) = 0$, and consequently $G_{\Phi,\eps}(t) = \id$, for all $t \leqslant t_0$. Two situations will be particularly interesting for the aims of this work: The first one concerns the adiabatic switching of a \emph{constant field} described by the conditions $f_1(t) = \ldots = f_d(t) = 1$ for all $t \in \R$. The second assumes that the switch function 
\begin{align}
	f_k(t) = \int_{\R} \dd \omega \; \e^{+ \ii \omega t} \, \hat{f}_k(\omega)
	\label{dynamics:eqn:freq_decomposition_modulation_function} 
\end{align}
is smooth and that the perturbation is switched off at $t = +\infty$ again, \ie it is the Fourier transform of $\hat{f}_k \in C_{\mathrm{c}}(\R)$ which are assumed to 
compactly supported and satisfy $\overline{\hat{f}_k(\omega)} = \hat{f}_k(-\omega)$ (so that the $f_k$ are \emph{real}-valued). Modulations with frequency expansions of the type \eqref{dynamics:eqn:freq_decomposition_modulation_function} has been studied in \cite{klein-lenoble-muller-07} in the contest of the derivation of the Mott formula for the ac-conductivity.

In view of Hypothesis~\ref{hypothesis:perturbation} the perturbed Hamiltonian $H_{\Phi,\eps}(t)$ defined by \eqref{main_results:eqn:definition_isospectrally_perturbed_hamiltonian} accomplishes the two important properties (see Section~\ref{dynamics:perturbed:gauge_perturbations}), (1) the \emph{isospectrality} $\spec \bigl ( H_{\Phi,\eps}(t) \bigr ) = \spec(H)$ for all $t \in \R$, and (2) the \emph{affiliation} $H_{\Phi,\eps}(t) \in \affil(\Alg)$ for all $t \in \R$. An isospectral perturbation $G_{\Phi,\eps}(t)$ also induces an automorphism of the von Neumann algebra $\Alg$ given by
\begin{align}
	\gamma^{\Phi,\eps}_t(A) := G_{\Phi,\eps}(t) \, A \, G_{\Phi,\eps}(t)^*
	, 
	&&
	A \in \Alg
	. 
	\label{main_results:eqn:interaction_dynamics}
\end{align}
Actions of this type extend to strongly continuous isometries on each of the Banach spaces $\mathfrak{L}^p(\Alg)$ (again Proposition~\ref{framework:prop:extension_isometry_Lp}). 
\begin{consequence}[Existence of the interaction dynamics]\label{main_results:conseq:interaction_evolution}
	Suppose Hypotheses~\ref{hypothesis:trace}–\ref{hypothesis:generators} hold, and let $t\mapsto G_{\Phi,\eps}(t)$ be an adiabatic isospectral perturbations in the sense of Hypothesis~\ref{hypothesis:perturbation}. Then the prescription \eqref{main_results:eqn:interaction_dynamics} induces a strongly continuous map of isometries 
	\begin{align*}
		\R \ni t \mapsto \gamma^{\Phi,\eps}_t \in \mathrm{Iso} \bigl ( \mathfrak{L}^p(\Alg) \bigr )
	\end{align*}
	for each $1 \leqslant p < \infty$ which is called \emph{interaction dynamics}.
\end{consequence}
%
The properties of $\gamma^{\Phi,\eps}_t$ are studied in Section~\ref{dynamics:perturbed:interaction_picture}. In particular, one has that 
\begin{align*}
	\norm{\cdot}_p-\lim_{t \to -\infty} \gamma^{\Phi,\eps}_t(A) &= A 
	= \norm{\cdot}_p-\lim_{\Phi \to 0} \gamma^{\Phi,\eps}_t(A) 
	&&
	\forall A \in \mathfrak{L}^p(\Alg)
	, 
\end{align*}
which means that $\gamma^{\Phi,\eps}_t$ converges strongly to the identity map $\id_{\mathfrak{L}^p}$ at the initial time $t = -\infty$ and in the limit of vanishing perturbation $\Phi \to 0$. Moreover, in the case of regular elements $A \in \mathfrak{W}^{1,p}(\Alg)$ one can differentiate (strongly) $\gamma^{\Phi,\eps}_t$ in $t$ obtaining 
\begin{align*}
	\frac{\dd}{\dd t}\gamma^{\Phi,\eps}_t(A) = \gamma^{\Phi,\eps}_t \Bigl ( \switch(t) \; \pmb{f}^{\Phi}(t) \cdot \nabla(A) \Bigr ) 
	= \gamma^{\Phi,\eps}_t \left ( \switch(t) \sum_{k = 1}^d \Phi_k \, f_k(t) \, \partial_{X_k}(A) \right ) 
	.
\end{align*}
%

\section{The perturbed dynamics: bridge to analysis} 
\label{main_results:bridge_to_analysis}
Ultimately, the core of LRT is a comparison between the interaction dynamics generated by the perturbation of $H$ and the {perturbed} dynamics defined by $H_{\Phi,\eps}(t)$ in the limit of small $\eps$. The proof of the existence of a \emph{unitary time propagator} $U_{\Phi,\eps}(t,s)$ which implements the perturbed dynamics generated by $H_{\Phi,\eps}(t)$ cannot be based on purely algebraic considerations, but requires the use of tools borrowed from functional analysis. For this reason we also need a set of technical assumptions about the interplay between the Hamiltonian $H$ and the generators $X_k$. 

To simplify the presentation of the next Hypothesis, let us introduce 
\begin{align*}
	\ad_{X_j}(H) := \ii [X_j,H] 
\end{align*}
for the commutator of two suitable operators $X_j$ and $H$. As usual, defining commutators of two potentially unbounded operators is fraught with technical problems, but that is something we will address below. It is tempting to identify $\ad_{X_j}(H)$ with $\partial_{X_j}(H)$, and while there are situations where the two coincide, we intentionally separate these two notions for reasons that we will elaborate upon in Chapter~\ref{framework:nc_Lp_Sobolev_spaces:Sobolev_spaces}: in general it turns out that $\ad_{X_j}(H)$ needs to be defined as a functional analytic object whereas we consider $\partial_{X_j}$ as an algebraic derivation on the Banach spaces $\mathfrak{L}^p(\Alg)$ — and in many cases of interest we in fact have $H \not\in \mathfrak{L}^p(\Alg)$ (\cf Example~\ref{framework:example:non_T_measurable_operators}). 

The notation for $\kappa_j$-fold commutators simplifies to the compact expression 
\begin{align*}
	\ad_{X_j}^{\kappa_j} := \underbrace{\ad_{X_j} \circ \cdots \circ \ad_{X_j}}_{\mbox{$\kappa_j$ times}}
\end{align*}
where by convention $\ad_{X_j}^0(H) := H$. Moreover, the fact that $\{ X_1 , \ldots , X_d \}$ commute leads to $\ad_{X_j} \circ \ad_{X_k} = \ad_{X_k} \circ \ad_{X_j}$ commuting amongst each other, and hence, we can use multiindex notation to simplify successive commutators with respect to different $X_j$'s, namely 
\begin{align*}
	\ad_X^{\kappa} := \ad_{X_1}^{\kappa_1} \circ \cdots \circ \ad_{X_d}^{\kappa_d}
\end{align*}
where $\kappa = (\kappa_1 , \ldots , \kappa_d) \in \N_0^d$ and as usual $\abs{\kappa} = \sum_{j = 1}^d \kappa_j$. With this notation in hand, we can stipulate our hypothesis on the current operators. 
\begin{hypothesis}[Current expansion]\label{hypothesis:current}
	The selfadjoint Hamiltonian $H \in \affil(\Alg)$ and the set $\{ X_1 , \ldots , X_d \}$ of $\mathcal{T}$-compatible generators with localizing domain $\domain_{\mathrm{c}}\subset\Hil$ meet the following assumptions:
	\begin{enumerate}[leftmargin=*,label=(\roman*)]
		\item The \emph{joint core} $\domain_{\mathrm{c}}(H) := \domain_{\mathrm{c}} \cap \domain(H)$ is a densely defined core for $H$ and $X_k[\domain_{\mathrm{c}}(H)] \subset \domain_{\mathrm{c}}(H)$ for all $k = 1 , \ldots , d$. 
		\item $H[\domain_{\mathrm{c}}(H)] \subset \domain_{\mathrm{c}}$ and the commutators
		\begin{align}
			J_{\kappa} := (-1)^{\abs{\kappa}} \, \ad_X^{\kappa}(H) 
			, 
			&&
			\kappa \in \N_0^d
			, 
			\label{main_results:eqn:definition_current_tensors}
		\end{align}
		are essentially selfadjoint on $\domain_{\mathrm{c}}(H)$, and therefore uniquely extend to selfadjoint operators (still denoted with the same symbol $J_{\kappa}$). With abuse of notation, we will also use $J_k := - \ad_{X_k}(H)$ in case of first-order current operators. 
		\item Assume that there exists $N \in \N_0$ with 
		\begin{align*}
			J_{\kappa} = 0
			&&
			\mbox{$\forall \kappa \in \N_0^d$ with $\abs{\kappa} > N$} 
		\end{align*}
		on $\domain_{\mathrm{c}}(H)$. The smallest such integer $N$ is called the \emph{order of} $H$ with respect to the family $\{ X_1 , \ldots , X_d \}$. 
		\item All the $J_{\kappa}$ are infinitesimally $H$-bounded in the sense that for any $\delta>0$ there are positive constants $a > 0$ and $\delta> b > 0$ such that
		\begin{align*}
			\bnorm{J_{\kappa} \varphi}_{\Hil} \leqslant a \, \bnorm{\varphi}_{\Hil} + b \, \bnorm{H \varphi}_{\Hil}
			,
			\quad\quad
			\forall \, \varphi \in \domain_{\mathrm{c}}(H)
			,
		\end{align*}
		for all $\kappa \in \N_0^d$. 
		\item The Hamiltonian $H$ has a (possible unbounded) \emph{spectral gap} marked by a real number $\xi \in \res(H)$ in the resolvent set. 
	\end{enumerate}
\end{hypothesis}
Hypothesis~\ref{hypothesis:current}, in its entirety, is quite strong and for this reason has several implications. Nevertheless many physical systems of interest fulfill the conditions listed above (see \eg Chapter~\ref{applications}). Item (i) in Hypothesis~\ref{hypothesis:current} ensures that for all $t \in \R$ the perturbed Hamiltonians $H_{\Phi,\eps}(t)$ are essentially selfadjoint on the common core $\domain_{\mathrm{c}}(H)$ (see Lemma~\ref{dynamics:lem:invariant_domain}~(1)). In fact, it would be appropriate to call $\domain_{\mathrm{c}}(H)$ localiz\emph{ed} (as opposed to localiz\emph{ing}) domain, but to better distinguish these two we will refer to it as joint domain instead. 

Item (ii) is needed to unambiguously define the family of selfadjoint operator-valued tensors: the (unperturbed) \emph{current density tensor} 
\begin{align}
	\mathbf{J}^{(r)} := \bigl \{ J_{\kappa} \bigr \}_{\abs{\kappa} = r}
	,
	&&
	r = 1 , \ldots , N
	,
	\label{main_results:eqn:collection_density_tensors}
\end{align}
is the collection of all the 
\begin{align*}
	J_{\kappa} = \ii^{\abs{\kappa}} \Bigl [ \Bigl [ \bigl [ [H,X_{k_1}] \, , \, X_{k_2} \bigr ] \, , \, \ldots \, \Bigr ] \, , \, X_{k_r} \Bigr ]
\end{align*}
of order $r$ where 
\begin{align*}
	\bigl ( k_1 , \ldots , k_r \bigr ) = \bigl ( \underbrace{1 , \ldots , 1}_{\mbox{$\kappa_1$ times}} \, , \, \ldots \, , \, \underbrace{d , \ldots , d}_{\mbox{$\kappa_d$ times}} \bigr )
\end{align*}
are suitably repeated indices. The condition $H[\domain_{\mathrm{c}}(H)] \subset \domain_{\mathrm{c}}$ is essential to define these operators as generalized commutators in the sense of Definition~\ref{framework:defn:generalized_commutators}. The conditions expressed in (i) and (ii) are also sufficient for the application of the \emph{Baker-Campbell-Hausdorff formula} (in the generalized setting of \cite[Section II.11.B]{schroeck-96}) to the perturbed Hamiltonian $H_{\Phi,\eps}(t)$. A straightforward calculation produces the relevant formula
\begin{align}
	H_{\Phi,\eps}(t) = H + W_{\Phi,\eps}(t)
	\label{main_results:eqn:additive_form_perturbed_hamiltonian}
\end{align}
which relates the isospectral perturbation of $H$ to the \emph{additive} perturbation $W_{\Phi,\eps}(t)$. The latter can be expanded in function of the density currents~\eqref{main_results:eqn:definition_current_tensors} according to the formula
\begin{align}
	W_{\Phi,\eps}(t) := \sum_{r = 1}^N \frac{(-1)^r}{r!} \sum_{\abs{\kappa} = r}^d w_{\kappa}^{\eps}(t) \, J_{\kappa} 
	\label{main_results:eqn:definition_additive_perturbation_W_w_kappa}
\end{align}
where the time dependent coefficients $w_{\kappa}^{\eps}$ are given in terms of  \eqref{main_results:eqn:definition_Phi_t} by
\begin{align}
	w_{\kappa}^{\eps}(t) := \prod_{j = 1}^d {\Phi_j^{\eps}(t)}^{\kappa_j}
	\label{main_results:eqn:perturbation_potential}
\end{align}
Formula \eqref{main_results:eqn:additive_form_perturbed_hamiltonian}, which provides the \emph{current expansion} of the isospectral perturbation, is proved in Lemma~\ref{dynamics:lem:invariant_domain}~(2). At this stage it is appropriate to point out that, without further assumptions, the equality \eqref{main_results:eqn:additive_form_perturbed_hamiltonian} makes sense only on the dense set $\domain_{\mathrm{c}}(H)$. In order to extend this equality in a suitable and useful way one would ask that also the right-hand side of \eqref{main_results:eqn:additive_form_perturbed_hamiltonian} is essentially selfadjoint on $\domain_{\mathrm{c}}(H)$. This is exactly the role of item (iii) in Hypothesis~\ref{hypothesis:current} which allows for the application of the \emph{Kato-Rellich Theorem} (see Lemma~\ref{dynamics:lem:invariant_domain}~(3) for the details). As a consequence one obtains that both sides of the \eqref{main_results:eqn:additive_form_perturbed_hamiltonian} are essentially selfadjoint operators on $\domain_{\mathrm{c}}(H)$ and so they define the same selfadjoint operator after the closure. Moreover, the domain of \eqref{main_results:eqn:additive_form_perturbed_hamiltonian} (and therefore that of \eqref{main_results:eqn:definition_isospectrally_perturbed_hamiltonian}) turns out to be independent of time, 
\begin{align}
	\domain \bigl ( H_{\Phi,\eps}(t) \bigr ) = \domain(H)
	,
	&& 
	\forall t \in \R
	. 
	\label{main_results:eqn:invariance_domain}
\end{align}
Finally, item (iv) and the definition of $H_{\Phi,\eps}(t)$ given by \eqref{main_results:eqn:definition_isospectrally_perturbed_hamiltonian} ensure that also 
\begin{align}
	\xi \in \res \bigl ( H_{\Phi,\eps}(t) \bigr )
	,
	&&
	\forall t \in \R
	,
	\label{main_results:eqn:xi_in_resolvent_set}
\end{align}
lies in a spectral gap of $H_{\Phi,\eps}(t)$. Properties \eqref{main_results:eqn:invariance_domain} and \eqref{main_results:eqn:xi_in_resolvent_set}, along with the regularity of the functions $w_{\kappa}^{\eps}$ in \eqref{main_results:eqn:perturbation_potential}, ensure the existence of a \emph{unitary} time propagator $U_{\Phi,\eps}(t,s)$ which implements the perturbed dynamics generated by $H_{\Phi,\eps}(t)$ (see Theorem~\ref{dynamics:prop:existence_perturbed_dynamics}). To some extent the propagator $U_{\Phi,\eps}(t,s)$ is the main object of LRT, and Hypothesis~\ref{hypothesis:current} has the principal purpose of ensuring the existence of a $U_{\Phi,\eps}(t,s)$ which is “regular enough” (in the sense of Proposition~\ref{dynamics:prop:existence_perturbed_dynamics}). As a matter of fact the mere existence of $U_{\Phi,\eps}(t,s)$ can be deduced under weaker hypotheses than those stated in Hypothesis~\ref{hypothesis:current}. However, Hypothesis~\ref{hypothesis:current} and especially the current expansion \eqref{main_results:eqn:additive_form_perturbed_hamiltonian} are key ingredients to make LRT rigorous. 

By construction $U_{\Phi,\eps}(t,s) \in \Alg$ and so it can be used to define a dynamics on $\Alg$ by means of the prescription
\begin{align}
	\alpha^{\Phi,\eps}_{(t,s)}(A) := U_{\Phi,\eps}(t,s) \; A \; U_{\Phi,\eps}(s,t)
	,
	&&
	t , s \in \R
	, \; 
	A \in \Alg
	. 
\label{main_results:eqn:definition_perturbed_dynamics_automorphism}
\end{align}
According to a standard argument (Proposition~\ref{framework:prop:extension_isometry_Lp}) this map extends to continuous family of isometries on each of the Banach spaces $\mathfrak{L}^p(\Alg)$.
\begin{consequence}[Existence of the perturbed dynamics]\label{main_results:conseq:perturbed_evolution}
	Suppose Hypotheses~\ref{hypothesis:trace}–\ref{hypothesis:current} hold true.
	Let $U_{\Phi,\eps}(t,s) \in \Alg$ be the unitary propagator generated by the isospectrally perturbed Hamiltonian $H_{\Phi,\eps}(t)$. Then for each $1 \leqslant p < \infty$ the prescription \eqref{main_results:eqn:definition_perturbed_dynamics_automorphism} induces an isometry 
	\begin{align*}
		\R \times \R \ni (t,s) \mapsto \alpha^{\Phi,\eps}_{(t,s)} \in \mathrm{Iso} \bigl ( \mathfrak{L}^p(\Alg) \bigr )
	\end{align*}
	which is jointly {strongly continuous} in $t$ and $s$. We refer to the mapping $(t,s) \mapsto \alpha^{\Phi,\eps}_{(t,s)}$ 
	as the \emph{perturbed dynamics}.
\end{consequence}
The joint strong continuity of the perturbed dynamics $\alpha^{\Phi,\eps}_{(t,s)}$ means that
\begin{align*}
	\lim_{t \to t_0} \; \Bnorm{\alpha^{\Phi,\eps}_{(t,r)}(A) - \alpha^{\Phi,\eps}_{(t_0,r)}(A)}_p
	= 0 
	= \lim_{s \to s_0} \; \Bnorm{\alpha^{\Phi,\eps}_{(r,s)}(A) - \alpha^{\Phi,\eps}_{(r,s_0)}(A)}_p 
\end{align*}
for all $A \in \mathfrak{L}^p(\Alg)$ and $r \in \R$. Moreover, this mapping is also a “perturbation” of the unperturbed dynamics $\alpha^0_t$ induced by $H$ in the sense that
\begin{align*}
	\lim_{\Phi \to 0} \; \Bnorm{\alpha^{\Phi,\eps}_{(t,s)}(A) - \alpha^{0}_{t-s}(A)}_p = 0
	,
	&&
	\forall A \in \mathfrak{L}^p(\Alg)
	, \;
	\forall t , s \in \R
\end{align*}
independently of $\eps > 0$. The last fact is just a consequence of the Duhamel formula (see Proposition~\ref{dynamics:prop:continuity_perturbed_dynamics_field_strength} for the details). Properties of the perturbed dynamics $\alpha^{\Phi,\eps}_{(t,s)}$ are investigated in full detail in Section~\ref{dynamics:perturbed:observables}.

Up to now, we have described three different types of dynamics on the space $\mathfrak{L}^p(\Alg)$: the \emph{unperturbed dynamics}, the \emph{interaction dynamics} and the \emph{perturbed dynamics}. All these maps can be used to define the different time evolutions of the (initial) equilibrium state $\rho$ enumerated in (H5). We are in position to state properly what “equilibrium” means.
\begin{definition}[Initial equilibrium state]\label{main_results:defn:initial_equilibrium_state}
	Let $H \in \affil(\Alg)$ be a selfadjoint Hamiltonian. An \emph{initial equilibrium state} for $H$ is any {positive} element $\rho \in \Alg^+$ such that $\alpha^0_t(\rho) = \rho$ where $\alpha^0_t$ is the unperturbed dynamics generated by $H$.
\end{definition}
The first observation is that an initial equilibrium state is totally insensitive to the effect of the unperturbed dynamics. Therefore, one is most interested in the behavior of $\rho$ under the interaction or perturbed dynamics. However, a careful analysis of the time evolution of the initial equilibrium state requires some extra regularity assumptions on $\rho$.
\begin{hypothesis}[$p$-regular initial equilibrium state]\label{hypothesis:state}
	Let $\rho$ be an initial equilibrium state in the sense of Definition~\ref{main_results:defn:initial_equilibrium_state}. We will assume that $\rho$ is \emph{$p$-regular} in the sense that
	\begin{enumerate}[leftmargin=*,label=(\roman*)]
		\item $\rho \in \Alg^+ \cap \rr{W}^{1,1}(\Alg) \cap \rr{W}^{1,p}(\Alg)$ (\emph{regularity}) and
		\item $\rho \in \rr{D}^{00}_{H,1} \cap \rr{D}^{00}_{H,p}$ (see equation~\eqref{framework:eqn:domain_maximal_generalized_commutator}) and $H \rho \in \rr{W}^{1,1}(\Alg) \cap \rr{W}^{1,p}(\Alg)$ (\emph{$H$-regularity}).
		\item $\partial_{X_k}(\rho) \in \rr{D}^{00}_{H,1} \cap \rr{D}^{00}_{H,p}$ holds for all $k = 1 , \ldots , d$. 
	\end{enumerate}
\end{hypothesis}
Item (i) also includes $\rho \in \mathfrak{L}^1(\Alg) \cap \mathfrak{L}^p(\Alg)$ and the equilibrium condition immediately implies that $\rho \in \ker \bigl ( \mathscr{L}_H^{(1)} \bigr ) \cap \ker \bigl ( \mathscr{L}_H^{(p)} \bigr )$ where $\mathscr{L}_H^{(p)}$ is the generator of $\alpha^0_t$ in $\mathfrak{L}^p(\Alg)$. Item (ii) ensures that also $H \, \rho \in \mathfrak{L}^1(\Alg) \cap \mathfrak{L}^p(\Alg)$. As a consequence of Proposition~\ref{dynamics:prop:density_core_Liouvillian} one immediately gets that the equality
\begin{align}
	H \, \rho = (H \, \rho)^* 
	\label{main_results:eqn:selfadjointness_product_H_rho}
\end{align}
holds true in $\mathfrak{L}^1(\Alg)$ and $\mathfrak{L}^p(\Alg)$. Equation \eqref{main_results:eqn:selfadjointness_product_H_rho} is a kind of generalized commutation rule between $\rho$ and $H$. In order to have a true commutation relation we need stronger assumptions that are not necessary for the moment. 
As a consequence of all the previous assumptions about the Hamiltonian $H$, the generators $X_k$ and the initial equilibrium state one can prove that $J_k \, \rho \in \mathfrak{L}^1(\Alg)\cap\mathfrak{L}^p(\Alg)$ and 
\begin{align}
	H \, \partial_{X_k}(\rho) = J_k \, \rho + \partial_{X_k}(H \, \rho) \in \mathfrak{L}^1(\Alg)\cap\mathfrak{L}^p(\Alg)\;
	\label{main_results:eqn:identity_H_partial_X_rho}
\end{align}
for all $k = 1 , \ldots , d$. The \eqref{main_results:eqn:identity_H_partial_X_rho}, which is proved in Lemma~\ref{Kubo_formula:lem:Lp_regularity_states}, will play a crucial role for the analysis of the perturbed dynamics of $\rho$. In particular \eqref{main_results:eqn:identity_H_partial_X_rho} ensures that $H \, \partial_{X_k}(\rho) \in \mathfrak{L}^1(\Alg)\cap\mathfrak{L}^p(\Alg)$ for all $k = 1 , \ldots , d$. Let us point out that even though the conditions for $\rho$ listed  above seems to be quite strong they are necessary and natural for the derivation of the Kubo formulas in the contest of a LRT. For instance, Assumption~5.1 in \cite{Bouclet_Germinet_Klein_Schenker:linear_response_theory_magnetic_Schroedinger_operators_disorder:2005} is equivalent (through \cite[Proposition 4.2]{Bouclet_Germinet_Klein_Schenker:linear_response_theory_magnetic_Schroedinger_operators_disorder:2005}) to the fact that $\rho$ has to be 2-regular. Let us also mention that the condition $\rho \in \rr{W}^{1,2}(\Alg)$ was originally identified in \cite{Bellissard_van_Elst_Schulz_Baldes:noncommutative_geometry_quantum_Hall_effect:1994} (see also \cite{Bellissard_Schulz_Baldes:quantum_transport_aperiodic_media:1998}) as the main requirement for the derivation of the Kubo formula. Let us point out that  a sufficient condition to construct an equilibrium state is to define $\rho := f(H)$ where $f$ is any positive function in $L^{\infty}(\R)$. In this way the equilibrium condition $\alpha_t^0(\rho) = \rho$ and the positivity condition $\rho \in \Alg^+$ are automatically satisfied. However, the class of positive function $\mathscr{S}^{1,p}_{\mathcal{T},X}(\R)\subset L^{\infty}(\R)$ such that $\rho$ verifies Hypothesis~\ref{hypothesis:state} relies strongly on the particular nature of the trace $\mathcal{T}$ and of the generators $X_k$. In many situations of physical interest the class $\mathscr{S}^{1,p}_{\mathcal{T},X}(\R)$ is composed of functions with a sufficiently rapid decay at infinity (\eg as the Schwartz functions in \cite{Bouclet_Germinet_Klein_Schenker:linear_response_theory_magnetic_Schroedinger_operators_disorder:2005}). This aspect will be shortly discussed in Section~\ref{Kubo_formula:comparing_evolutions:equilibrium_states}.

An initial equilibrium state $\rho$ is left invariant, by definition, under the unperturbed dynamics $\alpha^0_t$. However it can be evolved by the {interaction dynamics}:
\begin{align}
	\rho_{\mathrm{int}}(t) \equiv \rho_{\mathrm{int}}(t;\eps,\Phi) := \gamma^{\Phi,\eps}_t(\rho)
	,
	&& 
	t \in \R
	.
	\label{Kubo_formula:eqn:definition_interaction_evolution_intro}
\end{align}
The state $\rho$ can be evolved also by the perturbed dynamics $\alpha^{\Phi,\eps}_{(t,s)}$ 
through the prescription
\begin{align}
	\rho_{\mathrm{full}}(t) \equiv \rho_{\mathrm{full}}(t;\eps,\Phi) := \lim_{s \to -\infty} \; \alpha^{\Phi,\eps}_{(t,s)}(\rho) 
	, 
	&& 
	t \in \R
	, 
	\label{main_results:eqn:definition_full_evolution_int}
\end{align}
seen as a limit in $\mathfrak{L}^1(\Alg)$ and $\mathfrak{L}^p(\Alg)$ according to the regularity of $\rho$. We refer to \eqref{main_results:eqn:definition_full_evolution_int} as the \emph{full} evolution of the initial state $\rho$. While \emph{a priori} it is not at all clear whether this limit exists, under all these Hypotheses we can prove that the fully time-evolved state $\rho_{\mathrm{full}}(t)$ is indeed well-defined and it can be compared with $\rho_{\mathrm{int}}(t)$:%
\begin{theorem}[Comparison of the two dynamics]\label{main_results:thm:comparison_dynamics}
	Suppose Hypotheses~\ref{hypothesis:trace}–\ref{hypothesis:state} hold true, and let $r = 1 , p$ with $p$ being the regularity degree of $\rho$ from Hypothesis~\ref{hypothesis:state}.  
	\begin{enumerate}
		\item The limit which defines the full evolution $\rho_{\mathrm{full}}(t)$ given by \eqref{main_results:eqn:definition_full_evolution_int} exists in $\mathfrak{L}^r(\Alg)$, and can be expanded in terms of $\rho_{\mathrm{int}}(t)$, $\pmb{\Phi}$, and the operator-valued vector
		\begin{align}
			\mathbf{K}^{\Phi,\eps}(t) := - \int_{-\infty}^t \dd \tau \; 
			\switch(\tau) \, f_j(\tau) \;
			\alpha^{\Phi,\eps}_{(t,\tau)} \bigl ( \nabla \bigl ( \rho_{\mathrm{int}}(\tau) \bigr ) \bigr ) 
			\label{main_results:eqn:components_K_Phi_eps}
		\end{align}
		so that 
		\begin{align}
			\rho_{\mathrm{full}}(t) &= \rho_{\mathrm{int}}(t) + \pmb{\Phi} \cdot \mathbf{K}^{\Phi,\eps}(t) 
			.
			\label{main_results:eqn:comparison_rho_full_rho_int}
		\end{align}
		\item The full evolution $t \mapsto \rho_{\mathrm{full}}(t)$ is the unique solution of 
		\begin{align}
			\left \{
			\begin{aligned}
				\frac{\dd \rho_{\mathrm{full}}}{\dd t}(t) &= - \ii \bigl [ H_{\Phi,\eps}(t) \, , \, \rho_{\mathrm{full}}(t) \bigr ]_{\ddagger}
				\\
				\lim_{t \to -\infty} \rho_{\mathrm{full}}(t) &= \rho
			\end{aligned}
			\right .
			\label{main_results:eqn:full_dynamical_problem}
		\end{align}
		where the limit and the derivative are taken in $\mathfrak{L}^r(\Alg)$, and the generalized commutator $[ \, \cdot \, , \, \cdot \, ]_{\ddagger}$ is defined in \eqref{main_results:eqn:generalized_commutator}.
	\end{enumerate}
\end{theorem}
The proof of this theorem is postponed in Section~\ref{Kubo_formula:comparing_evolutions:difference_evolved_states}. Let us point out that Theorem~\ref{main_results:thm:comparison_dynamics} is nothing more than an “abstract” generalization of \cite[Theorem~1.1]{Bouclet_Germinet_Klein_Schenker:linear_response_theory_magnetic_Schroedinger_operators_disorder:2005}.

\section{Linear response and the Kubo formula} 
\label{main_results:linear_response}
In order to completely specify the context of LRT and to present the different incarnations of the Kubo-formula we need to discuss the role of (H4) in the initial list, \ie the instantaneous observable $t \mapsto J(t)$. The prototypical observables which enters in the LRT are the current density tensors $\mathbf{J}^{(r)}$ given in \eqref{main_results:eqn:collection_density_tensors}. As $\mathbf{J}^{(r)}$ is generated by the unperturbed Hamiltonian $H$ through iterated commutators, in the same way one can consider \emph{instantaneously} perturbed current density tensors generated by the perturbed Hamiltonian $H_{\Phi,\eps}(t)$. This leads to a family of operator-valued tensors
\begin{align*}
	\R \ni t \mapsto \mathbf{J}^{(r)}_{\Phi,\eps}(t) := G_{\Phi,\eps}(t) \; \mathbf{J}^{(r)} \; G_{\Phi,\eps}(t)^*
\end{align*}
which are well-defined if one assumes the validity of Hypothesis~\ref{hypothesis:generators} and Hypothesis~\ref{hypothesis:current}; \cf Section~\ref{Kubo_formula:Kubo_formula:net_current}. This relevant type of instantaneous observables possesses some peculiar properties (Proposition~\ref{Kubo_formula:prop:J_k_is_of_current_type}) which are sufficient to derive the Kubo-formula. This motivates the the attempt to generalize the family of instantaneous observable suitable for the LRT.
\begin{definition}[Current-type observable]\label{main_results:defn:current}
	Suppose Hypotheses~\ref{hypothesis:trace}–\ref{hypothesis:state} hold true. We say a time-dependent observable $\R \ni t\mapsto J_{\Phi,\eps}(t)$ is of \emph{current-type} with respect to $H_{\Phi,\eps}(t)$ and $\rho$ if the following holds: 
	\begin{enumerate}[leftmargin=*,label=(\roman*)]
		\item $J_{\Phi,\eps}(t)$ is an instantaneous perturbation, \ie there exists a selfadjoint operator $J  \in \affil(\Alg)$ with 
		\begin{align*}
			J_{\Phi,\eps}(t) := G_{\Phi,\eps}(t) \; J \; G_{\Phi,\eps}(t)^*
			,
			&&
			\forall t \in \R
			.
		\end{align*}
		\item $J_{\Phi,\eps}(t) \, \rho_{\mathrm{full}}(t) \in \mathfrak{L}^1(\Alg)$ for all $t \in \R$ and 
		\begin{align*}
			\lim_{t \to -\infty} J_{\Phi,\eps}(t) \, \rho_{\mathrm{full}}(t) = J \, \rho
		\end{align*}
		in the topology of $\mathfrak{L}^1(\Alg)$. 
		\item $\domain(H)\subseteq \domain(J)$ which in turn ensures $J \, \frac{1}{H - \xi} \in \Alg$ (\cf Lemma~\ref{framework:lem:extension_algebra_unbounded_operators}~(2)) and
		\begin{align*}
			J_{\Phi,\eps}(t) \, \frac{1}{H_{\Phi,\eps}(t) - \xi} = \gamma^{\Phi,\eps}_t \left ( J \, \frac{1}{H - \xi} \right ) 
		\end{align*}
		for all $t \in \R$.
	\end{enumerate}
\end{definition}
The central quantity for LRT is the \emph{macroscopic net current} 
\begin{subequations}
	\begin{align}
		\mathscr{J}^{\Phi,\eps}[J,\rho](t) :& \negmedspace= 
		\mathcal{T} \Bigl ( J_{\Phi,\eps}(t) \, \bigl ( \rho_{\mathrm{full}}(t) - \rho_{\mathrm{int}}(t) \bigr ) \Bigr )
		\label{main_results:eqn:net_density_current_1}
		\\ 
		&= \mathcal{T} \bigl ( J_{\Phi,\eps}(t) \, \rho_{\mathrm{full}}(t) \bigr ) - \mathcal{T}(J \, \rho)
		\label{main_results:eqn:net_density_current_2}
	\end{align}
\end{subequations}
associated to the initial equilibrium state $\rho$, evolved using the full evolution~\eqref{main_results:eqn:full_dynamical_problem}, and the current-type observable $J$. As a difference of expectation values this quantity measures the net flow of the \emph{macroscopic} current between the fully evolved state $\rho_{\mathrm{full}}(t)$ and the “dragged along” state $\rho_{\mathrm{int}}(t)$. Because we slowly switch on the perturbation field $\pmb{\Phi}$ at rate $\eps > 0$, heuristically we expect $\rho_{\mathrm{full}}(t) \approx \rho_{\mathrm{int}}(t)$ to hold in case $\eps$ is small. For the precise justification of the second equality \eqref{main_results:eqn:net_density_current_2} we refer to Section~\ref{Kubo_formula:Kubo_formula:net_current}. Item (ii) in Definition~\ref{main_results:defn:current} ensures the \emph{equilibrium condition}, \ie zero net flux in the distant past,
\begin{align*}
	\lim_{t \to -\infty} \mathscr{J}^{\Phi,\eps}[J,\rho](t) = 0 
	. 
\end{align*}
We will elaborate on this further in Remark~\ref{Kubo_formula:remark:equilibrium_conditions_current}. Moreover, under all the Hypotheses listed above we will prove in Lemma~\ref{Kubo_formula:lem:existence_integral_net_current} the absence of net current 
\begin{align*}
	\lim_{\Phi \to 0} \mathscr{J}^{\Phi,\eps}[J,\rho](t) = 0 
\end{align*}
in the limit of vanishing perturbations. Consequently, the first term in the “Taylor” expansion of $\mathscr{J}^{\Phi,\eps}[J,\rho](t)$ around $\pmb{\Phi} = 0$ vanishes, and the first non-trivial term in
\begin{align}
	\mathscr{J}^{\Phi,\eps}[J,\rho](t) = \sum_{k = 1}^d \Phi_k \, \sigma^{\eps}_k[J,\rho](t) + \order(\Phi^2)
	. 
	\label{main_results:eqn:net_macro_current_expansion_Phi}
\end{align}
describe the \emph{linear response} of the system to the perturbation. Mathematically speaking, our task is to ensure that $\mathscr{J}^{\Phi,\eps}[J,\rho](t)$ is sufficiently regular in the fields — mere continuity does not suffice. These first-order corrections are  collectively known as the 
\begin{definition}[Conductivity coefficients]\label{main_results:defn:conductivity_tensor}
	Let $\rho$ by an initial equilibrium state for the Hamiltonian $H \in \affil(\Alg)$ and $J$ a current-type observable. The $\eps$-dependent \emph{conductivity coefficients} generated by perturbing the system adiabatically at rate $\eps > 0$ via a field $\pmb{\Phi}$ between the initial time $-\infty$ and the final time $t$, is the $d$-dimensional vector with components
	\begin{align*}
		\sigma^{\eps}_k[J,\rho](t) := \left . \frac{\partial}{\partial \Phi_k} \mathscr{J}^{\Phi,\eps}[J,\rho](t) \right \vert_{\Phi=0}
		, 
		&&
		k = 1 , \ldots , d
		, 
	\end{align*}
	when they exist. Their adiabatic limits 
	\begin{align*}
		\sigma^k[J,\rho] := \lim_{\eps \rightarrow 0^+} \sigma^{\eps}_k[J,\rho](t)
		,
		&&
		k = 1 , \ldots , d
		, 
	\end{align*}
	are also referred to as \emph{conductivity coefficients} whenever they exist. The $\sigma^k[J,\rho]$ are often referred to as the \emph{Kubo coefficients}. 
\end{definition}
We are now in position to state the main result of LRT.
\begin{theorem}[The Kubo formula]\label{main_results:thm:Kubo_formula}
	Suppose that in addition to Hypotheses~\ref{hypothesis:trace}–\ref{hypothesis:state} we are given a current-type observable $J(t)$ in the sense of Definition~\ref{main_results:defn:current}. Then the $\eps$-dependent conductivity coefficients are given by the \emph{Kubo formula}
	\begin{align}
		\sigma^{\eps}_k[J,\rho](t) = - \ii \int_{-\infty}^t \dd \tau \; \switch(\tau) \, f_k( \tau) \, \mathcal{T} \left ( 
		J \; \alpha^{0}_{t - \tau} \bigl ( \partial_{X_k}(\rho) \bigr ) \right ) \;
		\label{main_results:eqn:Kubo_formula}
	\end{align}
	for each $k = 1 , \ldots , d$.
\end{theorem}
The proof of this result is postponed to Section~\ref{Kubo_formula:Kubo_formula:proof}. By inserting the explicit expression for the switch function $\switch$ into \eqref{main_results:eqn:Kubo_formula} we can rewrite the conductivity coefficients
\begin{align}
	\sigma^{\eps}_k[J,\rho](t) = \widetilde{\sigma}^k_{\eps}[J,\rho](t) + \delta^k_{\eps}[J,\rho](t)
	\label{main_results:eqn:Kubo_formula_splitting}
\end{align}
as the sum of two terms, a non-trivial contribution, 
\begin{align}
	\widetilde{\sigma}^k_{\eps}[J,\rho](t) := - \ii \e^{\eps t} \; \mathcal{T} \left ( \int^{+\infty}_0 \dd \tau \; \e^{-\eps \tau} \; f_k(t-\tau) \; 
	J \; \alpha^{0}_{\tau} \bigl ( \partial_{X_k}(\rho) \bigr ) \right ) 
	, 
	\label{main_results:eqn:Kubo_formula_non_trivial_contribution_splitting}
\end{align}
and a remainder which vanishes in the adiabatic limit (\cf Lemma~\ref{Kubo_formula:lem:residual_term_vanishes_in_adiabatic_limit}),
\begin{align*}
	\lim_{\eps \to 0^+} \delta^k_{\eps}[J,\rho](t) = 0
	. 
\end{align*}
In view of the fact that we are interested in the adiabatic limit, the last observation allows us to consider the quantity $\widetilde{\sigma}^k_{\eps}[J,\rho](t)$ instead of $\sigma^{\eps}_k[J,\rho](t)$. The following result uses that the Laplace transform relates the evolution automorphism $\alpha^0_{\tau}$ and the resolvent of the Liouvililan on $\mathfrak{L}^p(\Alg)$, 
\begin{align}
	\frac{1}{\mathscr{L}_{H}^{(p)} + \eps}(A) = \int_0^{+\infty} \dd \tau \, \e^{-\eps\tau} \; \alpha^0_{\tau}(A)
	,
	&&
	\forall A \in \mathfrak{L}^p(\Alg)
	. 
	\label{main_results:eqn:laplace_transform_resolvent}
\end{align}
\begin{corollary}\label{main_results:cor:conductivity_tensor_special_cases}
	For all $t \geq 0$ the conductivity coefficients from Theorem~\ref{main_results:thm:Kubo_formula} can be computed explicitly for the following two choices of $f_k$: 
	\begin{enumerate}
		\item If $f_k \equiv 1$ then
		\begin{align}
			\widetilde{\sigma}^k_{\eps}[J,\rho](t) = - \ii \e^{\eps t} \;
			\mathcal{T} \left ( J \; \frac{1}{\mathscr{L}_H^{(1)} + \eps} \bigl ( \partial_{X_k}(\rho) \bigr ) \right ) 
			. 
			\label{main_results:eqn:kubo_formula_f_eq_1}
		\end{align}
		\item In case $f_k$ is given by \eqref{dynamics:eqn:freq_decomposition_modulation_function}, then
		\begin{align}
			\widetilde{\sigma}^k_{\eps}[J,\rho](t) = - \ii \e^{\eps t} \;
			\int_{\R} \dd \kappa \; \e^{\ii\kappa t} \; \hat{f}_k(\kappa) \; \mathcal{T} \left ( J \; \frac{1}{\mathscr{L}_H^{(1)} + \eps + \ii \kappa} \bigl ( \partial_{X_k}(\rho) \bigr ) \right ) 
			. 
			\label{main_results:eqn:Kubo_formula_Fourier_transform_f}
		\end{align}
	\end{enumerate}
\end{corollary}
Also this result is proven in Section~\ref{Kubo_formula:Kubo_formula:proof}. Equation~\eqref{main_results:eqn:kubo_formula_f_eq_1} has been first obtained in \cite[eq.~(41)]{Thouless_Kohmoto_Nightingale_Den_Nijs:quantized_hall_conductance:1982} and \cite[Theorem~1]{Bellissard_Schulz_Baldes:quantum_transport_aperiodic_media:1998} in the approximation of bounded tight-binding operators where $C^*$-algebraic (as opposed to von Neumann algebraic) techniques have been used. The analogous formula for the magnetic Laplacian in the continuum has been derived in \cite[Corollary~5.10]{Bouclet_Germinet_Klein_Schenker:linear_response_theory_magnetic_Schroedinger_operators_disorder:2005}. A formula similar to equation \eqref{main_results:eqn:Kubo_formula_Fourier_transform_f} appears in \cite[eq.~(3.30)]{klein-lenoble-muller-07}.

\section{The adiabatic limit and the Kubo-Str\v{e}da formula} 
\label{main_results:adiabatic_limit}
Physically the \emph{adiabatic limit} $\eps \to 0$ means that the ramp speed at which we switch on the external, macroscopic perturbation becomes infinitesimally small compared to the time scale of the microscopic dynamics. Given that here it is expected that $\rho_{\mathrm{full}}(t) \to \rho_{\mathrm{int}}(t)$ holds in some sense, we expect that many of the details on \emph{how} the perturbation is switched on will be washed out. In comparison, the conductivity coefficients $\sigma^{\eps}_k[J,\rho](t)$ (or, equivalently, $\widetilde{\sigma}^k_{\eps}[J,\rho](t)$) depend on the entire history of the system until time $t$. Indeed, this is the primary purpose of the adiabatic limit $\eps \to 0$, it leads to a time averaging of the conductivity coefficients that averages away many details of the perturbation. While with our approach the time averaging emerges naturally, other authors, \eg \cite{Thouless_Kohmoto_Nightingale_Den_Nijs:quantized_hall_conductance:1982,Bellissard_Schulz_Baldes:quantum_transport_aperiodic_media:1998}, had to introduce it in an \emph{ad hoc} fashion to derive the Kubo formula. Their idea is to exploit the well-known fact that Cesàro summability implies Abel summability \cite[Chap.~8, Theorem~2.3]{widder-71} and
\begin{align*}
	\sexpval{a} := \lim_{T \to +\infty}{\frac{1}{T}} \int_0^T \dd \tau \, a(\tau) 
	= \lim_{\eps \to 0^+} \eps \int_0^{+\infty} \dd \tau \, \e^{-\eps \tau} \, a(\tau)
	. 
\end{align*}
And writing this time average as a Laplace transform allows one to relate this expression to the resolvent of the Liouvillian via \eqref{main_results:eqn:laplace_transform_resolvent}. In our derivation the presence of the time average in \eqref{main_results:eqn:Kubo_formula} can be traced back to the time integral in equation~\eqref{main_results:eqn:comparison_rho_full_rho_int} which computes the instantaneous difference between $\rho_{\mathrm{full}}(t)$ and $ \rho_{\mathrm{int}}(t)$. 

In order to state the main result about the adiabatic limit let us observe that the Liouvillian $\mathscr{L}_H^{(q)}$ is a linear operator on the Banach space $\mathfrak{L}^q(\Alg)$ and we are interested in computing the limit of  $\mathfrak{L}^q(\Alg)$ times its resolvent in $\eps$ when $\eps\to0^+$ with respect to the strong operator topology in $\mathfrak{L}^q(\Alg)$. The existence of this limit is proved in Section \ref{dynamics:unperturbed:projection} and the result is
\begin{align}
	\lim_{\eps \to 0^+} \frac{\mathscr{L}_H^{(q)}}{\mathscr{L}_H^{(q)} - \eps}(A) =   \mathscr{P}_H^{(q)\bot} (A)
	, 
	&&
	A \in \mathfrak{L}^q(\Alg)
	\label{main_results:eqn:limit_L_Lplus_eps_projection}
\end{align}
where $\mathscr{P}_H^{(q)\bot}$ is a suitable idempotent (Banach space projection) on acting on $\mathfrak{L}^q(\Alg)$. In the special case $p=2$ the idempotent $\mathscr{P}_H^{(2)\bot}$ turns out to be an orthogonal projection with respect to the Hilbert structure of $\mathfrak{L}^2(\Alg)$, and one has the relation $\mathscr{P}_H^{(2) \perp} := \id_{\mathfrak{L}^2(\Alg)} - \mathscr{P}_H^{(2)}$ where $\mathscr{P}_H^{(2)}$ is the  projection onto the kernel $\ker \bigl ( \mathscr{L}_H^{(2)} \bigr )$.
\begin{theorem}[Adiabatic limit of the Kubo formula]\label{main_results:thm:adiabatic_limit_Kubo_formula}
	Assume Hypotheses~\ref{hypothesis:trace}–\ref{hypothesis:state} hold, and $f_k \equiv 1$. Furthermore, suppose we are given a current-type observable $J(t)$ in the sense of Definition~\ref{main_results:defn:current} for which there exists a $Q_J \in \mathfrak{L}^q(\Alg)$ such that 
	\begin{align*}
		J = \mathscr{L}_H^{(q)}(Q_J)
	\end{align*}
	with ${p}^{-1} + {q}^{-1} = 1$. Then the adiabatic limit $\eps \to 0$ of the conductivity coefficients exists, and is given by
	\begin{align}
		\sigma^k[J,\rho] = \ii \;
		\mathcal{T} \Bigl (\mathscr{P}_H^{(q)\bot}(Q_J) \; \partial_{X_k}(\rho) \Bigr ) 
		\label{main_results:eqn:adiabatic_limit_conductivity_tensor}.
	\end{align}
	In the special case $p = q = 2$ the last formula can be recast as 
	\begin{align}
		\sigma^k[J,\rho] = \ii \Bdscpro{\mathscr{P}_H^{(2)\bot}(Q_J)^* \,}{\, \partial_{X_k}(\rho) }_{\mathfrak{L}^2}
		\label{main_results:eqn:adiabatic_limit_L2_conductivity_tensor}
	\end{align}
	where the Hilbert structure of $\mathfrak{L}^2(\Alg)$ has been used. 
\end{theorem}
The proof of this result is presented in Section~\ref{Kubo_formula:adiabatic_limit}.

The most important case considered in the physical literature is the computation of the Kubo's coefficients when the initial equilibrium state is a spectral projection of the Hamiltonian $P := P(H)$ and the observable is the $k$-th component of the density current \eqref{main_results:eqn:definition_current_tensors} $J := J^{(1)}_k$. In this case
the \eqref{main_results:eqn:net_macro_current_expansion_Phi} can be conveniently rewritten as
\begin{align}
	\mathscr{J}^{\Phi,\eps}_k[P](t) = \sum_{j = 1}^d \Phi_j \, \sigma^{kj}_{\eps}[P](t) + \order(\Phi^2)
	, 
	&&
	k = 1 , \ldots , d 
	, 
	\label{main_results:eqn:macroscopic_current_after_adiabatic_limit}
\end{align}
where $\mathscr{J}^{\Phi,\eps}_k[P](t) := \mathscr{J}^{\Phi,\eps}[J^{(1)}_k,P](t)$ can be interpreted as the $k$-th component of the macroscopic current, and the $\sigma^{kj}_{\eps}[P](t) := \sigma^{\eps}_k[J^{(1)}_j,P](t)$ are the components of a rank two tensor. The next result enumerates sufficient conditions under which the adiabatic limit
\begin{align*}
	\sigma^{kj}[P] := \lim_{\eps \to 0^+} \sigma^{kj}_{\eps}[P](t)
	, 
	&&
	k , j = 1 , \ldots , d 
	, 
\end{align*}
exists and \emph{can be computed explicitly}. The quantity 
\begin{align*}
	\pmb{\sigma}[P] := \bigl \{ \sigma^{kj}[P] \bigr \}_{k , j = 1 , \ldots , d }
\end{align*}
is called the \emph{conductivity tensor} and, as a matter of fact, it is the most relevant object in the physical applications of LRT. 
\begin{theorem}[The Kubo-Str\v{e}da formula]\label{main_results:thm:Kubo_Streda_formula}
	Suppose Hypotheses~\ref{hypothesis:trace}–\ref{hypothesis:current} hold true. Moreover, assume that $P$ is a spectral projection of $H$ with respect to a bounded portion of $\spec(H)$ which is also $2$-regular in the sense of Hypothesis~\ref{hypothesis:state}.
	
	Then the $\eps$-dependent conductivity coefficients
	\begin{align*}
		\sigma^{kj}_{\eps}[P](t) := \sigma^{\eps}_k[J^{(1)}_j,P](t)
		, 
		&&
		k , j = 1 , \ldots , d
		, 
	\end{align*}
	computed through the formula \eqref{main_results:eqn:Kubo_formula} are well-defined, their adiabatic limits exist and are given by the \emph{Kubo-Str\v{e}da formula}
	\begin{subequations}\label{main_results:eqn:Kubo_Streda_formula}
		\begin{align}
			\sigma^{kj}[P] := \lim_{\eps \to 0^+} \sigma^{kj}_{\eps}[P](t) &= + \ii \Bdscpro{\bigl [P,\partial_{X_k}(P)\bigr ]_{(2)} \,}{\,  \partial_{X_j}(P) }_{\mathfrak{L}^2}
			\\
			&= - \ii \mathcal{T} \Bigl (P \, \bigl [ \partial_{X_k}(P) , \partial_{X_j}(P) \bigr ]_{(1)} \Bigr ) 
		\end{align}
	\end{subequations}
	for each $k , j = 1 , \ldots , d$.
\end{theorem}
The proof of this important result is the main aim of Section~\ref{Kubo_formula:adiabatic_limit}. The symbol $[ \, \cdot \, , \, \cdot \, ]_{(r)}$ which appears in equation \eqref{main_results:eqn:Kubo_Streda_formula} means that the commutator takes values in the space $\mathfrak{L}^r(\Alg)$; we refer to Section~\ref{framework:commutators:t_measurable_operators} for details.

\section{Zero temperature limit and topological interpretation} 
\label{main_results:zero_temp}
We already discussed that a typical way of building initial equilibrium states associated to a Hamiltonian $H \in \affil(\Alg)$ is to choose suitable positive functions of $H$. For instance, in condensed matter problems concerning electron systems (usually in the approximation of non-interacting particles) the typical initial equilibrium states are described through the 
the \emph{Fermi-Dirac} distribution at \emph{inverse temperature} $\beta>0$ and \emph{Fermi
energy} $E_{\mathrm{F}}$, 
\begin{align*}
	\rho_{\beta}(x) :=
	\begin{cases}
		\frac{1}{1 + \e^{\beta(x-E_{\mathrm{F}})}} & \beta < \infty \\
		\chi_{(-\infty,E_{\mathrm{F}}]}(x) & \beta = \infty \\
	\end{cases}
	. 
\end{align*}
Here $\chi_{(-\infty,E_{\mathrm{F}}]}(x)$ denotes the characteristic function of $(-\infty,E_{\mathrm{F}}]$ , \ie $\chi_{(-\infty,E_{\mathrm{F}}]}(x)$ is $1$ if $x \leqslant E_{\mathrm{F}}$ and is $0$ otherwise. The associated initial equilibrium state is given by
\begin{align*}
	\rho_{\beta} := \rho_{\beta}(H) = 
	\begin{cases}
		\frac{1}{1+\e^{\beta(H-E_{\mathrm{F}})}} & \beta < \infty \\
		P & \beta = \infty \\
	\end{cases}
\end{align*}
where $P := \chi_{(-\infty,E_{\mathrm{F}}]}(H)$ is the spectral projection of $H$ for energies up to the Fermi energy $E_{\mathrm{F}}$, and is conventionally called \emph{Fermi projection}. Since the Fermi-Dirac distribution converges pointwise to the characteristic function of $(-\infty,E_{\mathrm{F}}]$ in the limit $\beta \to +\infty$, it follows that $\rho_{\beta} \to P$ with respect to the strong operator topology (SOT) of operators on the Hilbert space $\Hil$. The limit $\beta \to +\infty$ is known in condensed matter physics as the \emph{zero temperature limit}.

Suppose $\rho_{\beta} = \rho_{\beta}(H)$ is a net of 2-regular initial equilibrium states according to Hypothesis~\ref{hypothesis:state} with the following two additional properties: 
\begin{enumerate}[leftmargin=*,label=(\roman*)]
	\item[(iii)] $\rho_{\beta} \to P$ in the topology of $\mathfrak{L}^1(\Alg)$ where $P$ is a 2-regular spectral projection of $H$, and 
	\item[(iv)] $H \, \partial_{X_k}(\rho_{\beta}) \to H \, \partial_{X_k}(P)$ in the topology of $\mathfrak{L}^1(\Alg)$. 
\end{enumerate}
Under these conditions, in particular (iv), one has that 
\begin{align*}
	\lim_{\beta \to +\infty} \mathcal{T} \left ( J \; \alpha^{0}_{t} \bigl ( \partial_{X_k}(\rho_{\beta}) \bigr ) \right ) 
	= \mathcal{T} \left ( J \; \alpha^{0}_{t} \bigl ( \partial_{X_k}(P) \bigr ) \right ) 
\end{align*}
and the application of the Dominated Convergence Theorem in \eqref{main_results:eqn:Kubo_formula} leads to
\begin{align*}
	\lim_{\beta \to +\infty} \sigma^{\eps}_k[J,\rho_{\beta}](t) = \sigma^{\eps}_k[J,P](t)
	.
\end{align*}
The Kubo-Str\v{e}da formula is the adiabatic limit of $\sigma^{\eps}_k[J_j^{(1)},P](t)$, and hence one has
\begin{align*}
	\lim_{\eps \to 0^+} \lim_{\beta \to +\infty} \sigma^{\eps}_k[J_j^{(1)} , \rho_{\beta}](t) = \sigma^{kj}[P] 
\end{align*}
where $\sigma^{kj}[P]$ is given by \eqref{main_results:eqn:Kubo_Streda_formula}. In other words, under suitable conditions (\eg property (iv) above) the Kubo-Str\v{e}da formula can be seen as the “zero temperature limit” of the Kubo formula. Let us point out that the zero temperature limit has to be taken first and \emph{then} one computes the adiabatic limit. The two limits do not commute in general.

Let us briefly discuss one last aspect of the Kubo-Str\v{e}da formula. Under appropriate circumstances the right-hand side of \eqref{main_results:eqn:Kubo_Streda_formula} has the structure of a 2-cocycle over a (sub) $*$-algebra contained in $\Alg$. This associates the components $\sigma^{kj}[P]$ with Chern-Connes characters which are quantities of topological nature \cite[Part~IV]{connes-94}, \cite[Section~F]{Bellissard_van_Elst_Schulz_Baldes:noncommutative_geometry_quantum_Hall_effect:1994}. The equality of the components of the conductivity tensor as given by the Kubo-Str\v{e}da formula with Chern-Connes characters (up to physical proportionality constants) goes under the name of \emph{Kubo-Chern formula}. This is the crucial ingredient in the topological interpretation of the Quantum Hall Effect via LRT \cite{Thouless_Kohmoto_Nightingale_Den_Nijs:quantized_hall_conductance:1982,Bellissard_van_Elst_Schulz_Baldes:noncommutative_geometry_quantum_Hall_effect:1994}, which gave birth to the theory of \emph{topological insulators}, one of the most active areas of condensed matter physics today. However, the geometric aspects of the Kubo-Str\v{e}da are still quite a hot topic in many areas of condensed matter physics, and a more detailed analysis is beyond the scope of this work.

\section{The tight-binding type simplification} 
\label{main_results:tight_binding_simplification}
For simpler systems, most notably those described by tight-binding operators, the setting for LRT can be greatly simplified \cite{Thouless_Kohmoto_Nightingale_Den_Nijs:quantized_hall_conductance:1982,Bellissard_Schulz_Baldes:quantum_transport_aperiodic_media:1998}, because the assumptions avoid the main technical problems we tackle in this work. 
\begin{definition}[Tight-binding-type setting]\label{main_results:defn:tight_binding_setting}
	One is in a \emph{tight-binding-type setting} if $\mathcal{T}(\id) = 1$ and $H \in \Alg$.
\end{definition}
The first condition, $\mathcal{T}(\id) = 1$, implies the trace is in fact finite (as opposed to semifinite), and therefore $\affil(\Alg)$ agrees automatically with the $\ast$-algebra of measurable operators with respect to $\mathcal{T}$ (see Example~\ref{framework:example:T_measurable_operators}~(3) and references therein). Additionally, $\Alg$ itself is contained in all of the other $\mathfrak{L}^p(\Alg)$ spaces. The second condition, $H \in \Alg$, has several implications, most notably that many of the products of $H$ with operators in $\mathfrak{L}^p(\Alg)$ are unambiguously defined thanks to the $\Alg$-module structure of the $\mathfrak{L}^p(\Alg)$ spaces and the Leibniz rule (see Proposition~\ref{framework:prop:Leibniz_rule}~(2) and related comments).

Furthermore, $H \in \Alg$ implies that many of our Hypotheses are automatically verified. Due to the boundedness we know $\domain_{\mathrm{c}}(H) = \domain_{\mathrm{c}}$ holds and so item (i) of Hypothesis~\ref{hypothesis:current} is trivially satisfied. The same goes for item (iv) of Hypothesis~\ref{hypothesis:current}. Items (ii) and (iii) of Hypothesis~\ref{hypothesis:current} just say that the iterated commutators of $H$ with the $X_k$'s are well-defined as bounded operators, and thanks to Lemma~\ref{Kubo_formula:lem:Lp_regularity_states}~(1) we also know that these commutators have to be elements of the algebra $\Alg$. 

In summary, Hypothesis~\ref{hypothesis:current} just states that $H$ has to be a regular (indeed a smooth) element of $\Alg$ with respect the \emph{spatial derivations} induced on $\Alg$ by the generators $X_k$'s (see Definition~\ref{framework:defn:T_compatible_derivation} for more details). Let $\rr{C}^n(\Alg) \subset \Alg$ be the subset of elements which can be derived $n$-times inside the algebra (one can also write $\rr{C}^n(\Alg) \equiv \rr{W}^{n,\infty}(\Alg)$ in agreement with the notation previously introduced). Then, in the tight-binding type setting Hypothesis~\ref{hypothesis:current} can be simply replaced by
\begin{assumption}[Differentiability of $H$]\label{main_results:assumption:differentiability_tight_binding}
	$H \in \rr{C}^n(\Alg)$ for a sufficiently large $n$ (\eg $n = \infty$) and 
	\begin{align*}
		J_{\kappa} = (-1)^{\abs{\kappa}} \, \partial_X^{\kappa}(H) := (-1)^{\abs{\kappa}} \, \partial_{X_1}^{\kappa_1} \circ \cdots \circ \partial_{X_d}^{\kappa_d}(H) 
	\end{align*}
	for all $\kappa \in \N_0^d$ with $\abs{\kappa} \leqslant n$.
\end{assumption}
Hypothesis~\ref{hypothesis:state} can be simplified to the only regularity requirement (i), 
\begin{align*}
	\rho \in \Alg^+ \cap \rr{W}^{1,1}(\Alg) \cap \rr{W}^{1,p}(\Alg)
	. 
\end{align*}
The $H$-regularity described by item (ii) of Hypothesis~\ref{hypothesis:state} just follows because the $\Alg$-module structure of the spaces $\rr{L}^{p}(\Alg)$ and the smoothness of $H$.

Also Hypothesis~\ref{hypothesis:current} is automatically satisfied, and just states that any $J \in \Alg$ is a current-type observable which is sufficiently well-behaved to derive the Kubo's formula. Finally, the extra Assumption~(iv) in the statement of Theorem~\ref{main_results:thm:adiabatic_limit_Kubo_formula} becomes equivalent to $J = \ii [Q_J,H]$ for a $Q_J \in \Alg \cap \rr{L}^{q}(\Alg)$. Finally, \emph{any} spectral projection $P$ of $H$ automatically meets the requirement (iii) in the statement of Theorem~\ref{main_results:thm:Kubo_Streda_formula}.

In conclusion, even though problems described by tight-binding operators do not necessitate the level of generality we insist on here and avoid the main technical stumbling blocks, they are nevertheless covered by the framework we propose. We will revisit this point in the abstract in Chapter~\ref{unified} where we give the procedure to construct von Neumann algebra and the trace per unit volume, which applies to $\Z^d$ $\R^d$; And then in Chapter~\ref{applications} we will talk about specific examples of tight-binding operators. 
%
%
\chapter{Mathematical Framework} 
\label{framework}
The purpose of this Chapter is to introduce all the necessary mathematical notions, and make the remainder of the book self-contained. While none of mathematical objects covered here are new, our goal is to make this work accessible to as many people as possible. Moreover, we feel that many readers will benefit from a discussion of central issues such as measurability which will play crucial roles in the proofs in Chapter~\ref{dynamics} and \ref{Kubo_formula}.

\section{Algebra of observables} 
\label{framework:algebra_observables}
In order to develop a general “operator-theoretic” approach to {linear response theory} first of all we need to settle on a proper definition of \emph{observables}. A common point of view in quantum mechanics is to assume that the set of the relevant observable forms a \emph{von Neumann algebra} of bounded operators on a separable Hilbert space. However, many common observables such as position, momentum and angular momentum are in fact \emph{unbounded}, and to cover those, we need to introduce the notion of \emph{affiliation}. The theory of operators algebras is a well established subject and many monographs are devoted to them: among these we will refer mainly to \cite{Dixmier:C_star_algebras:1977,dixmier-81,kadison-ringrose-83,kadison-ringrose-86,
Bratteli_Robinson:operator_algebras_1:2002,takesaki-02,takesaki-03}.

\subsection{The von Neumann algebra of observables} 
\label{framework:algebra_observables:von_Neumann}
Throughout this monograph the symbol $\Hil$ will denote systematically a complex Hilbert space endowed with a sesquilinear inner product $\sscpro{\, \cdot \,}{\, \cdot \,}$; $\Hil$ need \emph{not} be separable. The ($C^{\ast}$)-algebra of the bounded operators of $\Hil$ will be denote by $\mathscr{B}(\Hil)$. This algebra can be endowed with several topologies. The \emph{uniform} operator topology (UOT) is the metric topology induced by the \emph{operator norm}
\begin{align*}
	\snorm{A}^2 := \sup_{\psi \in \Hil \setminus \{ 0 \}} \frac{\bscpro{\psi}{A^* A \psi}}{\sscpro{\psi}{\psi}} 
	= \sup_{\psi \in \Hil \setminus \{ 0 \}} \frac{\bnorm{A \psi}^2_{\Hil}}{\snorm{\psi}^2_{\Hil}}
	, 
	&&
	A \in \mathscr{B}(\Hil)
	. 
\end{align*}
The \emph{strong} operator topology (SOT) is the locally convex topology induced by the seminorms $A \mapsto \mathrm{s}_\psi(A) := \bnorm{A \psi}_{\Hil}$, as $\psi$ varies in $\Hil$; The \emph{weak} operator topology (WOT) is the locally convex topology induced by the seminorms $A \mapsto \mathrm{w}_\psi(A) := \babs{\scpro{\psi}{A \psi}}$, as $\psi$ varies in $\Hil$. The other topologies are labelled by sequences $\{ \psi_n \}_{n \in \N}$ with the $\ell^2$ property $\sum_{n \in \N} \snorm{\psi_n}_{\Hil}^2 < +\infty$: The \emph{ultra-strong} operator topology (uSOT) and \emph{ultra-weak} operator topology (uWOT) are the locally convex topologies induced by the families of seminorms $A \mapsto \mathrm{us}_{\{ \psi_n \}_{n \in \N}}(A) := \bigl ( \sum_{n \in \N} \bnorm{A \psi_n}_{\Hil}^2 \bigr )^{\nicefrac{1}{2}}$ and $A \mapsto \mathrm{uw}_{\{\psi_n\} , \{\phi_n\}}(A) := \babs{\sum_{n \in \Z} \scpro{\psi_n}{A \phi_n}}$, respectively. 
The notion of convergence of a net (or a sequence) $\bigl \{ A_{\alpha} \bigr \}_{\alpha \in I} \subset \mathscr{B}(\Hil)$ in each of the above topologies is summarized below:
\begin{align*}
	\lim_{\alpha} &\; A_{\alpha} = A
	&&
	\; \Longleftrightarrow \; 
	&& \lim_{\alpha} \; \bnorm{A_{\alpha} - A} =0 
	\\
	\slim_{\alpha} &\; A_{\alpha} = A
	&&
	\; \Longleftrightarrow \; 
	&& \lim_{\alpha} \; \bnorm{(A_{\alpha}-A) \psi}_{\Hil} = 0
	&&
	\forall \psi \in \Hil
	\\
	\wlim_{\alpha} &\; A_{\alpha} = A
	&&
	\; \Longleftrightarrow \; 
	&& \lim_{\alpha} \; \babs{\bscpro{\phi}{(A_{\alpha}-A)\psi}} = 0
	&&
	\forall \phi , \psi \in \Hil
\end{align*}
Ultra-strong and ultra-weak limits are defined analogously. One has that
\begin{align}
	\begin{array}{ccccc}
	\text{UOT} &\ \Rightarrow \ & \text{uSOT} & \ \Rightarrow \ & \text{uWOT} \\
	 & & 	\Downarrow & & 	\Downarrow \\
	 & & \text{SOT} & \ \Rightarrow \ & \text{WOT}
	\end{array}
\end{align}
where $X \Rightarrow Y$ means that the convergence of a net in the topology $X$ implies the convergence also in the topology $Y$. Let us recall some important facts \cite[Part~I, Chapter~3, Section~1]{dixmier-81}: On norm-bounded subsets of $\mathscr{B}(\Hil)$ one has that uSOT $=$ SOT and uWOT $=$ WOT. The operator product $(A,B) \mapsto A \, B$ as a map from $\mathscr{B}(\Hil) \times \mathscr{B}(\Hil)$ to $\mathscr{B}(\Hil)$ fails to be continuous with respect to these topologies. However, if $A$ varies in a bounded subset of $\mathscr{B}(\Hil)$ then $(A,B) \mapsto A \, B$ is continuous with respect to the uSOT. Finally, for each fixed $B \in \mathscr{B}(\Hil)$, the maps $A \mapsto B \, A$ and $A \mapsto A \, B$ are continuous with respect to the uWOT.

Given a subset $\Alg \subseteq \mathscr{B}(\Hil)$ we denote with $\Alg'$ its \emph{commutant}, \ie the set of all bounded operators on $\Hil$ commuting with every operator in $\Alg$. Clearly $\Alg'$ is a Banach algebra of operators containing the \emph{identity} (or \emph{unit}) $\id$.
\begin{definition}[Von Neumann algebra]\label{defi:VN_alg}
	A \emph{von Neumann algebra} on $\Hil$ is a unital $\ast$-subalgebra $\Alg \subseteq \mathscr{B}(\Hil)$ such that 
	\begin{align*}
		(\Alg')' = \Alg
		.
	\end{align*}
	The center of a von Neumann algebra $\Alg$ is defined by $\rr{Z}(\Alg) := \Alg \cap \Alg'$.  $\Alg$ is called a \emph{factor} if it has a trivial center, \ie if $\rr{Z}(\Alg) = \bigl \{ c \id \; \; \vert \; \; c \in \C \bigr \}$.
\end{definition}
The main characterization of a von Neumann algebra is given in terms of the topologies  described above. In fact the celebrated \emph{Bicommutant Theorem} \cite[Theorem~2.4.11]{Bratteli_Robinson:operator_algebras_1:2002} states that $\Alg$ is a von Neumann algebra if and only if it is closed with respect to one (and consequently with respect to all) of the operator topologies, the uWOT, the WOT, the uSOT or the SOT. 

Let us recall that an \emph{orthogonal projection} $P \in \mathscr{B}(\Hil)$ fulfills $P^* = P = P^2$. The set of all the orthogonal projections contained in the von Neumann algebra $\Alg$ will be denoted with $\proj(\Alg)$. The following useful facts are straightforward consequences of Definition~\ref{defi:VN_alg}.
\begin{proposition}\label{framework:prop:properties_vonNeumann_algebras}
	Let $\Alg \subseteq \mathscr{B}(\Hil)$ be a von Neumann algebra. 
	\begin{enumerate}
		\item If $A \in \Alg$ is a selfadjoint operator then $\Alg$ contains all spectral projections of $A$. 
		\item The set $\proj(\Alg)$ is dense in $\Alg$ with respect to the SOT.
		\item An operator $A \in \mathscr{B}(\Hil)$ lies in $\Alg$ if and only if $V \, A \, V^* = A$ for all unitary elements $V \in \Alg'$.
		\item Let $A = U \, \abs{A}$ be the polar decomposition of $A \in \Alg$. Then both $U$ and $\sabs{A}$ are elements of $\Alg$.
	\end{enumerate}
\end{proposition}
%

\subsection{The algebra of affiliated operators} 
\label{framework:algebra_observables:affiliation}
While von Neumann algebras are composed only of \emph{bounded} operators, the notion of \emph{affiliation} allows us to associate \emph{unbounded} operators to a von Neumann algebra. Such a \emph{principle of affiliation} can be seen as an extension of property~(3) in Proposition~\ref{framework:prop:properties_vonNeumann_algebras}.

\begin{definition}[Affiliation]\label{framework:defn:affiliation}
	Let $\Alg\subseteq \mathscr{B}(\Hil)$ be a von Neumann algebra and $A$ a \emph{closed} (not necessarily bounded) operator with dense domain $\domain(A) \subset \Hil$. If for each unitary $ V \in \Alg'$
	\begin{align*}
		V \bigl [ \domain(A) \bigr ] = \domain(A)
		\qquad 
		\text{and} 
		\qquad 
		V \, A \, V^{\ast} = A
		, 
	\end{align*}
	then one says that $A$ is \emph{affiliated} with $\Alg$. The set of all closed and densely defined operators affiliated with $\Alg$ will be denoted by $\affil(\Alg)$.
\end{definition}
By definition $A \in \Alg$ exactly when $A$ is bounded and affiliated to $\Alg$, \ie $\Alg = \affil(\Alg) \cap \mathscr{B}(\Hil)$. The next result provides a generalization of Proposition~\ref{framework:prop:properties_vonNeumann_algebras}. 
\begin{proposition}[{{\cite[Lemma~2.5.8]{Bratteli_Robinson:operator_algebras_1:2002}}}]\label{framework:prop:properties_vonNeumann_algebras_aff}
	Let $\Alg \subseteq \mathscr{B}(\Hil)$ be a von Neumann algebra and $A$ a closed and densely defined operator with polar decomposition $A = U \, \sabs{A}$. Then $A \in \affil(\Alg)$ if and only if $U \in \Alg$ and $\sabs{A} \in \affil(\Alg)$, and in this case all the spectral projections of $\sabs{A}$ lie in $\Alg$. Moreover, if $A = A^*$ is a selfadjoint operator one has that $A \in \affil(\Alg)$ if and only if $f(A) \in \Alg$ for all $f \in L^{\infty}(\R)$.
\end{proposition}
\begin{remark}[Algebraic operations on $\affil(\Alg)$]\label{framework:remark:algebraic_operations}
	Let us recall that the \emph{adjoint} of a densely defined unbounded operator $A$ is the operator $A^*$ with domain $\domain(A^*)$ given by the $\psi \in \Hil$ such that there exists a (necessarily unique) $\phi := A^* \psi$ which verifies the equality $\sscpro{\phi}{\varphi} = \sscpro{\psi}{A \varphi}$ for all $\varphi \in \domain(A)$. The operator $A^*$ turns out to be automatically closed and densely defined (see \cite[Theorem VIII.1]{Reed_Simon:M_cap_Phi_1:1972}). Moreover, starting from  Definition~\ref{framework:defn:affiliation} one can verify that $A \in \affil(\Alg)$ implies $A^* \in \affil(\Alg)$. Also sums and products of unbounded operators can be defined at the cost of some technicalities (see \eg \cite[Section~2.7]{kadison-ringrose-83}). Given a pair $A , B \in \affil(\Alg)$ one can define the \emph{sum operator} $A \overset{\circ}{+} B$ and a \emph{product operator} $A \overset{\circ}{\cdot} B$ as linear operators on $\Hil$ with domains
	\begin{align*}
		\domain \bigl ( A \overset{\circ}{+} B \bigr ) :& \negmedspace= \domain(A) \cap \domain(B)
		, 
		\\
		\domain \bigl ( A \overset{\circ}{\cdot} B \bigr ) :& \negmedspace= \bigl \{ \psi \in \domain(B) \; \; \vert \; \; B \psi \in \domain(A) \bigr \}
		. 
	\end{align*}
	These operations are associative so that expressions like $A_1 \overset{\circ}{+} A_2 \overset{\circ}{+} \ldots \overset{\circ}{+} A_n$ and $A_1 \overset{\circ}{\cdot} A_2 \overset{\circ}{\cdot} \ldots \overset{\circ}{\cdot} A_n$ describe well-defined operators. If the sum $A \overset{\circ}{+} B$ is \emph{closable} and \emph{densely defined}, then we can denote by $A + B$ its closure (\emph{strong sum}). Similarly, if the product $A \overset{\circ}{\cdot} B$ is \emph{closable} and \emph{densely defined}, then we denote by $AB$ its closure (\emph{strong product}). When the strong sum $A + B$ of a pair $A , B \in \affil(\Alg)$ is defined then $A + B \in \affil(\Alg)$. Similarly, if the strong product $A \, B$ is defined then $A \, B \in \affil(\Alg)$. Unfortunately, without extra assumption~on $\Alg$ (see Section~\ref{framework:algebra_observables:finite_vs_semifinite}), the \emph{full} set $\affil(\Alg)$ fails to be a $\ast$-algebra. In fact, it might happen that the domains $\domain \bigl ( A \overset{\circ}{+} B \bigr ) = \{ 0 \}$ or $\domain \bigl ( A\overset{\circ}{\cdot} B \bigr ) = \{ 0 \}$ are trivial even though $\domain(A)$ and $\domain(B)$ are dense; or $A \overset{\circ}{+} B$ and $A \overset{\circ}{\cdot} B$ might not be closable even though $A$ and $B$ are closed and densely defined (see \eg \cite[Exercise 2.8.43]{kadison-ringrose-83}). However, particular subsets of $\affil(\Alg)$ do form $\ast$-algebras; we will discuss important examples in Section~\ref{framework:nc_Lp_Sobolev_spaces:measure}. 
\end{remark}
Finally, let us recall that a linear subspace $\domain_0$ of $\Hil$ is called a \emph{core} for the closed operator $H$ if $H$ agrees with the closure of the restriction $H \vert_{\domain_0}$. 

\subsection{Finite vs{.} semi-finite von Neumann algebras} 
\label{framework:algebra_observables:finite_vs_semifinite}
One important axis along which to distinguish von Neumann algebras is the question of finiteness vs.\ semi-finiteness: Two projections $P_1$ and $P_2$ in a von Neumann algebra $\Alg$ are said to be \emph{equivalent}, written $P_1 \sim P_2$, if there exists a $W \in \Alg$ such that $P_1 = W \, W^*$ and $P_2 = W^* \, W$. This relation $\sim$ in fact defines an \emph{equivalence} relation on $\proj(\Alg)$ \cite[Proposition~6.1.5]{kadison-ringrose-86}. We say that $P_1$ is \emph{weaker} of $P_2$, written $P_1 \preceq P_2$, if there are projection $Q_1 , Q_2 \in \proj(\Alg)$ such that $P_1 \sim Q_1$, $P_2 \sim Q_2$ and $Q_2 - Q_1 \in \proj(\Alg)$ (\ie, $Q_1$ is a subprojection of $Q_2$). The relation $\preceq$ fixes a partial ordering in $\proj(\Alg)$ \cite[Proposition~6.2.4 \& Proposition~6.2.5]{kadison-ringrose-86} with respect to which $\proj(\Alg)$ is a complete lattice. A projection $P \in \Alg$ is said to be \emph{infinite} if it is equivalent to a proper subprojection of itself, namely if it exists a $Q \in \proj(\Alg)$ such that $P - Q \in \proj(\Alg)$ and $P \sim Q$. Otherwise, $P$ is said to be \emph{finite}. The von Neumann algebra $\Alg$ is called \emph{finite} if the unit $\id$ is a finite projection (namely if every isometry in $\Alg$ is a unitary) and is called \emph{semi-finite} if there exists an increasing net of finite projections $\{P_{\alpha}\}\subset \Alg$ such that $P_{\alpha} \to \id$ strongly. Although semi-finite von Neumann algebras are the main object of interest of this work, let us recall the following result which is directly related to the discussion in Remark \ref{framework:remark:algebraic_operations}.
\begin{theorem}[{{\cite[Theorem~3.11]{zhe-11}}}]\label{framework:thm:finite_von_Neumann_algebra_affiliated_algebra_star_algebra}
	If $\Alg \subseteq \mathscr{B}(\Hil)$ is a \emph{finite} von Neumann algebra then $\affil(\Alg)$ is a $\ast$-algebra with respect to the adjoint, the strong sum and the strong product.
\end{theorem}
%

\section{Non-commutative $\mathfrak{L}^p$-spaces} 
\label{framework:nc_Lp_Sobolev_spaces}
Von Neumann algebras equipped with a \emph{a faithful normal semi-finite} (f.n.s.) trace provide the setting for a non-commutative version of \emph{integration theory}. This line of research was initiated by Segal in \cite{segal-53} and subsequently developed by many other authors \cite{Nelson:noncommutative_integration:1974,Terp:noncommutative_Lp_spaces:1981,yeadon-73,fack-kosaki-86} (see also \cite[Chapter IX]{takesaki-03} for a detailed exposition).

\subsection{F.n.s.\ trace state} 
\label{framework:nc_Lp_Sobolev_spaces:fns_trace}
Let $\Alg\subseteq\mathscr{B}(\Hil)$ be a von Neumann algebra and $\Alg^+\subset \Alg$ be the convex cone of \emph{positive} elements of $\Alg$. Let us recall that $\Alg^+$ is endowed with a natural \emph{order relation}: One says that $A,B\in \Alg^+$ are in the relation $A\geqslant B$ if and only if $A-B\in \Alg^+$. The order allows us to consider increasing nets with a “sup” element. The possibility of developing the non-commutative analogs of $L^p$ theory on $\Alg$ is based on the existence of objects of the following type \cite[Chapter VII, Definition~1.1]{takesaki-03}:
\begin{definition}[F.n.s.\ trace]\label{framework:defn:fns_trace}
	A \emph{weight} on $\Alg$ is a map $\mathcal{T} : \Alg^+ \longrightarrow [0 , +\infty]$ satisfying
	\begin{align*}
		\mathcal{T} \bigl ( \lambda A + \mu B \bigr ) = \lambda \, \mathcal{T}(A) + \mu \, \mathcal{T}(B)
		,
		&&
		A , B \in \Alg^+
		, \; 
		\lambda , \mu \in \R_+
	\end{align*}
	with the convention $0 \cdot (+\infty) = 0$. 
	\begin{enumerate}
		\item[(f)] The weight $\mathcal{T}$ is called \emph{faithful} if $\mathcal{T}(A)=0$ if and only if $A = 0$.
		\item[(n)] The weight $\mathcal{T}$ is called \emph{normal} if for each increasing bounded net $\{A_{\alpha}\}\subset \Alg^+$ such that $A:=\sup_\alpha(A_{\alpha}) \in \Alg^+$,
		one has that $\sup_\alpha\mathcal{T}(A_{\alpha}) = \mathcal{T}(A)$.
		\item[(s)] The weight $\mathcal{T}$ is called \emph{semi-finite} if, 
		\begin{align*}
			\Alg^+_{\mathcal{T}} := \bigl \{ A \in \Alg^+ \; \; \vert \; \; \mathcal{T}(A) < +\infty \bigr \}
		\end{align*}
		is dense in $\Alg$ with respect to the uWOT (and is called \emph{finite} if $\Alg^+_{\mathcal{T}} = \Alg^+$).
		%
	\end{enumerate}
	A weight which fulfills the above conditions is called a \emph{faithful normal semi-finite} weight or \emph{f.n.s.} weight for short. When a weight fulfills 
	\begin{enumerate}
		\item[(t)] $\mathcal{T} \bigl (A \, A^* \bigr ) = \mathcal{T} \bigl (A^* \, A \bigr )$ for all $A \in \Alg$, 
	\end{enumerate}
	then it is called \emph{trace}.
\end{definition}
Every von Neumann algebra admits a f.n.s.\ weight \cite[Chapter VII, Theorem~2.7]{takesaki-03} but only semi-finite von Neumann algebras can admit a f.n.s.\ trace \cite[Part~I, Chapter~6, Proposition~9]{dixmier-81}. Put differently, the existence of a f.n.s.\ trace is more restrictive than that of f.n.s.\ weight. The normality property (n) in Definition~\ref{framework:defn:fns_trace} can be replaced by several equivalent properties which are listed, for instance, in \cite[Theorem~2.7.11]{Bratteli_Robinson:operator_algebras_1:2002} or in \cite[Part~I, Chapter~4, Theorem~1]{dixmier-81}. In particular, it is important to notice that the normality is equivalent to \emph{ultra-weak continuity} of $\mathcal{T}$, namely
\begin{align*}
  \mathrm{uw}-\lim_{\alpha} \; A_{\alpha} = A
  	\qquad \Longrightarrow 	\qquad \lim_{\alpha} \; \mathcal{T}(A_{\alpha}) = \mathcal{T}(A).
\end{align*}
Finally, let us mention that the trace property (t) can be equivalently stated in the following form: $\mathcal{T}(A) = \mathcal{T} \bigl ( U \, A \, U^* \bigr )$ for each $A \in \Alg^+$ and for all unitary operators $U \in \Alg$ \cite[Part~I, Chapter~6, Corollary~1]{dixmier-81}.

Henceforth, we will assume that $\Alg \subseteq \mathscr{B}(\Hil)$ is a (necessarily) semi-finite von Neumann algebra equipped with a f.n.s.\ trace $\mathcal{T}$. Along with $\Alg^+_{\mathcal{T}}$ let us introduce also the subset
\begin{align*}
	\mathscr{J}_{\mathcal{T}} := \bigl \{ A \in \Alg \; \; \vert \; \; \mathcal{T} \bigl (A^* \, A \bigr ) < +\infty \bigr \}
	.
\end{align*}
We recall some standard facts which will be useful in what follows: 
\begin{enumerate}
	\item The set $\Alg^+_{\mathcal{T}}$ is a hereditary convex subcone of $\Alg^+$.
	\item The set $\mathscr{J}_{\mathcal{T}}$ is a two-sided $\ast$-ideal of $\Alg$.
	\item Let $\Alg_{\mathcal{T}}$ be the complex linear span of $\Alg^+_{\mathcal{T}}$. Then $\Alg_{\mathcal{T}}$ is a two-sided $\ast$-ideal of $\Alg$ and
	\begin{align*}
		\Alg_{\mathcal{T}} = \bigl \{ A = B^* \, C \; \; \vert \; \; B , C \in \mathscr{I}_{\mathcal{T}} \bigr \}
		,
		&&
		\Alg_{\mathcal{T}} \cap \Alg^+ = \Alg^+_{\mathcal{T}}
		. 
	\end{align*}
	\item $\mathcal{T}$ extends to a linear functional on $\Alg_{\mathcal{T}}$ which is called, for this reason, the \emph{domain of definition} of the trace $\mathcal{T}$. Moreover, $\mathcal{T}(A \, B)=\mathcal{T}(B \, A)$ for all $A \in \Alg_{\mathcal{T}}$ and $B \in \Alg$. 
	\end{enumerate}
For a proof of this facts we refer the reader to \cite[Chapter~VII, Lemma~1.2]{takesaki-03} and \cite[Part~I, Chapter~6, Proposition~1]{dixmier-81}. 

\subsection{Convergence in measure and measurable operators} 
\label{framework:nc_Lp_Sobolev_spaces:measure}
For each $P \in \proj(\Alg)$ we denote with $P^{\perp} := \id - P \in \proj(\Alg)$ its \emph{orthogonal} complement.
\begin{definition}[$\mathcal{T}$-measurable operators {\cite{Nelson:noncommutative_integration:1974,Terp:noncommutative_Lp_spaces:1981}}]\label{framework:defn:measurable_operators}
	Let $\Alg \subseteq \mathscr{B}(\Hil)$ be a von Neumann algebra equipped with a f.n.s.\ trace $\mathcal{T}$. An element $A \in \affil(\Alg)$ with domain $\domain(A)$ is called \emph{$\mathcal{T}$-measurable} if for each $\delta>0$ there exists a projection $P \in \proj(\Alg)$ such that
	\begin{align}
		P[\Hil] \subseteq \domain(A)
		\quad
		\mbox{and}
		\quad
		\mathcal{T}(P^{\perp}) \leqslant \delta
		.
		\label{framework:eqn:T_measurable_operator_notion_density}
	\end{align}
	We denote with $\rr{M}(\Alg)$ the set of $\mathcal{T}$-measurable operators. A domain $\domain(A)\subset\Hil$ which verifies condition \eqref{framework:eqn:T_measurable_operator_notion_density} is called \emph{$\mathcal{T}$-dense}.
\end{definition}
 When $A \in \affil(\Alg)$ is measurable according to Definition~\ref{framework:defn:measurable_operators}, then the operator $A \, P$ is everywhere defined and closed, hence bounded by the \emph{closed graph theorem} \cite[Theorem III.12]{Reed_Simon:M_cap_Phi_1:1972}. This in turn implies that $A \, P \in \Alg$. For $\epsilon , \delta > 0$ let
\begin{align*}
	N(\epsilon, \delta) := \Bigl \{ A \in \rr{M}(\Alg) \; \; \big \vert \; \; \exists \, P \in \proj(\Alg) \; \text{such that} \; \snorm{A \, P} \leqslant \epsilon \; \text{and} \; \mathcal{T}(P^{\perp}) \leqslant \delta \Bigr \}
	.
\end{align*}
The most important characterization of the subset $\rr{M}(\Alg) \subset \affil(\Alg)$ of $\mathcal{T}$-measurable operators is provided by the following theorem.
\begin{theorem}[{{\cite{Nelson:noncommutative_integration:1974,Terp:noncommutative_Lp_spaces:1981}}}]\label{framework:thm:measure_topology}
	$\rr{M}(\Alg)$ is a $\ast$-algebra with respect to the usual adjoint, \emph{strong} sum and \emph{strong} product of unbounded operators. The collection of sets
	\begin{align*}
		\mathscr{N}_{\epsilon, \delta}(A) := \bigl \{ A + B \; \; \big \vert \; \; B \in N(\epsilon, \delta) \bigr \}
	\end{align*}
	labelled by $\epsilon , \delta > 0$ and $A \in \rr{M}(\Alg)$, form a basis for a topology on $\rr{M}(\Alg)$, called \emph{measure topology}, that turns $\rr{M}(\Alg)$ into a topological algebra. With respect to this topology $\rr{M}(\Alg)$ is a complete, first countable, Hausdorff $\ast$-algebra and $\Alg\subset \rr{M}(\Alg)$ is a dense subalgebra.
\end{theorem}
Let us establish some relevant facts about the notion of
$\mathcal{T}$-density which will be used many times throughout this book. First of all, it is possible to prove that a $\mathcal{T}$-dense domain is in particular dense \cite[Corollary~11]{Terp:noncommutative_Lp_spaces:1981}. Secondly, if $\domain$ is a $\mathcal{T}$-dense domain and $U$ is an unitary operator such that $U \, \Alg \, U^* = \Alg$ and $\mathcal{T}(U \, A \, U^*)=\mathcal{T}(A)$ for all $A\in\Alg$, then also $U[\domain]$ is $\mathcal{T}$-dense (to verify this one has only to conjugate the projections which enter in the Definition~\ref{framework:defn:measurable_operators} by $U$). The final, and extremely important property is that the intersection of two  $\mathcal{T}$-dense domains is still $\mathcal{T}$-dense. This property, essentially proved in \cite[Proposition~5~(i)]{Terp:noncommutative_Lp_spaces:1981}, is at the basis of the fact that the set $\rr{M}(\Alg)$ is closed under the sum.

According to \cite[Definition~1.3]{fack-kosaki-86} to each $A \in \rr{M}(\Alg)$ we can associate a \emph{distribution function} $\epsilon \mapsto \lambda_{\epsilon}(A)$ defined by
\begin{align}
	\lambda_{\epsilon}(A) := \mathcal{T} \bigl ( \chi_{(\epsilon,+\infty)}(\sabs{A}) \bigr )
	,
	&&
	\epsilon \geqslant 0
	, 
	\label{framework:eqn:distribution_function}
\end{align}
where $\chi_{(\epsilon,+\infty)}$ is the characteristic function of the open interval $(\epsilon,+\infty)$. The operator $A$ being $\mathcal{T}$-measurable, we have $\lambda_{\epsilon}(A)<+\infty$ for $\epsilon$ large enough and $\lim_{\epsilon \to +\infty} \lambda_{\epsilon}(A) = 0$. Moreover, the function $\epsilon \mapsto \lambda_{\epsilon}(A)$ is non-increasing and continuous from the right. The \emph{distribution function} provides a useful description for the basis of the measure topology {\cite[Lemma~7]{Terp:noncommutative_Lp_spaces:1981}}:
\begin{align}
	\mathscr{N}_{\epsilon, \delta}(A) = \bigl \{ B \in \rr{M}(\Alg) \; \; \vert \; \; \lambda_{\epsilon}(B-A) \leqslant \delta \bigr \}
	.
	\label{framework:eqn:N_eps_mu_distribution_function}
\end{align}
Finally, the distribution function can be used to define the (generalized) \emph{$\delta$-singular numbers} of a $\mathcal{T}$-measurable operator $A \in \rr{M}(\Alg)$ \cite[Proposition~2.2 \& Remark~2.3]{fack-kosaki-86}:
\begin{align*}
	\mu_{\delta}(A):= \inf_{\epsilon \geqslant 0} \bigl \{ \lambda_{\epsilon}(A) \leqslant \delta \bigr \} 
	= \inf_{\substack{P \in \proj(\Alg) \\ \mathcal{T}(P^{\perp}) \leqslant \delta}} \, \bigl \{\snorm{AP} \big\}
	.
\end{align*}
The function $\delta \mapsto \mu_{\delta}(A)$ is usually called \emph{decreasing rearrangement} \cite{dodds-dodds-pagter-93}.
\begin{remark}
	It is well known that the measure topology on $\rr{M}(\Alg)$ is metrizable \cite{yeadon-73} and so, as a corollary of Theorem~\ref{framework:thm:measure_topology}, one has that $\rr{M}(\Alg)$ is a \emph{Fréchet $\ast$-algebra}. An example of a metric is provided by the Fréchet norm 
	\begin{align*}
		\rho_{\mathcal{T}}(A) := \inf_{\substack{P \in \proj(\Alg) \\ AP \in \Alg}} \; \max \bigl \{\snorm{AP} \, , \, \mathcal{T}(P^{\perp}) \bigr \}
		.
	\end{align*}
\end{remark}
The following result will be used many times: 
\begin{proposition}[{{\cite[Proposition 12]{Terp:noncommutative_Lp_spaces:1981}}}]\label{framework:prop:agreement_operators_T_dense_set} 
	Let $A_1 , A_2 \in \affil(\Alg)$ such that there is a $\mathcal{T}$-dense set $\mathcal{E} \subset \domain(A_1) \cap \domain(A_2)$. Assume that the two operators agree when restricted to $\mathcal{E}$, \ie $A_1 |_{\mathcal{E}} = A_2 |_{\mathcal{E}}$ Then $A_1 , A_2\in\rr{M}(\Alg)$ and $A_1 = A_2$.
\end{proposition}
\begin{example}\label{framework:example:T_measurable_operators}
	Here are some significant examples of $\ast$-algebras of $\mathcal{T}$-measurable operators:
	\begin{enumerate}
		\item \emph{Commutative case.} Let $(\Omega,\mu)$ be a measure space and consider the commutative von Neumann algebra $L^{\infty}(\Omega,\mu)$ equipped with the trace $f \mapsto \int_{\Omega} f \, \dd\mu$. A function $f : \Omega \longrightarrow \C$ is in $\rr{M} \bigl ( L^{\infty}(\Omega,\mu) \bigr )$ if and only if it is a $\mu$-measurable function which is bounded except on a set of finite measure. Thus, $\rr{M} \bigl ( L^{\infty}(\Omega,\mu) \bigr )$ is large enough to contain all the classical $L^p$-spaces for $0 < p \leqslant \infty$. Also, $\rr{M} \bigl ( L^{\infty}(\Omega,\mu) \bigr )$ is the closure of $L^{\infty}(\Omega,\mu)$ with respect to the measure topology \cite[Section~22]{halmos-74} or \cite[Section~2.4]{folland-99}.
		\item \emph{Full algebra.} If $\Alg$ is the full von Neumann algebra $\mathscr{B}(\Hil)$ and $\mathcal{T}$ is the usual trace $\trace_{\Hil}$, then $\rr{M} \bigl ( \mathscr{B}(\Hil) \bigr ) = \mathscr{B}(\Hil)$. This is a consequence of \cite[Proposition~21 (iv)]{Terp:noncommutative_Lp_spaces:1981} and that $\trace_{\Hil}(P) < 1$ implies $P = 0$ if $P$ is a projection.
		\item \emph{Finite trace.} If the trace $\mathcal{T}$ is \emph{finite}, namely $\mathcal{T}(\id) < +\infty$, then $\rr{M}(\Alg) = \affil(\Alg)$ (see also \cite[Proposition~21~(vi)]{Terp:noncommutative_Lp_spaces:1981}). This result is connected to the fact that the existence of a finite trace for $\Alg$ implies that $\Alg$ is finite \cite[Part~I, Chapter~6, Proposition~9]{dixmier-81} and for a finite von Neumann algebra Theorem~\ref{framework:thm:finite_von_Neumann_algebra_affiliated_algebra_star_algebra} holds. 
		\item \emph{The framework described in Chapter~\ref{unified}.} Here, a von Neumann algebra and a trace per unit volume are constructed from an ergodic topological dynamical system and twisting $2$-cocycle. This covers the standard cases where the Hilbert spaces are either $\ell^2(\Z^d) \otimes \C^N$ or $L^2(\R^d) \otimes \C^N$. 
	\end{enumerate}
\end{example}
\begin{remark}[Non-$\mathcal{T}$-measurable Hamiltonians]\label{framework:example:non_T_measurable_operators}
	In many physically relevant situations the trace $\mathcal{T}$ is only \emph{semi}-finite and the $\ast$-algebra $\rr{M}(\Alg)$ turns out to be strictly smaller than $\affil(\Alg)$. On the other hand, the dynamical properties of the system are usually described by operators in $\affil(\Alg)$, and $\rr{M}(\Alg)$ turns out to be too small to develop a dynamical theory. A typical example is given by the von Neumann algebra $\Alg_{\mathrm{per}} \subset \mathscr{B} \bigl ( L^2(\R^d) \bigr )$ of operators which are invariant under $\Z^d$ lattice translations. In this case integration is provided by the \emph{trace per unit volume} $\mathcal{T}$, which is only semi-finite but not finite. Typically, the dynamics of the system are described by \emph{unbounded} Hamiltonians of the form $H := -\Delta + V_{\mathrm{per}}$ where $\Delta$ is the Laplacian operator and $V_{\mathrm{per}}$ is a $\Z^d$-periodic multiplication operator. Under mild technical assumptions on the potential $H \in \affil(\Alg_{\mathrm{per}})$ is affiliated, because spectral projections are $\Z^d$-periodic. However, \emph{such $H$ usually fail to be $\mathcal{T}$-measurable}. Let us consider, for instance, the free Hamiltonian $H_0 := -\Delta$. This operator is non-negative, hence $H_0 = \sabs{H_0}$. The trace per unit volume of the \emph{Fermi} projection $P_E := \chi_{(0,E)}(H_0) = \chi_{(0,E)}(\sabs{H_0})$ just provides the \emph{integrated density of states} 
	\begin{align*}
		\mathscr{N}(E) = \mathcal{T} \bigl ( \chi_{(0,E)}(\sabs{H_0}) \bigr ) = C_d \, E^{\nicefrac{d}{2}}
		,
		&&
		E \in [0, \infty)
		. 
	\end{align*}
	The last equation shows that {distribution function} \eqref{framework:eqn:distribution_function} for $H_0$ is always divergent, \ie, $\epsilon \mapsto \lambda_{\epsilon}(H_0) = +\infty$, which proves that $H_0 \notin \rr{M}(\Alg_{\mathrm{per}})$. The same deduction holds true also for perturbed operators $H = -\Delta + V$ with $V \in L^{\infty}(\R^d)$ (just an application of the \emph{minimax principle}).
\end{remark}
%

\subsection{Integration and $\mathfrak{L}^p$-spaces} 
\label{framework:nc_Lp_Sobolev_spaces:Lp_spaces}
Non-commutative $\mathfrak{L}^p$-spaces are defined akin to $p$-Schatten class operators in functional analysis: Given a positive selfadjoint operator $A \in \affil(\Alg)$ we set 
\begin{align}
	\mathcal{T}(A) := \sup_{n \in \N} \; \mathcal{T} \left ( \int_0^n \lambda \, \dd \mathds{E}(\lambda) \right ) 
	= \int_0^{\infty} \lambda \, \mathcal{T} \bigl (\dd\mathds{E} (\lambda) \bigr )
	\label{framework:eqn:definition_trace_via_functional_calculus}
\end{align}
where $A = \int_0^{\infty} \lambda \, \dd \mathds{E}(\lambda)$ is the spectral representation of $A$. For each $0 < p < \infty$ we can define
\begin{align*}
	\mathfrak{L}^p(\Alg) := \Bigl \{ A \in \affil(\Alg) \; \; \big \vert \; \; \snorm{A}_p := \mathcal{T} \bigl ( \sabs{A}^p \bigr )^{\nicefrac{1}{p}}< +\infty \Bigr \}
	.
\end{align*}
The intersection $\mathfrak{L}^1(\Alg) \cap \Alg = \Alg_{\mathcal{T}}$ coincides with the definition two-sided ideal $\Alg_{\mathcal{T}}$ introduced in Section~\ref{framework:nc_Lp_Sobolev_spaces:fns_trace}. Moreover, since $\Alg_{\mathcal{T}}$ is an ideal, one as that $\Alg_{\mathcal{T}} \subseteq \mathfrak{L}^p(\Alg) \cap \Alg$ for all $1 \leqslant p < \infty$. Finally, it is common to identify the von Neumann algebra $\Alg$ with $\mathfrak{L}^{\infty}(\Alg)$.

The most important facts concerning the structure of the spaces $ \mathfrak{L}^p(\Alg)$ are summarized below.
\begin{theorem}[{{\cite{Nelson:noncommutative_integration:1974,Terp:noncommutative_Lp_spaces:1981,takesaki-03}}}]\label{framework:thm:measurability_products_in_Lp}
	Let $1 \leqslant p \leqslant \infty$. Each $\mathfrak{L}^p(\Alg)$ endowed with the norm $\norm{\cdot}_p$ is a Banach space which is continuously embedded in $\rr{M}(\Alg)$. As a consequence the $\mathfrak{L}^p(\Alg)$ can be identified with
	\begin{align*}
		\mathfrak{L}^p(\Alg) = \bigl \{ A \in \rr{M}(\Alg) \; \; \vert \; \; \snorm{A}_p < +\infty \bigr \}
		. 
	\end{align*}
	Each $\mathfrak{L}^p(\Alg)$ is a selfadjoint (weak Fréchet) $\Alg$-bimodule with an associative product
	\begin{align*}
		B_1 \, A \, B_2 := B_1 \, (A \, B_2) = (B_1 \, A) \, B_2
		,
		&&
		A \in \mathfrak{L}^p(\Alg)
		, \; 
		B_1 , B_2 \in \Alg
		, 
	\end{align*}
	compatible with the adjoint operation
	\begin{align*}
		\bigl ( B_1 \, A \, B_2 \bigr )^* = B_2^* \; A^* B_1^*
		.
	\end{align*}
	The adjoint operation is isometric $\norm{A^*}_p = \snorm{A}_p$ and the norm bound
	\begin{align}
		\bnorm{B_1 \, A \, B_2}_p \leqslant \snorm{B_1} \, \snorm{B_2} \, \snorm{A}_p
		,
		&& 
		A \in \mathfrak{L}^p(\Alg)
		, \; 
		B_1 ,B_2 \in \Alg
		\label{framework:eqn:Lp_norm_estimate_double_product}
	\end{align}
	holds. Finally, the ideal $\Alg_{\mathcal{T}}$ is dense in each of the $\mathfrak{L}^p(\Alg)$.
\end{theorem}
\begin{remark}[Hilbert space structure on $\mathfrak{L}^2(\Alg)$]\label{framework:remark:abstract_Hilbert_Schmidt_space}
	We point out that the sesquilinear form
	\begin{align*}
		\sdscpro{A}{B}_{\mathfrak{L}^2}
		:= \mathcal{T}(A^* \, B)
		,
		&&
		A , B \in \mathfrak{L}^2(\Alg)
	\end{align*}
	defines an Hilbert space structure on $\mathfrak{L}^2(\Alg)$. For this reason elements in $\mathfrak{L}^2(\Alg)$ can be called \emph{$\mathcal{T}$-Hilbert-Schmidt} operators.
\end{remark}
\begin{remark}[Symmetric operator spaces]\label{framework:remark:symmetric_spaces}
	The spaces $\mathfrak{L}^p(\Alg)$, $1\leqslant p\leqslant \infty$ provide special examples of non-commutative (fully) \emph{symmetric Banach function spaces} in the sense of \cite{dodds-dodds-pagter-89,pagter-07}. These spaces are associated by means of the {decreasing rearrangement} function $\delta\mapsto \mu_{\delta}$ to the classical spaces $L^p(0,+\infty)$ which are commutative symmetric (or rearrangement-invariant) Banach function spaces \cite{krein-petunin-semenov-82}. In particular, the spaces $L^p(0,+\infty)$ are \emph{separable} (see \eg \cite[Theorem~4.13]{brezis-11}), and this fact allows us to apply the results of \cite{Pagter_Sukochev:commutator_estimates_R_flows:2007} in the next sections. Let us finally observe that also the intersections $L^p(0,+\infty)\cap L^q(0,+\infty)$ ($p<q$) are symmetric Banach spaces with respect to the \emph{max-norm} $\norm{\cdot}_{p \cap  q} := \max \bigl \{ \norm{\cdot}_{p} , \norm{\cdot}_{q} \bigr \}$ \cite[Part~II, Lemma~4.5]{krein-petunin-semenov-82}. This fact and \cite[Theorem~4.5]{dodds-dodds-pagter-89} imply that also $\mathfrak{L}^p(\Alg)\cap\mathfrak{L}^q(\Alg)$ are non-commutative Banach function spaces with respect to the related max-norm $\norm{\cdot}_{p \cap q}$.
\end{remark}
As explained in \cite[Chapter~IX, Lemma~2.12~(iii)]{takesaki-03} or in \cite[Section~3]{dodds-dodds-pagter-93} the trace $\mathcal{T}$ can be naturally extended to the positive cone of $\rr{M}(\Alg)$. In this way one obtains a continuous extension from the definition ideal $\Alg_{\mathcal{T}}$ to $\mathfrak{L}^1(\Alg)$ and the bilinear form $\mathfrak{L}^1(\Alg) \times \Alg \ni (A,B) \mapsto \mathcal{T}(AB) \in \C$ identifies $\mathfrak{L}^1(\Alg)$ with the pre-dual of $\Alg$ \cite[Chapter IX, Lemma~2.12 (iii)]{takesaki-03}. Standard density arguments apply to \cite[Part~I, Chapter~6, Corollary~1]{dixmier-81} allow to deduce the trace property
\begin{align*}
	\mathcal{T}(A \, B) = \mathcal{T}(B \, A)
	,
	&&
	A \in \mathfrak{L}^1(\Alg)
	, \; 
	B \in \Alg
	.
\end{align*}
This kind of results can be extended to other $p \neq 1$ by means of the non-commutative analog of the Hölder inequality which states that \cite[Theorem~4.2]{fack-kosaki-86}
\begin{align}
	\snorm{AB}_r \leqslant \snorm{B}_q \, \snorm{A}_p
	,
	&&
	A \in \mathfrak{L}^p(\Alg)
	, \; 
	B \in \mathfrak{L}^q(\Alg) 
	, 
	\label{framework:eqn:generic_Hoelder_inequality}
\end{align}
where $p , q , r > 0$ are related by $p^{-1}+q^{-1}=r^{-1}$. In turn this implies that the product $A \, B$ is an element of $\mathfrak{L}^r(\Alg)$. The most interesting case occurs when $r = 1$: 
\begin{theorem}[{{\cite{Nelson:noncommutative_integration:1974,Terp:noncommutative_Lp_spaces:1981,takesaki-03}}}]\label{theo_PQ_ineq}
	Let $1 \leqslant p , q \leqslant \infty$ be such that $p^{-1}+q^{-1}=1$. The trace $\mathcal{T}$ defines bilinear forms
	\begin{align*}
		\mathfrak{L}^p(\Alg) \times \mathfrak{L}^q(\Alg) \ni (A,B) \longmapsto \mathcal{T}(AB) \in \C
	\end{align*}
	which make $\mathfrak{L}^p(\Alg)$ and $\mathfrak{L}^q(\Alg)$ dual spaces of each other. Moreover, the estimate
	\begin{align}
		\babs{\mathcal{T}(A \, B)} \leqslant \snorm{A \, B}_1 \leqslant \snorm{B}_q \, \snorm{A}_p
		,
		&&
		A \in \mathfrak{L}^p(\Alg)
		, \, 
		B \in \mathfrak{L}^q(\Alg)
		, 
		\label{framework:eqn:trace_Lp_Lq_estimate}
	\end{align}
	and the trace property
	\begin{align}
		\mathcal{T}(A \, B) = \mathcal{T}(B \, A)
		,
		&&
		A \in \mathfrak{L}^p(\Alg)
		, \; 
		B \in \mathfrak{L}^q(\Alg)
		, 
		\label{framework:eqn:trace_property}
	\end{align}
	hold.
\end{theorem}
The trace property \eqref{framework:eqn:trace_property} can be deduced from the density of $\Alg_{\mathcal{T}}$ in each space $\mathfrak{L}^p(\Alg)$. It is useful to recall the following representation for the $p$-norms \cite[Proposition~24]{Terp:noncommutative_Lp_spaces:1981}
\begin{align*}
	\snorm{A}_p = \sup \Bigl \{ \mathcal{T}(AB) \; \; \big \vert \; \; B \in \mathfrak{L}^q(\Alg), \; \snorm{B}_q \leqslant 1 , \; p^{-1}+q^{-1}=1 \Bigr \}
	, 
	 \end{align*}
which immediately implies:
\begin{lemma}\label{framework:lem:weak_zero_vector_zero}
	Let $A \in \mathfrak{L}^p(\Alg)$ with $1 \leqslant p \leqslant \infty$ and suppose that $\mathcal{T}(A \, B) = 0$ holds true for all $B \in \mathfrak{L}^q(\Alg)$ with $p^{-1}+q^{-1}=1$. Then $A = 0$ necessarily vanishes.
\end{lemma}
The Hölder inequality~\eqref{framework:eqn:generic_Hoelder_inequality} can be used to prove the non-commutative version of the \emph{log-convexity} of the $L^p$-norms. More precisely, for any $A \in \mathfrak{L}^p(\Alg) \cap \mathfrak{L}^q(\Alg)$ with $0 < p \leqslant q$, the $r_{\theta}$ norm can be estimated by interpolation, $\snorm{A}_{r_{\theta}} \leqslant \snorm{A}_{p}^{1-\theta} \snorm{A}_{q}^{\theta}$, where $0 \leqslant \theta \leqslant 1$ and $r_{\theta} := {p \, q} \, \bigl ( \theta p + (1 - \theta) q \bigr )^{-1}$. In particular this implies that
\begin{align}
	\mathfrak{L}^p(\Alg) \cap \mathfrak{L}^q(\Alg)\subset \mathfrak{L}^r(\Alg)
	&&
	\forall r \in [p,q]
	.
	\label{framwork:eqn:interpolation_Lp_Lq_norms}
\end{align}
The last result of this section concerns a peculiar property of convergence for $\mathcal{T}$-measurable projections. First of all, notice that an element $P \in \proj(\Alg)$ is in $\mathfrak{L}^1(\Alg)$ if and only if it is in $\mathfrak{L}^r(\Alg)$ for every $r > 0$. This just follows from the equality $P = \sabs{P} = \sabs{P}^r$ which can be trivially derived from functional calculus and implies $\norm{P}_1 = \norm{P}_r^r$ for every $r > 0$. To decide whether sequences of projections converge to projections, considering a single $\mathfrak{L}^p$ space does not suffice: 
\begin{lemma}\label{framework:lem:project_limit}
	Let $\bigl \{ P_{\alpha} \bigr \} \subset \proj(\Alg) \cap \mathfrak{L}^p(\Alg)\cap \mathfrak{L}^q(\Alg)$, $1 \leqslant p < 2p \leqslant q < \infty$, be a net of projections converging to $P$ in the {max-norm} $\norm{\cdot}_{p \cap q}$. Assume also $\bnorm{P_{\alpha}}_{q} \leqslant K$ for all $\alpha$. Then $P \in \proj(\Alg)$ and for all $r > 0$ the net $\{ P_{\alpha} \}$ converges to $P$ in $\mathfrak{L}^r(\Alg)$. 
\end{lemma}
\begin{proof}
	From \eqref{framwork:eqn:interpolation_Lp_Lq_norms} one has that $P_{\alpha}$ converges to $P$ in the topology of $\mathfrak{L}^{s}(\Alg)$ for all $s \in [p,q]$. Moreover, by the isometry of the adjoint operation (\cf Theorem~\ref{framework:thm:measurability_products_in_Lp}) and the uniqueness of the limit we deduce from
	\begin{align*}
		\bnorm{P^* - P_{\alpha}}_{s} &= \bnorm{P^* - P_{\alpha}^*}_{s}
		= \bnorm{P - P_{\alpha}}_{s} \longrightarrow 0 
		\,\qquad\quad \forall\ s\in[p,q]
	\end{align*}
	that $P$ is automatically selfadjoint on all the spaces $\mathfrak{L}^{s}(\Alg)$ with $s \in [p,q]$. An application of the Hölder inequality leads to
	\begin{align}
		\bnorm{P^2 - P_{\alpha}}_{p} &\leqslant \bnorm{P \, \bigl ( P - P_{\alpha} \bigr )}_{p} + \bnorm{\bigl ( P - P_{\alpha} \bigr ) \, P_{\alpha}}_{p} 
		\notag \\
		&\leqslant \bigl ( \bnorm{P}_{2p} + \bnorm{P_{\alpha}}_{2p} \bigr ) \, \bnorm{P - P_{\alpha}}_{2p} 
		\notag \\
		&\leqslant \bigl ( \bnorm{P}_{2p} + K^{\frac{q}{2p}} \bigr ) \, \bnorm{P - P_{\alpha}}_{2p} 
		\label{framework:eqn:estimate_sequence_projections}
	\end{align}	
	where we used $\bnorm{P_{\alpha}}_{2p}^{2p} = \bnorm{P_{\alpha}}_{q}^{q} \leqslant K^q$. Thus, $P_{\alpha}$ converges to to $P^2$ in $\mathfrak{L}^p(\Alg)$ and the uniqueness of the limit ensures $P^2 = P$ as elements of $\mathfrak{L}^p(\Alg)$ and so of $\mathfrak{L}^{s}(\Alg)$. Any $P \in \mathfrak{L}^{s}(\Alg)$ with $P^2 = P = P^*$ is necessarily an element of $\proj(\Alg)$: by definition $P$ is affiliated to $\Alg$, selfadjoint and non-negative. Hence, functional calculus yields $\spec(P) \subseteq \{ 0 , 1 \}$, and $P$ is indeed a \emph{bounded} operator. The following inequality 
	\begin{align*}
		\bnorm{P - P_{\alpha}}_r^r &= \mathcal{T} \Bigl ( \babs{P - P_{\alpha}}^{r-p} \, \babs{P - P_{\alpha}}^p \Bigr ) 
		\\
		&\leqslant \bnorm{P - P_{\alpha}}^{r-p} \, \bnorm{P - P_{\alpha}}_p^p 
		\leqslant 2^{r-p} \, \bnorm{P - P_{\alpha}}_p^p 
		\longrightarrow 0 
		, 
	\end{align*}
	shows that $P_{\alpha}$ converges to $P$ in every space $\mathfrak{L}^{r}(\Alg)$ with $r \geqslant p$. Moreover, by replacing $p$ with $\nicefrac{p}{2}$ in inequality~\eqref{framework:eqn:estimate_sequence_projections} we deduce that $P_{\alpha} \to P$ also in the topology of $\mathfrak{L}^{\nicefrac{p}{2}}(\Alg)$. Interpolation~\eqref{framwork:eqn:interpolation_Lp_Lq_norms} ensures that $P_{\alpha}$ converges to $P$ for all intermediate $r \geqslant \frac{p}{2}$. Iterating this argument one finally shows the convergence for all $r > 0$.
\end{proof}
%

\subsection{Isospectral transformations and induced isometries} 
\label{framework:nc_Lp_Sobolev_spaces:gauge_transformations}
Let us start with a result which provides the continuity of the $\Alg$-bimodule structure of $\mathfrak{L}^p(\Alg)$ with respect to the strong convergence of sequences (or nets) in $\Alg$. It is the analog of a well-known result for classical Schatten ideals \cite{grumm-73}. 
\begin{lemma}\label{framework:lem:strong_convergence_trace_product}
	Let $\{ B_{\alpha} \}_{\alpha \in I} , \{ C_{\alpha} \}_{\alpha \in I} \subset \Alg$ be two nets such that $B_{\alpha} \to B$ and $C^*_{\alpha} \to C^*$ with respect to the SOT of $\Alg$ and $\sup_{\alpha} \bigl \{ \bnorm{B_{\alpha}} , \bnorm{C_{\alpha}} \bigr \} = K < +\infty$. Then, for every $A \in \mathfrak{L}^p(\Alg)$, $1 \leqslant p < \infty$, one has that $B_{\alpha} \, A \, C_{\alpha} \to B \, A \, C$ in the topology of $\mathfrak{L}^p(\Alg)$ and consequently $\bnorm{B_{\alpha} \, A \, C_{\alpha}}_p \to \bnorm{B \, A \, C}_p$. Moreover, for $p = 1$ one also gets $ \mathcal{T} \bigl ( B_{\alpha} \, A \, C_{\alpha} \bigr ) \to \mathcal{T}(B \, A \, C)$. 
\end{lemma}
\begin{proof}
	Let us start with $A \in \mathfrak{L}^1(\Alg) \cap \Alg = \Alg_{\mathcal{T}}$. According to \cite[Part I, Chapter 6, Proposition~1]{dixmier-81} and \cite[Part I, Chapter 3, Theorem~1]{dixmier-81} the linear form on $\Alg$ defined by $\Alg \in T \mapsto \mathcal{T}(A \, T) = \mathcal{T}(T \, A) \in \C$ is ultra-weakly continuous and hence ultra-strongly continuous. Moreover, on each closed ball of $\mathscr{B}(\Hil)$ one has that the uWOT and the WOT (as well the uSOT and the SOT) coincide. The net $T_{\alpha} := \bigl ( B_{\alpha} - B \bigr )^* \, \bigl ( B_{\alpha} - B \bigr )$ is equibounded by $4 K^2$ and converges strongly (hence ultra-strongly) to $0$ since
	\begin{align*}
		\bnorm{T_{\alpha} \psi}_{\Hil} \leqslant \bnorm{B_{\alpha} - B} \, \bnorm{(B_{\alpha} - B) \psi}_{\Hil}
		\leqslant 2K \, \bnorm{(B_{\alpha} - B) \psi}_{\Hil}
	\end{align*}
	for all $\psi \in \Hil$. It suffices to assume $A > 0$. Inequality \eqref{framework:eqn:trace_Lp_Lq_estimate} provides
	\begin{align*}
		\bnorm{(B_{\alpha} - B) \, A}_1^2 &\leqslant \bnorm{(B_{\alpha} - B) \, \sqrt{A}}_2^2 \, \bnorm{\sqrt{A}}_2^2
		\\
		&= \mathcal{T}(A) \; \mathcal{T} \Bigl ( \sqrt{A} \, T_{\alpha} \, \sqrt{A} \Bigr )
		= \mathcal{T}(A) \; \mathcal{T} \bigl ( A \, T_{\alpha} \bigr )
		, 
	\end{align*}
	and the ultra-strong convergence of $T_{\alpha}$ implies $\lim_{\alpha} \mathcal{T} \bigl ( A \, T_{\alpha} \bigr ) = 0$. This fact along with the linearity of $\mathcal{T}$, and the density of $\Alg_{\mathcal{T}}$ in $\mathfrak{L}^1(\Alg)$ ensures that $\lim_{\alpha} \bnorm{(B_{\alpha} - B) \, A}_1 \to 0$ for all $A \in \mathfrak{L}^1(\Alg)$. Similarly, we can show $\lim_{\alpha} \bnorm{A \, (C_{\alpha} - C)}_1 \to 0$ for all $A \in \mathfrak{L}^1(\Alg)$ by using the strong convergence of $C^*_{\alpha} \to C^* \in \Alg$ and the fact that the involution is an isometry, $\bnorm{A \, (C_{\alpha} - C)}_1 = \bnorm{(C_{\alpha}^*-C^*) \, A^*}_1$. The inequality
	\begin{align}
		\bnorm{B_{\alpha} \, A \, C_{\alpha} - B \, A \, C}_1 &\leqslant \bnorm{B_{\alpha} \, A \, (C_{\alpha} - C)}_1 + \bnorm{(B_{\alpha} - B) \, A \, C}_1
		\notag \\
		&\leqslant K \, \Bigl ( \bnorm{A \, (C_{\alpha} - C)}_1 + \bnorm{(B_{\alpha} - B) \, A}_1 \Bigr )
		\label{framwork:eqn:estimate_continuity_trace_product_nets}
	\end{align}
	concludes the argument for $p = 1$. Moreover, the first inequality in \eqref{framework:eqn:trace_Lp_Lq_estimate} provides
	\begin{align*}
		\babs{\mathcal{T}(B_{\alpha} \, A \, C_{\alpha}) - \mathcal{T}(B \, A \, C)} &= 
		\babs{\mathcal{T}(B_{\alpha} \, A \, C_{\alpha}-B \, A \, C)} 
		\leqslant \bnorm{B_{\alpha} \, A \, C_{\alpha}-B \, A \, C}_1 
	\end{align*}
	which ensures the continuity of $\mathcal{T}(B_{\alpha} \, A \, C_{\alpha})$.
	
	These arguments generalize from $p = 1$ to the case $p > 1$ via \eqref{framework:eqn:trace_Lp_Lq_estimate} so that 
	\begin{align*}
		\bnorm{(B_{\alpha} - B) \, A}_p^p &= \mathcal{T} \bigl ( \babs{(B_{\alpha} - B) A}^{p-1} \, \babs{(B_{\alpha} - B) \, A} \bigr ) 
		\\
		&\leqslant \Bnorm{\babs{(B_{\alpha} - B) \, A}^{p-1}} \; \bnorm{(B_{\alpha} - B) \, A}_1
	\end{align*}
	if $A \in \Alg_{\mathcal{T}}$. The boundness of the net $B_{\alpha} - B$ and the density of $\Alg_\mathcal{T}$ in $\mathfrak{L}^p(\Alg)$ imply $\lim_{\alpha} \bnorm{(B_{\alpha} - B)A}_p = 0$ for all $A \in \mathfrak{L}^p(\Alg)$. Just as above we also obtain $\lim_{\alpha} \bnorm{A \, (C_{\alpha} - C)}_p = 0$ for all $A \in \mathfrak{L}^p(\Alg)$. Finally, inequality~\eqref{framwork:eqn:estimate_continuity_trace_product_nets} holds also if we replace the norm $\norm{\cdot}_1$ with the norm $\norm{\cdot}_p$.
\end{proof}
The next notion will play a crucial role throughout this entire work.
\begin{definition}[Isospectral transformation]\label{framework:defn:gauge_transformation}
	Let $I \subseteq \R$ be an open interval. A map $I \ni t \mapsto G(t) \in \mathscr{B}(\Hil)$ satisfying 
	\begin{enumerate}[leftmargin=*,label=(\roman*)]
		\item $t \mapsto G(t)$ and $t \mapsto G(t)^*$ are both continuous with respect to the SOT, 
		\item $G(t)^* = G(t)^{-1}$ is unitary for each $t \in I$, 
		\item $G(t) \, A \, G(t)^* \in \Alg$ for all $A \in \Alg$ and $t \in I$, and 
		\item $\mathcal{T} \bigl ( G(t) \, A \, G(t)^* \bigr ) = \mathcal{T}(A)$ for all $A \in \Alg^+$ and $t \in I$, 
	\end{enumerate}
	is called an \emph{isospectral transformation} for the pair $(\Alg , \mathcal{T})$. 
\end{definition}
Condition (i) takes care of the fact that the adjoint operation is usually discontinuous in the SOT
(see \eg the example in \cite[Proposition~2.4.1]{Bratteli_Robinson:operator_algebras_1:2002}). Notice that in the special situation $G(t) \in \Alg$ the conditions (iii) and (iv) are automatically implied by the (ii).
\begin{proposition}[Isospectral dynamics]\label{framework:prop:extension_isometry_Lp}
	Let $I \ni t \mapsto G(t) \in \mathscr{B}(\Hil)$ be an isospectral transformation for the $(\Alg,\mathcal{T})$. For each $t\in I$ consider the $\mathcal{T}$-preserving isometric $\ast$-automorphism of $\Alg$ defined by
	\begin{align*}
		\gamma_t(A) := G(t) \, A \, G(t)^*
		,
		&&
		A \in \Alg
		.
	\end{align*}
	The map $I \ni t \mapsto\gamma_t(A)$ is ultra-weakly continuous for all $A \in \Alg$. Moreover, each $\gamma_t$ extends uniquely to a $\ast$-automorphism of $\rr{M}(\Alg)$ (still denoted with the same symbol) which is continuous with respect the measure topology. In addition:
	\begin{enumerate}
		\item If $f \in L^{\infty}_{\rm loc}(\R)$ and $A \in \rr{M}(\Alg)$ is a selfadjoint element $A = A^*$ then $\gamma_t \bigl ( f(A) \bigr ) = f \bigl ( \gamma_t(A) \bigr )$. 
		\item $\gamma_t$ is trace preserving, that is,
		$\mathcal{T} \bigl ( \gamma_t(A) \bigr ) = \mathcal{T}(A)$ for all $A \in \rr{M}(\Alg)$, $A \geqslant 0$.
		\item In particular, if $A \in \mathfrak{L}^1(\Alg)$ then also $\gamma_t(A) \in \mathfrak{L}^1(\Alg)$ and $\mathcal{T}(\gamma_t(A))=\mathcal{T}(A)$.
		\item For all $1 \leqslant p < \infty$ the map $\gamma_t : \mathfrak{L}^p(\Alg) \longrightarrow \mathfrak{L}^p(\Alg)$ is a $\ast$-isometry, namely $\bnorm{\gamma_t(A)}_p = \snorm{A}_p$. Moreover, the map $I \ni t \mapsto \gamma_t \in \mathrm{Iso}(\mathfrak{L}^p(\Alg))$ is \emph{strongly continuous} with respect to the topology of $\mathfrak{L}^p(\Alg)$, namely
		\begin{align*}
			\lim_{t \to t_0} \bnorm{\gamma_t(A) - \gamma_{t_0}(A)}_p = 0
			&&
			\forall \, A \in \mathfrak{L}^p(\Alg)
			.
		\end{align*}
		For $p=2$ the isometries $\gamma_t$ are in particular unitary with respect to the Hilbert structure of $\mathfrak{L}^2(\Alg)$.
	\end{enumerate}
\end{proposition}
\begin{proof}
	The properties of the isospectral transformation $t \mapsto G(t)$ and the fact that the multiplication on norm-bounded subsets of $\mathscr{B}(\Hil)$ is strongly continuous implies that $t \mapsto \gamma_t(A)$ is continuous in the SOT for all $A \in \Alg$. Moreover, on norm-bounded subsets one has that SOT $=$ uSOT, and the continuity with respect to the uSOT implies the continuity in the uWOT. The rest of the claim, concerning the canonical extension of $\gamma_t$ to $\rr{M}(\Alg)$, and properties (1), (2) and (3) are proved in \cite[Proposition~3.3]{Pagter_Sukochev:commutator_estimates_R_flows:2007}. (4) is a consequence of (1), which implies that $\gamma_t(\snorm{A}^p) = \babs{\gamma_t(A)}^p$, and (3), which guarantees $\mathcal{T} \bigl ( \gamma_t(\sabs{A}^p) \bigr ) = \mathcal{T} \bigl ( \sabs{A}^p \bigr )$. The strong continuity of $t \mapsto \gamma_t$ in $\mathfrak{L}^p(\Alg)$ is now a consequence of Lemma~\ref{framework:lem:strong_convergence_trace_product}.
\end{proof}
%

\section{Generalized commutators} 
\label{framework:commutators}
One of the key technical issues in the proofs of Chapters~\ref{dynamics} and \ref{Kubo_formula} is a proper definition of products and commutators of measurable and not necessarily measurable operators, and measurability is a prerequisite for them to be an element of $\mathfrak{L}^p(\Alg)$. For instance, $\mathcal{T}(J \, \rho)$ enters the definition of the net current, and this expression only makes sense if the product $J \, \rho$ lies in $\mathfrak{L}^1(\Alg)$. Similarly, the dynamics of observables are defined in terms of generalized commutators, and therefore we need to introduce various notions of commutators which all morally reduce to $[A , B] = A \, B - B \, A$ but are mathematically distinct.

\subsection{Commutators between $\mathcal{T}$-measurable operators} 
\label{framework:commutators:t_measurable_operators}
The $\ast$-algebraic structure of $\rr{M}(\Alg)$ and the non-commutative Hölder inequality~\eqref{framework:eqn:generic_Hoelder_inequality} for the spaces $\mathfrak{L}^p(\Alg)$ allows us to define the algebraic commutators
\begin{align*}
	[A,B]_{(0)} := A \, B - B \, A
	 \in \rr{M}(\Alg)
	, 
	&&
	A , B \in \rr{M}(\Alg)
	,
\end{align*}
and similarly on the $\mathfrak{L}^p$-spaces 
\begin{align*}
	[A,B]_{(r)} := A \, B - B \, A \in \mathfrak{L}^r(\Alg)
	,
	&&
	A \in \mathfrak{L}^p(\Alg)
	, \; 
	B \in \mathfrak{L}^q(\Alg)
	,
\end{align*}
where the indices $p , q , r > 0$ need to satisfy ${p}^{-1} + {q}^{-1} = {r}^{-1}$. Furthermore, the trace property~\eqref{framework:eqn:trace_property} leads to 
\begin{align}
	&\mathcal{T} \bigl ( A \, [B,C]_{(1)} \bigr ) = \mathcal{T} \bigl ( [A,B]_{(p)} \, C \bigr )
	,
	&&
	A \in \Alg
	, \; 
	B \in \mathfrak{L}^p(\Alg)
	, \; 
	C \in \mathfrak{L}^q(\Alg)
	, 
	\label{framework:eqn:trace_product_with_commutator_switch}
\end{align}
where we need to impose $1 \leqslant p , q \leqslant \infty$ and ${p}^{-1} + {q}^{-1} = 1$ to ensure that left-hand and right-hand side are well-defined. 

Now let us consider a projection $P \in \proj(\Alg)$ and its orthogonal complement $P^{\perp} := \id - P \in \proj(\Alg)$. Thanks to the $\Alg$-bimodule structure of the $\mathfrak{L}^p(\Alg)$ 
\begin{align*}
	[P , A]_{(p)} := P \, A - A \, P \in \mathfrak{L}^p(\Alg)
\end{align*}
is well-defined for each $A \in \mathfrak{L}^p(\Alg)$, and 
\begin{align*}
	[P^{\perp} , \, A]_{(p)} = - [P , A]_{(p)}
\end{align*}
holds. The orthogonality relation $P \, (\id-P) = 0 = (\id-P) \, P$ and the distributivity of the $\Alg$-bimodule structure yield a commutator identity
\begin{align}
	\bigl [ P , [P , A]_{(p)} \bigr ]_{(p)} = P \, A \, P^{\perp} + P^{\perp} \, A \, P
	\label{framework:eqn:double_commutator_projection}
\end{align}
which is valid for any $A \in \mathfrak{L}^p(\Alg)$. The identity \eqref{framework:eqn:double_commutator_projection} will be relevant in Section~\ref{Kubo_formula:adiabatic_limit}.

\subsection{Commutators between $\mathcal{T}$-measurable and affiliated operators} 
\label{framework:commutators:hamiltonian_t_measurable_operators}
Unfortunately, commutators inside $\rr{M}(\Alg)$ are not sufficient for our purposes when the trace $\mathcal{T}$ is only semi-finite but not finite (\cf Example \ref{framework:example:T_measurable_operators} (3)). In fact, in this case, we need to work also with commutators between $\mathcal{T}$-measurable operators and (selfadjoint) operators $H \in \affil(\Alg)$ which are \emph{not} $\mathcal{T}$-measurable (\cf Example \ref{framework:example:non_T_measurable_operators}). The following result shows that under certain conditions it is possible to define a “good” (left) multiplication between elements of $\affil(\Alg)$ and $\mathcal{T}$-measurable operators in $\rr{M}(\Alg)$.
\begin{proposition}[Left multiplication]\label{framework:prop:extended_product}
	Suppose $A \in \rr{M}(\Alg)$ and $H \in \affil(\Alg)$ has a dense domain $\domain(H) \subset \Hil$ such that $H \not\in \rr{M}(\Alg)$ (which is only possible if $\mathcal{T}$ is semi-finite but not finite). Define the initial domain of the product 
	\begin{align*}
		\domain \bigl (H \overset{\circ}{\cdot} A \bigr ) := \bigl \{ \varphi \in \domain(A) \; \; \vert \; \; A \varphi \in \domain(H) \bigr \} 
		\subseteq \domain(A)
		, 
	\end{align*}
	and assume that 
	\begin{enumerate}[leftmargin=*,label=(\roman*)]
		\item $\domain(H \overset{\circ}{\cdot} A)$ is {$\mathcal{T}$-dense}\footnote{Notice that $A[\domain(H \overset{\circ}{\cdot} A)]\subseteq \domain(H)$ fails to be $\mathcal{T}$-dense when $H\notin \rr{M}(\Alg)$ as consequence of the discussion in Remark \ref{framework:remark:not_T_dense}} and
		\item $H \overset{\circ}{\cdot} A$ is closable on $\domain \bigl (H \overset{\circ}{\cdot} A \bigr )$.
	\end{enumerate}
	Then the closure (strong product) $H \, A$ is $\mathcal{T}$-measurable, \ie $H \, A \in \rr{M}(\Alg)$.
\end{proposition}
\begin{proof}
	The product $H \, A$ is well-defined on $\domain \bigl (H \overset{\circ}{\cdot} A \bigr )$ and defines an element affiliated with $\Alg$ (see \eg \cite[Remark 2]{Terp:noncommutative_Lp_spaces:1981}). Moreover, $\domain \bigl (H \overset{\circ}{\cdot} A \bigr )$ is {$\mathcal{T}$-dense} and so dense. The closure of $H \, A$ is an element of $\affil(\Alg)$ with a {$\mathcal{T}$-dense} domain, hence a $\mathcal{T}$-measurable operator.
\end{proof}
To simplify the presentation let us associate to each $H \in \affil(\Alg)$ the domain 
\begin{align*}
	\rr{Left}_{H}^p := \Bigl \{ A \in \rr{M}(\Alg) \; \; \big \vert \; \; \, H \, A \in \mathfrak{L}^p(\Alg)\ \ \text{in the sense of Proposition \ref{framework:prop:extended_product}} \Bigr \}
\end{align*}
of left $H$-multiplication with values in $\mathfrak{L}^p(\Alg)$. By extending this notation we can use $\rr{Left}_{H}^0$  for the set of  $\mathcal{T}$-measurable operators such that the left $H$-multiplication take values only in $\rr{M}(\Alg)$.  Then we can introduce 
\begin{definition}[Generalized commutators]\label{framework:defn:generalized_commutators}
	Let $H \in \affil(\Alg)$ (not necessarily $\mathcal{T}$-measurable) and $A \in \rr{M}(\Alg)$ such that $A \in \rr{Left}_{H}^0$ and $A^* \in \rr{Left}_{H^*}^0$. Then we define the generalized commutator to be 
	\begin{align}
		[H,A]_{\ddagger} := H \, A - \bigl ( H^* \, A^* \bigr )^* \in \rr{M}(\Alg)
		. 
		\label{framework:eqn:generalized_commutators}
	\end{align}
	Moreover, we set $[A,H]_{\ddagger} := -[H,A]_{\ddagger}$. If $A \in \rr{Left}_{H}^p$ and $A^* \in \rr{Left}_{H^*}^p$ for some $1 \leqslant p < \infty$, then we say that the commutator between $H$ and $A$ is $p$-measurable and $[H,A]_{\ddagger} \in \mathfrak{L}^p(\Alg)$. 
\end{definition}
Evidently, in case $H$ is selfadjoint one has that
\begin{align}
	\rr{D}^{00}_{H,p} := \Bigl \{ A \in \mathfrak{L}^p(\Alg) \; \; \big \vert \; \; A \,,\, A^* \in \rr{Left}_{H}^p \Bigr \}
	\label{framework:eqn:domain_maximal_generalized_commutator}
\end{align}
is the maximal domain in $\mathfrak{L}^p(\Alg)$ where the generalized commutator with $H$ can be defined in the sense of Definition~\ref{framework:defn:generalized_commutators}.

In the following two remarks we discuss some aspects related with the difficulty in defining a $\mathcal{T}$-measurable \emph{right} multiplication and the link between the generalized and the usual commutator.
\begin{remark}[Right multiplication]\label{framework:remark:not_T_dense}
	The right product $A \overset{\circ}{\cdot} H$ has as domain $\domain \bigl ( A \overset{\circ}{\cdot} H \bigr ) \subset \domain(H)$. If $H \in \affil(\Alg)$ but $H \notin \rr{M}(\Alg)$ then $\domain(A \overset{\circ}{\cdot} H)$ cannot be $\mathcal{T}$-dense (otherwise $H$ would have a $\mathcal{T}$-dense domain). For this reason the definition of right multiplication is much more problematic than left multiplication. Indeed, even though $A \overset{\circ}{\cdot} H$ is closable there is no guarantee that the strong product (closure) $H \, A$ has a $\mathcal{T}$-dense domain. On the other hand this does not mean that the right multiplication is always ill-defined. For instance, if $H$ is selfadjoint and $f : \R \longrightarrow \R$ is a sufficiently rapidly decreasing function one has that $f(H) \in \Alg$ and $f(H) \, H = H \, f(H) \in \Alg\subset \rr{M}(\Alg)$. In summary, Proposition~\ref{framework:prop:extended_product} only establishes \emph{sufficient} conditions for the existence of a left multiplication with $H$. In particular, this does not exclude the possibility to define a left multiplication as well as a right multiplication in particular situations not covered by Proposition~\ref{framework:prop:extended_product}. Finally, let us notice that when $H \in \rr{M}(\Alg)$ the $\ast$-algebraic structure of $\rr{M}(\Alg)$ implies that $\bigl (H^* \, A^* \bigr )^* = \bigl ( (A \, H)^* \bigr )^* = A \, H$ and \eqref{framework:eqn:generalized_commutators} reduces to the usual commutator in $\rr{M}(\Alg)$.
\end{remark}
\begin{remark}[From {generalized} to usual commutators]\label{framework:remark:generalized_to_usual_commutator}
	Let us characterize the domain of $\Hil$ where the {generalized} commutator introduced in Definition~\ref{framework:defn:generalized_commutators} agrees with the \emph{usual} commutator. Let $H \in \affil(\Alg)$ (not necessarily $\mathcal{T}$-measurable) and $A \in \rr{Left}_{H}^0$,  $A^* \in \rr{Left}_{H^*}^0$. This means that $\domain(H\overset{\circ}{\cdot}A)$ and $\domain(H^*\overset{\circ}{\cdot}A^*)$ are both $\mathcal{T}$-dense, hence dense, domains. Consider the domain $\domain(A\overset{\circ}{\cdot}H):=\{\psi \in \domain(H)\ |\ H\psi \in \domain(A)\}$. If $H$ is not $\mathcal{T}$-measurable, then $\domain(A\overset{\circ}{\cdot}H)$ is surely not $\mathcal{T}$-dense (\cf Remark~\ref{framework:remark:not_T_dense}), and it in fact need not even be dense. Seeing as $\domain \bigl ( H^* \overset{\circ}{\cdot} A^* \bigr )$ is dense in $\domain \bigl ( H^* \, A^* \bigr )$, we deduce that for all $\varphi \in \domain \bigl ( H^* \overset{\circ}{\cdot} A^* \bigr )$ and $\psi \in \domain \bigl ( A \overset{\circ}{\cdot} H \bigr )$ we have 
	\begin{align*}
		\bscpro{\psi \, }{\, \bigl ( H^*A^* \bigr ) \varphi} = \bscpro{H \psi \,}{\, A^* \varphi} 
		= \bscpro{(A \, H) \psi \,}{\, \varphi}
		. 
	\end{align*}
	This implies $\domain \bigl ( A \overset{\circ}{\cdot} H \bigr ) \subseteq \domain \bigl ( (H^* \, A^*)^* \bigr ) \cap \domain(H)$ so that on $\domain(H \, A) \cap \domain \bigl ( A \overset{\circ}{\cdot} H \bigr )$ the action of the generalized commutator coincides with that of the usual commutator, 
	\begin{align*}
		[H,A]_{\ddagger} \psi = \bigl ( H \, A - A \, H \bigr ) \psi
		,
		&&
		\psi \in \domain(H \, A) \cap \domain \bigl ( A \overset{\circ}{\cdot} H \bigr )
		. 
	\end{align*}
	However, even though $\domain(H \, A)$ is $\mathcal{T}$-dense (hence dense), the intersection $\domain(H \, A) \cap \domain \bigl ( A \overset{\circ}{\cdot} H \bigr )$ need not be dense — it could even be $\{ 0 \}$ without extra assumptions.
\end{remark}
The next two results state that left multiplication described in Proposition~\ref{framework:prop:extended_product} is well-defined for each $H \in \affil(\Alg)$ and in every space $\mathfrak{L}^p(\Alg)$.

\begin{lemma}\label{framework:lem:product_affiliated_algebra_density}
	Let $H \in \affil(\Alg)$. For every $B \in \mathfrak{L}^p(\Alg) \cap \Alg$, $1 \leqslant p < \infty$, there exists a sequence $\{ B_n \}_{n \in \N} \subset \mathfrak{L}^p(\Alg) \cap \Alg$ such that: 
	\begin{enumerate}
		\item $B_n \to B$ in the SOT and $\snorm{B_n} < K$, and the left multiplication $H \, B_n \in \mathfrak{L}^p(\Alg)$ is well-defined.
		\item $B_n \to B$ in the topology of $\mathfrak{L}^p(\Alg)$. 
		\item Left multiplication by $H$ is defined on a dense domain in $\mathfrak{L}^p(\Alg)$.
	\end{enumerate}
\end{lemma}
\begin{proof}
	\begin{enumerate}
		\item Pick $B \in \mathfrak{L}^p(\Alg) \cap \Alg$. Furthermore, let $H = U \, \sabs{H}$ be the polar decomposition of $H$ and $P_n(H) \in \Alg$ the spectral projection of $\sabs{H}$ for the interval $[0,+n] \subset \R$. Define $B_n := P_n(H) \, B \, P_n(H)\in \mathfrak{L}^p(\Alg) \cap \Alg$; this uses that the spaces $\mathfrak{L}^p(\Alg)$ are $\Alg$-modules. Clearly, $\snorm{B_n} \leqslant \snorm{B}$ and $B_n \to B$ in the SOT. Moreover, the equality $H \, B_n = U \, \bigl ( P_n(H) \, \sabs{H} \, P_n(H) \bigr ) \, B_n\in\mathfrak{L}^p(\Alg) \cap \Alg$ shows that the left multiplication is well-defined.
		\item In view of Lemma~\ref{framework:lem:strong_convergence_trace_product} we know $P_n(H) \, B \to B$ and $B \, P_n(H) \to B$ converge in the topology of $\mathfrak{L}^p(\Alg)$. Then the straight-forward estimate 
		\begin{align*}
			\bnorm{B_n - B}_p 
			&\leqslant \bnorm{P_n(H)} \, \bnorm{B \, P_n(H) - B}_p + \bnorm{P_n(H) \, B - B}_p
		\end{align*}
		and $\snorm{P_n(H)} = 1$ imply that $\bnorm{B_n - B}_p \to 0$ when $n \to \infty$. 
		\item Item (2) shows that the left multiplication by $H$ is well-defined on a domain that is $\norm{\cdot}_p$-dense in $\mathfrak{L}^p(\Alg) \cap \Alg$. However, the latter is in turn dense in $\mathfrak{L}^p(\Alg)$ since $\Alg_{\mathcal{T}} \subseteq \mathfrak{L}^p(\Alg) \cap \Alg$ and $\Alg_{\mathcal{T}}$ is dense in $\mathfrak{L}^p(\Alg)$ by Theorem~\ref{framework:thm:measurability_products_in_Lp}.
	\end{enumerate}
\end{proof}
\begin{corollary}\label{framework:cor:product_affiliated_algebra_density}
	Let $H \in \affil(\Alg)$ be selfadjoint. For each $1 \leqslant p < \infty$ the domain $\rr{D}^{00}_{H,p}$ defined in \eqref{framework:eqn:domain_maximal_generalized_commutator} is dense in $\mathfrak{L}^p(\Alg)$.
\end{corollary}
\begin{proof}
	It is enough to complete the proof of Lemma~\ref{framework:lem:product_affiliated_algebra_density} with the observation that also $B_n^* = P_n(H) \, B^* \, P_n(H)$ and $H \, B_n^* = U \, \bigl ( P_n(H) \, \sabs{H} \, P_n(H) \bigr ) \, B^*_n$ are in $\mathfrak{L}^p(\Alg) \cap \Alg$. This implies that $\{ B_n \}_{n \in \N} \subset \rr{D}^{00}_{H,p}$, and we obtain the density from the same argument as in the proof of Lemma~\ref{framework:lem:product_affiliated_algebra_density}~(3).
\end{proof}
Proposition~\ref{framework:prop:extended_product} provides the way to extend the module structure of the spaces $\mathfrak{L}^p(\Alg)$ with a left multiplication by $\affil(\Alg)$. It turns out that this extension preserves many of the properties of the standard left module structure. The following lemma, which will be used time and again, makes this precise. 
\begin{lemma}[Extended left-module structure]\label{framework:lem:extension_algebra_unbounded_operators}
	Suppose we are given $A \in \mathfrak{L}^p(\Alg)$, and let $H \in \affil(\Alg)$ so that $A \in \mathfrak{Left}^p_H$. Then the following facts hold true:
	\begin{enumerate}
		\item Let $B \in \Alg$ such that $\domain_{H,A}(B) := \bigl \{ \varphi \in \Hil \; \; \vert \; \; B \varphi \in \domain \bigl ( H \overset{\circ}{\cdot} A \bigr ) \bigr \}$ is $\mathcal{T}$-dense. Then the left multiplication $H \, (A \, B)$ is well-defined and associative so that
		\begin{align*}
			H \, A \, B := (H \, A) \, B = H \, (A \, B)
			,
			&&
		\end{align*}
		holds as elements of $\mathfrak{L}^p(\Alg)$. Moreover, the associativity holds true automatically if $B \in \Alg$ is invertible in $\Alg$.
		\item Assume in addition that $H^{-1} \in \Alg$, and let $J \in \affil(\Alg)$ be such that $\domain(H) \subseteq \domain(J)$ and $A \in \mathfrak{Left}^p_J$ holds. Then $J \, H^{-1} \in \Alg$ and the two products 
		\begin{align*}
			J \, A = \bigl ( J \, H^{-1} \bigr ) \, \bigl ( H \, A \bigr )
			\in \mathfrak{L}^p(\Alg)
		\end{align*}
		agree. In particular, this applies in case $J \in \Alg$ itself lies in the algebra. 
	\end{enumerate}
\end{lemma}
\begin{proof}
	\begin{enumerate}
		\item Because the inclusion $\domain \bigl ( H \overset{\circ}{\cdot} A \bigr ) \subseteq \domain(H \, A)$ implies $\domain_{H,A}(B) \subseteq \domain \bigl ( H \overset{\circ}{\cdot} (A \, B) \bigr )$, we know the domain $\domain \bigl ( H \overset{\circ}{\cdot} (A \, B) \bigr )$ is necessarily $\mathcal{T}$-dense. Moreover, one can check that
		\begin{align}
			(H \, A) \, (B \varphi) = H \bigl (A (B \varphi) \bigr ) = H \, \bigl ( (A \, B) \varphi \bigr )
			,
			&&
			\varphi \in \domain_{H,A}(B)
			\label{framework:eqn:mixed_product}
		\end{align}
		where the second equality is justified by $B \varphi \in \domain(A)$ which that ensures $(A \, B) \varphi = A \, (B \varphi)$. Let us recall that the operator $A \, B \in \mathfrak{L}^p(\Alg)$ exists for the $\Alg$-bimodule structure of the space $\mathfrak{L}^p(\Alg)$. Equation~\eqref{framework:eqn:mixed_product} says that the (not closed) operators $(H \, A) \overset{\circ}{\cdot} B$ and $H \overset{\circ}{\cdot} (A \, B)$ agree on the $\mathcal{T}$-dense domain $\domain_{H,A}(B)$. The $\Alg$-bimodule structure of $\mathfrak{L}^p(\Alg)$ implies that $(H \, A) \overset{\circ}{\cdot}B$ is closable (\cf\cite[Proposition~24 (1)]{Terp:noncommutative_Lp_spaces:1981}) and $(H \, A) \, B\in\mathfrak{L}^p(\Alg)$. As a consequence of the \eqref{framework:eqn:mixed_product} also $H \overset{\circ}{\cdot} (A \, B)$ is closable and the two closed operators (strong multiplications) $(H \, A) \, B = H \, (A \, B)$ are equal as elements of $\rr{M}(\Alg)$ in view of Proposition~\ref{framework:prop:agreement_operators_T_dense_set}, and in turn also as elements of $\mathfrak{L}^p(\Alg)$. If $B$ is invertible then $\domain_{H,A}(B)=B^{-1}[\domain(H \overset{\circ}{\cdot} A )]$ and an inspection to Definition \ref{framework:defn:measurable_operators} shows that the left multiplication by $B^{-1}$ preserves the $\mathcal{T}$-density of the transformed domain.
		\item The conditions $\domain(H) \subseteq \domain(J)$ and $J \in \affil(\Alg)$ ensure that the strong product $J \overset{\circ}{\cdot} H^{-1}$ is globally defined on $\Hil$ and closable. Therefore, the strong product defines an element $J \, H^{-1}\in\affil(\Alg)$ (see Remark \ref{framework:remark:algebraic_operations}) which turns out to be bounded as a consequence of the closed graph theorem. Then, one has $J \, H^{-1} \in \Alg$. The condition concerning the domains also implies that $\domain(H\overset{\circ}{\cdot}A)\subseteq \domain(JA)$ and this ensures that for any $\varphi \in \domain(H\overset{\circ}{\cdot}A)$ it holds that
		\begin{align*}
			\bigl (J \, H^{-1} \bigr ) \, \bigl ( H \, A \bigr ) \varphi = J \, \bigl ( H^{-1} \bigl ( H (A \varphi) \bigr ) \bigr )
			= J \, (A \varphi) 
			= (J \, A) \varphi. 
		\end{align*}
		Hence, the two operators $\bigl ( J \, H^{-1} \bigr ) \overset{\circ}{\cdot} (H \, A)$ and $J \overset{\circ}{\cdot} A$ agree on the $\mathcal{T}$-dense domain $\domain(H\overset{\circ}{\cdot}A)$ and are both closable. Thus, $J \, A = \bigl ( J \, H^{-1} \bigr ) \, \bigl ( H \, A \bigr )$ first as elements of $\rr{M}(\Alg)$ (Proposition~\ref{framework:prop:agreement_operators_T_dense_set}) and consequently as elements of $\mathfrak{L}^p(\Alg)$.
	\end{enumerate}
\end{proof}
These results extend to the product of finite sums and integrals. 
\begin{corollary}[Linearity of the left multiplication]
	Suppose we are given a family $\bigl \{ A_1 , \ldots , A_n \bigr \} \subset \mathfrak{L}^p(\Alg)$, and let $J , H \in \affil(\Alg)$ be such that $\domain(H) \subseteq \domain(J)$, $H^{-1} \in \Alg$, and $A_1 , \ldots , A_n \in \mathfrak{Left}^p_J \cap \mathfrak{Left}^p_H$. Then the extended left multiplication is distributive, \ie the two sides of the equations
	\begin{align*}
		J \, \bigl ( A_1 + \ldots + A_n \bigr ) = J \, A_1 + \ldots + J \, A_n
	\end{align*}
	agree as elements of $\mathfrak{L}^p(\Alg)$. 
\end{corollary}
\begin{proof}
	The strong sum $J \, A_1 + \ldots + J \, A_n$ is well-defined in $\mathfrak{L}^p(\Alg)$ and by using Lemma~\ref{framework:lem:extension_algebra_unbounded_operators}~(2) and the $\Alg$-bimodule structure of  $\mathfrak{L}^p(\Alg)$
	we have 
	\begin{align*}
		J \, A_1 + \ldots + J \, A_n &= \bigl ( J \, H^{-1} \bigr ) \, \bigl ( H \, A_1 \bigr ) + \ldots + \bigl ( J \, H^{-1} \bigr ) \, \bigl ( H \, A_n \bigr )
		\\
		&= \bigl ( J \, H^{-1} \bigr ) \, \bigl ( H \, A_1 + \ldots + H \, A_n \bigr )
		. 
	\end{align*}	
	The domain $\domain \bigl ( H \, A_1 + \ldots + H \, A_n \bigr )$ contains the intersection of finitely many $\mathcal{T}$-dense domains, $\domain \bigl ( H \overset{\circ}{\cdot} A_k \bigr )$, and is therefore still $\mathcal{T}$-dense (\cf with the argument in \cite[Proposition~5 (ii)]{Terp:noncommutative_Lp_spaces:1981}). The following equality
	\begin{align*}
		J \, A_1 + \ldots + J \, A_n &= J \, \Bigl ( H^{-1} \, \bigl ( H \, A_1 + \ldots + H \, A_n \bigr ) \Bigr )
	\end{align*}
	is first justified on the $\mathcal{T}$-dense intersection of the $\domain \bigl ( H \overset{\circ}{\cdot} A_k \bigr )$ and then extended to the strong product via Proposition~\ref{framework:prop:agreement_operators_T_dense_set}. Finally, the $\Alg$-bimodule structure of $\mathfrak{L}^p(\Alg)$ provides the desired equality.
\end{proof}
The above argument can be repeated verbatim when the sum is replaced with a converging series or a Bochner integrals. In particular we will make use of the following result:
\begin{corollary}\label{framework:cor:integral_linearity_product}
	Suppose $t \mapsto A(t) \in L^1 \bigl ( \R , \mathfrak{L}^p(\Alg) \bigr )$ is a Bochner-integrable function, and let $J , H \in \affil(\Alg)$ be such that $\domain(H) \subseteq \domain(J)$, $H^{-1} \in \Alg$, $A(t) \in \mathfrak{Left}^p_J \cap \mathfrak{Left}^p_H$ for almost all $t \in \R$, and $t \mapsto H \, A(t) \in L^1 \bigl ( \R , \mathfrak{L}^p(\Alg) \bigr )$. Then we can pull $J$ inside the integral, \ie 
	\begin{align*}
		J \int_{\R} \dd t \, A(t) = \int_{\R} \dd t \, J \, A(t)
	\end{align*}
	holds as elements of $\mathfrak{L}^p(\Alg)$. 
\end{corollary}
%

\subsection{Commutators between unbounded operators} 
\label{framework:commutators:hamiltonian_unbounded_operators}
There is a third type of commutator that has to be considered for the purposes of this work where \emph{neither} operator is measurable. The current operators from Hypothesis~\ref{hypothesis:current}, for instance, are of this form. Assume that: 
\begin{enumerate}[leftmargin=*,label=(\roman*)]
	\item $H$ is selfadjoint and there is a dense domain $\domain_{\mathrm{c}}\subset \Hil$ (called \emph{localizing domain}) such that $\domain_{\mathrm{c}}(H) := \domain_{\mathrm{c}} \cap \domain(H)$ is a core for $H$. 
	\item $X$ is a closed operator with dense domain $\domain(X) \supseteq \domain_{\mathrm{c}}$ so that $X[\domain_{\mathrm{c}}] \subseteq \domain(H)$. 
	\item $H[\domain_{\mathrm{c}}(H)] \subseteq \domain_{\mathrm{c}}$.
\end{enumerate}
Under these conditions $\domain_{\mathrm{c}}$ is contained in both, $\domain(H \, X)$ and $\domain(X \, H)$, and hence, 
\begin{align}
	[X,H] \psi := \bigl ( X \, H - H \, X \bigr ) \psi
	,
	&&
	\psi \in \domain_{\mathrm{c}}(H)
	, 
	\label{framework:eqn:formal_commutator}
\end{align}
is well-defined on this joint core $\domain_{\mathrm{c}}(H) \subset \Hil$. In case the expression \eqref{framework:eqn:formal_commutator} is closable on $\domain_{\mathrm{c}}(H)$ we denote $\ii$ times its closure with 
\begin{align*}
	\ad_{X}(H) := \ii \overline{[X,H]} 
	. 
\end{align*}
If in addition $X$ is selfadjoint, then $\ad_X(H)$ is selfadjoint as well. 

\section{Non-commutative Sobolev spaces} 
\label{framework:Sobolev_spaces}
Just like in ordinary $L^p$-theory, there is a non-commutative analog of Sobolev spaces, whose elements have additional regularity properties. Here, the “derivatives” are associated to a set of generators $\{ X_1 , \ldots , X_d \}$ (\eg Hypothesis~\ref{hypothesis:generators}). \emph{Formally} speaking, these derivations are commutators $\ii \, [X_k , \, \cdot \;]$, although as we have seen in the previous subsection such a naïve definition would lead to problems of measurability.

\subsection{$\mathcal{T}$-compatible spatial derivations} 
\label{framework:nc_Lp_Sobolev_spaces:Sobolev_spaces}
Instead, we choose a different approach: Let $\R \ni t \mapsto G(t) \in \mathscr{B}(\Hil)$ be an isospectral transformation according to Definition~\ref{framework:defn:gauge_transformation}. Assume in addition that $t \mapsto G(t)$ is a one-parameter unitary group. Then, according to the Stone's theorem \cite[Theorem~VIII.8]{Reed_Simon:M_cap_Phi_1:1972}, the group admits an exponential representation $G(t) = \e^{+ \ii t X}$ where $X$ is a (possibly unbounded) selfadjoint operator with domain $\domain(X) \subset \Hil$. A comparison with Hypothesis~\ref{hypothesis:generators} allows us to say that $X$ is the $\mathcal{T}$-compatible generator associated to the isospectral transformation $G(t)$. Moreover, the first part of Proposition \ref{framework:prop:extension_isometry_Lp} ensures that $X$ generates a one-parameter group of $\ast$-automorphisms $\R \ni t \mapsto \eta^X_t \in \mathrm{Aut}(\Alg)$ defined by
\begin{align}
	\eta^X_t(A) := \e^{+ \ii t X} \, A \, \e^{- \ii t X}
	,
	&&
	A \in \Alg
	, \; 
	t \in \R
	\label{framework:eqn:definition_X_flow}
\end{align}
such that $t \mapsto \eta^X_t(A)$ is continuous in the uWOT (and also in the SOT)
for all $A \in \Alg$. Then, the map $t \mapsto \eta^X_t$ defines an \emph{$\R$-flow}\footnote{It is sometimes also called $C^*_0$-group according to \cite[Definition~3.1.2]{Bratteli_Robinson:operator_algebras_1:2002}} on $\Alg$ according to the nomenclature introduced in \cite[Section~4 \& 6]{Pagter_Sukochev:commutator_estimates_R_flows:2007}. The prescription 
\begin{align}
	\partial_X(A) := \lim_{t \to 0} \; \frac{\eta_t^X(A) - A}{t}
	,
	&&
	A \in \rr{D}_{X,\infty} \subset \Alg
	, 
	\label{framework:eqn:derivation_as_derivative_automorphism}
\end{align}
seen as a limit in the uWOT, defines an \emph{unbounded derivation} on $\Alg$ with domain $\rr{D}_{X,\infty}$ \cite[Definition~3.1.5]{Bratteli_Robinson:operator_algebras_1:2002}. It is well-known that elements $A\in \rr{D}_{X,\infty}$ are characterized by the condition $A[\domain(X)]\subseteq \domain(X)$ and the representation as commutator, 
\begin{align}
	\partial_X(A) = \ii [X,A] 
	= \ad_X(A) 
	\in \Alg
	, 
	\label{framework:eqn:derivation_as_commutator}
\end{align}
where $[X,A] := X \, A - A \, X$ is closable on $\domain(X)$, and its closure belongs to $\Alg$. One usually refers to the operator $\partial_X : \rr{D}_{X,\infty} \longrightarrow \Alg$ as given in \eqref{framework:eqn:derivation_as_derivative_automorphism}, or equivalently in \eqref{framework:eqn:derivation_as_commutator}, as a \emph{spatial derivation} of $\Alg$ \cite[Section~3.2.5]{Bratteli_Robinson:operator_algebras_1:2002}. Clearly, $\partial_X$ is a linear operator which, in addition fulfills the \emph{Leibniz rule}
\begin{align*}
	\partial_X(A \, B) = \partial_X(A) \, B + A \, \partial_X(B)
	&&
	\forall A , B \in \rr{D}_{X,\infty}
\end{align*}
and is symmetric in the sense that 
\begin{align*}
	\partial_X(A^*) = \partial(A)^*
	&&
	\forall A \in \rr{D}_{X,\infty}
	. 
\end{align*}
Therefore the domain $\rr{D}_{X,\infty}$ turns out to be a $\ast$-subalgebra of $\Alg$. One can show that a spatial derivation obeys $\bnorm{A - \lambda \partial_X(A)} \geqslant A$ for all $A \in \rr{D}_{X,\infty}$ and $\lambda \in \R$ \cite[Corollary~3.2.56]{Bratteli_Robinson:operator_algebras_1:2002}. In the case $\rr{D}_{X,\infty}=\Alg$ one says that $\partial_X$ is a \emph{bounded} spatial derivation, and this is possible if and only if $\norm{X} < +\infty$.

The notion of spatial derivation can be extended from $\Alg$ to the Banach spaces $\mathfrak{L}^p(\Alg)$. Proposition~\ref{framework:prop:extension_isometry_Lp} ensures that the prescription \eqref{framework:eqn:definition_X_flow} first extends canonically to a one-parameter group of $\ast$-automorphisms of $\rr{M}(\Alg)$ and then define a strongly continuous one-parameter group of $\ast$-isometries $\R \ni t \mapsto \eta^X_t \in \mathrm{Iso}(\mathfrak{L}^p(\Alg))$ for each $1 \leqslant p < \infty$. This allows us to define spatial derivations on $\mathfrak{L}^p(\Alg)$ by using the theory of $C_0$-groups \cite{Bratteli_Robinson:operator_algebras_1:2002,Engel_Nagel:one_parameter_semigroup_linear_evolution_equations:2000} on Banach spaces.
\begin{definition}[$\mathcal{T}$-compatible spatial derivation]\label{framework:defn:T_compatible_derivation}
	A $\mathcal{T}$-compatible generator $X$ defines a $\mathcal{T}$-compatible spatial derivation on each Banach space $\mathfrak{L}^p(\Alg)$, $1 \leqslant p < \infty$, according to the formula
	\begin{align*}
		\partial_X(A) := \lim_{t \to 0} \frac{\eta_t^X(A) - A}{t}
		,
		&&
		A \in \rr{D}_{X,p} \subset \mathfrak{L}^p(\Alg)
		, 
	\end{align*}
	where the limit is taken with respect to the uniform topology of the norm $\norm{\cdot}_p$. The domains are given by $\rr{D}_{X,p} := \bigl \{ A \in \mathfrak{L}^p(\Alg) \; \; \vert \; \; \partial_X(A) \in \mathfrak{L}^p(\Alg) \bigr \}$.
\end{definition}
We notice that a $\mathcal{T}$-compatible spatial derivation $\partial_X$ on $\mathfrak{L}^p(\Alg)$ has automatically a norm dense domain $\rr{D}_{X,p}$ and it is closed (see \eg \cite[Corollary~7.3]{Engel_Nagel:one_parameter_semigroup_linear_evolution_equations:2000}). Since $\eta_t^X$ is a bounded group action, the spectrum $\spec(\partial_X)$ of the generator is contained in $\ii \R$. For $\lambda \in \C$ with $\Re(\lambda)>0$, the resolvent of $\partial_X$ is given by a Laplace transform (see \eg \cite[Theorem~1.10]{Engel_Nagel:one_parameter_semigroup_linear_evolution_equations:2000}), 
\begin{align}
	\frac{1}{\partial_X \pm \lambda}(A) = \pm \int_0^{+\infty} \dd t \, \e^{- \lambda t} \, \eta_{\mp t}^X(A)
	,
	&&
	A \in \mathfrak{L}^p(\Alg)
	, 
	\label{framework:eqn:resolvent_equation_integral}
\end{align}
where the right-hand side is seen as a norm convergent Bochner integral in $\mathfrak{L}^p(\Alg)$. The set $\rr{D}_{X,p}\cap\Alg$ turns out to be sufficiently large to determine the derivations $\partial_X$. 
\begin{proposition}\label{framework:prop:core_derivation}
	Let $X$ by a $\mathcal{T}$-compatible generator, $\R \ni t \mapsto \eta^X_t \in \mathrm{Iso}(\mathfrak{L}^p(\Alg))$, $1 \leqslant p < \infty$, the associated one-parameter group and $\partial_X$ the related derivation. Then:
	\begin{enumerate}
		\item The set 
		\begin{align*}
			\rr{D}^{0}_{X,p} := \Bigl \{ A \in \rr{D}_{X,p} \cap \Alg \; \; \big \vert \; \; \partial_X(A) \in \Alg \Bigr \}
		\end{align*}
		is a core for $\partial_X$ on $\mathfrak{L}^p(\Alg)$ that can be characterized as 
		\begin{align}
			\rr{D}^{0}_{X,p} = \Bigl \{ A \in \rr{D}_{X,p} \cap \Alg \; \; \big \vert \; \; A[\domain(X)] \subset \domain(X), \; [X,A] \in \mathfrak{L}^p(\Alg) \cap \Alg \Bigr\}
			\label{framework:eqn:initial_domain_derivation}
		\end{align}
		where $\partial_X(A) = \ii \, [X,A]$ for all $A \in \rr{D}^{0}_{X,p}$. 
		\item Let $C^{1+\delta}(\R)$ with $delta > 0$ be the space of continuously differentiable functions on $\R$ with bounded derivative $f'$ which satisfies $\babs{f'(x)-f'(y)} \leqslant C \, \sabs{x-y}^{\delta}$ for all $x,y \in \R$ and $C \geqslant 0$.
		Then, if $A = A^* \in \rr{D}_{X,p}$ and $f \in C^{1+\delta}(\R)$ with $f(0) = 0$ then $f(A) \in \rr{D}_{X,p}$.
	\end{enumerate}
\end{proposition}
\begin{proof}
	The first part of item (1) is proved in \cite[Theorem~4.3]{Pagter_Sukochev:commutator_estimates_R_flows:2007} and implies, in particular, that $\rr{D}_{X,p}\cap\Alg$ is a core for $\partial_X$. The characterization \eqref{framework:eqn:initial_domain_derivation} and the equality $\partial_X(A) = \ii \, [X,A]$ are justified in \cite[Theorem~7.3]{Pagter_Sukochev:commutator_estimates_R_flows:2007} and say that on the core $\rr{D}^{0}_{X,p}$ the operator $\partial_X$ acts as a spatial derivation. Item (2) is proved in \cite[Corollary~5.9]{Pagter_Sukochev:commutator_estimates_R_flows:2007}.
\end{proof}
The case $p = \infty$ is slightly different since $\mathfrak{L}^{\infty}(\Alg) = \Alg$ contains the unit and the $\R$-flow $t \mapsto \eta^X_t$ is only ultra-weakly continuous. Nevertheless, analogs of (1) and (2) exist and have been proven in \cite[Proposition~3.1.6]{Bratteli_Robinson:operator_algebras_1:2002} and \cite[Theorem~3.3.7]{sakai-91}. 

Many of the usual properties of a classical derivative are still
valid in this non-commutative framework.
\begin{proposition}\label{framework:prop:Leibniz_rule}
	The following facts hold true:
	\begin{enumerate}
		\item $\mathcal{T}\big(\partial_X(A)\big)=0$ for all $A \in \rr{D}_{X,1}$. 
		\item For all $A \in \rr{D}_{X,p}$ and $B \in \rr{D}_{X,q}$, $p^{-1} + q^{-1} = 1$ the \emph{Leibniz rule}
		\begin{align}
			\partial_X(A \, B) = \partial_X(A) \, B + A \, \partial_X(B)
			,
			&&
			\label{framework:eqn:Lebnizs_rule_L_p_spaces}
		\end{align}
		holds and we can perform \emph{integration by parts}		
		\begin{align}
			\mathcal{T} \bigl ( A \, \partial_X(B) \bigr ) = -\mathcal{T} \bigl ( \partial_X(A) \, B \bigr )
			.
			\label{framework:eqn:int_by_parts_L_p_spaces}
		\end{align}
	\end{enumerate}
\end{proposition}
\begin{proof}
	\begin{enumerate}
		\item The condition $A \in \rr{D}_{X,1}$ means that in the limit the sequence 
		\begin{align*}
			\Delta_n(A) := n \, \bigl ( \eta_{\nicefrac{1}{n}}^X(A) - A \bigr ) \in \mathfrak{L}^1(\Alg)
		\end{align*}
		converges to $\partial_X(A) \in \mathfrak{L}^1(\Alg)$, and therefore the estimate 
		\begin{align*}
			\Babs{\mathcal{T} \bigl ( \Delta_n(A) \bigr ) - \mathcal{T} \bigl ( \partial_X(A) \bigr )} 
			&\leqslant \mathcal{T} \Bigl ( \babs{\Delta_n(A) - \partial_X(A)} \Bigr ) = \bnorm{\Delta_n(A) - \partial_X(A)}_1
		\end{align*}
		shows that also $\lim_{n \to \infty} \mathcal{T} \bigl ( \Delta_n(A) \bigr ) = \mathcal{T} \bigl ( \partial_X(A) \bigr )$ holds. However, $\mathcal{T} \bigl ( \Delta_n(A) \bigr ) = 0$ since the $\ast$-morphisms $\eta^X_t$ are trace preserving. This proves $\mathcal{T} \bigl ( \partial_X(A) \bigr ) = 0$.
		\item This is a slight modification of \cite[Proposition~4.5]{Pagter_Sukochev:commutator_estimates_R_flows:2007}. The crucial point is that $A \, B \in \mathfrak{L}^1(\Alg)$ due to the non-commutative Hölder inequality, and so the flow $t \mapsto \eta_t^X(A \, B) = \eta_t^X(A) \, \eta_t^X(B)$
		is well-defined in $\mathfrak{L}^1(\Alg)$. Exploiting this factorization, and adding and subtracting terms suitably yields 
		\begin{align*}
			\frac{\eta_t^X(A \, B) - A \, B}{t} = \bigl ( \eta_t^X(A) - A \bigr ) \; \frac{\eta_t^X(B) - B}{t} + A \; \frac{\eta_t^X(B) - B}{t} + \frac{\eta_t^X(A)-A}{t} \; B
			.
		\end{align*}
		Due to the assumptions on $A$ and $B$ we can apply Hölder's inequality to each of the three terms in the sum, and hence, it suffices to estimate each term separately. The difference quotients converge to $\partial_X(A)$ and $\partial_X(B)$, respectively. Hence, the last two terms combine to give the right-hand side of \eqref{framework:eqn:Lebnizs_rule_L_p_spaces}. The first term vanishes as the $\R$-flow is \emph{strongly} continuous, \ie $\lim_{t \to 0} \, \bnorm{\eta_t^X(A) - A}_p = 0$ holds for all $A \in \mathfrak{L}^p(\Alg)$. Thus, the product $A \, B \in \rr{D}_{X,1}$ is differentiable and one gets the Leibniz rule~\eqref{framework:eqn:Lebnizs_rule_L_p_spaces}. Formula~\eqref{framework:eqn:int_by_parts_L_p_spaces} is an immediate consequence of (1) and the Leibniz rule.
	\end{enumerate}
\end{proof}
\begin{remark}\label{rk:leibniz_incrnations}
	The original result in \cite[Proposition~4.5]{Pagter_Sukochev:commutator_estimates_R_flows:2007} establishes the Leibniz rule for the pair $A , B \in \rr{D}_{X,p} \cap \Alg$ (which is different from  our Proposition where $A$ and $B$ are elements in conjugate $\mathfrak{L}^r(\Alg)$-spaces). This implies that $\rr{D}_{X,p} \cap \Alg$ is a $\ast$-subalgebra of $\Alg$. We point out that the Leibniz rule cannot hold true for generic pairs of operators inside the same domain $\rr{D}_{X,p}$ since $\mathfrak{L}^p(\Alg)$ is not closed under the operator product. Finally the Leibniz rule could fail to be true if $A\in \rr{D}_{X,\infty}$ but $A \not\in \mathfrak{L}^p(\Alg)$. Indeed, $A \in \rr{D}_{X,\infty}$ means that the difference quotient \eqref{framework:eqn:derivation_as_derivative_automorphism} converges with respect to the uWOT but not with respect to the norm topology of $\Alg$.
\end{remark}
As an application of the Leibniz rule let us consider a projection $P \in \rr{D}_{X,p} \cap \proj(\Alg)$. Using $P^2 = P$ and the $\Alg$-bimodule structure of $\mathfrak{L}^p(\Alg)$ one gets
\begin{align}
	\partial_X(P) = P \, \partial_X(P) + \partial_X(P) \, P 
	\in \mathfrak{L}^p(\Alg)
	. 
	\label{framework:eqn:derivation_projectionX}
\end{align}
From \eqref{framework:eqn:derivation_projectionX} and the distributivity of the $\Alg$-bimodule structure we immediately deduce $P \, \partial_X(P) \, P = 0 = P^{\perp} \, \partial_X(P) \, P^{\perp}$ where $P^{\perp} = \id - P \in \proj(\Alg)$ is the projection onto the complement. This means that $\partial_X(P)$ is purely offdiagonal, 
\begin{align}
	\partial_X(P) = P \, \partial_X(P) \, P^{\perp} + P^{\perp} \, \partial_X(P) \, P
	. 
	\label{framework:eqn:derivation_projection}
\end{align}
A comparison between \eqref{framework:eqn:double_commutator_projection} and \eqref{framework:eqn:derivation_projection} provides the following result which will be useful in Section~\ref{Kubo_formula:adiabatic_limit}.
\begin{lemma}\label{framework:lem:double_commutator_identity}
	For all $P \in \rr{D}_{X,p} \cap \proj(\Alg)$, $1 \leqslant p \leqslant \infty$, we have the following identity
	\begin{align*}
		\Bigl [ \, P \, , \, \bigl [ P \, , \, \partial_X(P) \bigr ]_{(p)} \Bigr ]_{(p)} = \partial_X(P)
		\in \mathfrak{L}^p(\Alg)
		. 
	\end{align*}
\end{lemma}

\subsection{Non-commutative gradient and Sobolev spaces} 
The notion of $\mathcal{T}$-compatible spatial derivation allows us to define the non-commutative version of the Sobolev spaces.
\begin{definition}[Gradient and Sobolev spaces]\label{framework:defn:non_commutative_Sobolev}
	A \emph{non-commutative ($d$-dimensio\-nal) gradient} in the Banach space $\mathfrak{L}^p(\Alg)$, $1 \leqslant p \leqslant \infty$, is a family of $\mathcal{T}$-compatible spatial derivations $\nabla := \bigl ( \partial_{X_1} , \ldots , \partial_{X_d} \bigr )$ with (maximal) common domain 
	\begin{align*}
		\mathfrak{W}^{1,p}(\Alg) := \bigcap_{j = 1}^d \rr{D}_{X_j,p} 
		= \Bigl \{ A \in \mathfrak{L}^p(\Alg) \; \; \big \vert \; \; \nabla(A) \in \mathfrak{L}^p(\Alg)^{\times d} \Bigr \}
		. 
	\end{align*}
	such that 
	\begin{align}
		\partial_{X_j} \bigl ( \partial_{X_k}(A) \bigr ) = \partial_{X_k} \bigl ( \partial_{X_j}(A) \bigr),
		&&
		A \in \mathfrak{W}^{1,p}(\Alg)
		, 
		\label{framework:eqn:commutativity_derivation}
	\end{align}
	for all $j , k \in \{ 1 , \ldots , d \}$. The domain $\mathfrak{W}^{1,p}(\Alg)$ is called \emph{non-commutative Sobolev space}, and it is a Banach space if endowed with the norm 
	\begin{align*}
		\snorm{A}_{1,p} := \snorm{A}_p + \sum_{j = 1}^d \bnorm{\partial_{X_j}(A)}_p
		,
		&&
		A \in \mathfrak{W}^{1,p}(\Alg)
		.
	\end{align*}
\end{definition}
\begin{remark}[Strongly commuting generators]\label{framework:remark:strongly_commuting_generators}
	The property \eqref{framework:eqn:commutativity_derivation} about the commutativity of the derivations is, ultimately, a condition about the commutativity of the $\R$-flows which generate the derivative, \ie, $\eta^{X_j}_t \circ \eta^{X_k}_s = \eta^{X_k}_s \circ \eta^{X_j}_t$ for each pair of generators $X_j$ and $X_k$ and times $t , s \in \R$. A simple computation shows that this condition is equivalent to require that expressions of the type $\e^{- \ii t X_j} \, \e^{- \ii s X_k} \, \e^{+ \ii t X_j} \, \e^{+ \ii s X_k}$ have to lie in the commutant $\Alg'$. Of course, this is the case if $\e^{- \ii t X_j} \, \e^{- \ii s X_k} \, \e^{+ \ii t X_j} \, \e^{+ \ii s X_k} = \id$, namely if the generators $X_j$ and $X_k$ \emph{strongly commute} \cite[Section~VIII.5]{Reed_Simon:M_cap_Phi_1:1972} or \cite[Section~5.6]{schmudgen-12}. In this situation there exists an invariant common core $\domain_{\mathrm{c}}\subset\Hil$ such that $\domain_{\mathrm{c}} \subset \domain(X_j)$ and $X_j[\domain_{\mathrm{c}}] \subset \domain_{\mathrm{c}}$ for all $j = 1 , \ldots , d$ \cite[Corollary~5.28]{schmudgen-12}. More specifically, $\domain_{\mathrm{c}}$ is a dense set of \emph{analytic vectors} (in the sense of \cite[Section~X.6]{Reed_Simon:M_cap_Phi_2:1975}) for each $X_j$, and we will refer to $\domain_{\mathrm{c}}$ as the \emph{localizing domain} of the $d$-tupel $\{ X_1 , \ldots , X_d \}$. The existence of a joint spectral resolution for strongly commuting families of operators (\eg \cite[Theorem~5.23]{schmudgen-12}) allows to prove that also linear combinations of the form $\lambda_1 \, X_1 + \ldots + \lambda_d \, X_d$ with $\lambda_1 , \ldots , \lambda_d \in \R$ are essentially selfadjoint on $\domain_{\mathrm{c}}$, and therefore uniquely define selfadjoint operators (see \eg \cite[Lemma~2.13]{arai-05} for more details).
\end{remark}
%
%
%
%
\chapter{A Unified Framework for Common Physical Systems} 
\label{unified}
The setting we have detailed in Chapters~\ref{main_results} and \ref{framework} is still rather abstract, so we will spend a few pages on a more concrete framework that directly applies to the most common examples treated in the literature. That includes quantum systems with and without magnetic fields, on the discrete and the continuum, periodic or random. In Chapter~\ref{applications} we anticipate a new application, namely to random Maxwell operators where the Hilbert structure is defined in terms of random weights.

\section{Von Neumann algebra associated to ergodic topological dynamical systems} 
\label{unified:topological_dynamic_systems}
The first ingredient is an ergodic topological dynamical system $\bigl ( \Omega , \mathbb{G} , \tau , \mathbb{P} \bigr )$ consisting of an abelian group $\mathbb{G}$ acting on a probability space $(\Omega , \mathbb{P})$ via the action $\tau$.
\begin{definition}[Ergodic topological dynamical system]\label{unified:defn:topological_dynamical_system}
	An \emph{ergodic topological dynamical system} is a quadruple $(\Omega,\mathbb{G},\tau,\mathbb{P})$ consisting of 
	\begin{enumerate}[(a)]
		\item a separable, metrizable, locally compact abelian group $\mathbb{G}$, 
		\item a standard probability space $(\Omega,\bb{F},\mathbb{P})$ where $\Omega$ is a compact metrizable (hence separable) space, $\bb{F}$ is the Borel $\sigma$-algebra and $\mathbb{P}$ is a Borel measure such that $\mathbb{P}(\Omega)=1$, and 
		\item a representation $\tau:\mathbb{G}\to{\rm Homeo}(\Omega)$ of the group $\mathbb{G}$ by means of homeomorphisms of the space $\Omega$. 
	\end{enumerate}
	These structures are related to each other by the following assumptions:
	\begin{enumerate}[leftmargin=*,label=(\roman*)]
		\item The group action $\tau : \mathbb{G}\times \Omega \to \Omega$ given by $(g,\omega) \mapsto \tau_g(\omega)$ is \emph{jointly} continuous. 
		\item The measure $\mathbb{P}$ is $\tau$-invariant, \ie $\mathbb{P}(\tau_g(\mathtt{B}))=\mathbb{P}( \mathtt{B})$ for all $g \in \mathbb{G}$ and $\mathtt{B}\in \bb{F}$. 
		\item The measure is ergodic, namely if $\mathtt{B}\in \bb{F}$ meets $\tau_g(\mathtt{B}) = \mathtt{B}$ for all $g \in \mathbb{G}$ then $\mathbb{P}(\mathtt{B}) = 1$ or $\mathbb{P}(\mathtt{B}) = 0$.
	\end{enumerate}
\end{definition}
Physically speaking, the probability space $(\Omega , \mathbb{P})$ describes in what way the physical system depends on the random variable $\omega \in \Omega$, and how the different configurations are distributed. The group $\mathbb{G}$ is both, the configuration space and seen as a group of translations acting on the probability space as well as the relevant vector space
\begin{align*}
	\Hil_{\ast} := L^2(\mathbb{G}) \otimes \C^N
	. 
\end{align*}
Here, $L^2(\mathbb{G})$ and the other $L^p(\mathbb{G})$ spaces are defined in terms of the (unique up to scale factors) Haar measure $\mu_\mathbb{G}$, and we have added $\C^N$ to take into account \eg spin-type degrees of freedom. Under the conditions imposed on $\mathbb{G}$ and $(\Omega , \mathbb{P})$, both, $L^2(\mathbb{G})$ and $L^2(\Omega)$ (defined in terms of $\mathbb{P}$) are separable \cite[Théorème~3-4]{de_La_Rue:Lebesgue_spaces:1993}. The most common examples are $\mathbb{G} = \R^d , \Z^d , \T^d$, although this approach can also accommodate more general situations where \eg $\mathbb{G}$ is replaced by a groupoid describing quasicrystals \cite{Lenz:algebras_operators_delone_dynamical_systems:2003,Lenz:groupoids_Neumann_lgebras_integrated_density_states:2007}.

\subsection{Projective representations of $\mathbb{G}$} 
\label{unified:topological_dynamic_systems:projective_representations}
Another choice which influences by the precise physical setting is the type of representation of $\mathbb{G}$ we choose on $\Hil_{\ast}$. For non-magnetic quantum systems, for instance, the generators of translations commute amongst one another, and it is therefore appropriate to choose the standard representation of $\mathbb{G}$ via $\varphi(h) \mapsto \varphi \bigl ( h \, g^{-1} \bigr )$. As a matter of convention we denote the group law with multiplication. 

However, for a lot of interesting applications we instead have to choose a \emph{projective} representation
\begin{align}
	(S_g \varphi)(h) := \Theta \bigl ( g , h \, g^{-1} \bigr ) \, \varphi \bigl ( h \, g^{-1} \bigr ) 
	\label{unified:eqn:G_action}
\end{align}
which now includes an additional  multiplication by $\Theta \bigl ( g , h \, g^{-1} \bigr ) \in \mathbb{U}(1)$. Put another way the product of two translations 
\begin{align*}
	S_{g_1} \, S_{g_2} = \Theta(g_1,g_2) \; S_{g_1 \, g_2}
\end{align*}
is said phase factor times a translation. This phase factor can be thought of as the abstract version of the exponential of $\ii$ times the magnetic circulation $\int_{[x,x-y]} A$ in transversal gauge; We will make this link explicit when we discuss magnetic quantum systems in Chapter~\ref{applications:quantum_Hall_effect:continuum}. 

Inspired by magnetic systems on the one hand and twisted crossed products on the other \cite{Edwards_Lewis:twisted_group_algebras_1:1969,Packer:twisted_group_Cstar_algebras:1989,Mantoiu_Purice_Richard:twisted_X_products:2004,Bedos_Conti:twisted_Fourier_analysis:2009}, it is natural to phrase the discussion in terms of cohomology theory. 
\begin{definition}[Twisting group 2-cocycle]\label{unified:defn:2_cocycle}
	A \emph{twisting group 2-cocycle} is a map $\Theta : \mathbb{G} \times \mathbb{G} \longrightarrow \mathbb{U}(1)$ such that the following holds:
	\begin{enumerate}[leftmargin=*,label=(\roman*)]
		\item Let $e \in \mathbb{G}$ the unit then  $\Theta(e,g) = 1 = \Theta(g,e)$ for all $g \in \mathbb{G}$ (identity property). 
		\item $\Theta(g_1,g_2) \; \Theta(g_1 \, g_2,g_3) = \Theta(g_2,g_3) \; \Theta(g_1,g_2 \, g_3)$  for all $g , g_1 , g_2 , g_3 \in \mathbb{G}$ (cocycle property). 
		\item $\Theta \bigl ( g^{-1} , g \bigr ) = 1$ for all $g \in \mathbb{G}$ (normalization property).
	\end{enumerate}
\end{definition}
Note that the normalization property is strictly speaking not necessary, although it \emph{does} simplify many of the subsequent equations. The following formulas for the the inverse
$$
\Theta(g_1,g_2)^{-1} = \Theta(g_1\,g_2, g_2^{-1}) = \Theta(g_1^{-1},g_1 \, g_2)
$$
turn out to be quite useful in the computations.

In principle, we could refine our arguments to include \emph{generalized} projective representations, where $\Theta$ takes values in some abelian Polish group rather than just $\mathbb{U}(1)$. In this formulation the cocycle property implies that $\Theta = \delta^1(\Lambda)$ is the boundary of some $1$-cochain, and certain aspects such as “changes of gauge” have a natural interpretation. But for the sake of simplicity, we shall not embark on this endeavor; a very elegant exposition in the $C^*$-algebraic context can be found in \cite[Section~2.3]{Mantoiu_Purice_Richard:twisted_X_products:2004}. 

\subsection{Randomly weighted Hilbert spaces} 
\label{unified:topological_dynamic_systems:randomly_weighted_hilbert_spaces}
One complication compared to most of the literature arises when we want to treat random Maxwell operators $M_{\omega} = W_{\omega}^{-1} \, D$, namely that the scalar product 
\begin{align}
	\sscpro{\phi}{\psi}_{\omega} := \bscpro{\phi}{W_{\omega} \, \psi}_{\Hil_{\ast}} 
	\label{unified:eqn:weighted_scalar_product}
\end{align}
includes random weights $\omega \mapsto W_{\omega}$. We shall always impose the following conditions:%
\begin{definition}[Field of weights]\label{unified:defn:random_weights}
	A \emph{field of weights} is a mapping $\Omega \ni \omega \mapsto W_{\omega}\in L^{\infty}(\mathbb{G}) \otimes \mathrm{Mat}_{\C}(N)$\footnote{One quick note on tensor products: the (topological) tensor product of Hilbert spaces, von Neumann algebras and all vector spaces which appear in this chapter are unambiguously defined (see \eg \cite[Section~2.7.2]{Bratteli_Robinson:operator_algebras_1:2002} for von Neumann algebras and \cite[Proposition~40.2]{Treves:topological_vector_spaces:1967}).} such that 
	\begin{enumerate}[leftmargin=*,label=(\roman*)]
		\item there is a matrix valued measurable function $w : \Omega \to \mathrm{Mat}_{\C}(N)$ such that $W_{\omega}(g) := w \bigl ( \tau_{g^{-1}}(\omega) \bigr )$ for $\mu_{\mathbb{G}}$-almost all $g \in \mathbb{G}$ and 
		$\mathbb{P}$-almost all $\omega\in\Omega$, and 
		\item $W_{\omega}$ is a positive  and invertible element of the von Neumann algebra $L^{\infty}(\mathbb{G}) \otimes \mathrm{Mat}_{\C}(N)$ for all  $\omega \in \Omega$.
	\end{enumerate}
\end{definition}
The first condition states that the matrix-valued function satisfies the \emph{covariance relation}
\begin{align}
	W_{\omega} \bigl ( h \, g^{-1} \bigr ) &= W_{\tau_g(\omega)}(h) 
	, 
	\label{unified:eqn:covariance_relation}
\end{align}
which holds for almost all $g , h \in \mathbb{G}$ and $\omega \in \Omega$. Sometimes condition~(i) is replaced by the stronger requirement that $\omega \mapsto W_{\omega}$ be continuous where $L^{\infty}(\Omega) \otimes \mathrm{Mat}_{\C}(N)$, seen as dual of $L^{1}(\Omega) \otimes \mathrm{Mat}_{\C}(N)$,  is endowed with the $\ast$-weak topology. 

Condition~(ii) implies that $W_{\omega}$ and its inverse $W_{\omega}^{-1}$ are bounded from $0$ and $\infty$, and therefore the norm $\norm{\psi}_{\omega} := \sscpro{\psi}{\psi}_{\omega}^{\nicefrac{1}{2}}$ is equivalent to the $\norm{\cdot}_{\Hil_{\ast}}$-norm. Hence, we define 
\begin{align*}
	\Hil_{\omega} := L^2_{W_{\omega}}(\mathbb{G},\C^N)
\end{align*}
as the \emph{Banach} space $\Hil_{\ast}$ endowed with the $\omega$-dependent scalar product~\eqref{unified:eqn:weighted_scalar_product}. If we pick the representation of $\mathbb{G}$ on $\Hil_{\ast}$ via \eqref{unified:eqn:G_action}, we can rewrite \eqref{unified:eqn:covariance_relation} as $S_g \, W_{\omega} \, S_g^{-1} = W_{\tau_g(\omega)}$. Of course, here it was crucial that the $2$-cocycle can be viewed as a $\mathbb{U}(1)$-valued multiplication operator, and therefore commutes with all weights $W_{\omega}$. Thanks to the covariance condition the operator $S_g$ defines a family of unitaries
\begin{align}
	U_{g,\tau_g(\omega)} : \Hil_{\omega} \longrightarrow \Hil_{\tau_g(\omega)}
	, 
	\; \; 
	\psi_{\omega} \mapsto S_g \psi_{\omega}
	, 
	\label{unified:eqn:unitary_translations_Hil_omega}
\end{align}
indexed by $\omega \in \Omega$; Note, though, that in general $S_g$ does \emph{not} act unitarily if seen as a map from $\Hil_{\omega}$ to itself. 

\subsection{Direct integral of Hilbert spaces} 
\label{unified:topological_dynamic_systems:direct_integral}
We will now show how these $\Hil_{\omega}$ can be “glued together” in the form of a \emph{direct integral of Hilbert spaces}
\begin{align*}
	\Hil := \int^{\oplus}_{\Omega} \dd \mathbb{P}(\omega) \; \Hil_{\omega}
	, 
\end{align*}
which can the thought of as a generalization of the tensor product space $L^2(\Omega) \otimes \Hil_{\ast}$. We outline the relevant two key ideas, namely the notions of \emph{measurability} and \emph{covariance}; The interested reader can find detailed accounts in \cite[Part~II, Chapters~1–5]{dixmier-81}. To identify a subset of measurable vector fields of the \emph{field of Hilbert spaces} 
\begin{align*}
	\mathcal{F} := \prod_{\omega \in \Omega} \Hil_{\omega} 
	, 
\end{align*}
it is necessary to select a \emph{fundamental sequence of basis vectors}. In our case, the $\Hil_{\omega}$ are all isomorphic to one another, and therefore the question of measurability \emph{by itself} is trivial \cite[p.~167, Proposition~3]{dixmier-81}. However, what is \emph{a priori} not clear is whether the measurable structure is \emph{compatible with the covariance relation}~\eqref{unified:eqn:covariance_relation} because the latter induces a relation between the “fiber spaces” $\Hil_{\omega}$ and $\Hil_{\tau_g(\omega)}$. Elements of $\mathcal{F}$ are merely maps $\widehat{\varphi} : \omega \mapsto \varphi_{\omega} \in \Hil_{\omega}$, although we will also write $\widehat{\varphi} = \{ \varphi_{\omega} \}_{\omega \in \Omega}$. Now we pick any basis $\{ \eta_\alpha \}_{\alpha \in \mathcal{I}}$ of $\Hil_{\ast}$ where $\mathcal{I} \subseteq \N$ is finite or countably infinite depending on whether $\mathbb{G}$ is finite or not, and define $\widehat{\varphi}_{\alpha} := \bigl \{ \varphi_{\alpha,\omega} \bigr \}_{\omega \in \Omega}$, $\varphi_{\alpha,\omega} := \eta_\alpha$, in terms of the fixed basis $\eta_\alpha$. Note that these vectors are actually \emph{in}dependent of $\omega$ and still form a basis of $\Hil_{\omega}$ as the latter agrees with $\Hil_{\ast}$ as a vector space. In view of \cite[p.~167, Proposition~4]{dixmier-81} the measurability of 
\begin{align*}
	\omega \mapsto \bscpro{\varphi_{\alpha,\omega}}{\varphi_{\beta,\omega}}_{\omega} = \bscpro{\eta_\alpha}{W_{\omega} \, \eta_\beta}_{\Hil_{\ast}} 
\end{align*}
for all indices $\alpha , \beta \in \mathcal{I}$ suffices to uniquely identify a measurable structure. But the measurability of the above expression is an immediate consequence of our measurability assumptions on $W_{\omega}$ (Definition~\ref{unified:defn:random_weights}) and Fubini's Theorem. The set of measurable vectors now forms a subspace of $\mathcal{F}$ which consists of exactly those vectors $\widehat{\psi} = \{ \psi_{\omega} \}_{\omega \in \Omega}$ for which   $\omega \mapsto \bscpro{\varphi_{\alpha,\omega}}{\psi_{\omega}}_{\omega}$ is measurable for all $\alpha \in \mathcal{I}$. 

Our specific choice of fundamental sequence of basis vectors is also \emph{compatible with the covariance relation} which translates to~\eqref{unified:eqn:unitary_translations_Hil_omega}. For this we set $\widehat{U}_g := \{ U_{g,\tau_g(\omega)} \}_{\omega \in \Omega}$ as the collection of the unitaries from equation~\eqref{unified:eqn:unitary_translations_Hil_omega}. Given a measurable vector field $\widehat{\psi}$ one has that $\widehat{U}_g$ acts on it according to 
\begin{align*}
	\bigl ( \widehat{U}_g \widehat{\psi} \bigr )_{\tau_g(\omega)}(h) &= \Theta(g , h \, g^{-1}) \, \psi_{\omega} \bigl ( h \, g^{-1} \bigr ) \in \Hil_{\tau_g(\omega)}
\end{align*}
From this formula one can infer that also the collection of the $\widehat{U}_g \widehat{\varphi}_\alpha$ consists of measurable vector fields which span the $\Hil_{\omega}$, and therefore they give an equivalent fundamental sequence of basis vectors. 

To obtain the direct integral of Hilbert spaces, we merely need to identify identify two measurable vector fields if they agree on a set of full measure (with respect to $\mathbb{P}$) and require that their norm 
\begin{align*}
	\bnorm{\widehat{\psi}}_{\widehat{\Hil}} := \sqrt{\int_{\Omega} \dd \mathbb{P}(\omega) \, \bnorm{\psi_{\omega}}_{\omega}^2} 
	< \infty
\end{align*}
be finite. Endowed with the scalar product 
\begin{align*}
	\bscpro{\widehat{\varphi}}{\widehat{\psi}}_{\widehat{\Hil}} := \int_{\Omega} \dd \mathbb{P}(\omega) \, \bscpro{\varphi_{\omega}}{\psi_{\omega}}_{\omega} 
	, 
\end{align*}
we denote the resulting Hilbert space, the direct integral, with 
\begin{align*}
	\widehat{\Hil} := \int^{\oplus}_{\Omega} \dd \mathbb{P}(\omega) \, \Hil_{\omega} 
	. 
\end{align*}
It is straightforward to check that that the operators $\widehat{U}_g$ are isometric and invertible, hence unitary with respect to the Hilbert structure of $\widehat{\Hil}$, and this is due in large part to the invariance of the measure $\mathbb{P}$. In case the field of weights is constant, \eg $W_{\omega} = \id_{\Hil_{\ast}}$ for all $\omega$, the direct integral we have just construct can be \emph{canonically} identified with $L^2(\Omega) \otimes \Hil_{\ast}$ \cite[p.~174, Proposition~11]{dixmier-81}. Even if the field of weights is non-trivial, the two Hilbert spaces $\widehat{\Hil} \simeq L^2(\Omega) \otimes \Hil_{\ast}$ are still isomorphic, albeit no longer canonically isomorphic. 

\subsection{The algebra of covariant random operators} 
\label{unified:topological_dynamic_systems:covariant_random_operators}
The relevant von Neumann algebra is a subalgebra of $\mathscr{B}(\widehat{\Hil})$ composed of covariant operators. Our discussion here mirrors the construction of $\widehat{\Hil}$ where we need to first deal with questions of measurability and then, in a second step, with covariance. 

Let us start with the notion of \emph{random operators}: this is a bounded-operator valued map $\omega \mapsto A_{\omega} \in \mathscr{B}(\Hil_{\omega})$, collectively denoted by $\widehat{A} = \{ A_{\omega} \}_{\omega \in \Omega}$, so that
\begin{enumerate}[leftmargin=*,label=(\roman*)]
	\item $\omega \mapsto \bscpro{\eta_\alpha}{A_{\omega} \eta_\beta}_{\omega}$ is measurable for all $\alpha , \beta \in \mathcal{I}$, where the $\eta_\alpha$ are the $\omega$-independent basis vectors which have determined the set of measurable vector fields and 
	\item the family is \emph{essentially bounded}, $\mathrm{ess \, sup}_{\omega \in \Omega} \bnorm{A_{\omega}}_{\omega} < \infty$. 
\end{enumerate}
They act on $\widehat{\varphi} \in \widehat{\Hil}$ as $\bigl ( \widehat{A} \; \widehat{\varphi} \bigr )_{\omega} := A_{\omega} \, \varphi_{\omega}$. Evidently, we need to identify $\widehat{A}$ and $\widehat{B}$ if $A_{\omega} = B_{\omega}$ agree for $\mathbb{P}$-almost all $\omega$. We call such operators \emph{random} or \emph{decomposable} bounded operators, and denote the set they form with $\mathrm{Rand}(\widehat{\Hil}) \subseteq \mathscr{B}(\widehat{\Hil})$. The $\mathscr{B}(\widehat{\Hil})$-norm of $\widehat{A} \in \mathrm{Rand}(\widehat{\Hil})$ equals its essential supremum \cite[Part~II, Chapter~2, Proposition~2]{dixmier-81}, 
\begin{align*}
	\bnorm{\widehat{A}}: = \mathrm{ess \, sup}_{\omega \in \Omega} \bnorm{A_{\omega}}_{\omega}
	. 
\end{align*}
An important subalgebra of $\mathrm{Rand}(\widehat{\Hil})$ is $L^{\infty}(\Omega)$ itself through the identification $f \mapsto \bigl \{ f(\omega) \, \id_{\Hil_{\omega}} \bigr \}_{\omega \in \Omega}$, and is usually referred to as \emph{diagonal} decomposable operators. Evidently, $L^{\infty}(\Omega) \subset \mathrm{Rand}(\widehat{\Hil})$ is an abelian von Neumann algebra, and it is straight-forward to show that in fact $L^{\infty}(\Omega)' = \mathrm{Rand}(\widehat{\Hil})$ \cite[p.~188, Corollary]{dixmier-81}. That concludes questions of measurability. 

Now on to covariance: This is formulated in terms of a projective $\mathbb{G}$-representation on $\widehat{\Hil}$: given $g \in \mathbb{G}$ we define the operator $\widehat{U}_g := \bigl \{ U_{g,\tau_g(\omega)} \bigr \}_{\omega \in \Omega}$ through \eqref{unified:eqn:unitary_translations_Hil_omega}. Evidently, such operators \emph{cannot} be elements of $\mathrm{Rand}(\widehat{\Hil})$ (unless, of course, $g = e$ is the identity in $\mathbb{G}$), because the $U_{g,\tau_g(\omega)} : \Hil_{\omega} \longrightarrow \Hil_{\tau_g(\omega)}$ map between \emph{different} Hilbert spaces. Nevertheless, we easily verify $\widehat{U}_g^{-1} = \widehat{U}_g^* = \widehat{U}_{g^{-1}}$ and $\widehat{U}_e = \id_{\widehat{\Hil}}$, so that indeed $g \mapsto \widehat{U}_g$ forms a projective unitary representation of $\mathbb{G}$. Covariant operators are now those decomposable operators on $\widehat{\Hil}$ that commute with all the $\widehat{U}_g$. 
\begin{definition}[Covariant random operators]
	The set of random covariant operators 
	\begin{align*}
		\Alg(\Omega,\mathbb{P},\mathbb{G},\Theta) := \mathrm{Span}_{\mathbb{G}} \bigl \{ \widehat{U}_g \bigr \}' \cap \mathrm{Rand}(\widehat{\Hil})
	\end{align*}
	consists of those decomposable operators which are \emph{covariant} with respect to the projective unitary representation $\mathbb{G} \ni g \mapsto \widehat{U}_g \in \mathscr{B}(\widehat{\Hil})$, \ie those for which
	\begin{align}
		U_{g,\tau_g(\omega)} \, A_{\omega} \, U_{g,\tau_g(\omega)}^{-1} = A_{\tau_{g}(\omega)}
		\label{unified:eqn:covariance_condition}
	\end{align}
	holds for all $g \in \mathbb{G}$ and $\mathbb{P}$-almost all $\omega \in \Omega$. 
\end{definition}
\begin{proposition}
	The set $\Alg(\Omega,\mathbb{P},\mathbb{G},\Theta)$ is a von Neumann sub-algebra of $\mathscr{B}(\widehat{\Hil})$. 
\end{proposition}
\begin{proof}
	Both $\mathrm{Rand}(\widehat{\Hil}) = L^{\infty}(\Omega)'$ and $\mathrm{Span}_\mathbb{G}\{\widehat{U}_g\}'$ are von Neumann algebras since they are communtants of $\ast$-subalgebras of $\mathscr{B}(\widehat{\Hil})$. 
	Then $\Alg(\Omega,\mathbb{P},\mathbb{G},\Theta)$ is the intersection of two von Neumann algebras, and therefore itself a von Neumann algebra \cite[Part~I, Chapter~1, Section~1, Proposition~1]{dixmier-81}. 
\end{proof}
An important example of a covariant random operator is the field of weights $\widehat{W} := \{ W_{\omega} \}_{\omega \in \Omega} \in \Alg(\Omega,\mathbb{P},\mathbb{G},\Theta)$ which defines $\widehat{\Hil}$ in the first place as well as its inverse $\widehat{W}^{-1}$. The second one noteworthy are \emph{random potentials} $\widehat{V}$, namely those composed of $\omega \mapsto V_{\omega} \in L^{\infty}(\mathbb{G}) \otimes \mathrm{Mat}_{\C}(N)$ that satisfy the measurability and covariance conditions. 
The last important example is the constant field $\widehat{H} := \{ H \}_{\omega \in \Omega}$ associated to a selfadjoint operator on $\Hil_{\ast}$ that satisfies $S_g \, H \, S_g^* = H$ for all $g \in \mathbb{G}$. Due to the presence of the weights $W_{\omega}$, the operator $H$ is no longer selfadjoint on the weighted Hilbert space $\Hil_{\omega}$ even though its domain $\domain(H) \subseteq \Hil_{\omega}$ is still dense. Should $H$ and therefore $\widehat{H}$ be bounded, then $\widehat{H}$ is an element of $\Alg(\Omega,\mathbb{P},\mathbb{G},\Theta)$. Even in case $H$ is unbounded, $\widehat{H}$ still commutes with all the $\widehat{U}_g$'s and the diagonal operators in $L^{\infty}(\Omega) \subseteq \Alg(\Omega,\mathbb{P},\mathbb{G},\Theta)$. Therefore, $\widehat{H}$ is in any case affiliated to $\Alg(\Omega,\mathbb{P},\mathbb{G},\Theta)$. 

The reason to include the weights in the definition of the $\Hil_{\omega}$ is to establish a principle of affiliation (\cf Definition~\ref{framework:defn:affiliation}) for operators of the form 
\begin{align}
	\widehat{H}_{W,V} := \widehat{W}^{-1} \, \widehat{H} + \widehat{V} 
	\label{unified:eqn:Maxwell_Schroedinger_hybrid_operator}
\end{align}
where $\widehat{W}^{-1} \in \Alg(\Omega,\mathbb{P},\mathbb{G},\Theta)$ is the inverse of the field of weights, $\widehat{V} \in \Alg(\Omega,\mathbb{P},\mathbb{G},\Theta)$ is a random potential, and $\widehat{H} := \{ H \}_{\omega \in \Omega}$ a constant field as above. Hence, this covers random Schrödinger operators $\widehat{H} + \widehat{V}$ and random Maxwell operators $\widehat{W}^{-1} \, \widehat{H}$ where $\widehat{W} = \id_{\widehat{\Hil}}$ and $\widehat{V} = 0$, respectively. 
\begin{proposition}\label{unified:prop:affiliation}
	Let $\widehat{W} := \{ W_{\omega} \}_{\omega \in \Omega}$ be the operator associated to the field of weights, $\widehat{V} := \{ V_{\omega} \}_{\omega \in \Omega}$ a random potential, and $\widehat{H} := \{ H \}_{\omega \in \Omega}$ an unbounded operator with constant fiber such that $H$ is selfadjoint on $\Hil_{*}$ and $S_g \, H \, S_g^* = H$. Assume in addition that $V_{\omega}$ is selfadjoint on $\Hil_{\ast}$ and $H$-bounded for $\mathbb{P}$-almost all $\omega \in \Omega$, and $\bigl [ \widehat{W} , \widehat{V} \bigr ] = 0$. Then $\widehat{H}_{W,V}$ is selfadjoint and $\widehat{H}_{W,V} \in \affil \bigl ( \Alg(\Omega,\mathbb{P},\mathbb{G},\Theta) \bigr )$.
\end{proposition}
\begin{proof}
	Since $\widehat{H}$ is affiliated with $\Alg(\Omega,\mathbb{P},\mathbb{G},\Theta)$ and $\widehat{W}$ is a bounded and invertible element in $\Alg(\Omega,\mathbb{P},\mathbb{G},\Theta)$ it follows that the product $\widehat{W}^{-1} \, \widehat{H}$ is densely defined and closable, and so define a closed densely defined operator affiliated with $\Alg(\Omega,\mathbb{P},\mathbb{G},\Theta)$, namely $\widehat{W}^{-1} \, \widehat{H}\in \affil \bigl ( \Alg(\Omega,\mathbb{P},\mathbb{G},\Theta) \bigr )$. It is immediate to check fiberwise that $W_{\omega}^{-1}H$ is a selfadjoint element in $\Hil_{\omega}$, hence $\widehat{W}^{-1} \, \widehat{H}$ turns out to be selfadjoint on the direct integral $\widehat{\Hil}$. The commutation relation $\bigl [ \widehat{W},\widehat{V} \bigr ] = 0$ ensures that $\widehat{V}$ is a $\widehat{H}$-bounded selfadjoint element in $\affil \bigl ( \Alg(\Omega,\mathbb{P},\mathbb{G},\Theta) \bigr )$. The the sum $\widehat{W}^{-1} \, \widehat{H} + \widehat{V}$ is clearly affiliated with $\affil \bigl ( \Alg(\Omega,\mathbb{P},\mathbb{G},\Theta) \bigr )$. To prove that is selfadjoint (hence densely defined and closed) it is enough to use the Kato-Rellich theorem fiberwise.
\end{proof}
Let us lastly mention another interesting subalgebra: for any $f \in C_{\mathrm{c}}(\Omega \times \mathbb{G}) \otimes \mathrm{Mat}_{\C}(N)$, \ie $f$ is continuous and has compact support, we can associate a \emph{twisted convolution operator} 
\begin{align}
	\bigl ( \widehat{K}_f \widehat{\varphi} \bigr )_{\omega}(h) := \int_{\mathbb{G}} \dd \mu_{\mathbb{G}}(g) \; f \bigl ( \tau_{h^{-1}}(\omega) , g \, h^{-1} \bigr ) \, \Theta(g,hg^{-1}) \, \varphi_{\omega}(g) 
	\label{unified:eqn:twisted_convolution_operator}
\end{align}
on $\widehat{\Hil}$, and it is easy to see that this defines a bounded decomposable operator which commutes with the $\widehat{U}_g$'s, namely
an element of $\Alg(\Omega,\mathbb{P},\mathbb{G},\Theta)$ (\cf with \cite[Proposition 1.2.1]{Lenz:random_operators_crossed_products:1999}. In fact, the map $f \mapsto \widehat{K}_f \in \Alg(\Omega,\mathbb{P},\mathbb{G},\Theta)$ is a faithful morphism, \ie $\widehat{K}_f = 0$ if and only if $f = 0$ (this fact is essentially due to the amenability of the abelian group $\mathbb{G}$ \cite[Section~7.7]{Pedersen1979}), and we denote its image with $\Alg_0 \subseteq \Alg(\Omega,\mathbb{P},\mathbb{G},\Theta)$. The set $\Alg_0$ has the structure of a \emph{pre}-$C^*$-algebra and its closure with respect to the operator norm of $\widehat{\Hil}$ defines a (faithful representation of a) $C^*$-algebra usually called \emph{(twisted) crossed product} in the literature \cite{Pedersen1979,Mantoiu_Purice_Richard:twisted_X_products:2004,Williams2007}. Standard approximation arguments show that 
$\Alg_0''= \Alg(\Omega,\mathbb{P},\mathbb{G},\Theta)$, namely the von Neumann algebra generated by $\Alg_0$ agree with the full algebra $\Alg(\Omega,\mathbb{P},\mathbb{G},\Theta)$ (\cf \cite[Theorem 6]{Bellissard_van_Elst_Schulz_Baldes:noncommutative_geometry_quantum_Hall_effect:1994} and references therein).

What makes $\widehat{K}_f$ interesting is its connection to pseudodifferential theory: equation~\eqref{unified:eqn:twisted_convolution_operator} gives rise to a calculus which can be regarded as an abstract version of the Wigner-Weyl calculus \cite{Mantoiu_Purice_Richard:twisted_X_products:2004,Lein_Mantoiu_Richard:anisotropic_mag_pseudo:2009,Belmonte_Lein_Mantoiu:mag_twisted_actions:2010} that connects algebra and functional analysis. Because $f \mapsto \widehat{K}_f$ is faithful, the product and the involution on the von Neumann algebra can be pulled back to the realm of suitable functions on $\Omega \times \mathbb{G}$. This allows one to infer properties of the “quantization” $\widehat{K}_f$ from properties of the “symbol” $f$. 

\section{The trace per unit volume} 
\label{unified:trace}
The aforementioned algebra also plays a role in the definition of the \emph{canonical} f.n.s.\ trace, the so-called \emph{trace per unit volume} $\mathcal{T}_{\mathbb{P}}$. We follow and generalize the construction in \cite{Lenz:random_operators_crossed_products:1999}. In our specific situation the new difficulties are presented by the presence of random weights $W_{\omega}$ and the extra degrees of freedom described by the tensor product with $\C^N$. For reasons of space we only sketch the main points of the construction here and postpone a more detailed analysis to a future work \cite{DeNittis_Lein:LRT_light:2017}.

It turns out the $\ast$-subalgebra $\Alg_0$ from the previous Section can be endowed with a scalar product, 
\begin{align}
	\bscpro{\widehat{K}_{f_1}}{\widehat{K}_{f_2}}_{\Alg_0} := \int_{\Omega \times \mathbb{G}} \dd \mathbb{P}(\omega) \, \dd \mu_{\mathbb{G}}(g) \; \mathrm{Tr}_{\C^N} \Bigl ( f_1(\omega,g)^* \, W_{\omega}(e) \, f_2(\omega,g) \, W_{\omega}(g) \Bigr )
	,  
	\label{unified:eqn:scalar_product_Alg_0}
\end{align}
which makes $\Alg_0$ into a (quasi-)Hilbert algebra in the sense of \cite[Part~I, Chapter~5]{dixmier-81}. To make this claim rigorous we need to assume that the weights $W_{\omega}$ are not only selfadjoint but also real valued, \ie $W_{\omega}(g)^* = W_{\omega}(g) = \overline{W_{\omega}(g)}$ for all $g \in \mathbb{G}$ and $\omega \in \Omega$. Clearly, the completion of the Hilbert algebra $\Alg_0$ with respect to its Hilbert structure defines an Hilbert space which is isomorphic (at least as Banach spaces) to $L^2(\Omega \times \mathbb{G}) \otimes \mathrm{Mat}_{\C}(N)$ under the identification between operators and symbols. The relevant fact is that the left-standard von Neumann algebra associated with $\Alg_0$ turns out to be (spatially) isomorphic to $\Alg(\Omega,\mathbb{P},\mathbb{G},\Theta)$ (this is a slight modification of \cite[Theorem~2.1.4]{Lenz:random_operators_crossed_products:1999} due to the presence of the $2$-cocycle $\Theta$).

We can extend both, the scalar product~\eqref{unified:eqn:scalar_product_Alg_0} and the representation~\eqref{unified:eqn:twisted_convolution_operator}. However, while the scalar product \eqref{unified:eqn:scalar_product_Alg_0} still makes sense for $f \in L^2(\Omega \times \mathbb{G}) \otimes \mathrm{Mat}_{\C}(N)$, we can no longer guarantee that the associated $\widehat{K}_f$ actually defines an element of $\Alg(\Omega,\mathbb{P},\mathbb{G},\Theta)$. Nevertheless, if we impose the latter as an added condition, we obtain the selfadjoint ideal 
\begin{align*}
	\Alg_{L^2} :& \negmedspace= \Bigl \{ \widehat{K}_f \; \; \big \vert \; \; f \in L^2(\Omega \times \mathbb{G}) \otimes \mathrm{Mat}_{\C}(N) , \; \widehat{K}_f \in \Alg(\Omega,\mathbb{P},\mathbb{G},\Theta) \Bigr \} 
	\\
	&\subset \Alg(\Omega,\mathbb{P},\mathbb{G},\Theta)
	, 
\end{align*}
which contains $\Alg_0 \subset \Alg_{L^2}$ \cite[Proposition~2.1.6]{Lenz:random_operators_crossed_products:1999}, and therefore also $\Alg_{L^2}$ is dense in $\Alg(\Omega,\mathbb{P},\mathbb{G},\Theta)$ with respect to the SOT. Any $\widehat{K}_f \in \Alg_{L^2}$ acts fiberwise on $\Hil_{\omega}$ as an integral Carleman operator ${K}_{f,\omega}$ with kernel
\begin{align}
	f_{\omega}(h,g) := f \bigl ( \tau_{h^{-1}}(\omega) , g \, h^{-1} \bigr ) \, \Theta(g,hg^{-1})
	.
	\label{eq:extended_kernel}
\end{align}
Evidently, the kernel $f_{\omega}$ can be seen as an element of $\Hil_{\omega} \otimes \Hil_{\omega}$ by means of the (Banach space) isomorphism $\Hil_{\omega} \otimes \Hil_{\omega} \simeq \bigl ( L^2(\mathbb{G}) \otimes L^2(\mathbb{G}) \bigr ) \otimes \mathrm{Mat}_{\C}(N)$ which is essentially due to $\C^N \otimes \C^N \simeq \mathrm{Mat}_{\C}(N)$. More precisely the following expression
\begin{align*}
	f_{\omega}(h,g)_{j,k} := \sum_{\alpha,\beta \in \mathcal{I}} F_{\alpha,\beta}(\omega) \, \varphi_{\alpha,\omega}^j(h) \otimes \overline{\varphi_{\beta,\omega}^k(g)}
	,
	&&
	j , k = 1 , \ldots , N
	,
\end{align*}
connects the matrix entries of $f_{\omega}$ with the “double” eigenvector expansion via an orthonormal basis of $\Hil_{\omega}$ given by $\varphi_{\alpha,\omega} := \bigl ( \varphi_{\alpha,\omega}^1 , \ldots , \varphi_{\alpha,\omega}^N \bigr )$ where $\alpha \in \mathcal{I}$ and $\mathcal{I} \subseteq \N$ is a countable or uncountable subset depending on whether $\mathbb{G}$ is finite or not. The standard Hilbert-Schmidt argument (\cf \cite[Theorem~VI.23]{Reed_Simon:M_cap_Phi_1:1972}) provides
\begin{align}
	\mathrm{Tr}_{\Hil_{\omega}} \bigl ({K}_{f,\omega}^*\, {K}_{f,\omega} \bigr ) &= \sum_{\alpha,\beta\in\mathcal{I}}\big|F_{\alpha,\beta}(\omega)\big|^2= \snorm{f_{\omega}}_{\Hil_{\omega} \otimes \Hil_{\omega}}^2 
	.
	\label{eq:tra_H_norm_HH}
\end{align}
It is useful to notice that the norm in the space  
$\Hil_{\omega} \otimes \Hil_{\omega}$ which appears in the equation above can be rewritten in matrix notation according to 
\begin{align*}
	\snorm{f_{\omega}}_{\Hil_{\omega} \otimes \Hil_{\omega}}^2 = \int_{\mathbb{G} \times \mathbb{G}} \dd \mu_{\mathbb{G}}(h) \, \dd \mu_{\mathbb{G}}(g) \, \mathrm{Tr}_{\C^N} \Bigl ( f_{\omega}(h,g)^* \, W_{\omega}(h) \, f_{\omega}(h,g)\, W_{\omega}(g) \Bigr )
	 .
\end{align*}
Now, the crucial aspect of this construction is that left-standard von Neumann algebras associated with Hilbert algebras admit a unique (up to isomorphisms) \emph{natural} f.n.s.\ trace \cite[Part~I, Chapter~6, Section~2]{dixmier-81}. This trace can be transported by the aforementioned isomorphism to $\Alg(\Omega,\mathbb{P},\mathbb{G},\Theta)$, thereby obtaining a f.n.s.\ trace $\mathcal{T}_{\mathbb{P}}$ on $\Alg(\Omega,\mathbb{P},\mathbb{G},\Theta)$ characterized by the property
\begin{align}
	\mathcal{T}_{\mathbb{P}} \bigl ( \widehat{K}_f^* \; \widehat{K}_f \bigr ) := \snorm{f}_{\Alg_0}^2 
	, 
	&&
	\forall \widehat{K}_f \in \Alg_{L^2}
	. 
	\label{unified:eqn:characterizing_property_trace_per_unit_volume}
\end{align}
Here, the norm $\norm{\cdot}_{\Alg_0}$ is obtained from the scalar product \eqref{unified:eqn:scalar_product_Alg_0}. For the detailed justification of \eqref{unified:eqn:characterizing_property_trace_per_unit_volume} in the case of a trivial field of weights one can follow the strategy of \cite[Proposition~2.1.6 and Theorem~2.2.2]{Lenz:random_operators_crossed_products:1999}. There is an useful formula to compute this trace.
\begin{proposition}[Trace per unit volume]\label{unified:prop:trace_per_unit_volume}
	Let $\lambda\in L^{\infty}(\mathbb{G})$ be any positive normalized function, $\int_{\mathbb{G}} \dd \mu_{\mathbb{G}}(g) \, \lambda(g)^2 = 1$. Let $\widehat{M}_{\lambda} \in  \Alg(\Omega,\mathbb{P},\mathbb{G},\Theta)$ be the constant operator which acts on each fiber of the direct integral $\widehat{\Hil}$ as the multiplication operator $M_{\lambda} := \lambda \otimes \id_{\C^N}$. Then the formula
	\begin{align}
		\mathcal{T}_{\mathbb{P}}(\widehat{A}) := \int_{\Omega} \dd \mathbb{P}(\omega) \; \mathrm{Tr}_{\Hil_{\omega}} \bigl ( M_{\lambda} \, A_{\omega} \, M_{\lambda} \bigr ) 
		,
		&&
		\widehat{A} \in \Alg(\Omega,\mathbb{P},\mathbb{G},\Theta)^+
		, 
		\label{unified:eqn:formula_trace_per_unit_volume}
	\end{align}
	agrees with the f.n.s.\ trace on $\Alg(\Omega,\mathbb{P},\mathbb{G},\Theta)$ defined on $ \Alg_{L^2}$ by \eqref{unified:eqn:characterizing_property_trace_per_unit_volume}.
\end{proposition}
\begin{proof}
	This result can be obtained by generalizing the arguments in the proofs of \cite[Lemma~2.2.6 and Theorem~2.2.7]{Lenz:random_operators_crossed_products:1999}. First of all, we note that the multiplication operator $M_{\lambda}$ is selfadjoint on each space $\Hil_{\omega}$ since it commutes with $W_{\omega}$. Let $\chi_{\mathtt{B}}$ be the characteristic function of a borelian set $\mathtt{B}\subset\mathbb{G}$  and $\widehat{A} := \{ A_{\omega} \}_{\omega \in \Omega}$ a positive element such that $A_{\omega} = C_{\omega}^* \, C_{\omega}$. Then we have 
	\begin{align*}
		\int_{\Omega} \dd \mathbb{P}(\omega) \;  \mathrm{Tr}_{\Hil_{\omega}} \bigl ( M_{\chi_{\mathtt{B}}} \, A_{\omega} \, M_{\chi_{\mathtt{B}}} \bigr )
		= \int_{\Omega} \dd \mathbb{P}(\omega) \; \mathrm{Tr}_{\Hil_{\omega}} \bigl ( C_{\omega} \,  M_{\chi_{\mathtt{B}}} \, C_{\omega}^* \bigr ) 
		=: \nu_{\hat{A}}(\mathtt{B})
	\end{align*}
	where we used the cyclicity of the trace. The mapping $\mathtt{B} \mapsto \nu_{\hat{A}}(\mathtt{B})$ defines a measure, and simple monotone convergence arguments show that
	\begin{align*}
		\int_{\mathbb{G}} \dd \nu_{\hat{A}}(g) \, \lambda(g)^2 = \int_{\Omega} \dd \mathbb{P}(\omega) \; \mathrm{Tr}_{\Hil_{\omega}} \bigl ( C_{\omega} \, M_{\lambda^2} \, C_{\omega}^* \bigr ) 
		= \int_{\Omega} \dd \mathbb{P}(\omega) \; \mathrm{Tr}_{\Hil_{\omega}} \bigl ( M_{\lambda} \, A_{\omega} \, M_{\lambda} \bigr ) 
		.
	\end{align*}
	The arguments in \cite[Lemma~2.2.6]{Lenz:random_operators_crossed_products:1999} ensure that the measure $\nu_{\hat{A}}$ is $\mathbb{G}$-invariant: covariance tells us $M_{\chi_{g \, \mathtt{B}}}= U_{g,\tau_g(\omega)}^{-1} \, M_{\chi_{\mathtt{B}}} \, U_{g,\tau_g(\omega)}$, and consequently, the measure of the shifted set 
	\begin{align*}
		\nu_{\hat{A}}(g \, \mathtt{B}) &= \int_{\Omega} \dd \mathbb{P}(\omega) \; \mathrm{Tr}_{\Hil_{\omega}} \Bigl ( U_{g,\tau_g(\omega)}^{-1} \, C_{\tau_g(\omega)} \,  M_{\chi_{\mathtt{B}}} \, C_{\tau_g(\omega)}^* \, U_{g,\tau_g(\omega)} \Bigr )
		\\
		&= \int_{\Omega} \dd \mathbb{P}(\omega) \; \mathrm{Tr}_{\Hil_{\tau_g(\omega)}} \bigl ( C_{\tau_g(\omega)} \,  M_{\chi_{\mathtt{B}}} \, C_{\tau_g(\omega)}^* \bigr ) 
		= \nu_{\hat{A}}(\mathtt{B})
		, 
	\end{align*}
	coincides with $\nu_{\hat{A}}(\mathtt{B})$. The last equality is guaranteed by the $\mathbb{G}$-invariance of the measure $\mathbb{P}$. At this point the uniqueness of the Haar measure on $\mathbb{G}$ implies that 
	\begin{align*}
		\int_{\mathbb{G}} \dd \nu_{\hat{A}}(g) \, \lambda(g)^2 = \Lambda_{\mathbb{P}}(\hat{A}) \, \int_{\mathbb{G}} \dd \mu_{\mathbb{G}}(g) \; \lambda(g)^2 
		= \Lambda_{\mathbb{P}}(\hat{A})
	\end{align*}
	for a unique $\Lambda_{\mathbb{P}}(\hat{A}) \in [0,+\infty]$. 
	
	To prove the equivalence between $\Lambda_{\mathbb{P}}$ and $\mathcal{T}_{\mathbb{P}}$ it is enough to prove the equality $\Lambda_{\mathbb{P}} \bigl ( \widehat{K}_f^* \; \widehat{K}_f \bigr ) = \mathcal{T}_{\mathbb{P}} \bigl ( \widehat{K}_f^* \; \widehat{K}_f \bigr )$ for generic elements $\widehat{K}_f \in \Alg_{L^2}$, and to use the SOT density of $\Alg_{L^2}$ in $\Alg(\Omega,\mathbb{P},\mathbb{G},\Theta)$. This follows as in \cite[Theorem~2.2.7]{Lenz:random_operators_crossed_products:1999} modulo the necessary modifications to include the field of weights. The central equality is
	\begin{align*}
		\Lambda_{\mathbb{P}} \bigl ( \widehat{K}_f^* \; \widehat{K}_f \bigr ) &= \int_{\Omega} \dd \mathbb{P}(\omega) \; \mathrm{Tr}_{\Hil_{\omega}} \bigl (M_{\lambda}\,{K}_{f,\omega}^*\, {K}_{f,\omega}\, M_{\lambda} \bigr )\\
		&= \int_{\Omega} \dd \mathbb{P}(\omega) \int_{\mathbb{G} \times \mathbb{G}} \dd \mu_{\mathbb{G}}(h) \, \dd \mu_{\mathbb{G}}(g) \, \lambda(g)^2 \, 
		\cdot \\
		&\qquad \qquad \qquad \qquad 
		\cdot
		\mathrm{Tr}_{\C^N} \Bigl ( f_{\omega}(h,g)^* \, W_{\omega}(h)\,  f_{\omega}(h,g)\, W_{\omega}(g) \Bigr )
	\end{align*}
	where we used the fact that ${K}_{f,\omega} \, M_{\lambda}$ has as its kernel $f_{\omega}(h,g) \, \lambda(g)$ with $\lambda(g)$ a scalar, and equation \eqref{eq:tra_H_norm_HH} which connects $\mathrm{Tr}_{\Hil_{\omega}}$ with the norm on $\Hil_{\omega}\otimes \Hil_{\omega}$. The Fubini Theorem allows to rewrite the last equality as
	\begin{align*}
		&\Lambda_{\mathbb{P}} \bigl ( \widehat{K}_f^* \; \widehat{K}_f \bigr ) = \int_{\mathbb{G}} \dd \mu_{\mathbb{G}}(g) \, \lambda(g)^2\; T(g) = T
	\end{align*}
	where the function $T(g) = T$ is indeed independent of $g$, 	
	\begin{align*}
		T(g) :& \negmedspace= \int_{\mathbb{G} \times \Omega} \dd \mu_{\mathbb{G}}(h) \, \dd \mathbb{P}(\omega) \, \mathrm{Tr}_{\C^N} \Bigl ( f_{\omega}(h,g)^* \,W_{\omega}(h) \, f_{\omega}(h,g)\, W_{\omega}(g) \Bigr )
		\\
		&= \int_{\mathbb{G} \times \Omega} \dd \mu_{\mathbb{G}}(h)\, \dd \mathbb{P}(\omega) \, \mathrm{Tr}_{\C^N} \Bigl ( f(\omega,h)^* \, W_{ \omega}(e) \, f(\omega,h) \, W_{ \omega}(h) \Bigr ) 
		= T
		.
	\end{align*}
	In the passage from the first to the second line we used the extended expression \eqref{eq:extended_kernel} for the kernel $f_{\omega}(h,g)$, the covariance relations $W_{\omega}(h) = W_{ \tau_{h^{-1}}(\omega)}(e)$, $W_{\omega}(g) = W_{ \tau_{h^{-1}}(\omega)} \bigl ( g \, h^{-1} \bigr )$ and the invariance of the measures $\mathbb{P}$ and $\mu_{\mathbb{G}}$. A comparison with \eqref{unified:eqn:scalar_product_Alg_0} shows that $T = \snorm{f}_{\Alg_0}^2$, proving the equality \eqref{unified:eqn:formula_trace_per_unit_volume} for elements of the type $\widehat{K}_f^* \; \widehat{K}_f$. At this point the proof can be concluded by a density argument.
\end{proof}
By assumption $\mathbb{G}$ is a locally compact abelian, hence \emph{amenable} group. This means that $\mathbb{G}$ admits a left invariant mean for the von Neumann algebra $L^\infty(\mathbb{G})$, or equivalently that $\mathbb{G}$ admits at least one \emph{F{\o}lner exhausting sequence} $\Lambda_n \nearrow \mathbb{G}$ (see \cite{Greenleaf1969} for more details). In particular this applies for the physically interesting cases $\mathbb{G} = \R^d$ of $\mathbb{G} = \Z^d$. Therefore, under the conditions stipulated in Definition~\ref{unified:defn:topological_dynamical_system} the \emph{Mean Ergodic Theorem} \cite[Corollary 3.5]{Greenleaf1973} holds:
\begin{align*}
	\lim_{n \to \infty} \frac{1}{\sabs{\Lambda_n}} \int_{\Lambda_n} \dd \mu_{\mathbb{G}}(h) \, f \bigl ( \tau_{h^{-1}}(\omega) \bigr ) 
	= \int_\Omega\dd \mathbb{P}(\omega') \, f(\omega') 
	, 
	&&
	f \in L^1(\Omega)
	, 
\end{align*}
where $\sabs{\Lambda_n} := \mu_{\mathbb{G}}(\Lambda_n)$ and $\omega$ is any point in a set $\Omega_f \subseteq \Omega$ of full measure. When $\widehat{K}_f \in \Alg_{L^2}$ the function
\begin{align*}
	Z(\omega) := \int_{\mathbb{G}} \dd \mu_{\mathbb{G}}(g) \, \mathrm{Tr}_{\C^N} \Bigl ( f(\omega,g)^* \, W_{ \omega}(e) \, f(\omega,g) \, W_{ \omega}(g) \Bigr )
\end{align*}
defines an element of $L^1(\Omega)$, and the formula above applies, 
\begin{align*}
	\lim_{n \to \infty} \frac{1}{\sabs{\Lambda_n}} \int_{\Lambda_n} \dd \mu_{\mathbb{G}}(h) \, Z \bigl ( \tau_{h^{-1}}(\omega) \bigr ) 
	= \int_{\Omega} \dd \mathbb{P}(\omega') \, Z(\omega') 
	= \mathcal{T}_{\mathbb{P}} \bigl ( \widehat{K}_f^* \; \widehat{K}_f \bigr )
	.
\end{align*}
On the other hand we can check directly from \eqref{eq:tra_H_norm_HH} that 
\begin{align*}
	\int_{\Lambda_n} \dd \mu_{\mathbb{G}}(g) \, Z \bigl ( \tau_{h^{-1}}(\omega) \bigr ) =
	\mathrm{Tr}_{\Hil_{\omega}} \bigl ( P_{\Lambda_n} \, {K}_{f,\omega}^* \, {K}_{f,\omega} \, P_{\Lambda_n} \bigr )
\end{align*}
where $P_{\Lambda_n}$ is the multiplication operator by the characteristic function for $\Lambda_n$ (in fact, a projection). After putting all the pieces together we obtain
\begin{align}
	\mathcal{T}_{\mathbb{P}} (\widehat{A}) &= \lim_{n \to \infty} \frac{1}{\sabs{\Lambda_n}} \, \int_{\Omega} \dd \mathbb{P}(\omega) \; \mathrm{Tr}_{\Hil_{\omega}}
	\bigl ( P_{\Lambda_n} \, A_{\omega} \, P_{\Lambda_n} \bigr )
	\notag \\
	&= \lim_{n \to \infty} \frac{1}{\sabs{\Lambda_n}} \, \mathrm{Tr}_{\Hil_{\omega}}
	\bigl ( P_{\Lambda_n} \, A_{\omega} \, P_{\Lambda_n} \bigr ) 
	\label{unified:eqn:trace_per_unit_volume_as_infinite_volume_limit}
\end{align}
where the first equality is exact and the second holds $\mathbb{P}$-almost surely. Equation~\eqref{unified:eqn:trace_per_unit_volume_as_infinite_volume_limit} justifies the name \emph{name trace per unit volume} for the trace $\mathcal{T}_{\mathbb{P}}$. For more details we refer to \cite[pp.~208-209]{Lenz:random_operators_crossed_products:1999} and references therein.

\section{Generators compatible with the trace per unit volume} 
\label{unified:generators}
Let us construct particular examples of $\mathcal{T}_{\mathbb{P}}$-compatible generators in the sense of Hypothesis~\ref{hypothesis:generators}. Suppose $x : \mathbb{G} \to \R$ is a measurable function which provides an additive representation of $\mathbb{G}$, namely
\begin{align*}
	x(g_1 \, g_2) = x(g_1) + x(g_2)
	&&
	\forall g_1 , g_2 \in \mathbb{G}
	.
\end{align*}
This function defines a selfadjoint multiplication operator on $L^2(\mathbb{G})$ by $(X \varphi)(g) := x(g) \, \varphi(g)$ for all $g \in \mathbb{G}$, and the compactly supported functions are a core. However, $X$ is not left invariant by the application of the $S_g$; Indeed a simple computation shows
\begin{align*}
	S_g \, X \, S_g^* = X - x(g) \, \id
	&&
	\forall g \in \mathbb{G}
	.
\end{align*}
This gives rise to the constant field of operators $\widehat{X} := \bigl \{ X_{\omega} \equiv X \bigr \}_{\omega \in \Omega}$ defined on the direct integral $\widehat{\Hil}$. The measurability of $\widehat{X}$ follows from the measurability of the function $x$; Moreover, it is selfadjoint because $X$ commutes with each $W_{\omega}$, albeit not necessarily bounded. The associated evolution group $\e^{+ \ii s \widehat{X}}$ are elements of $\mathrm{Rand}(\widehat{\Hil})$ for all $s \in \R$ but not of the von Neumann algebra $\Alg(\Omega,\mathbb{P},\mathbb{G},\Theta)$ since
\begin{align*}
	\widehat{U}_g \, \e^{+ \ii s \widehat{X}} \, \widehat{U}_g^* = \e^{- \ii s x(g)} \, \e^{+ \ii s \widehat{X}}
	&&
	\forall g \in \mathbb{G}
	, \; 
	t \in \mathbb{R}
	.
\end{align*}
However, conjugating with $\e^{+ \ii s \widehat{X}}$ leaves $\Alg(\Omega,\mathbb{P},\mathbb{G},\Theta)$ invariant: a quick computation shows 
\begin{align*}
	\widehat{U}_g \, \bigl ( \e^{+ \ii s \widehat{X}} \, \widehat{A} \, \e^{- \ii s \widehat{X}} \bigr ) \, \widehat{U}_g^{*}=
	\e^{+ \ii s \widehat{X}} \, \bigl ( \widehat{U}_g \, \widehat{A} \, \widehat{U}_g^{*} \bigr )  \, \e^{- \ii s \widehat{X}}
	= \e^{+ \ii s \widehat{X}} \, \widehat{A} \, \e^{- \ii s \widehat{X}}
\end{align*}
which ensures that also $\e^{+ \ii s \widehat{X}} \, \widehat{A} \, \e^{- \ii s \widehat{X}} \in \Alg(\Omega,\mathbb{P},\mathbb{G},\Theta)$ if $\widehat{A}$ is. Similarly, conjugating with $\e^{+ \ii s \widehat{X}}$ leaves the trace invariant: with the help of equation~\eqref{unified:eqn:formula_trace_per_unit_volume} we obtain 
\begin{align*}
	\mathcal{T}_{\mathbb{P}}(\e^{+ \ii s \widehat{X}}\widehat{A} \e^{- \ii s \widehat{X}}) &= \int_{\Omega} \dd \mathbb{P}(\omega) \; \mathrm{Tr}_{\Hil_{\omega}} \Bigl ( M_f \, \bigl ( \e^{+ \ii s \widehat{X}} \, \widehat{A} \, \e^{- \ii s \widehat{X}} \bigr )_{\omega} \, M_f \Bigr ) 
	\\
	&= \int_{\Omega} \dd \mathbb{P}(\omega) \; \mathrm{Tr}_{\Hil_{\omega}} \bigl ( M_f \, A_{\omega} \, M_f \bigr ) 
	= \mathcal{T}_{\mathbb{P}}(\widehat{A})
	, 
\end{align*}
where we have used the fact that $\e^{\pm \ii s X}$ commutes with $M_f$ on each $\Hil_{\omega}$ and the cyclicity of the trace. Since this equality holds on a dense set, it holds for each element of the maximal ideal on which $\mathcal{T}_{\mathbb{P}}$ is defined. In conclusion we have proven the following 
\begin{proposition}\label{unified:prop:compatibility_generators}
	The operator $\widehat{X} := \bigl \{ X_{\omega} \equiv X \bigr \}_{\omega \in \Omega}$ is a $\mathcal{T}_{\mathbb{P}}$-compatible generators in the sense of Hypothesis~\ref{hypothesis:generators}.
\end{proposition}
Clearly, it is possible to consider a finite family of generators $\widehat{X}_1 , \ldots , \widehat{X}_d$ associated with measurable functions $x_j : \mathbb{G} \longrightarrow \R$ that preserve the group structure. Evidently, this ensures the strong commutativity of these generators.

\section{Reduction to the non-random case} 
\label{unified:no_randomness}
The \emph{non-random case} corresponds to the special situation where the disorder set $\Omega = \{\omega_*\}$ collapses to a point, and consequently the ergodic measure $\mathbb{P} = \delta_{\omega_*}$ reduces to the Dirac measure concentrated on $\omega_*$. The dynamical system $\tau : \mathbb{G} \to \id_{\Omega}$ becomes trivial in the sense that $\tau_{g}(\omega_*)=\omega_*$ for all $g \in \mathbb{G}$. In this case the field of weights reduces to just a single positive $W \in L^{\infty}(\mathbb{G}) \otimes \mathrm{Mat}_{\C}(N)$ with bounded inverse. The $\mathbb{G}$-\emph{co}variance condition simplifies to $\mathbb{G}$-invariance, \ie $W \bigl ( h \, g^{-1} \bigr ) = W(h)$ holds for all $g\in \mathbb{G}$. 
The direct integral has only one fiber and so reduces to the twisted Hilbert space $\Hil_W$ (the Banach space $\Hil_{\ast}$ endowed with the scalar product $\scpro{\, \cdot \,}{W \, \cdot \,}_{\Hil_{\ast}}$). The set of the decomposable bounded operators conincides with the algebra $\mathscr{B}(\Hil_W)$. The covariance property is just defined by the projective representation $\mathbb{G} \ni g \mapsto S_g \in \mathscr{B}(\Hil_W)$ according to
\begin{align}
	S_g \, A \, S_g^* = A
	&&
	g \in \mathbb{G}.
	\label{unified:eqn:G_invariance_operator}
\end{align}
Operators which meet this condition are called \emph{$\mathbb{G}$-periodic}, and these form the \emph{von Neumann algebra of $\mathbb{G}$-periodic operators}, 
\begin{align*}
	\Alg_{\mathbb{G}} := \mathrm{Span}_{\mathbb{G}} \{ S_g \}' \cap \mathscr{B}(\Hil_W)
	. 
\end{align*}
The trace per unit volume is realized by
\begin{align*}
	\mathcal{T}_{\mathbb{G}}(A) := \mathrm{Tr}_{\Hil_W} \bigl (M_{\lambda} \, A_{\omega} \, M_{\lambda} \bigr ) 
\end{align*}
where $\lambda \in L^{\infty}(\mathbb{G})$ is any positive normalized function $\int_{\mathbb{G}} \dd \mu_{\mathbb{G}}(g) \, \lambda(g)^2 = 1$ (\eg the normalized characteristic function of a subset $\Lambda\subset\mathbb{G}$). Finally, the multiplicative operators $(X \varphi)(h) = x(h) \, \varphi(h)$ described in the beginning of Section~\ref{unified:generators} provides typical examples of $\mathcal{T}_{\mathbb{G}}$-compatible generators. 

In case the algebra $\mathrm{Span}_{\mathbb{G}} \{ S_g \}$ contains a (non-trivial) commutative $C^*$-algebra $\mathscr{S}_{\mathbb{G}} \subseteq \mathrm{Span}_{\mathbb{G}} \{ S_g \}$, we can invoke the \emph{von Neumann's complete spectral theorem} \cite[Part~II, Chapter~6, Theorem~1]{dixmier-81}: it states there exists a direct integral Hilbert space
\begin{align*}
	\Hil_{\mathscr{S}} := \int^{\oplus}_{\mathbb{B}} \dd \mu(k) \; \Hil_k 
\end{align*}
where $\mathbb{B} = \mathbb{B}(\mathscr{S}_{\mathbb{G}})$ is the Gel'fand spectrum of $\mathscr{S}_{\mathbb{G}}$ (also referred to as \emph{Brillouin zone}) and $\mu$ is a basic spectral measure, and a unitary map
\begin{align*}
	\mathcal{F}_{\mathscr{S}} : \Hil_W \longrightarrow \Hil_{\mathscr{S}}
\end{align*}
called the \emph{Gel'fand-Fourier} (Bloch-Floquet) transform such that
$\mathcal{F}_{\mathscr{S}} \, \Alg_{\mathbb{G}} \, \mathcal{F}_{\mathscr{S}}^{-1}$ is contained in the bounded decomposable operators over $\Hil_{\mathscr{S}}$. Or, said differently, 
\begin{align*}
	\mathcal{F}_{\mathscr{S}} \, A \, \mathcal{F}_{\mathscr{S}}^{-1} = \{ A_k \}
	\equiv \int^{\oplus}_{\mathbb{B}} \dd \mu(k) \, A_k 
	,
	&&
	\forall A \in \Alg_{\mathbb{G}}
	.
\end{align*}
In this representation the trace per unit volume is the Brillouin zone average 
\begin{align*}
	\mathcal{T}_{\mathbb{G}}(A) = \frac{1}{\mu(\mathbb{B})} \int_{\mathbb{B}} \dd \mu(k) \; \mathrm{Tr}_{\Hil_k}(A_k) 
	.
\end{align*}
Lastly, let us notice that when $\Theta = 1$ (the twist is trivial) we can choose $\mathscr{S}_{\mathbb{G}} = \mathrm{Span}_{\mathbb{G}} \{ S_g \}$. 
%
%
\chapter{Studying the Dynamics} 
\label{dynamics}
The purpose of this chapter is to study the unperturbed, perturbed and interaction dynamics within the framework of Chapter~\ref{framework}. Among the things on our list is to characterize cores and domains of the generators, prove existence of the propagators, show their continuity in the perturbation parameter, and compare the different evolutions to one another.

\section{Unperturbed dynamics} 
\label{dynamics:unperturbed}
Let $H \in \affil(\Alg)$ be a selfadjoint operator and consider the (unitary) \emph{unperturbed} propagator 
\begin{align}
	U_0(t) := \e^{- \ii t H}
	,
	&&
	t \in \R
	.
	\label{dynamics:eqn:unperturbed_unitary_evolution}
\end{align}
The map $t \mapsto U_0(t)$ verifies all the conditions of Definition~\ref{framework:defn:gauge_transformation} and so Proposition~\ref{framework:prop:extension_isometry_Lp} applies. This is the crucial fact which allows to define an \emph{unperturbed dynamics} on each of the Banach spaces $\mathfrak{L}^p(\Alg)$. 

As a consequence of Proposition~\ref{framework:prop:extension_isometry_Lp} the $\R$-flow $t \mapsto \alpha^0_t$ on $\Alg$ by
\begin{align}
	\alpha^0_t(A) := U_0(t) \; A \; U_0(-t)
	,
	&&
	t \in \R
	, \ 
	A \in \Alg\, 
	, 
\end{align}
extends canonically to a one-parameter group of $\ast$-automorphisms of $\rr{M}(\Alg)$ and secondly defines strongly continuous one-parameter groups of isometries on each of the $\mathfrak{L}^p(\Alg)$, $1 \leqslant p < \infty$. The the $\R$-flow
\begin{align*}
	\R \ni t \mapsto \alpha^0_t \in \mathrm{Iso} \bigl ( \mathfrak{L}^p(\Alg) \bigr ) 
\end{align*}
is called \emph{unperturbed dynamics} (see Consequence~\ref{main_results:conseq:unperturbed_evolution}). Notice that for $p = 2$ the isometries $\alpha_t$ act as unitary operators with respect to the Hilbert space structure of $\mathfrak{L}^2(\Alg)$, and its generator is anti-selfadjoint.

\subsection{The generator of the unperturbed dynamics} 
\label{dynamics:unperturbed:generator}
Also for all other $1 \leqslant p < \infty$ the unperturbed dynamics admit an infinitesimal generator $\mathscr{L}_H^{(p)}$ in each space $\mathfrak{L}^p(\Alg)$, which we refer to as the \emph{$p$-Liouvillian}. The main properties of $\mathscr{L}_H^{(p)}$ are described in \cite[Section~7]{Pagter_Sukochev:commutator_estimates_R_flows:2007} (\cf also Proposition~\ref{framework:prop:core_derivation} and subsequent comments), and its properties can be summarized as follows:
\begin{proposition}
	Suppose Hypotheses~\ref{hypothesis:trace} and \ref{hypothesis:Hamiltonian} are satisfied, and let $\mathscr{L}_H^{(p)}$ be the Liouvillian on $\mathfrak{L}^p(\Alg)$ associated to $H \in \affil(\Alg)$. Then, $\mathscr{L}_H^{(p)}$ is a closed operator defined on the norm-dense domain $\rr{D}_{H,p}\subset \mathfrak{L}^p(\Alg)$. A core for $\mathscr{L}_H^{(p)}$ is given by
	\begin{align}
		\rr{D}^0_{H,p} := \Bigl \{ A \in \rr{D}_{H,p} \cap \Alg \; \; \big \vert \; \; A[\domain(H)] \subset \domain(H), \; [ H , \, A ]  \in \mathfrak{L}^p(\Alg) \cap \Alg \Bigr \}
		\label{dynamics:eqn:core_p_Liouvillian}
	\end{align}
	and $\mathscr{L}_H^{(p)}(A) = - \ii \, [H , \, A]$ for all $A \in \rr{D}^0_{H,p}$. Moreover, if $A = A^* \in \rr{D}_{H,p}$ and $f \in C^{1+\delta}(\R)$ with $\delta > 0$ and $f(0) = 0$, then $f(A) \in \rr{D}_{H,p}$. 
\end{proposition}
The reason for all this commutator calisthenics lies with the non-measurability of $H$: as we have argued in Remark~\ref{framework:example:non_T_measurable_operators} there are a lot of physically relevant cases where $H \not\in \rr{M}(\Alg)$, and this necessitates the different notions of commutators and derivatives covered in Section~\ref{framework:commutators}. Nevertheless, there are situations where $H$ \emph{is} measurable, and there, the above statement simplifies. 
\begin{remark}[$\mathcal{T}$-measurable Hamiltonians]
	When $H \in \rr{M}(\Alg)$, such is the case if the trace $\mathcal{T}$ is finite (\cf Example~\ref{framework:example:T_measurable_operators}~(3)), the description of the domain of the Liouvillians it becomes simpler. First of all, the $\R$-flow $\alpha_t^0 : \rr{M}(\Alg) \longrightarrow \rr{M}(\Alg)$ is \emph{everywhere} differentiable \cite[Lemma~6.1]{Pagter_Sukochev:commutator_estimates_R_flows:2007}, namely
	\begin{align*}
		\lim_{t \to 0} \frac{\alpha_t^0(A) - A}{t} = - \ii \, [H,A]
		&&
		\forall A \in \rr{M}(\Alg)
		,
	\end{align*}
	where the limit is meant with respect to the measure topology and the commutator $[H,A] = H \, A - A \, H$ is defined through the algebraic structure of $\rr{M}(\Alg)$. Moreover, if $A \in \mathfrak{L}^p(\Alg)$ and the derivative of $\alpha_t^0(A)$ exists with respect to the topology of $\mathfrak{L}^p(\Alg)$ one has that $[H,A] \in \mathfrak{L}^p(\Alg)$ and $\mathscr{L}_H^{(p)}(A) = - \ii \, [H,A]$ \cite[Corollary~6.3]{Pagter_Sukochev:commutator_estimates_R_flows:2007}. In particular, in this special case one can prove that the domain $\rr{D}_{H,p}$ of the $p$-Liouvillian $\mathscr{L}_H^{(p)}$ is described by 
	\begin{align*}
		\rr{D}_{H,p} := \bigl \{ A \in \mathfrak{L}^p(\Alg) \; \; \vert \; \; [H,A] \in \mathfrak{L}^p(\Alg) \bigr \}
	\end{align*}
	and $\mathscr{L}_H^{(p)}(A) = - \ii \, [H,A]$ for all $A \in \rr{D}_{H,p}$ \cite[Theorem~6.8]{Pagter_Sukochev:commutator_estimates_R_flows:2007}.
\end{remark}
Let us recall the definition of the domain $\rr{D}^{00}_{H,p} \subset \mathfrak{L}^p(\Alg)$ given in \eqref{framework:eqn:domain_maximal_generalized_commutator}. Corollary~\ref{framework:cor:product_affiliated_algebra_density} ensures that $\rr{D}^{00}_{H,p}$ is dense. We can say a little more.
\begin{proposition}\label{dynamics:prop:density_core_Liouvillian}
	Suppose Hypotheses~\ref{hypothesis:trace} and \ref{hypothesis:Hamiltonian} hold true. 
	\begin{enumerate}
		\item For each $1 \leqslant p < \infty$ the domain $\rr{D}^0_{H,p} \cap \rr{D}^{00}_{H,p}$ is dense in $\mathfrak{L}^p(\Alg)$. Moreover, if $A \in \rr{D}^{00}_{H,p}$ then
		\begin{align}
			\alpha_t^0 \bigl ( H \, A \bigr ) = H \, \alpha_t^0(A)
			&&
			\forall t \in \R
			\label{dynamics:eqn:commutator_H_alpha}
		\end{align}
		as elements of $\mathfrak{L}^p(\Alg)$.
		\item The domain $\rr{D}^{00}_{H,p}$ is an operator core for the $p$-Liouvillian $\mathscr{L}_H^{(p)}$ associated to $H \in \affil(\Alg)$. Moreover, $\mathscr{L}_H^{(p)}$ acts on $\rr{D}^{00}_{H,p}$ as a \emph{generalized commutator} in the sense of Definition~\ref{framework:defn:generalized_commutators}, namely
		\begin{align*}
			\mathscr{L}_H^{(p)}(A) = - \ii [H,A]_{\ddagger} = - \ii \bigl ( H \, A - \bigl ( H \, A^* \bigr )^* \bigr )
			,
			&&
			A \in \rr{D}^{00}_{H,p}
			.
		\end{align*}
	\end{enumerate}
\end{proposition}
\begin{proof}
	\begin{enumerate}
		\item Just observe that the sequences $B_n$ and $B_n^*$ defined in the proofs of Lemma~\ref{framework:lem:product_affiliated_algebra_density} and Corollary~\ref{framework:cor:product_affiliated_algebra_density} are actually contained in the intersection $\rr{D}^0_{H,p} \cap \rr{D}^{00}_{H,p}$. Equality \eqref{dynamics:eqn:commutator_H_alpha} follows from Lemma~\ref{framework:lem:extension_algebra_unbounded_operators}~(1) and the fact that the two factors in $U_0(t) \, H = H \, U_0(t)$ commute.
		\item According to \cite[Chapter~II, Proposition~1.7]{Engel_Nagel:one_parameter_semigroup_linear_evolution_equations:2000} the set $\rr{D}^0_{H,p} \cap \rr{D}^{00}_{H,p}$ is a core for $\mathscr{L}_H^{(p)}$, because it has a norm-dense subset invariant under $\alpha^0_t$ (composed of vectors of the form $B_n = P_n(H) \, B \, P_n(H)$ as in the proof of Lemma~\ref{framework:lem:product_affiliated_algebra_density}). This in turn implies that also $\rr{D}^{00}_{H,p}$ is a core for $\mathscr{L}_H^{(p)}$ since it contains the core $\rr{D}^0_{H,p} \cap \rr{D}^{00}_{H,p}$. Let us start with the identity
		\begin{align}
			\frac{\alpha^0_t(A) - A}{t} &= \frac{U_0(t) - \id}{t} \, A \, U_0(-t) + U_0(t) \, A \, \frac{U_0(-t) - \id}{t}
			\label{dynamics:eqn:generator_dynamics_p_Liouvillian}
		\end{align}
		and observe that 
		\begin{align*}
			\frac{U_0(t) - \id}{t} \, A \, U_0(-t) = \left ( \frac{U_0(t) - \id}{t} \, \frac{1}{H - \xi} \right ) \, \bigl ( (H - \xi) \, A \bigr ) \, U_0(-t)
		\end{align*}
		where $\xi \notin \spec(H)$ and we used $A \in \rr{D}^{00}_{H,p}$ along with Lemma~\ref{framework:lem:extension_algebra_unbounded_operators}~(2). According to Stone's theorem the limit 
		\begin{align*}
			\slim_{t \to 0} \frac{U_0(t) - \id}{t} \, \frac{1}{H - \xi} = - \ii \frac{H}{H - \xi}
		\end{align*}
		exists in the SOT because the resolvent maps $\Hil$ to $\domain(H)$. Moreover, we can use functional calculus to estimate 
		\begin{align}
			\norm{\frac{U_0(\pm t) - \id}{t} \, \frac{1}{H - \xi}} &= \sup_{E \in \spec(H)} \abs{\frac{\e^{\pm \ii t E} - \id}{t} \, \frac{1}{E - \xi}}
			\notag \\
			&\leqslant \sup_{E \in \spec(H)} \abs{\frac{E}{E - \xi}}
			=: \kappa_{\xi}
			\label{dynamics:eqn:definition_constant_E_E_minus_xi}
		\end{align}
		independently of $t \in \R$. This fact allows us to use Lemma~\ref{framework:lem:strong_convergence_trace_product} which provides
		\begin{align}
			\norm{\cdot}_p-\lim_{t \to 0} \frac{U_0(t) - \id}{t} \, A \, U_0(-t) = - \ii\left ( \frac{H}{H - \xi} \right ) \, \bigl ( (H - \xi) \, A \bigr ) 
			= - \ii H \, A
			\label{dynamics:eqn:computation_p_Liouvillian_term_1}
		\end{align}
		where we used again Lemma~\ref{framework:lem:extension_algebra_unbounded_operators}~(2). The second summand in \eqref{dynamics:eqn:generator_dynamics_p_Liouvillian} can be rewritten as
		\begin{align*}
			A \; \frac{U_0(-t) - \id}{t} &= \left ( \left ( \frac{U_0(-t) - \id}{t} \right )^* \, A^* \right )^*
			\\
			&= \left ( \left ( \frac{U_0(t) - \id}{t} \, \frac{1}{H - \xi} \right ) \, \bigl ( (H - \xi) \, A^* \bigr ) \right )^*
		\end{align*}
		As above one can prove that 
		\begin{align*}
			\norm{\cdot}_p-\lim_{t \to 0} \left ( \frac{U_0(t) - \id}{t} \, \frac{1}{H - \xi} \right ) \, \bigl ( (H - \xi) \, A^* \bigr ) =- \ii H \, A^*
			.
		\end{align*}
		Moreover, the adjoint map $B \mapsto B^*$ is an isometry in $\mathfrak{L}^p(\Alg)$ and so
		\begin{align}
			\norm{\cdot}_p-\lim_{\epsilon \to 0} A \, \frac{U_0(-t) - \id}{t} = \bigl ( - \ii H \, A^* \bigr )^*
			.
			\label{dynamics:eqn:computation_p_Liouvillian_term_2}
		\end{align}
		By plugging \eqref{dynamics:eqn:computation_p_Liouvillian_term_1} and \eqref{dynamics:eqn:computation_p_Liouvillian_term_2} into \eqref{dynamics:eqn:generator_dynamics_p_Liouvillian} one obtains $\mathscr{L}_H^{(p)}(A) =- \ii [H,A]_{\ddagger}$ for $A \in \rr{D}^{00}_{H,p}$.
	\end{enumerate}
\end{proof}
%

\subsection{A formula for the projection in Theorem~\ref{main_results:thm:adiabatic_limit_Kubo_formula}} 
\label{dynamics:unperturbed:projection}
When considering the adiabatic limit, we encounter the limit of the net of operators
\begin{align*}
	\mathscr{Z}_{\eps}^{(p)} := \frac{\mathscr{L}_H^{(p)}}{\mathscr{L}_H^{(p)} - \eps}
	, 
	&&
	\eps > 0
	, 
\end{align*}
as $\eps \to 0^+$ in the strong sense with respect to the topology of $\mathfrak{L}^p(\Alg)$, namely we want to know whether the limit 
\begin{align}
	\lim_{\eps \to 0^+} \mathscr{Z}_{\eps}^{(p)}(A)
	\stackrel{?}{=}\mathscr{Z}_{0}^{(p)}(A)
	\label{dynamics:eqn:hard_lim_banach1}
\end{align}
exists for each $A\in\mathfrak{L}^p(\Alg)$ and compute the limit operator $\mathscr{Z}_{0}^{(p)}$ explicitly. 

Considering this limit in the $\mathfrak{L}^p(\Alg)$ spaces are much more tricky than for $p = 2$ as there is no one-to-one correspondence between projections and closed subspaces of $\mathfrak{L}^p(\Alg)$. Indeed, what makes these projections unique for $p = 2$ is the requirement that they be orthogonal — a notion which does not exist in Banach spaces. Fortunately, the $\mathscr{Z}_{\eps}^{(p)}$ and the $p$-Liouvillians with respect to which they are defined have additional properties, because the $\mathscr{L}^{(p)}_H$ are infinitesimal generators of $\R$-flows in $\mathfrak{L}^p(\Alg)$. First of all, from equation~\eqref{framework:eqn:resolvent_equation_integral} we obtain the standard norm estimate of the resolvent \cite[Theorem~1.10]{Engel_Nagel:one_parameter_semigroup_linear_evolution_equations:2000}, 
\begin{align*}
	\norm{\frac{1}{\mathscr{L}_H^{(p)} - \eps}}_{\mathscr{B}(\mathfrak{L}^p(\Alg))} \leqslant \frac{1}{\eps}
	, 
	&& 
	\eps > 0
	. 
\end{align*}
With the help of the identity 
\begin{align}
	\mathscr{Z}_{\eps}^{(p)} = \id_{\mathfrak{L}^p(\Alg)} + \frac{\eps}{\mathscr{L}_H^{(p)} - \eps}
	\label{dynamics:eqn:hard_lim_banach2}
\end{align}
we deduce that $\mathscr{Z}_{\eps}^{(p)}$ is an equibounded net with bound
\begin{align*}
	\Bnorm{\mathscr{Z}_{\eps}^{(p)}}_{\mathscr{B}(\mathfrak{L}^p(\Alg))} \leqslant 2
	.
\end{align*}
This bound will be crucial in the proof. Furthermore, also Proposition~\ref{framework:prop:Leibniz_rule}~(2) applies. This means that if 
$A \in \rr{D}_{H,p}$ and $B \in \rr{D}_{H,q}$ with $p^{-1} + q^{-1} = 1$ then 
\begin{align}
	\mathcal{T} \Bigl ( A \; \mathscr{L}_H^{(q)}(B) \Bigr ) = - \mathcal{T} \Bigl ( \mathscr{L}_H^{(p)}(A) \; B \Bigr ) 
	\label{dynamics:eqn:anti_adjoint}
\end{align}
meaning that the $p$-Liouvillian $\mathscr{L}_H^{(p)}$ and the $q$-Liouvillian $\mathscr{L}_H^{(q)}$ are anti-adjoints of one another with respect to the duality induced by the trace $\mathcal{T}$. In particular, in the relevant case $p=2$ the Liouvillian $\mathscr{L}_H^{(2)}$ turns out to be an \emph{anti}-selfadjoint operator on the Hilbert space $\mathfrak{L}^2(\Alg)$.
Finally, let us notice that the kernels of the $\mathscr{L}_H^{(p)}$ are all generally non-empty as, for instance, all functions of $H$ which are in $\mathscr{L}_H^{(p)}$ are automatically in $\ker \bigl ( \mathscr{L}_H^{(p)} \bigr )$\footnote{We are tacitly assuming here that the abelian  von Neumann algebra $L^\infty(A)\subset \Alg$ generated by the bounded Borelian functions of $H$ has a non-empty intersection with  $\mathfrak{L}^p(\Alg)$. Of course this property depend on $H$ and on the trace $\mathcal{T}$.}. 

Let us start with the easy case $p = 2$.
 \begin{lemma}\label{dynamics:lem:limit_projection_p=2}
	Suppose Hypotheses~\ref{hypothesis:trace} and \ref{hypothesis:Hamiltonian} are satisfied. Let $\mathscr{P}_H^{(2)}$ be the projection on the kernel of the $2$-Liouvillian $\mathscr{L}_H^{(2)}$ and ${\mathscr{P}_H^{(2)\bot}}:=\id_{\mathfrak{L}^2(\Alg)} - \mathscr{P}_H^{(2)}$ its orthogonal complement with respect to the Hilbert space structure of $\mathfrak{L}^2(\Alg)$.
	Then the product of the 2-Liouvillian with its resolvent converges strongly in $\mathfrak{L}^2(\Alg)$ to the projection ${\mathscr{P}_H^{(2)\bot}}$ in the sense that
	\begin{align*}
		\mathscr{Z}_{0}^{(2)}(A):=	\lim_{\eps \to 0^+} \frac{\mathscr{L}_H^{(2)}}{\mathscr{L}_H^{(2)} - \eps}(A) = {\mathscr{P}_H^{(2)\bot}}(A)
		&&
		\forall A \in \mathfrak{L}^p(\Alg)
		.
	\end{align*}
\end{lemma}
\begin{proof}
	The crucial ingredient for this proof is the spectral functional calculus which can be used in view of the fact that  $\mathfrak{L}^2(\Alg)$ is an Hilbert space and $\ii \mathscr{L}_H^{(2)}$ is a selfadjoint operator. Therefore,  the spectral representation allows us to control the strong limit above just by computing the pointwise limit $\eps \to 0^+$ of the function
	\begin{align*}
	    f_{\eps}(x) := 
		\begin{cases}
			\frac{x}{x-\ii\eps} & x \in \spec \bigl ( \mathscr{L}_H^{(2)} \bigr ) \\
			0 & x \not\in \spec \bigl ( \mathscr{L}_H^{(2)} \bigr ) \\
		\end{cases}
	\end{align*}
	to the characteristic function $\chi_{\{ 0 \}}(x)$ for the set $\{ 0 \}$, and the the proof follows.
\end{proof}
Lemma~\ref{dynamics:lem:limit_projection_p=2} says that in the case $p = 2$ the limit operator $\mathscr{Z}_{0}^{(2)}$ in \eqref{dynamics:eqn:hard_lim_banach1} exists and it is equal to the complement of the kernel projection $\mathscr{P}_H^{(2)}$. What about $p \neq 2$? In this case one is tempted to conjecture that also
\begin{align*}
	\mathscr{Z}_{0}^{(p)}= \id_{\mathfrak{L}^p(\Alg)} - \mathscr{P}_H^{(p)}
	.
\end{align*}
However, this formula is delicate for the following reasons: first of all, even though $\ker \bigl ( \mathscr{L}_H^{(p)} \bigr )\neq 0$ is a closed subspace, it is \emph{a priori} not true that it admits a Banach space projection $\mathscr{P}_H^{(p)}$. And even if such a projection exists, it need not be unique. Finally, the strategy used to prove Lemma~\ref{dynamics:lem:limit_projection_p=2} cannot be trivially extended to the Banach spaces $\mathfrak{L}^p(\Alg)$ due to the lack of a sufficiently general spectral functional calculus (\eg \cite{Davies:functional_calculus:1995} requires the functions to be sufficiently regular, and $\chi_{\{ 0 \}}$ is not).
 
Let us recall that the two-sided ideal $\Alg_{\mathcal{T}}=\mathfrak{L}^1(\Alg) \cap \Alg$ is dense in each Banach space $\mathfrak{L}^p(\Alg)$. Hence, the inclusion $\Alg_{\mathcal{T}}\subset \mathfrak{L}^p(\Alg) \cap \mathfrak{L}^2(\Alg) \neq \{ 0 \}$ shows that the interpolation Banach spaces $\mathfrak{L}^p(\Alg) \cap \mathfrak{L}^2(\Alg)$ lie dense in each $\mathfrak{L}^p(\Alg)$ for all $p \geqslant 1$. Moreover, from the definition of the unperturbed dynamics one has that the maps $\alpha^0_t$ acts in the same way on each $\mathfrak{L}^p(\Alg)$. This observation, along with the strong limit which defines the resolvent (\cf \cite[Chapter~II, eqn.~(1.13)]{Engel_Nagel:one_parameter_semigroup_linear_evolution_equations:2000}), implies that
\begin{align*}
	\left . \frac{1}{\mathscr{L}_H^{(p)} - \eps} \right \vert_{ \mathfrak{L}^p(\Alg) \cap \mathfrak{L}^2(\Alg)} = \left . \frac{1}{\mathscr{L}_H^{(2)} - \eps} \right |_{\mathfrak{L}^p(\Alg)\cap \mathfrak{L}^2(\Alg)}
	,
\end{align*}
which is the analog of \cite[Proposition~2.1]{Hempel1986}. The last equality along with the identity \eqref{dynamics:eqn:hard_lim_banach2} also implies
\begin{align*}
	\left . \mathscr{Z}_{\eps}^{(p)} \right |_{\mathfrak{L}^p(\Alg) \cap \mathfrak{L}^2(\Alg)} = \left . \mathscr{Z}_{\eps}^{(2)} \right |_{\mathfrak{L}^p(\Alg) \cap \mathfrak{L}^2(\Alg)}
	.
\end{align*}
From this and Lemma~\ref{dynamics:lem:limit_projection_p=2} we deduce that 
\begin{align*}
	\lim_{\eps \to 0^+} \mathscr{Z}_{\eps}^{(p)}(A) =
	\begin{cases}
		0 & A \in \mathfrak{L}^p(\Alg) \cap \ker \bigl ( \mathscr{L}_H^{(2)} \bigr ) \\
		A & A \in \mathfrak{L}^p(\Alg) \cap \ker \bigl ( \mathscr{L}_H^{(2)} \bigr )^\bot \\
	\end{cases}
\end{align*}
Let $\mathfrak{K}^p := \overline{\mathfrak{L}^p(\Alg) \cap \ker \bigl ( \mathscr{L}_H^{(2)} \bigr )}^{\norm{\cdot}_p}$ and $\mathfrak{K}^{p\bot} := \overline{\mathfrak{L}^p(\Alg) \cap \ker \bigl ( \mathscr{L}_H^{(2)} \bigr )^\bot}^{\norm{\cdot}_p}$. Of course $\mathfrak{K}^p \cup \mathfrak{K}^{p\bot} = \mathfrak{L}^p(\Alg)$ and $\mathfrak{K}^p \cap \mathfrak{K}^{p\bot} = \{ 0 \}$. Moreover, using the fact that $\mathscr{Z}_{\eps}^{(p)}$ is equibounded, we can prove that 
\begin{align*}
	\lim_{\eps \to 0^+} \mathscr{Z}_{\eps}^{(p)}(A) =
	\begin{cases}
		0 & A \in \mathfrak{K}^p \\
		A & A \in \mathfrak{K}^{p\bot} \\
	\end{cases}
\end{align*}
namely $\mathscr{Z}_{0}^{(p)}$ exists and it is an idempotent (Banach space projection). We have proved the following result:
 \begin{lemma}\label{dynamics:lem:limit_projection_pneq2}
	Suppose Hypotheses~\ref{hypothesis:trace} and \ref{hypothesis:Hamiltonian} are satisfied. Let $\mathscr{P}_H^{(p)\bot}$ be the projection on the closed subspace 
	\begin{align*}
		\mathfrak{K}^{p\bot} := \overline{\mathfrak{L}^p(\Alg) \cap \ker \bigl ( \mathscr{L}_H^{(2)} \bigr )^\bot}^{\norm{\cdot}_p}
		.
	\end{align*}
	Then $\mathscr{Z}_{\eps}^{(p)}$ converges strongly in $\mathfrak{L}^p(\Alg)$ to 
	\begin{align*}
		\mathscr{Z}_{0}^{(p)}(A) := \lim_{\eps \to 0^+} \frac{\mathscr{L}_H^{(2)}}{\mathscr{L}_H^{(2)} - \eps}(A) =  \mathscr{P}_H^{(p)\bot}
		&&
		\forall A \in \mathfrak{L}^p(\Alg)
		.
	\end{align*}
\end{lemma}
Even though we have the inclusion $\ker \bigl ( \mathscr{L}_H^{(p)} \bigr ) \subseteq \mathfrak{K}^p$ (since $\mathscr{Z}_{\eps}^{(p)}(A) = 0$ for all $\eps > 0$ implies that $\mathscr{Z}_{0}^{(p)}(A) = 0$), it is not clear (and maybe false) whether $\ker \bigl ( \mathscr{L}_H^{(p)} \bigr ) = \mathfrak{K}^p$. 

\section{Perturbed dynamics} 
\label{dynamics:perturbed}
There are many ways to perturb a Hamiltonian, the most common one being 
\begin{align}
	\tilde{H}_{\Phi,\eps}(t) = H + \dot{F}_{\Phi,\eps}(t)
	\label{dynamics:eqn:additive_perturbation}
\end{align}
being the addition of some potential $\dot{F}_{\Phi,\eps}(t)$ that vanishes in the limit $\Phi \to 0$ (the use of this unconventional notation will be justified in Section \ref{dynamics:perturbed:additive_vs_multiplictaive}). However, to work with such additive perturbations, usually unbounded and not small compared with $H$ in the most relevant cases, is quite complicated for different technical questions which will be discuss in detail in Section~\ref{dynamics:perturbed:additive_vs_multiplictaive}.

Hence, we will study a \emph{class} of perturbations which are initially \emph{multiplicative}, but
can be rephrased in an additive form like \eqref{dynamics:eqn:additive_perturbation} under appropriate conditions. The new additive perturbation has the virtue to be $H$-bounded in contrast to $\dot{F}_{\Phi,\eps}(t)$. Of course we have the duty to justify that the two approaches describe consistency the same physics and we will pay our debt in Section \ref{dynamics:perturbed:additive_vs_multiplictaive}. For the moment let us only mention that there are a number of systems, most notably Maxwell's equations and other classical wave equations, where perturbations \emph{are} multiplicative (\cf Chapter~\ref{applications}). The purpose of this section is to characterize the perturbations considered here, and prove existence and uniqueness of the dynamics on the $\mathfrak{L}^p(\Alg)$ spaces, and which is defined in terms of a unitary propagator associated to the \emph{time-dependent} perturbed
Hamiltonian.

\subsection{Adiabatic isospectral perturbations} 
\label{dynamics:perturbed:gauge_perturbations}
The class of perturbations we are interested in are \emph{adiabatic isospectral perturbations} also studied in \cite{Bouclet_Germinet_Klein_Schenker:linear_response_theory_magnetic_Schroedinger_operators_disorder:2005,Elgart_Schlein:Kubo_for_Landau:2004}, which are sometimes also referred to as “gauge-type” perturbations. More specifically, the perturbed Hamiltonian 
\begin{align}
	H_{\Phi,\eps}(t) := G_{\Phi,\eps}(t) \, H \, G_{\Phi,\eps}(t)^*
	,
	&&
	t \in \R
	,
	\label{dynamics:eqn:definition_perturbed_hamiltonian}
\end{align}
is obtained by conjugating with $G_{\Phi,\eps}(t)$, the unitary defined in Hypothesis~\ref{hypothesis:perturbation} through a family of generators $\{ X_1 , \ldots , X_d \}$. Evidently, this defines an isospectral transformation in the sense of Definition~\ref{framework:defn:gauge_transformation} that in addition to time depends on the perturbation parameters $\pmb{\Phi} = \bigl ( \Phi_1 , \ldots , \Phi_d \bigr )$ and the adiabatic parameter $\eps > 0$. The perturbation is switched on at $t_0 < 0$ (which in principle could be $-\infty$), \ie $H_{\Phi,\eps}(t) = H$ for all $t \leqslant t_0$, and $\eps$ quantifies how quickly the perturbation is ramped up. Consequently, we impose $G_{\Phi,\eps}(t) = \id$ for all $t \leqslant t_0$; In case $t_0 = -\infty$ we take this to mean $\lim_{t \to -\infty} G_{\Phi,\eps}(t) = \id$ in the SOT. 

We have not yet exploited all the Hypothesis, though: only finitely many iterated commutators of $H$ with the $X_j$ are non-zero (Hypothesis~\ref{hypothesis:current}~(iii)), and this in fact allows us to express the multiplicative form of the perturbed Hamiltonian, equation~\eqref{dynamics:eqn:definition_perturbed_hamiltonian}, to an additive form~\eqref{dynamics:eqn:additive_perturbation}. 
\begin{lemma}\label{dynamics:lem:invariant_domain}
	Suppose Hypotheses~\ref{hypothesis:trace}–\ref{hypothesis:current} hold true. 
	\begin{enumerate}
		\item The perturbed Hamiltonian $H_{\Phi,\eps}(t) = G_{\Phi,\eps}(t) \, H \, G_{\Phi,\eps}(t)^*$ is essentially selfadjoint on the joint core $\domain_{\mathrm{c}}(H)$ from Hypothesis~\ref{hypothesis:current}~(i) for all $t \in \R$.
		\item $H_{\Phi,\eps}(t)$ can be represented on the joint core $\domain_{\mathrm{c}}(H)$ as a sum
		\begin{align}
			H_{\Phi,\eps}(t) = H + W_{\Phi,\eps}(t)
			\label{dynamics:eqn:additive_form_perturbation}
		\end{align}
		where the perturbation $W_{\Phi,\eps}$ can be expressed in terms of the density currents $J_{\kappa}$ (\cf equation~\eqref{main_results:eqn:definition_current_tensors}) and the functions $w_{\kappa}^{\eps}$ via equation~\eqref{main_results:eqn:definition_additive_perturbation_W_w_kappa}. 
		\item The domain of $H_{\Phi,\eps}(t)$ is independent of $t$,
		\begin{align*}
			\domain \bigl ( H_{\Phi,\eps}(t) \bigr ) = \domain(H)
			&&
			\forall t \in \R
			, 
		\end{align*}
		and the representation \eqref{dynamics:eqn:additive_form_perturbation} holds true on the entire domain $\domain(H)$.
		\item $\spec(H)=\spec(H_{\Phi,\eps}(t))$ for all $t \in \R$.
		\end{enumerate} 
\end{lemma}
\begin{proof}
	\begin{enumerate}
		\item From $X_j[\domain_{\mathrm{c}}(H)] \subseteq \domain_{\mathrm{c}}$ and from Hypothesis~\ref{hypothesis:perturbation} that fixes the form of $G_{\Phi,\eps}(t)$ and $G_{\Phi,\eps}(t)^*$, one infers that $G_{\Phi,\eps}(t)[\domain_{\mathrm{c}}(H)] = \domain_{\mathrm{c}}(H)$ for all $t \in \R$ (\ie $\domain_{\mathrm{c}}(H)$ is a dense set of analytic vectors for the $X_j$). This implies that $H_{\Phi,\eps}(t)$ is defined on $\domain_{\mathrm{c}}$. Moreover, the unitary equivalence implies that $H_{\Phi,\eps}(t)$ has the same deficiency indices of $H$ on $\domain_{\mathrm{c}}$ and so one concludes that $H_{\Phi,\eps}(t)$ is essentially selfadjoint on $\domain_{\mathrm{c}}$.
		\item Hypotheses~\ref{hypothesis:perturbation} and \ref{hypothesis:current}~(i) ensure that $\domain_{\mathrm{c}}(H)$ is a dense set of analytic vectors for the $X_j$ and that the iterated commutators of the $X_j$ with $H$ are well-defined on $\domain_{\mathrm{c}}(H)$. This fact allows us to use the \emph{Baker-Campbell-Hausdorff formula} in its formulation for unbounded operators (see also \cite[Proposition~5.2.4]{Bratteli_Robinson:operator_algebras_2:2003}). Formula \eqref{dynamics:eqn:additive_form_perturbation} follow after an explicit computation.
		\item This is a consequence of the Kato-Rellich theorem \cite[Theorem~X.12]{Reed_Simon:M_cap_Phi_2:1975} and of the fact that Hypothesis~\ref{hypothesis:current}~(iv) guarantees that all the currents are infinitesimally bounded with respect to $H$.
	\item This follows from the unitary equivalence between $H$ and $H_{\Phi,\eps}(t)$
	\end{enumerate}
\end{proof}
%

\subsection{Additive vs.\ multiplicative perturbations} 
\label{dynamics:perturbed:additive_vs_multiplictaive}
In order to explain many physical phenomena we need to perturb a given Hamiltonian $H$ with a time dependent potential $\dot{F}_{\Phi,\eps}(t)$ that vanishes in the limit $\Phi \to 0$ and that it is ramped up slowly at rate $\eps$. For instance, in the case of the Quantum Hall Effect $H := (- \ii \nabla - A)^2 + V$ is perturbed by an uniform time-dependent electric field $\dot{F}_{\Phi,\eps}(t) := \Phi_1 \, \switch(t) \, X_1$ where $X_1$ is the position operator in the spatial direction $1$ according to a given reference system \cite{Elgart_Schlein:Kubo_for_Landau:2004,Bouclet_Germinet_Klein_Schenker:linear_response_theory_magnetic_Schroedinger_operators_disorder:2005, klein-lenoble-muller-07}. Another example is given by the AC-Stark Effect where the unperturbed Hamiltonian of the hydrogen atom $H := -\Delta - e \, \abs{x}^{-1}$ is perturbed by the oscillating electric field $\dot{F}_{\Phi,\eps}(t) := \Phi_1 \, \cos(\omega t) \, X_1$ \cite{Graffi_Yajima:1983}. In all these cases common sense would suggest to study the perturbed  Hamiltonian
\begin{align}
	\widetilde{H}_{\Phi,\eps}(t) = H + \dot{F}_{\Phi,\eps}(t)
	\label{dynamics:eqn:additive_perturbation_bis}
	.
\end{align}
However, working with the Hamiltonian~\eqref{dynamics:eqn:additive_perturbation_bis} directly can introduce a number of complications: (1)~The perturbation $\dot{F}_{\Phi,\eps}(t)$ is in general not small with respect to $H$, and this makes it difficult to study domain and spectrum of the perturbed Hamiltonian or to just decide whether $\widetilde{H}_{\Phi,\eps}(t)$ is selfadjoint; (2)~It is not clear if $\widetilde{H}_{\Phi,\eps}(t)$ is \emph{still affiliated} to the same von Neumann algebra of $H$ and if bounded functions of $\widetilde{H}_{\Phi,\eps}(t)$ are measurable in the sense of Section~\ref{framework:nc_Lp_Sobolev_spaces:measure}; (3) Finally, there are no general arguments that ensure the existence of the unitary propagator associated to $\widetilde{H}_{\Phi,\eps}(t)$.

Instead of working with a complicated object like \eqref{dynamics:eqn:additive_perturbation_bis}, and in accordance with the philosophy followed in \cite{Graffi_Yajima:1983,Elgart_Schlein:Kubo_for_Landau:2004,Bouclet_Germinet_Klein_Schenker:linear_response_theory_magnetic_Schroedinger_operators_disorder:2005,klein-lenoble-muller-07} (as well as in many other papers), we prefer to work with the isospectrally perturbed Hamiltonian
\begin{align}
	H_{\Phi,\eps}(t) := G_{\Phi,\eps}(t) \, H \, G_{\Phi,\eps}(t)^*
	\label{dynamics:eqn:definition_perturbed_hamiltonian}
\end{align}
where the unitary is related to $\dot{F}_{\Phi,\eps}(t)$ via 
\begin{align*}
	G_{\Phi,\eps}(t) := \e^{+ \ii \int_{-\infty}^t\dot{F}_{\Phi,\eps}(\tau)\dd \tau} 
	.
\end{align*}
The benefits of working with $H_{\Phi,\eps}(t)$ instead of $\widetilde{H}_{\Phi,\eps}(t)$ are evident and have already been discussed in the Section \ref{dynamics:perturbed:gauge_perturbations}. The purpose of this section is different: We want to pay the debt of explaining \emph{in what sense the physics described by $\widetilde{H}_{\Phi,\eps}(t)$ and $H_{\Phi,\eps}(t)$ is the same}. Our line of argumentation follows \cite[Section~2.2]{Bouclet_Germinet_Klein_Schenker:linear_response_theory_magnetic_Schroedinger_operators_disorder:2005}. Let us assume that $\R \ni t \mapsto \psi(t) \in \Hil$ solves the Schrödinger equation
\begin{align}\label{eq:rk:guge_dyn_1}
	\ii \frac{\dd \psi}{\dd t}(t) = H_{\Phi,\eps}(t) \, \psi(t)
	, 
	&&
	\psi \in \domain(H)
	.
\end{align}
with $H_{\Phi,\eps}(t)$ given by \eqref{dynamics:eqn:definition_perturbed_hamiltonian}. Then, at least formally, the time evolution of $\widetilde{\psi}(t) := G_{\Phi,\eps}(t)^* \, \psi(t)$ is governed by the differential equation
\begin{align}
	\ii \frac{\dd \widetilde{\psi}}{\dd t}(t) 
	&= \dot{F}_{\Phi,\eps}(t) \, G_{\Phi,\eps}(t)^* \bigl ( G_{\Phi,\eps}(t) \, \widetilde{\psi}(t) \bigr ) + G_{\Phi,\eps}(t)^* \, H_{\Phi,\eps}(t) \, \bigl ( G_{\Phi,\eps}(t) \, \widetilde{\psi}(t) \bigr )
	\notag \\
	&= \widetilde{H}_{\Phi,\eps}(t) \, \widetilde{\psi}(t)
	.
	\label{dynamics:eqn:dynamics_interaction_representation_bis}
\end{align}
This shows that $\widetilde{\psi}$ is evolved by the perturbed Hamiltonian $\widetilde{H}_{\Phi,\eps}(t)$ given by \eqref{dynamics:eqn:additive_perturbation_bis}. However, a careful inspection shows that the differential equation \eqref{dynamics:eqn:dynamics_interaction_representation_bis} can make sense only if $\widetilde{\psi}(t)$ is in the domain of $\widetilde{H}_{\Phi,\eps}(t)$, a fact which is usually false and, in any case not straight-forward to verify. There is also a second important aspect: In general it is difficult to establish whether the Schrödinger equation~\eqref{dynamics:eqn:dynamics_interaction_representation_bis} admits \emph{strong solutions} or, equivalently, whether $\widetilde{H}_{\Phi,\eps}(t)$ satisfies conditions which guarantee the existence of a unitary propagator (\cf \eg Theorem~\ref{dynamics:thm:existence_perturbed_evolution}).

Then \emph{in which sense can we compare the solutions of \eqref{eq:rk:guge_dyn_1} and \eqref{dynamics:eqn:dynamics_interaction_representation_bis}?} The answer is that  we have to  pay attention to \emph{weak solutions} $t \mapsto \psi(t)$ which meet
\begin{align}
	\ii \frac{\dd }{\dd t} \bscpro{\phi}{\psi(t)} = \bscpro{H_{\Phi,\eps}(t) \phi \,}{\, \psi(t)}
	&&
	\forall \phi \in \domain_{\mathrm{c}}(H)
	\label{dynamics:eqn:weak_dynamics_interaction_representation_bis}
\end{align}
where $\domain_{\mathrm{c}}(H)$ is a suitable time-independent core for $H_{\Phi,\eps}(t)$. The existence of such a core is guaranteed in many situations of interest (see \eg Lemma~\ref{dynamics:lem:invariant_domain}). If this is the case, equation~\eqref{dynamics:eqn:weak_dynamics_interaction_representation_bis} says that $\psi(t)$ is a {weak solution} of the Schrödinger equation generated by $H_{\Phi,\eps}(t)$ \emph{if and only if} $\widetilde{\psi}(t)$ is a weak solution of the Schrödinger equation generated by $\widetilde{H}_{\Phi,\eps}(t)$. We can straightforwardly generalize these arguments to where states are selfadjoint, non-negative operators of trace $1$: Let $\rho := \ketbra{\psi}{\psi}$ be a state, and $\rho_{\rm full}(t) := \ketbra{\psi(t)}{\psi(t)}$ its evolution under the dynamics generated by $H_{\Phi,\eps}(t)$ between a fixed initial time $t_0$ and $t$ (the definition of $\rho_{\rm full}(t)$ for $t_0 = -\infty$ requires some care as discussed in Section~\ref{Kubo_formula:comparing_evolutions:difference_evolved_states}). Similarly, set $\widetilde{\rho}(t) := \ketbra{\widetilde{\psi}(t)}{\widetilde{\psi}(t)}$ to be the evolution of the same state under the dynamics generated by $\widetilde{H}_{\Phi,\eps}(t)$. In view of the discussion above the two evolved states are connected by the transformation $\rho_{\rm full}(t) = G_{\Phi,\eps}(t) \, \widetilde{\rho}(t) \, G_{\Phi,\eps}(t)^*$. Let $J_k = \ii [H , X_k]$ be the unperturbed density current generated by $H$ and consider the two time-dependent density currents
\begin{align*}
	J_k(t) = \ii \, \bigl [ H_{\Phi,\eps}(t) \, , \, X_k \bigr ] = G_{\Phi,\eps}(t) \, J_k \, G_{\Phi,\eps}(t)^*
\end{align*}
and 
\begin{align*}
	\tilde{J}_k(t) = \ii \, \bigl [ \widetilde{H}_{\Phi,\eps}(t) \, , \, X_k \bigr ] = J_k
\end{align*}
generated by $H_{\Phi,\eps}(t)$ and $\widetilde{H}_{\Phi,\eps}(t)$, respectively. Then the invariance of the trace $\mathcal{T}$ with respect to the adjoint action of $G_{\Phi,\eps}(t)$ at least formally implies 
\begin{align}
	\mathcal{T} \bigl ( J_k(t) \; \rho_{\mathrm{full}}(t) \bigr ) = 
	\mathcal{T} \Bigl ( G_{\Phi,\eps}(t) \, J_k \; \widetilde{\rho}(t) \, G_{\Phi,\eps}(t)^* \Bigr )
	= \mathcal{T} \bigl ( \widetilde{J}_k(t) \, \widetilde{\rho}(t) \bigr )
	.
	\label{dynamics:eqn:invariance_expectation_values}
\end{align}
This last equation expresses a crucial fact for our analysis, namely \emph{the time-dependent expectation values of the density current $J_k$ with respect to the state $\rho$ under the evolution defined by the two distinct Hamiltonians $H_{\Phi,\eps}(t)$ and $\widetilde{H}_{\Phi,\eps}(t)$ are the same.} 

From a physical point of view the dynamical equations~\eqref{dynamics:eqn:additive_perturbation_bis} and \eqref{dynamics:eqn:dynamics_interaction_representation_bis} are related by a time-dependent unitary which facilitates a \emph{time-dependent change of representation} to a “co-moving frame”. Due to its similarity to the interaction representation in quantum mechanics, we will formally introduce $A \mapsto G_{\Phi,\eps}(t) \, A \, G_{\Phi,\eps}(t)^*$ in Section~\ref{dynamics:perturbed:interaction_picture} as “interaction evolution” of $A$. Viewing $\rho_{\mathrm{full}}$ and $\widetilde{\rho}(t)$ as the evolution of the same physical state in different representations immediately explains why the two expectation values necessarily coincide. As pointed out in \cite[Section~2.2]{Bouclet_Germinet_Klein_Schenker:linear_response_theory_magnetic_Schroedinger_operators_disorder:2005}, there is no physical reason to prefer one representation over the other because none of the physics depends on the choice of representation. But mathematically it is advantageous, perhaps even necessary to work with the much more “benign” operator, the “unusual”, isospectrally perturbed Hamiltonian $H_{\Phi,\eps}(t)$. 

\subsection{Existence of the unitary propagator} 
\label{dynamics:perturbed:perturbed_propagator}
The second main task of this Chapter is to establish the existence of the perturbed evolution, and recast it in the language of Chapter~\ref{framework}. Under the Hypotheses we need to prove the existence of a \emph{unitary propagator} associated to the \emph{time-dependent} Hamiltonian that solves 
\begin{align}
	\ii \frac{\dd \psi}{\dd t}(t) = H_{\Phi,\eps}(t) \, \psi(t)
	,
	&&
	t \in \R
	, \; 
	\psi(t) \in \Hil
	.
	\label{main_results:eqn:Schroedinger_Yosida_equation}
\end{align}
We recall the definition \cite[Section~X.12]{Reed_Simon:M_cap_Phi_2:1975}:
\begin{definition}[Unitary propagator]\label{defi:unitary_propagator}
	A \emph{unitary propagator} for the equation \eqref{main_results:eqn:Schroedinger_Yosida_equation} is a two-parameter family of unitary operators $\R^2 \ni (t,s) \mapsto U_{\Phi,\eps}(t,s) \in \mathscr{B}(\Hil)$ such that 
	\begin{enumerate}[leftmargin=*,label=(\roman*)]
		\item $U_{\Phi,\eps}(t,s) \, U_{\Phi,\eps}(s,r) = U_{\Phi,\eps}(t,r)$ for all $r , s , t \in \R$, 
		\item $U_{\Phi,\eps}(t,t) = \id$ for all $t \in \R$, and 
		\item $(t,s) \mapsto U_{\Phi,\eps}(t,s)$ is strongly jointly continuous in $s$ and $t$.
	\end{enumerate}
\end{definition}
The existence of the unitary propagator for the Schrödinger equation~\eqref{main_results:eqn:Schroedinger_Yosida_equation} with time-dependent Hamiltonian $H_{\Phi,\eps}(t)$ can be guaranteed by imposing additional, technical conditions:
\begin{definition}[Regular, time-dependent Hamiltonian]\label{dynamics:defn:optimally_perturbed_hamiltonian}
	Suppose $\R \ni t \mapsto H(t) \in \affil(\Alg)$ takes values in the selfadjoint operators, and set 
	\begin{align*}
		C(t,s) := \bigl ( H(t) - H(s) \bigr ) \, \frac{1}{H(s) - \xi}
		. 
	\end{align*}
	A regular, time-dependent Hamiltonian $H(t)$ has the following properties: 
	%
	\begin{enumerate}[leftmargin=*,label=(\roman*)]
		\item The $H(t)$ have common dense domain $\domain \subseteq \Hil$ and there exists an $\xi \in \R$ such that $\xi \in \res \bigl ( H(t) \bigr )$ for all $t \in \R$.
		\item The maps $\R \times \R \ni (s,t) \mapsto \norm{C(t,s)}$ and $\R \times \R \ni (s,t) \mapsto \frac{\norm{C(t,s)}}{\sabs{s - t}}$
		%
		are uniformly continuous and uniformly bounded in $t$ and $s$ for $t \neq s$ lying in any fixed compact subinterval of $\R$.
		\item The norm limit $C(t) := \lim_{s \rightarrow t} \frac{C(t,s)}{s - t}$
		%
		exists uniformly for $t$ in every compact subinterval of $\R$ and $t \mapsto \snorm{C(t)}$ is continuous. Moreover, 
		\begin{align}
			 \int_{-\infty}^t \dd \tau \; \bnorm{C(\tau)} \leqslant +\infty
			&&
			\forall t \in \R
			.
			\label{dynamics:eqn:regular_time_dependent_hamiltonian_integrability_condition}
		\end{align}
	\end{enumerate}
\end{definition}
Note that these conditions are stronger than necessary (compare \eg with \cite[Theorem~X.70]{Reed_Simon:M_cap_Phi_2:1975}). The crucial result for the integration of equation \eqref{main_results:eqn:Schroedinger_Yosida_equation} is provided by the following classical theorem (\cf \cite[Theorem~XIV.4.1]{yosida-80} or the aforementioned \cite[Theorem~X.70]{Reed_Simon:M_cap_Phi_2:1975}):
\begin{theorem}\label{dynamics:thm:existence_perturbed_evolution}
	Let $H(t)$ is a regular, time-dependent Hamiltonian in the sense of Definition~\ref{dynamics:defn:optimally_perturbed_hamiltonian}. Then for all $s , t \in \R$ the unique unitary propagator $U(t,s) \in \Alg$ exists, \ie it leaves the domain invariant, $U(t,s) \domain = \domain$, and solves 
	\begin{align}
		\ii \frac{\dd \psi_s}{\dd t}(t) &= H(t) \, \psi_s(t)
		,
		&&
		\psi_s(s) = \psi_0 \in \domain
		, 
		\label{dynamics:eqn:Schroedinger_equation_s_dependent_initial_condition}
	\end{align}
	in the sense that $U(t,s) \, \psi_0 = \psi_s(t)$. 
\end{theorem}
Observe that condition (i) of Definition~\ref{dynamics:defn:optimally_perturbed_hamiltonian} implies that $0 \in \res \bigl ( H(t) - \xi \bigr )$ for all $t \in \R$. Moreover, in view of the norm estimate on the resolvent 
\begin{align*}
	\norm{\frac{1}{H(t) - (\xi + \ii\lambda)}} \leqslant \frac{1}{\lambda}
	, 
	&&
	\forall \lambda > 0 
	, \; 
	t \in \R
	, 
\end{align*}
the Hille-Yosida theorem \cite[Theorem~X.47a]{Reed_Simon:M_cap_Phi_2:1975} applies, and $\ii \bigl ( H(t) - \xi \bigr )$ generates a \emph{contraction semigroup}. This shows that Definition~\ref{dynamics:defn:optimally_perturbed_hamiltonian}~(i) and condition~(a) in \cite[Theorem~X.70]{Reed_Simon:M_cap_Phi_2:1975} agree (see also the discussion after the statement of \cite[Theorem~X.70]{Reed_Simon:M_cap_Phi_2:1975}). The \emph{unitarity} of the propagator $U(t,s)$ hinges on $H(t)$ (or, equivalently, $H$) being bounded from below or the existence of a gap (see also the proof of \cite[Theorem~2.7]{Bouclet_Germinet_Klein_Schenker:linear_response_theory_magnetic_Schroedinger_operators_disorder:2005}). Conditions~(ii) and (iii) of Definition~\ref{dynamics:defn:optimally_perturbed_hamiltonian} are significantly stronger than the usual conditions necessary to prove Theorem~\ref{dynamics:thm:existence_perturbed_evolution}. More precisely, the construction of the unitary propagator requires only the continuity of the maps $C(t,s) \, (s-t)^{-1}$ and $C(t)$ with respect to the SOT and the latter is evidently implied by the norm-continuity required in conditions~(ii) and (iii). Anyway, condition (ii) of Definition~\ref{dynamics:defn:optimally_perturbed_hamiltonian} is satisfied as the boundedness of the operator $C(t,s)$ follows from the invariance of the domain $\domain$ and the {closed graph theorem} \cite[Theorem III.12]{Reed_Simon:M_cap_Phi_1:1972}. Equation~\eqref{dynamics:eqn:Schroedinger_equation_s_dependent_initial_condition} can be rewritten as
\begin{align}
	\ii \frac{\dd U(t,s)}{\dd t} \psi = H(t) \, U(t,s) \psi
	, 
	&&
	\forall \psi \in \domain
	, 
	\label{dynamics:eqn:Schroedinger_equation_propagator_t}
\end{align}
and the use of the chain rule for $U(t,s) \, U(s,t) = \id$ also implies \cite[Theorem~2.7]{Bouclet_Germinet_Klein_Schenker:linear_response_theory_magnetic_Schroedinger_operators_disorder:2005}
\begin{align}
	\ii \frac{\dd U(t,s)}{\dd s}\psi = - U(t,s) \, H(s) \psi
	, 
	&&
	\forall \psi \in \domain
	.
	\label{dynamics:eqn:Schroedinger_equation_propagator_s}
\end{align}
Let us to point out that the core of the proof of Theorem~\ref{dynamics:thm:existence_perturbed_evolution} is based on the definition of the \emph{approximate propagators} of order $k \in \N$
\begin{align*}
	U^{(k)}(t,s) := \e^{- \ii (t-s) \, H( m + \nicefrac{(j-1)}{k}) } \, \e^{+ \ii (t-s) \xi}
	,
	&&
	m + \tfrac{j-1}{k} \leqslant s , t \leqslant m + \tfrac{j}{k}
	, 
\end{align*}
where $m \in \Z$ and $j = 1 , 2 , \ldots , k$. These are glued together by
\begin{align*}
	U^{(k)}(t,r) := U^{(k)}(t,s) \; U^{(k)}(s,r)
	&&
	\forall t , s , r \in \R
	.
\end{align*}
The unitary propagator $U^{(k)}(t,s)$ entering the Theorem~\ref{dynamics:thm:existence_perturbed_evolution} is then given by
\begin{align*}
	U(t,s) := \slim_{k \rightarrow \infty} U^{(k)}(t,s)
\end{align*}
The last formula proves that $U(t,s) \in \Alg$. In fact the $U^{(k)}(t,s) \in \Alg$ by construction for all $k \in \N$ and the von Neumann algebra $\Alg$ is closed under strong limits.

An important piece of information contained in the proof of Theorem~\ref{dynamics:thm:existence_perturbed_evolution} is that 
\begin{align}
	M(t,s) := \bigl ( H(t) - \xi \bigr ) \, U(t,s) \, \frac{1}{H(s) - \xi}
	,
	&&
	t , s \in \R
	,
	\label{dynamics:eqn:definition_M_ts}
\end{align}
is a bounded operator-valued function which is jointly strongly continuous in $t$ and $s$. The boundedness also implies that $M(t,s) \in \Alg$ (\cf Remark~\ref{framework:remark:algebraic_operations}). Moreover, one has the norm bound
\begin{align}
	\bnorm{M(t,s)} \leqslant \e^{ \int_{\min\{t,s\}}^{\max\{t,s\}} \dd \tau \, \snorm{C(\tau)}}
	\label{dynamics:eqn:norm_bound_M_ts}
\end{align}
which can be extracted from the proof of \cite[Theorem~XIV.4.1, eqn.~(14)]{yosida-80}. By combining this bound with the integrability condition for $C(\tau)$ stated in Definition~\ref{dynamics:defn:optimally_perturbed_hamiltonian}~(iii) one gets
\begin{align}
	\sup_{s \leqslant t} \, \bnorm{M(t,s)} \leqslant \e^{ \int_{-\infty}^{t} \dd \tau \, \snorm{C(\tau)}} =: K_t < +\infty
	.
	\label{dynamics:eqn:sup_s_norm_bound_M_ts}
\end{align}
Let us notice that the rather strong conditions~(ii) and (iii) of Definition~\ref{dynamics:defn:optimally_perturbed_hamiltonian} lead to better regularity of the unitary propagator. A modification of the proof of \cite[Theorem~XIV.4.1]{yosida-80} based on the assumption that $C(t,s)$ and $C(t)$ are \emph{norm}-continuous and not just strongly continuous implies that the maps
\begin{align}\label{dynamics:eqn:norm_continuity_regularized_unitary_evolution}
	(t,s) \mapsto U(t,s) \, \frac{1}{H(s) - \xi}
	,
	&& (t,s) \mapsto \frac{1}{H(t) - \xi} \, U(t,s)
	, 
\end{align}
are both jointly continuous in $t$ and $s$ with respect to the \emph{norm} topology. As a consequence one obtains also that the equations
\begin{subequations}\label{dynamics:eqn:norm_derivatives_unitary_propagator_resolvent_squared}
	\begin{align}
		\ii \frac{\dd }{\dd s} \left ( \frac{1}{(H(t) - \xi)^2} \, U(t,s) \right ) &= - \frac{1}{(H(t) - \xi)^2} \, U(t,s) \, H(s)\label{dynamics:eqn:norm_derivatives_unitary_propagator_resolvent_squared:l}
		\\
		\ii \frac{\dd }{\dd t} \left ( U(t,s) \, \frac{1}{(H(s) - \xi )^2} \right ) &= + H(t) \, U(t,s) \, \frac{1}{(H(s) - \xi)^2}\label{dynamics:eqn:norm_derivatives_unitary_propagator_resolvent_squared:r}
	\end{align}
\end{subequations}
hold true in operator norm (for the details see the proof of \cite[Theorem~2.7]{Bouclet_Germinet_Klein_Schenker:linear_response_theory_magnetic_Schroedinger_operators_disorder:2005}). 

Hamiltonians modified by an adiabatic isospectral perturbations of the type described in Hypothesis~\ref{hypothesis:perturbation} naturally generate time evolutions given by a \emph{unitary} propagator, provided certain technical conditions are met. The next result summarizes a set of sufficient conditions which ensure that an adiabatically perturbed Hamiltonian is also regular in the sense of Definition~\ref{dynamics:defn:optimally_perturbed_hamiltonian}. 
\begin{proposition}\label{dynamics:prop:existence_perturbed_dynamics}
	Suppose Hypotheses~\ref{hypothesis:trace}–\ref{hypothesis:current} hold true. Then the isospectrally perturbed Hamiltonian $H_{\Phi,\eps}(t) = G_{\Phi,\eps}(t) \, H \, G_{\Phi,\eps}(t)^*$ is a regular, time-dependent Hamiltonian in the sense of Definition~\ref{dynamics:defn:optimally_perturbed_hamiltonian}, and the associated unitary propagator $U_{\Phi,\eps}(s,t) \in \Alg$ exists. $U_{\Phi,\eps}(t,s)$ satisfies the continuity conditions described by \eqref{dynamics:eqn:norm_continuity_regularized_unitary_evolution} and \eqref{dynamics:eqn:norm_derivatives_unitary_propagator_resolvent_squared}.
\end{proposition}
\begin{proof}
	The gap around $\xi \in \R$ exists in view of Hypothesis~\ref{hypothesis:current}~(v) and persists because of the isospectrality $\spec \bigl ( H_{\Phi,\eps}(t) \bigr ) = \spec(H)$ in Lemma~\ref{dynamics:lem:invariant_domain}~(4). Moreover, the invariance of the domain is guaranteed by Lemma~\ref{dynamics:lem:invariant_domain}~(3), and therefore the conditions in Definition~\ref{dynamics:defn:optimally_perturbed_hamiltonian}~(i) are satisfied. 
	
	For the remainder of the proof, the crucial quantity to study is $C_{\Phi,\eps}(s,t)$ which involves the difference 
	\begin{align*}
		H_{\Phi,\eps}(t) - H_{\Phi,\eps}(s) &= W_{\Phi,\eps}(t) - W_{\Phi,\eps}(s)
		\\
		&= \sum_{\substack{\kappa \in \N_0^d \\ 1 \leqslant \abs{\kappa} \leqslant N}} \frac{(-1)^{\abs{\kappa}}}{\abs{\kappa}!} \Bigl ( w_{\kappa}^{\eps}(t) - w_{\kappa}^{\eps}(s) \Bigr ) \, J_{\kappa} 
		. 
	\end{align*}
	Thanks to Lemma~\ref{dynamics:lem:invariant_domain}~(2) we know that this difference can be expressed as above, is densely defined on the core $\domain_{\mathrm{c}}(H) = \domain(H) \cap \domain_{\mathrm{c}}$ and is $H$-bounded. 
	
	First, we need norm estimates of the current operators and the resolvent, namely for any multi-index $\kappa \in \N_0^d$ 
	\begin{align*}
		\norm{J_{\kappa} \, \frac{1}{H_{\Phi,\eps}(s) - \xi}} &\leqslant \norm{J_{\kappa} \, \frac{1}{H - \xi}(H - \xi) \, G_{\Phi,\eps}(s) \, \frac{1}{H - \xi}}
		\\
		&\leqslant \norm{J_{\kappa} \, \frac{1}{H - \xi}} \; \norm{(H - \xi) \, G_{\Phi,\eps}(s) \, \frac{1}{H - \xi}}
		\\
		&\leqslant c_J \, \left ( \abs{\xi} \, c_{\xi} + \norm{H_{-\Phi,\eps}(s) \, \frac{1}{H - \xi}} \right )
	\end{align*}
	has a bound in terms of the constants 
	\begin{align}
		c_J := \max_{\substack{\kappa \in \N_0^d \\ 1 \leqslant \abs{\kappa} \leqslant N}}
		\norm{J_{\kappa} \; \frac{1}{H - \xi}} 
		,
		&&
		c_{\xi} := \norm{\frac{1}{H - \xi}} 
		\label{dynamics:eqn:definition_constants_norm_estimate_product_current_resolvent}
	\end{align}
	and the norm of $H_{-\Phi,\eps}(s)$ times the resolvent. The latter is no larger than 
	\begin{align*}
		\norm{{H}_{-\Phi,\eps}(s) \; \frac{1}{H - \xi}} \leqslant 
		\kappa_{\xi} + c_J \, g^{(1)}(s)
	\end{align*}
	where $\kappa_{\xi}$ is the constant defined in \eqref{dynamics:eqn:definition_constant_E_E_minus_xi} and the time-dependent function 
	\begin{align*}
		g^{(1)}(s) := \sum_{\substack{\kappa \in \N_0^d \\ 1 \leqslant \abs{\kappa} \leqslant N}} \frac{1}{\abs{\kappa}!} \babs{w_{\kappa}^{\eps}(s)} 
	\end{align*}
	stems from expressing $H_{-\Phi,s}(t)$ as the sum of $H$ and a potential via equation~\eqref{dynamics:eqn:additive_form_perturbation}. These coefficients $s \mapsto w_{\kappa}^{\eps}(s) = \prod_{j = 1}^d \Phi_j^{\eps}(s)^{\kappa_j}$, and therefore their finite sum above, are evidently elements in $C^1(\R)$. 
	
	This means the norm of $C_{\Phi,\eps}(t,s)$ has a bound 
	\begin{align}
		\bnorm{C_{\Phi,\eps}(t,s)} &= \norm{\bigl ( W_{\Phi,\eps}(t) - W_{\Phi,\eps}(s) \bigr ) \, \frac{1}{H_{\Phi,\eps}(s) - \xi}}
		\notag \\
		&\leqslant g^{(2)}(t,s) \; c_J \left ( \abs{\xi} \, c_{\xi} + \kappa_{\xi} + c_J \, g^{(1)}(s) \right )
		\label{dynamics:eqn:uniform_continuity_C_t_s}
	\end{align}
	in terms of the above constants, $g^{(1)}$ and 
	\begin{align*}
		g^{(2)}(t,s) := \sum_{\substack{\kappa \in \N_0^d \\ 1 \leqslant \abs{\kappa} \leqslant N}} \frac{1}{\abs{\kappa}!} \babs{w_{\kappa}^{\eps}(t) - w_{\kappa}^{\eps}(s)} 
		. 
	\end{align*}
	The two functions, $g^{(1)}(s)$ and $g^{(2)}(t,s)$ are $C^1$ in their time variables, so estimate~\eqref{dynamics:eqn:uniform_continuity_C_t_s} in fact makes sure that both, $C_{\Phi,\eps}(t,s)$ and $\frac{1}{t-s}C_{\Phi,\eps}(t,s)$ (with $t \neq s$) are uniformly continuous and uniformly bounded in operator norm for $t$ and $s$ restricted to a compact interval. This proves our Hamiltonian satisfies Definition~\ref{dynamics:defn:optimally_perturbed_hamiltonian}~(ii).
	
	For the last item, we need to show the existence of the limit 
	\begin{align*}
		C_{\Phi,\eps}(t) = \lim_{s \to t} \frac{C_{\Phi,\eps}(s,t)}{s - t}
	\end{align*}
	and prove the continuity of $t \mapsto \bnorm{C_{\Phi,\eps}(t)}$. To address the existence, we note that the difference quotient $\bigl ( W_{\Phi,\eps}(t) - W_{\Phi,\eps}(s) \bigr ) \, (t-s)^{-1}$ is $H$-bounded. 
	Thanks to that, we compute 
	\begin{align*}
		C_{\Phi,\eps}(t) = \sum_{\substack{\kappa \in \N_0^d \\ 1 \leqslant \abs{\kappa} \leqslant N}} \left ( \frac{(-1)^{\abs{\kappa}}}{\abs{\kappa}!} \frac{\dd w_{\kappa}^{\eps}}{\dd t}(t) \, J_{\kappa} \, \frac{1}{H_{\Phi,\eps}(t) - \xi} \right ) 
		, 
	\end{align*}
	which immediately allows us to estimate this quantity in terms of the time derivative of the coefficients of the potential, 
	\begin{align*}
		{g'}^{(1)}(t) := c_J \, \sum_{\substack{\kappa \in \N_0^d \\ 1 \leqslant \abs{\kappa} \leqslant N}} \frac{1}{\abs{\kappa}!} \abs{\frac{\dd w_{\kappa}^{\eps}}{\dd t}(t)} 
		. 
	\end{align*}
	Then adapting the arguments which led to \eqref{dynamics:eqn:uniform_continuity_C_t_s}, and obtain 
	\begin{align*}
		\bnorm{C_{\Phi,\eps}(t)} &\leqslant {g'}^{(1)}(t) \; c_J \Bigl ( \abs{\xi} \, c_{\xi} + \kappa_{\xi} + c_J \, g^{(1)}(t) \Bigr )
		\\
		&\leqslant c \, {g'}^{(1)}(t) \bigl ( 1 + g^{(1)}(t) \bigr )
		. 
	\end{align*}
	Lastly, we will verify the integrability condition~\eqref{dynamics:eqn:regular_time_dependent_hamiltonian_integrability_condition}: essentially, we have to show that $t \mapsto {g'}^{(1)}(t) \in C^1(\R)$ as well as the the product $t \mapsto {g'}^{(1)}(t) \; g^{(1)}(t) \in C^1(\R)$ of these two $C^1$ functions decays sufficiently rapidly in time as $t \to -\infty$. 
	
	This estimate can be made more explicit after we plug in \eqref{main_results:eqn:perturbation_potential} and \eqref{main_results:eqn:definition_Phi_t}, 
	\begin{align}
		\bnorm{C_{\Phi,\eps}(t)} \leqslant \switch(t) \sum_{j = 1}^d \sabs{\Phi_j} \, \abs{f_j(t)} \, h_j^{\Phi}(t)
		\label{dynamics:eqn:norm_bound_M_ts0}
	\end{align}	
	where the $h_j^{\Phi} \in C^0(\R)$ are suitable continuous functions such that $h_j^{\Phi}(t) \to 0$ when $t \to -\infty$. Thus, the integrability condition~\eqref{main_results:eqn:integrability_condition} for the $\switch \, \sabs{f_j}$ also ensures the integrability of $\snorm{C_{\Phi,\eps}(t)}$ required in Definition~\ref{dynamics:defn:optimally_perturbed_hamiltonian}~(iii). This concludes the proof. 
\end{proof}
%

\subsection{Evolution of observables} 
\label{dynamics:perturbed:observables}
Let $(t,s) \mapsto U_{\Phi,\eps}(t,s) \in \Alg$ be a unitary propagator in the sense of Definition~\ref{defi:unitary_propagator}. This induces a \emph{dynamics} $(t,s) \mapsto \alpha^{\Phi,\eps}_{(t,s)} \in {\rm Aut}(\Alg)$ defined by
\begin{align}
	\alpha_{(t,s)}^{\Phi,\eps}(A) := U_{\Phi,\eps}(t,s) \, A \, U_{\Phi,\eps}(s,t)
	,
	&&
	t , s \in \R
	, \; 
	A \in \Alg
	, 
	\label{dynamics:eqn:full_evolution_observables}
\end{align}
The automorphisms $\alpha^{\Phi,\eps}_{(t,s)}$ are trace-preserving and the mapping $(t,s) \mapsto \alpha^{\Phi,\eps}_{(t,s)}$ is jointly ultra-weakly continuous in $t$ and $s$. As a consequence one can adapt the proof of Proposition~\ref{framework:prop:extension_isometry_Lp} to conclude that each $\alpha^{\Phi,\eps}_{(t,s)}$ extends canonically to a $\ast$-automorphism of $\rr{M}(\Alg)$ and consequently restricts to a linear isometry on each of the $\mathfrak{L}^p(\Alg)$ (and to a unitary operator for $p = 2$). Moreover, the mapping $(t,s) \mapsto \alpha^{\Phi,\eps}_{(t,s)} \in \mathrm{Iso} \bigl ( \mathfrak{L}^p(\Alg) \bigr )$ turns out to be jointly \emph{strongly} continuous in $t$ and $s$. Finally, from the defining properties of the unitary propagator $U^{\Phi,\eps}(t,s)$ one derives similar properties for the $\alpha^{\Phi,\eps}_{(t,s)}$. More precisely one can verify that on each Banach space $\mathfrak{L}^p(\Alg)$ or on $\rr{M}(\Alg)$ the following properties hold:
\begin{enumerate}[leftmargin=*,label=(\roman*)]
	\item $\alpha^{\Phi,\eps}_{(t,s)}\circ\alpha^{\Phi,\eps}_{(s,r)} = \alpha^{\Phi,\eps}_{(t,r)}$
	\item $\alpha^{\Phi,\eps}_{(t,t)} = \id$
	\item $\alpha^{\Phi,\eps}_{(t,s)}(A)^* = \alpha_{(t,s)}(A^*)$
\end{enumerate}
We refer to the mapping 
\begin{align*}
	\R \times \R \ni (t,s) \mapsto \alpha^{\Phi,\eps}_{(t,s)} \in \mathrm{Iso} \bigl ( \mathfrak{L}^p(\Alg) \bigr )
\end{align*}
as the \emph{perturbed} or \emph{full dynamics} induced by $H_{\Phi,\eps}(t)$ on $\mathfrak{L}^p(\Alg)$ (\cf Consequence~\ref{main_results:conseq:perturbed_evolution}). The next result concerns differential properties of the dynamics $\alpha^{\Phi,\eps}_{(t,s)}$.
\begin{proposition}\label{dynamics:prop:properties_perturbed_dynamics}
	Suppose Hypotheses~\ref{hypothesis:trace}–\ref{hypothesis:current} hold true, and assume in addition that we are given an $A \in \mathfrak{L}^r(\Alg)$, $r = 1 , p$, so that at some $t_0 \in \R$ it lies in $\rr{D}^{00}_{H_{\Phi,\eps}(t_0),r}$ defined by \eqref{framework:eqn:domain_maximal_generalized_commutator}. Then the following holds:
	\begin{enumerate}
		\item $A \in \rr{D}^{00}_{H_{\Phi,\eps}(t),r}$ holds for all $t \in \R$.
		\item The map $\R \ni s \mapsto \alpha^{\Phi,\eps}_{(t,s)}(A) \in \mathfrak{L}^r(\Alg)$ is differentiable and
		\begin{align}
			\ii \frac{\dd}{\dd s} \alpha^{\Phi,\eps}_{(t,s)}(A) = - \alpha^{\Phi,\eps}_{(t,s)} \bigl ( \bigl [ H_{\Phi,\eps}(s) \, , \, A \bigr ]_{\ddagger} \bigr )
			\label{dynamics:eqn:alpha_solves_dynamical_equation_in_s}
		\end{align}
		where, according to Definition~\ref{framework:defn:generalized_commutators}, 
		\begin{align*}
			\bigl [ H_{\Phi,\eps}(s) \, , \, A \bigr ]_{\ddagger} = H_{\Phi,\eps}(s) \, A - \bigl ( H_{\Phi,\eps}(s) \, A^* \bigr )^* 
			. 
		\end{align*}
		\item $\alpha^{\Phi,\eps}_{(t,s)}(A) \in \rr{D}^{00}_{H_{\Phi,\eps}(t),p}$ holds for all $s,t \in \R$. Moreover, the map $\R \ni t \mapsto \alpha^{\Phi,\eps}_{(t,s)}(A) \in \mathfrak{L}^r(\Alg)$ is differentiable and
		\begin{align}
			\ii \frac{\dd}{\dd t} \alpha^{\Phi,\eps}_{(t,s)}(A) = \bigl [ H_{\Phi,\eps}(t) \, , \, \alpha_{(t,s)}^{\Phi,\eps}(A) \bigr ]_{\ddagger}
			\label{dynamics:eqn:perturbed_evolution_observables}
		\end{align}
		where 
		\begin{align*}
			\bigl [ H_{\Phi,\eps}(t) \, , \, \alpha^{\Phi,\eps}_{(t,s)}(A) \bigr ]_{\ddagger} = H_{\Phi,\eps}(t) \, \alpha_{(t,s)}^{\Phi,\eps}(A) - \bigl ( H_{\Phi,\eps}(t) \, \alpha^{\Phi,\eps}_{(t,s)}(A^*) \bigr )^* 
			.
		\end{align*}
	\end{enumerate} 
\end{proposition}
\begin{proof}
	\begin{enumerate}
		\item Notice that
		\begin{align*}
			H_{\Phi,\eps}(t) &= C_{\Phi,\eps}(t,t_0) \, \bigl ( H_{\Phi,\eps}(t_0) - \xi \bigr ) + H_{\Phi,\eps}(t_0)
			\\
			&= \bigl ( C_{\Phi,\eps}(t,t_0) + \id \bigr ) \, H_{\Phi,\eps}(t_0) - \xi \, C_{\Phi,\eps}(t,t_0)
			.
		\end{align*}
		can be written in terms of $H_{\Phi,\eps}(t_0)$ and $C_{\Phi,\eps}(t,t_0)$ from Definition~\ref{dynamics:defn:optimally_perturbed_hamiltonian}. Because $C_{\Phi,\eps}(t,t_0)$ is affiliated to $\Alg$ and bounded, it is also an element of the von Neumann algebra. Therefore, the $\Alg$-bimodule property of $\mathfrak{L}^r(A)$ implies that $H_{\Phi,\eps}(t) \, A \in \mathfrak{L}^r(\Alg)$ whenever $H_{\Phi,\eps}(t_0) \, A \in \mathfrak{L}^r(\Alg)$.
		\item By starting from the identity
		\begin{align}
			\frac{\alpha^{\Phi,\eps}_{(t,s+\delta)}(A) - \alpha^{\Phi,\eps}_{(t,s)}(A)}{\delta} &= \frac{U_{\Phi,\eps}(t,s + \delta) - U_{\Phi,\eps}(t,s)}{\delta} \, A \, U_{\Phi,\eps}(s + \delta,t)
			+ \notag \\
			&\qquad 
			+ U_{\Phi,\eps}(t,s) \, \left ( \frac{U_{\Phi,\eps}(t,s + \delta) - U_{\Phi,\eps}(t,s)}{\delta} \, A^* \right )^*
			\label{eq:TRAX01}
		\end{align}
		and by using (1) the proof here follows along the same lines as that of Proposition~\ref{dynamics:prop:density_core_Liouvillian}. Here the crucial ingredients are the strong limit 
		\begin{align*}
			\slim_{\delta \to 0} \frac{U_{\Phi,\eps}(t,s + \delta) - U_{\Phi,\eps}(t,s)}{\delta} \, \frac{1}{H_{\Phi,\eps}(s) - \xi} = 
			\ii \, U_{\Phi,\eps}(t,s) \, \frac{H_{\Phi,\eps}(s)}{H_{\Phi,\eps}(s) - \xi}
		\end{align*}
		where we have used that $U_{\Phi,\eps}(t,s)$ solves the Schrödinger equation~\eqref{dynamics:eqn:Schroedinger_equation_propagator_s} as well as the norm bound 
		\begin{align}
			&\norm{\frac{U_{\Phi,\eps}(t,s+\delta) - U_{\Phi,\eps}(t,s)}{\delta} \, \frac{1}{H_{\Phi,\eps}(s) - \xi}} = 
			\notag \\
			&\qquad \qquad = 
			\norm{\frac{U_{\Phi,\eps}(s,s + \delta) - \id}{\delta} \, \frac{1}{H_{\Phi,\eps}(s) - \xi}} 
			\leqslant C
			\label{eq:bound_norm}
		\end{align}
		where $C > 0$ is a constant independent of $\delta$. The bound \eqref{eq:bound_norm} can be proved in the following way. Observe that for all $\psi \in \Hil$ the following equality holds
		\begin{align*}
			\frac{U_{\Phi,\eps}(s,s + \delta) - \id}{\delta} \; &\frac{1}{H_{\Phi,\eps}(s) - \xi} \; \psi 
			= \\
			&= \frac{1}{\delta} \, \left[ U_{\Phi,\eps}(s,\tau) \; \frac{1}{H_{\Phi,\eps}(s) - \xi} \; \psi \right ]_{\tau = s}^{\tau = s + \delta}
			\\
			&= \frac{1}{\delta} \int_s^{s+\delta} \dd \tau \, \frac{\dd}{\dd \tau} U_{\Phi,\eps}(s,\tau) \; \frac{1}{H_{\Phi,\eps}(s) - \xi} \; \psi
			\\
			&= \frac{\ii}{\delta} \int_s^{s + \delta} \dd \tau \; U_{\Phi,\eps}(s,\tau) \, H_{\Phi,\eps}(\tau) \; \frac{1}{H_{\Phi,\eps}(s) - \xi} \; \psi
			.
		\end{align*}
		This gives rise to the estimate 
		\begin{align*}
			&\norm{\frac{U_{\Phi,\eps}(s,s + \delta) - \id}{\delta} \, \frac{1}{H_{\Phi,\eps}(s) - \xi} \psi}_{\Hil}
			\leqslant \\
			&\qquad 
			\leqslant \frac{1}{\delta} \int_s^{s + \delta} \dd \tau \, \norm{H_{\Phi,\eps}(\tau) \, \frac{1}{H_{\Phi,\eps}(s) - \xi}} \, \snorm{\psi}_{\Hil}
			\\
			&\qquad 
			\leqslant \sup_{\tau \in [s,s + \delta]} \norm{H_{\Phi,\eps}(\tau) \, \frac{1}{H_{\Phi,\eps}(s) - \xi}} \, \snorm{\psi}_{\Hil}
		\end{align*}
		where the right-hand side can be controlled by 
		\begin{align*}
			\sup_{\tau \in [s,s + \delta]} &\norm{H_{\Phi,\eps}(\tau) \, \frac{1}{H_{\Phi,\eps}(s) - \xi}} \leqslant
			\\
			&\leqslant \sup_{\tau \in [s,s + \delta]} \left ( \norm{\bigl ( H_{\Phi,\eps}(\tau) - H_{\Phi,\eps}(s) \bigr ) \, \frac{1}{H_{\Phi,\eps}(s) - \xi}} 
			+ \right . 
			\\
			&\qquad \qquad \qquad \left . 
			+ \norm{H_{\Phi,\eps}(s) \, \frac{1}{H_{\Phi,\eps}(s) - \xi}} \right ) 
			\\
			&\leqslant \max_{\tau \in [s,s + \delta]} \bnorm{C_{\Phi,\eps}(\tau,s)} + \sup_{E \in \spec(H)} \abs{\frac{E}{E - \xi}}
			. 
		\end{align*}
		Put another way, we have just shown that 
		\begin{align*}
			\norm{\cdot}_p-\lim_{\delta \to 0} \; &\frac{U_{\Phi,\eps}(t,s + \delta) - U_{\Phi,\eps}(t,s)}{\delta} \; A \; U_{\Phi,\eps}(s + \delta,t) 
			= \\
			&\qquad 
			= \ii \, U_{\Phi,\eps}(t,s) \, H_{\Phi,\eps}(s) \, A \; U_{\Phi,\eps}(s,t)
		\end{align*}
		and
		\begin{align*}
			\norm{\cdot}_p-\lim_{\delta \to 0} \frac{U_{\Phi,\eps}(t,s + \delta) - U_{\Phi,\eps}(t,s)}{\delta} \; A^* = \ii \, U_{\Phi,\eps}(t,s) \, H_{\Phi,\eps}(s) \, A^*
		\end{align*}
		and the \eqref{dynamics:eqn:alpha_solves_dynamical_equation_in_s} follows by combining the last two limits with equation \eqref{eq:TRAX01} and by using the fact that the adjoint map $B \mapsto B^*$ is an isometry in $\mathfrak{L}^r(\Alg)$.
		\item Rewriting the product 
		\begin{align*}
			H_{\Phi,\eps}(t) \, U_{\Phi,\eps}(t,s) \, \frac{1}{H_{\Phi,\eps}(s) - \xi} = M_{\Phi,\eps}(t,s) + \xi \; U_{\Phi,\eps}(t,s) \; \frac{1}{H_{\Phi,\eps}(s) - \xi}
			,
		\end{align*}
		so as to involve the bounded operator $M_{\Phi,\eps}(t,s) \in \Alg$ from \eqref{dynamics:eqn:definition_M_ts} and another operator which is evidently an element of $\Alg$ means also the left-hand side is in $\Alg$. From (1) we even know that $\bigl ( H_{\Phi,\eps}(s) - \xi \bigr ) \, A$ is in fact an element of $\mathfrak{L}^r(\Alg)$, and Lemma~\ref{framework:lem:extension_algebra_unbounded_operators} lets us conclude that also 
		\begin{align}
			&\left ( H_{\Phi,\eps}(t) \, U_{\Phi,\eps}(t,s) \, \frac{1}{H(s) - \xi} \right ) \, \Bigl ( \bigl ( H_{\Phi,\eps}(s) - \xi \bigr ) \, A \Bigr ) 
			= \notag \\
			&\qquad \qquad 
			= H_{\Phi,\eps}(t) \, U_{\Phi,\eps}(t,s) \, A \in \mathfrak{L}^r(\Alg)
			\label{dynamics:eqn:product_H_U_A_in_Lp}
		\end{align}
		is in fact in $\mathfrak{L}^r(\Alg)$ as well. By using again the associativity of the right $\Alg$-module structure of $\mathfrak{L}^r(\Alg)$ (Lemma~\ref{framework:lem:extension_algebra_unbounded_operators}) one immediately concludes
		\begin{align*}
			H_{\Phi,\eps}(t) \; \alpha^{\Phi,\eps}_{(t,s)}(A) = \Bigl ( H_{\Phi,\eps}(t) \; U_{\Phi,\eps}(t,s) \; A \Bigr ) \; U_{\Phi,\eps}(s,t) \in \mathfrak{L}^r(\Alg)
			.
		\end{align*}
		The same result also holds when $A$ is replaced by $A^*$. The last part of the proof follows by considering the difference quotient 
		\begin{align}
			\frac{\alpha^{\Phi,\eps}_{(t + \delta,s)}(A) - \alpha^{\Phi,\eps}_{(t,s)}(A)}{\delta} &= \frac{U_{\Phi,\eps}(t+\delta,s) - U_{\Phi,\eps}(t,s)}{\delta} \, A \, U_{\Phi,\eps}(s,t + \delta)
			\, + \notag \\
			&\qquad 
			+ U_{\Phi,\eps}(t,s) \, \left ( \frac{U_{\Phi,\eps}(t+\delta,s) - U_{\Phi,\eps}(t,s)}{\delta} \; A^* \right )^*. 
			\label{dynamics:eqn:difference_quotient_perturbed_dynamics}
		\end{align}
		We adapt the arguments from the proof of (2) to conclude that
		\begin{align*}
			\norm{\cdot}_p-\lim_{\delta \to 0} \, &\frac{U_{\Phi,\eps}(t + \delta,s) - U_{\Phi,\eps}(t,s)}{\delta} \, A \, U_{\Phi,\eps}(s,t+\delta) 
			= 
			\\
			&= - \ii \, \Bigl ( H_{\Phi,\eps}(t) \, U_{\Phi,\eps}(t,s) \, A \Bigr ) \, U_{\Phi,\eps}(s,t)
			\\
			&= - \ii \, H_{\Phi,\eps}(t) \, \alpha^{\Phi,\eps}_{(t,s)}(A)
		\end{align*}
		and similarly 
		\begin{align*}
			\norm{\cdot}_p-\lim_{\delta \to 0} \, \frac{U_{\Phi,\eps}(t+\delta,s) - U_{\Phi,\eps}(t,s)}{\delta} \, A^* &= - \ii \, H_{\Phi,\eps}(t) \, U_{\Phi,\eps}(t,s) \, A^*
		\end{align*}
		for the second term. By inserting these result into \eqref{dynamics:eqn:difference_quotient_perturbed_dynamics} yields
		\begin{align*}
			\ii \frac{\dd}{\dd t}\alpha^{\Phi,\eps}_{(t,s)}(A) = H_{\Phi,\eps}(t) \, \alpha^{\Phi,\eps}_{(t,s)}(A) - U_{\Phi,\eps}(t,s) \, \bigl ( H_{\Phi,\eps}(t) \, U_{\Phi,\eps}(t,s) \, A^* \bigr )^*
			.
		\end{align*}
		The proof is completed by observing that
		\begin{align*}
			U_{\Phi,\eps}(t,s) \, \bigl ( H_{\Phi,\eps}(t) \, U_{\Phi,\eps}(t,s) \, A^* \bigr )^* &= \Bigl ( \bigl ( H_{\Phi,\eps}(t) \, U_{\Phi,\eps}(t,s) \, A^* \bigr ) \, U_{\Phi,\eps}(s,t) \Bigr )^*
			\\
			&= \bigl ( H_{\Phi,\eps}(t) \, \alpha^{\Phi,\eps}_{(t,s)}(A^*) \bigr )^*
		\end{align*}
		where the first equality is justified by fact that $H_{\Phi,\eps}(t) \, U_{\Phi,\eps}(t,s) \, A^*$ is in $\mathfrak{L}^r(\Alg)$ while the second equality follows from the associativity of the right $\Alg$-module structure of $\mathfrak{L}^r(\Alg)$ (Lemma~\ref{framework:lem:extension_algebra_unbounded_operators}). This concludes the proof. 
	\end{enumerate}
\end{proof}
Under the assumptions of Proposition~\ref{dynamics:prop:properties_perturbed_dynamics} one immediately obtains estimates for the $p$-norms of the product of $H_{\Phi,\eps}(t) - \xi$ with the perturbed evolution of $A$, 
\begin{align*}
	\Bnorm{\bigl ( H_{\Phi,\eps}(t) - \xi \bigr ) \, \alpha_{(t,s)}^{\Phi,\eps}(A)}_p 
	&= \Bnorm{M_{\Phi,\eps}(t,s) \; \bigl ( H_{\Phi,\eps}(s) - \xi \bigr ) \, A \, U_{\Phi,\eps}(s,t)}_p
	\\
	&\leqslant \, \bnorm{M_{\Phi,\eps}(t,s)} \, \bnorm{\bigl ( H_{\Phi,\eps}(s) - \xi \bigr ) \, A}_p
	, 
\end{align*}
in terms of $M_{\Phi,\eps}(t,s)$ defined by \eqref{dynamics:eqn:definition_M_ts}, a similar inequality for the time-evolution of $A^*$ and of the commutator
\begin{align}
	&\Bnorm{\bigl [ H_{\Phi,\eps}(t) \, , \, \alpha_{(t,s)}^{\Phi,\eps}(A) \bigr ]}_p 
	\leqslant \nonumber\\
	&\qquad \qquad 
	\leqslant \Bnorm{(H_{\Phi,\eps}(s) - \xi) \, \alpha_{(t,s)}^{\Phi,\eps}(A)}_p + \Bnorm{(H_{\Phi,\eps}(s) - \xi) \, \alpha^{\Phi,\eps}_{(t,s)}(A^*)}_p
	\notag
	\\
	&\qquad \qquad 
	\leqslant 2 \, \bnorm{M_{\Phi,\eps}(t,s)} \, \max_{B = A , A^*} \bnorm{\bigl ( H_{\Phi,\eps}(s) - \xi \bigr ) \, B}_p
	.
	\label{dynamics:eqn:norm_bound_commutator_H_evolved_observable}
\end{align}
%

\subsection{Interaction evolution of observables} 
\label{dynamics:perturbed:interaction_picture}
%
According to Proposition~\ref{framework:prop:extension_isometry_Lp} the isospectral perturbations $G_{\Phi,\eps}(t)$ of Hypothesis~\ref{hypothesis:perturbation} induce a strongly continuous family of isometries 
\begin{align*}
	\R \ni t \mapsto \gamma^{\Phi,\eps}_t \in \mathrm{Iso} \bigl ( \mathfrak{L}^p(\Alg) \bigr )
\end{align*}
through the prescription
\begin{align}
	\gamma^{\Phi,\eps}_t(A) := G_{\Phi,\eps}(t) \, A \, G_{\Phi,\eps}(t)^*
	,
	&&
	A \in \mathfrak{L}^p(\Alg)
	. 
	\label{dynamics:eqn:definition_interaction_evolution}
\end{align}
The $t \mapsto \gamma^{\Phi,\eps}_t$ is called \emph{interaction dynamics}. The important properties of of this map are summarized in the following
\begin{proposition}\label{dynamics:prop:interaction_dynamics}
	Suppose Hypotheses~\ref{hypothesis:trace}–\ref{hypothesis:current} are satisfied. Then we have:
	\begin{enumerate}
		\item $\gamma^{\Phi,\eps}_t \to \id_{\mathfrak{L}^p}$ strongly when $t \to -\infty$. 
		\item $\gamma^{\Phi,\eps}_t \to \id_{\mathfrak{L}^p}$ strongly when $\Phi \to 0$.
		\item Let $\mathfrak{W}^{1,p}(\Alg)$ be the non-commutative Sobolev space from Definition~\ref{framework:defn:non_commutative_Sobolev}. Then for all $A \in \mathfrak{W}^{1,p}(\Alg)$ the map $t \mapsto \gamma^{\Phi,\eps}_t(A)$ can be differentiated at each $t \in \R$, 
		\begin{align}
			\norm{\cdot}_p-\lim_{\delta \to 0} \frac{\gamma^{\Phi,\eps}_{t + \delta}(A) - \gamma^{\Phi,\eps}_t(A)}{\delta} \,&=\, \gamma^{\Phi,\eps}_t \left ( \switch(t) \, \pmb{f}^{\Phi}(t)\cdot \nabla(A) \right ) 
			\notag \\
			&= \switch(t) \; \pmb{f}^{\Phi}(t)\cdot \nabla \left ( \gamma^{\Phi,\eps}_t(A) \right ) 
			.
			\label{dynamics:eqn:generator_perturbed_evolution}
		\end{align}
	\end{enumerate} 
\end{proposition}
\begin{proof}
	\begin{enumerate}
		\item follows from $\lim_{t \to -\infty} G_{\Phi,\eps}(t) = \id$ in the SOT
		and Lemma \ref{framework:lem:strong_convergence_trace_product}.
		\item This just relies on the observation that $G_{\Phi,\eps}(t)$ is a $d$-parameter group of unitary operators that are continuous in $\pmb{\Phi}$ with respect to the SOT. Then Proposition~\ref{framework:prop:extension_isometry_Lp} says that for fixed $t$ the interaction evolution $\gamma^{\Phi,\eps}_t$ is strongly continuous in $\pmb{\Phi}$ on each $\mathfrak{L}^p(\Alg)$. Therefore, the claim is just a consequence of $\slim_{\Phi \to 0}G_{\Phi,\eps}(t) = \id$.
		\item Let us start by justifying the second equality in \eqref{dynamics:eqn:generator_perturbed_evolution}. Let us consider the $\R$-flow $\eta^{X_j}_s(A) := \e^{+ \ii s X_j} \, A \, \e^{- \ii s X_j}$ which defines the derivation $\partial_{X_j}$. Since $\eta^{X_j}_s \circ \gamma^{\Phi,\eps}_t = \gamma^{\Phi,\eps}_t \circ \eta^{X_j}_s$ for all $(t,s) \in \R^2$ we conclude that $\gamma^{\Phi,\eps}_t(A) \in \rr{D}_{X_j,p}$ if and only if $A \in \rr{D}_{X_j,p}$ for all $j = 1 , \ldots , d$. Hence, $\gamma^{\Phi,\eps}_t : \mathfrak{W}^{1,p}(\Alg) \longrightarrow \mathfrak{W}^{1,p}(\Alg)$ is well-defined for all $t \in \R$ and $\partial_{X_j} \bigl ( \gamma^{\Phi,\eps}_t(A) \bigr ) = \gamma^{\Phi,\eps}_t \bigl ( \partial_{X_j}(A) \bigr )$. At this point the second equality in \eqref{dynamics:eqn:generator_perturbed_evolution} just follows from  linearity. For the first equality let us observe that $\bigl ( \gamma^{\Phi,\eps}_t \bigr )^{-1}(A) = G_{\Phi,\eps}(t)^* \, A \, G_{\Phi,\eps}(t)$ and so we have
		\begin{align*}
			\gamma^{\Phi,\eps}_{t + \delta}(A) - \gamma^{\Phi,\eps}_t(A) = {\gamma^{\Phi,\eps}_t} \, \left ( \bigl ( \gamma^{\Phi,\eps}_t \bigr )^{-1} \circ \gamma^{\Phi,\eps}_{t + \delta}(A) - A \right ) 
			. 
		\end{align*}
		But since ${\gamma^{\Phi,\eps}_t}$ is an isometry in $\mathfrak{L}^p(\Alg)$ the only thing that is left to prove is 
		\begin{align*}
			\norm{\cdot}_p-\lim_{\delta \to 0} \frac{\bigl ( \gamma^{\Phi,\eps}_t \bigr )^{-1} \circ \, \gamma^{\Phi,\eps}_{t + \delta}(A) - A}{\delta} = \gamma^{\Phi,\eps}_t \left ( \sum_{j = 1}^d \Phi_j \; \switch(t)\, f_j(\eps t) \, \partial_{X_j}(A) \right ) 
			.
		\end{align*}
		The map $\bigl ( \gamma^{\Phi,\eps}_t \bigr )^{-1} \circ \, \gamma^{\Phi,\eps}_{t + \delta}$ acts on $A$ by conjugation with the unitary operator 
		\begin{align*}
			\widetilde{G}_{\Phi,\eps,\delta}(t) := 
			\prod_{j = 1}^d \e^{+ \ii \Phi_j \bigl ( \int_{t}^{t+\delta} \dd \tau \, \switch(\tau)\,f_j(\tau) \bigr ) X_j}
			. 
		\end{align*}
		Evidently, this operator depends continuously on $\delta$ in the SOT, and in the limit $\widetilde{G}_{\Phi,\eps,\delta = 0}(t) = \id$ reduces to the identity. Therefore the two derivatives 
		\begin{align*}
			\frac{\dd}{\dd \delta} \e^{\pm \ii \Phi_j \bigl ( \int_{t}^{t+\delta} \dd \tau \, \switch(\tau)\,f_j( \tau) \bigr ) X_j} \psi = \frac{\dd}{\dd \delta} \e^{\pm \ii \delta \Phi_j \, \switch(t)\, f_j(t) X_j} \psi
		\end{align*}
		agree for all $\psi \in \domain_{\mathrm{c}}$. 
		Consequently, we can invoke Lemma~\ref{framework:lem:strong_convergence_trace_product} and because the observable $A \in \mathfrak{W}^{1,p}(\Alg)$ is in the appropriate Sobolev space we can replace the more complicated exponential by $\e^{\pm \ii \delta \Phi_j \, \switch(t)\, f_j(t) X_j}$ in 
		\begin{align*}
			\norm{\cdot}_p-\lim_{\delta \to 0} &\frac{\e^{\ii \Phi_j \bigl ( \int_{t}^{t+\delta} \dd \tau \, \switch(\tau) \, f(\tau) \bigr ) X_j} \, A \; \e^{- \ii \Phi_j \bigl ( \int_{t}^{t+\delta} \dd \tau \, \switch(\tau)\,f_j(\tau) \bigr ) X_j} - A}{\delta}
			=
			\\ 
			&= \norm{\cdot}_p-\lim_{\delta \to 0} \frac{\e^{+ \ii \delta \Phi_j \, \switch(t)\, f_j(t) X_j} \, A \; \e^{- \ii \delta \Phi_j \, \switch(t)\, f_j(t) X_j} - A}{\delta}
			\\
			&= \Phi_j \, \switch(t) \, f_j(t) \, \partial_j(A)
			, 
		\end{align*}
		and obtain an explicit expression for the derivative of the automorphism. Derivatives along arbitrary directions can be computed component-wise as the $X_j$ all commute amongst each other. 
	\end{enumerate}
\end{proof}
We point out that property (1) in the above Proposition becomes $\gamma^{\Phi,\eps}_t = \id_{\mathfrak{L}^p}$ for all $t \leqslant t_0$ along with the continuity condition $\slim_{t \to t_0} \gamma^{\Phi,\eps}_t = \id_{\mathfrak{L}^p}$ in case of a finite initial time $t_0 > -\infty$.

Let us now combine the notion of left-module structure by elements of $\affil(\Alg)$ discussed in Section~\ref{framework:commutators:hamiltonian_t_measurable_operators} with the interaction evolution.
\begin{lemma}\label{dynamics:lem:auxilary_difference_dynamics}
	Suppose Hypotheses~\ref{hypothesis:trace}–\ref{hypothesis:current} are satisfied, and assume $A \in \rr{Left}_{H}^p$. Then also the product
	\begin{align}
		H_{\Phi,\eps}(t) \; \gamma^{\Phi,\eps}_t(A) = \gamma^{\Phi,\eps}_t(H \, A) \in \mathfrak{L}^p(\Alg)
		\label{dynamics:eqn:evolution_H_rho}
	\end{align}
	remains in the same Banach space $\mathfrak{L}^p(\Alg)$ for all $t \in \R$.
\end{lemma}
\begin{proof}
	By assumption $A \in \rr{Left}_{H}^p$ implies that the domain
	\begin{align*}
		\domain(H \overset{\circ}{\cdot} A) = \bigl \{ \psi \in \Hil \; \; \vert \; \; A \psi \in \domain(H) \bigr \} \subseteq \Hil
	\end{align*}
	is $\mathcal{T}$-dense, $H \, A \in \mathfrak{L}^p(\Alg)$, and consequently  $\gamma^{\Phi,\eps}_t(H \, A) \in \mathfrak{L}^p(\Alg)$ for all $t \in \R$. Moreover, the equality \eqref{dynamics:eqn:evolution_H_rho} holds on the domain $G_{\Phi,\eps}^*(t)\bigl [ \domain(H \overset{\circ}{\cdot} A) \bigr ]$ which is still $\mathcal{T}$-dense and hence holds on $\mathfrak{L}^p(\Alg)$ as a consequence of Proposition~\ref{framework:prop:agreement_operators_T_dense_set}.
\end{proof}
%

\section{Comparison of perturbed and unperturbed dynamics} 
\label{dynamics:comparison}
To help along with proving the first main result, Theorem~\ref{main_results:thm:comparison_dynamics}, we factor out some basic facts on the dependence of $U_{\Phi,\eps}(t)$ on $\Phi$ and $\eps$ as $\pmb{\Phi} \to 0$ or $\eps \to 0$. A Duhamel argument will quantify the difference of the unperturbed dynamics $U_0(t-s) := \e^{- \ii (t-s) H}$ generated by $H$ with the perturbed dynamics $U_{\Phi,\eps}(t,s)$ generated by the time-dependent Hamiltonian $H_{\Phi,\eps}(t)$. Let us start with the general result.
\begin{theorem}[Duhamel formula]\label{dynamics:thm:Duhamel_formula}
	Let $H$ be a selfadjoint operator with domain $\domain(H)$ and propagator $U_0(t) := \e^{- \ii t H}$. Let $\R \ni t \mapsto H(t)$ be a family of selfadjoint operators with fixed domains $\domain \bigl ( H(t) \bigr ) = \domain(H)$ for all $t$ and such that all the conditions of Theorem~\ref{dynamics:thm:existence_perturbed_evolution} for the existence the unitary propagator $U(t,s)$ are met. Then, for all $t,s \in \R$ one has
	\begin{subequations}
		\begin{align}
			\bigl ( U(t,s) - U_0(t-s) \bigr ) \psi &= \ii \int_s^t \dd \tau \; U_0(t-\tau) \, \bigl ( H - H(\tau) \bigr ) \, U(\tau,s) \psi 
	 		\label{dynamics:eqn:Duhamel_formula_U0_UA}
			\\
			&= \ii \int_s^t \dd \tau \; U(t,\tau) \, \bigl ( H - H(\tau) \bigr ) \, U_0(\tau-s) \psi
			\label{dynamics:eqn:Duhamel_formula_UA_U0}
		\end{align}
	\end{subequations}
for all $\psi \in \domain(H)$.	Moreover, the equality extends to the full Hilbert space $\Hil$ when the difference operator $H - H(t)$ is bounded for all $t$. 
\end{theorem} 
\begin{proof}[Sketch]
	The Duhamel formula is a consequence of the identity
	\begin{align*}
		\frac{\dd}{\dd \tau} U_0(s - \tau) \, U(\tau,s) \psi = \ii U_0(s - \tau) \, H \, U(\tau,s) \psi - \ii U_0(s - \tau) \, H(\tau) \, U(\tau,s) \psi
	\end{align*}
	which follows from \eqref{dynamics:eqn:Schroedinger_equation_propagator_s} and is valid for all $\psi \in \domain(A)$. By integrating both sides between $s$ and $t$ one obtains
	\begin{align*}
		U_0(s - t) \, U(t,s) \psi - \psi = \ii \int_s^t \dd \tau \; U_0(s - \tau) \, \bigl ( H - H(\tau) \bigr ) \, U(\tau,s) \psi
		, 
	\end{align*}
	and the formula \eqref{dynamics:eqn:Duhamel_formula_U0_UA} simply follows by multiplying bot sides by $U_0(t-s)$. The formula \eqref{dynamics:eqn:Duhamel_formula_UA_U0} can be proved in a similar way starting from
	\begin{align*}
		\frac{\dd}{\dd \tau} U(s,\tau)\, U_0(\tau-s) \psi = \ii U(s,\tau) \, H(\tau) \, U_0(\tau-s) \psi - \ii U(s,\tau)\, H \, U_0(\tau-s) \psi
		.
	\end{align*}
	When $H - H(t)$ is a bounded operator the \eqref{dynamics:eqn:Duhamel_formula_U0_UA} and \eqref{dynamics:eqn:Duhamel_formula_UA_U0} are equalities between bounded operators that holds true on a dense domain and so can be extended by continuity to the whole Hilbert space $\Hil$.
\end{proof}
Under the Hypotheses the Duhamel formula \eqref{dynamics:eqn:Duhamel_formula_UA_U0} can be written as
\begin{align}
	\bigl ( U_{\Phi,\eps}(t,s) - U_0(t-s) \bigr ) \psi = - \ii \int_s^t \dd \tau \; U_{\Phi,\eps}(t,\tau) \, W_{\Phi,\eps}(\tau) \, U_0(\tau-s) \psi
	\label{dynamics:eqn:Duhamel_formula_W}
\end{align}
where $\psi \in \domain(H)$ and the operator $W_{\Phi,\eps}(\tau)$ is defined by \eqref{main_results:eqn:definition_additive_perturbation_W_w_kappa}. The next result describes the behavior of the perturbed propagator when the field strength vanishes.
\begin{corollary}
	\label{dynamics:cor:Phi_0_limit_perturbed_dynamics} 
	Suppose Hypotheses~\ref{hypothesis:trace}–\ref{hypothesis:current} are satisfied. Then the unitary propagator is strongly continuous at $\pmb{\Phi} = 0$, and 
	\begin{align*}
		\slim_{\Phi \to 0} U_{\Phi,\eps}(t,s) = U_0(t-s)
	\end{align*}
	holds for all $t , s \in \R$ independently of $\eps > 0$. 
\end{corollary}
\begin{proof}
	Notice that the propagator $U_{\Phi,\eps}(t,s)$ in the right-hand side of the \eqref{dynamics:eqn:Duhamel_formula_W} preserves the domain $\domain(H)$ by construction and that $W_{\Phi,\eps}(\tau)$ is well-defined on $\domain(H)$ for all $\tau$. The norm of the left-hand side of the \eqref{dynamics:eqn:Duhamel_formula_W} can be estimated as follows
	\begin{align*}
		\Bnorm{\bigl ( U_{\Phi,\eps}(t,s) - U_0(t-s) \bigr ) \psi} &\leqslant (t-s) \, \sup_{\min\{s,t\}\leqslant\tau\leqslant\max\{s,t\}} \bnorm{W_{\Phi,\eps}(\tau) \,
		U_0(\tau-s)\psi}
		\\
		&\leqslant (t-s) \sum_{\substack{\kappa \in \N_0^d \\ 1 \leqslant \abs{\kappa} \leqslant N}} \frac{1}{\abs{\kappa}!} \left ( \prod_{j = 1}^d \sabs{\Phi_j}^{\kappa_j} \right ) \, C_{\kappa} \, K_{\kappa}^{\psi}
	\end{align*}
	where the quantities 
	\begin{align*}
		C_{\kappa} :& \negmedspace= C_{\kappa}(t,s)
		= \prod_{j = 1}^d \left ( \int_{-\infty}^{\max\{s,t\}} \dd \tau \; \switch(\tau) \, \babs{f_j(\tau)} \right )^{\kappa_j}
		\\
		K_{\kappa}^{\psi} :& \negmedspace= K_{\kappa}^{\psi}(t,s) \; \sup_{\min\{s,t\} \leqslant \tau \leqslant \max\{s,t\}} \Bnorm{J_{\kappa} \, U_0(\tau-s) \psi}
	\end{align*}
	are well-defined and independent of the components of $\pmb{\Phi}$. This proves that 
	\begin{align*}
		\lim_{\Phi \to 0} \, \Bnorm{\bigl ( U_{\Phi,\eps}(t,s) - U_0(t-s) \bigr ) \psi} = 0
		,
		&&
		\forall \psi \in \domain(H)
		, 
	\end{align*}
	and the result follows due to the density of $\domain(H)$.
\end{proof}
The next technical result will be important in the following.
\begin{proposition}\label{dynamics:prop:M_limit_Phi}
	Suppose Hypotheses~\ref{hypothesis:trace}–\ref{hypothesis:current} are satisfied, and let $M_{\Phi,\eps}(t,s) \in \Alg$ be the operator defined by \eqref{dynamics:eqn:definition_M_ts}. Then for fixed $\eps > 0$ and $t , s \in \R$ we have 
	\begin{align*}
		\slim_{\Phi \to 0} M_{\Phi,\eps}(t,s) = U_0(t-s)
		. 
	\end{align*}
	Moreover, the sequence is equibounded for $\Phi \leqslant 1$. 
\end{proposition}
\begin{proof}	
	The inequality \eqref{dynamics:eqn:norm_bound_M_ts0} shows that $\snorm{C_{\Phi,\eps}(t)}$ is a continuous function of the field $\pmb{\Phi}$ such that $\lim_{\Phi \to 0} \snorm{C_{\Phi,\eps}(t)} = 0$ (observe that by construction also $h_j^{\Phi}(t) \to 0$ if $\Phi\to 0$). This means that the left-hand side of the inequality \eqref{dynamics:eqn:norm_bound_M_ts} is also a continuous function of the field $\pmb{\Phi}$ which has limit 1 when $\Phi \to 0$. Therefore, by continuity and compactness, the norm $\bnorm{M_{\Phi,\eps}(t,s)}$ can be bounded by a constant independent of $\pmb{\Phi}$ for values of the field in the sphere $\Phi\leqslant 1$. To prove the limit let us start by observing that
	\begin{align}
		\slim_{\Phi \to 0} U_{\Phi,\eps}(t,s) \; \frac{1}{H_{\Phi,\eps}(s) - \xi} 
		= U_0(t-s) \; \frac{1}{H - \xi}
		= \frac{1}{H - \xi} \; U_0(t-s)
		.
		\label{dynamics:eqn:limit_Phi_propagator_resolvent}
	\end{align}
	This follows since $U_{\Phi,\eps}(t,s)$ converges in the SOT to $U_0(t-s)$ in view of Corollary~\ref{dynamics:cor:Phi_0_limit_perturbed_dynamics} and $\bigl ( H_{\Phi,\eps}(s) - \xi \bigr )^{-1}$ converges in the SOT to $(H - \xi)^{-1}$ in view of the fact $\slim_{\Phi \to 0} G_{\Phi,\eps}(t) = \id$. Exchanging the order of the strong limit and the product is possible since the two sequences are equibounded. The second equality is just a consequence of the commutativity between the free propagator $U_0(t-s)$ and $H$. By using the decomposition $H_{\Phi,\eps}(t) = H + W_{\Phi,\eps}(t)$ one can split the difference 
	\begin{align}
		M_{\Phi,\eps}(t,s) - U_0(t-s) &= (H - \xi) \, \Delta U_{\Phi,\eps,\xi}(t,s) + W_{\Phi,\eps}(t) \, \Delta U_{\Phi,\eps,\xi}(t,s) 
		\, + \notag \\
		&\qquad 
		+ W_{\Phi,\eps}(t) \; U_0(t-s) \; \frac{1}{H - \xi}
		\label{dynamics:eqn:difference_M_Phi_eps_U_0}
	\end{align}
	into three terms where 
	\begin{align*}
		\Delta U_{\Phi,\eps,\xi}(t,s) &= U_{\Phi,\eps}(t,s) \, \frac{1}{H_{\Phi,\eps}(s) - \xi}-U_0(t-s)\, \frac{1}{H - \xi} 
		. 
	\end{align*}
	Observe that by virtue of the presence of the resolvents each of these three terms is globally defined on $\Hil$, hence bounded. Let $\psi \in \Hil$ be arbitrary. Note that
	\begin{align*}
		\lim_{\Phi \to 0} (H - \xi) \, \Delta U_{\Phi,\eps,\xi}(t,s) \psi = (H - \xi) \psi^{\Phi}(t,s) = 0
	\end{align*}
	since $H$ is a closed operator and
	\begin{align*}
		\psi^{\Phi}(t,s) := \left ( U_{\Phi,\eps}(t,s) \; \frac{1}{H_{\Phi,\eps}(s) - \xi}-U_0(t-s) \; \frac{1}{H - \xi} \right ) \psi
	\end{align*}
	converges to $0$ due to \eqref{dynamics:eqn:limit_Phi_propagator_resolvent}. On the other hand
	\begin{align*}
		\lim_{\Phi \to 0} W_{\Phi,\eps}(t) \; U_0(t-s) \; \frac{1}{H - \xi} \psi = W_{\Phi,\eps}(t) \widetilde{\psi}(t-s) = 0 
	\end{align*}
	since $W_{\Phi,\eps}(t)$ depends linearly on the components of $\pmb{\Phi}$ and
	\begin{align*}
		\widetilde{\psi}(t-s) := U_0(t-s) \, \frac{1}{H - \xi} \psi
	\end{align*}
	is in the domain of $W_{\Phi,\eps}(t)$. Therefore, combining these two arguments, we can conclude $\lim_{\Phi \to 0} W_{\Phi,\eps}(t) \, \Delta U_{\Phi,\eps,\xi}(t,s) \psi = 0$. Hence, every one of the terms on the right-hand side of \eqref{dynamics:eqn:difference_M_Phi_eps_U_0} vanishes, proving the claim. 
\end{proof}	
We conclude this section with a result concerning the continuity of the perturbed dynamics with respect to the perturbation parameter $\Phi = \sabs{\pmb{\Phi}}$.
\begin{proposition}
	\label{dynamics:prop:continuity_perturbed_dynamics_field_strength}
	Suppose Hypotheses~\ref{hypothesis:trace}–\ref{hypothesis:current} hold true. Then $\alpha^{\Phi,\eps}_{(t,s)}$ is strongly continuous in $\Phi$ at $0$, \ie 
	\begin{align*}
		\lim_{\Phi \to 0} \; \Bnorm{\alpha^{\Phi,\eps}_{(t,s)}(A) - \alpha^0_{t-s}(A)}_p = 0 
		&&
		\forall A \in \mathfrak{L}^p(\Alg)
	\end{align*}
	holds independently of $\eps > 0$ and $t,s \in \R$.
\end{proposition}
\begin{proof}
	Corollary~\ref{dynamics:cor:Phi_0_limit_perturbed_dynamics} ensures the convergence of $U_{\Phi,\eps}(t,s)$ to $U_0(t-s)$ in the SOT, and combining this with Proposition~\ref{framework:prop:extension_isometry_Lp}~(4) yields the claim. 
\end{proof}
%
%
%
\chapter{The Kubo Formula and its Adiabatic Limit} 
\label{Kubo_formula}
With the lion's share of the work now behind us, we can proceed to prove the main results of Chapter~\ref{main_results}, Theorems~\ref{main_results:thm:comparison_dynamics}, \ref{main_results:thm:Kubo_formula}, \ref{main_results:thm:adiabatic_limit_Kubo_formula} and \ref{main_results:thm:Kubo_Streda_formula}. The six Hypotheses we have imposed in Chapter~\ref{main_results} are meant to strike a balance so as to emphasize the structure of the proofs rather than technicalities. We reckon many of our results can be shown to hold in more general circumstances, though. Roughly speaking, there are three distinct steps: 
\begin{enumerate}[leftmargin=*,label=\textbf{Step \arabic*:}]
	\item Obtain an explicit expression for  
	\begin{align}
		\rho_{\mathrm{full}}(t) = \rho_{\mathrm{int}}(t) + \pmb{\Phi} \cdot \mathbf{K}^{\Phi,\eps}(t)
		, 
		\label{Kubo_formula:eqn:expansion_rho}
	\end{align}
	and prove that $\mathbf{K}^{\Phi,\eps}(t)$ is given by \eqref{main_results:eqn:components_K_Phi_eps} and depends continuously on $\pmb{\Phi}$. 
	\item Compute the conductivity tensor by deriving the macroscopic net current $\mathscr{J}^{\Phi,\eps}[J,\rho](t)$ (\cf equation~\eqref{main_results:eqn:net_density_current_1}) for finite $\eps > 0$ with respect to $\Phi_k$ at $\pmb{\Phi} = 0$. 
	\item Take the adiabatic limit $\eps \to 0$ of the conductivity tensor, and find a simplified expression for pure states associated to finite spectral ranges of $H$. 
\end{enumerate}
These steps will fundamentally remain the same even if we change the hypotheses or the setting, such as the precise form of the perturbation.

\section{Comparing the evolutions of equilibrium states} 
\label{Kubo_formula:comparing_evolutions}
The first main step is to expand the fully evolved state in the perturbation parameter $\pmb{\Phi}$ around $0$ (Theorem~\ref{main_results:thm:comparison_dynamics}), \ie we compare the fully evolved state 
\begin{align}
	\rho_{\mathrm{full}}(t) = \lim_{s \to -\infty} \alpha^{\Phi,\eps}_{(t,s)}(\rho)
\end{align}
with the \emph{interaction evolution}
\begin{align*}
	\rho_{\mathrm{int}}(t) = \gamma_t^{\Phi,\eps}(\rho) = G_{\Phi,\eps}(t) \, \rho \, G_{\Phi,\eps}(t)^* 
\end{align*}
that has been “dragged along”. Therefore, the first order of business is to make sure $\rho_{\mathrm{full}}(t)$ and $\rho_{\mathrm{int}}(t)$ exist as elements in $\mathfrak{L}^p(\Alg)$. 
While this is straightforward for $\rho_{\mathrm{int}}(t) = \gamma^{\Phi,\eps}_t(\rho)$, the interaction evolution $\gamma^{\Phi,\eps}_t$ is a strongly continuous isometry on $\mathfrak{L}^p(\Alg)$ (\cf Proposition~\ref{framework:prop:extension_isometry_Lp}) and $\rho \in \mathfrak{L}^p(\Alg)$ by Hypothesis~\ref{hypothesis:state}, \emph{a priori} it is not even clear whether the limit through which $\rho_{\mathrm{full}}(t)$ is defined actually exists. For states which are initially at equilibrium and in addition satisfy further technical assumptions, we will show that indeed $\rho_{\mathrm{full}}(t)$ exists. First, let us collect some facts on the states we are interested in.

\subsection{Initial equilibrium states} 
\label{Kubo_formula:comparing_evolutions:equilibrium_states}
The conditions on $\rho$ stipulated in Hypothesis~\ref{hypothesis:state} arise naturally from the present context. Let us start with the equilibrium conditions. A sufficient condition to construct an equilibrium state is to define $\rho = f(H)$ where $f$ is any non-negative function in $L^{\infty}(\R)$. A typical ansatz in the literature for the initial equilibrium states are (see \eg \cite[Assumption~5.1]{Bouclet_Germinet_Klein_Schenker:linear_response_theory_magnetic_Schroedinger_operators_disorder:2005})
\begin{enumerate}[leftmargin=*,label=(\roman*)]
	\item $\rho := f(H)$ with $f \in S(\R)$ a Schwartz function, or 
	\item $\rho := f(H) \, b(H)$ with $f \in S(\R)$ a Schwartz function and $b \in M(\R)$ a measurable (aka borelian) function.
\end{enumerate}
For such states the equilibrium condition $\alpha_t^0(\rho) = \rho$ and the positivity condition $\rho \in \Alg^+$ are automatically satisfied. However, in Hypothesis~\ref{hypothesis:state} we do \emph{not} assume that $\rho$ is necessarily a function of $H$. Instead, the equilibrium condition merely forces $\rho \in \ker(\mathscr{L}_H^{(1)}) \cap \ker(\mathscr{L}_H^{(p)})$ to lie in the kernel of the Liouvillians. 

Now on to the regularity of $\rho$: The integrability condition Hypothesis~\ref{hypothesis:state}~(i) instead depends on a combination of various ingredients. In fact, from equation~\eqref{framework:eqn:definition_trace_via_functional_calculus} we infer 
\begin{align}
	\mathcal{T}(\rho^p) = \int_0^{+\infty} f(\lambda)^p \; \dd \mu^H_{\mathcal{T}}(\lambda)
	\label{Kubo_formula:eqn:p_integrability_condition_rho}
\end{align}
where $\dd \mu^H_{\mathcal{T}}$ is the spectral measure of $H$ induced by the trace $\mathcal{T}$ according to 
\begin{align*}
	\mu^H_{\mathcal{T}}(\mathtt{B}) := \mathcal{T} \bigl ( \chi_{\mathtt{B}}(H) \bigr )
	,
	&&
\mathtt{B} \mbox{ Borel set}
	. 
\end{align*}
We notice that if $\mathcal{T}$ is not finite then the measure $\mu^H_{\mathcal{T}}$ is not normalized. In case $\dd \mu^H_{\mathcal{T}}(\lambda) = h_{\mathcal{T}}(\lambda) \, \dd \lambda$ is absolutely continuous with respect to the Lebesgue measure $\dd \lambda$, the integrability properties of $\rho$ can be checked in terms of the $L^1$ property of $f^p \, h_{\mathcal{T}}$. A similar discussion holds for the $H$-regularity of $H \, \rho \in \mathfrak{L}^p(\Alg)$ (Hypothesis~\ref{hypothesis:state}~(ii)) which requires the finiteness of 
\begin{align}
	\mathcal{T} \bigl ( \sabs{H \, \rho}^p \bigr ) = \int_0^{+\infty} f(\lambda)^p \, \sabs{\lambda}^p \, \dd \mu^H_{\mathcal{T}}(\lambda)
	, 
	\label{Kubo_formula:eqn:p_integrability_condition_H_rho}
\end{align}
and, at least with respect to the continuous part of the measure, this condition is related to a sufficiently fast decay of the function $f$. 

Differentiability properties like $\rho \in \mathfrak{W}^{1,p}(\Alg)$ or $H \, \rho \in \mathfrak{W}^{1,p}(\Alg)$ depend on the particular nature of the derivation and, in turn, on the particular choice of generators $\{ X_1 , \ldots , X_d \}$. 

States satisfying Hypothesis~\ref{hypothesis:state} have the following properties: 
\begin{lemma}\label{Kubo_formula:lem:Lp_regularity_states}
	Suppose Hypotheses~\ref{hypothesis:trace}–\ref{hypothesis:state} are satisfied, and let $r = 1 , p$. Then the following holds true: 
	\begin{enumerate}
		\item $J_k \in \affil(\Alg)$, $J_k \, (H - \xi)^{-1} \in \Alg$ and $J_k \, \rho \in \mathfrak{L}^r(\Alg)$ for all $k = 1 , \ldots , d$, where $\xi \in \res(H)$ lies in the resolvent set and $J_k$ is the closure of the commutator $- \ii [ X_k , H ]$ initially defined on the joint core $\domain_{\mathrm{c}}(H)$. 
		\item $H \, \eta^{X_k}_t(\rho) \in \mathfrak{L}^r(\Alg)$ holds for all $t \in \R$, and the limit $\lim_{t \to 0} H \, \eta^{X_k}_t(\rho) = H \, \rho$ exists in $\mathfrak{L}^r(\Alg)$.
		\item We have $H \, \partial_{X_k}(\rho) = J_k \, \rho + \partial_{X_k}(H \, \rho)$ as elements in $\mathfrak{L}^r(\Alg)$. 
	\end{enumerate}
\end{lemma}
\begin{remark}\label{Kubo_formula:remark:Lp_regularity_dropped_nonnegativity}
	A closer inspection of the proof below reveals that we do not use the non-negativity of $\rho$, and the Lemma holds if we drop the assumption $\rho \geq 0$. 
\end{remark}
\begin{proof}
	\begin{enumerate}
		\item The idea is to write $J_k$ as the limit of the difference quotient 
		\begin{align*}
			Y_k(t) := - \frac{\e^{+ \ii t X_k} \, H \, \e^{- \ii t X_k} - H}{t}
		\end{align*}
		which is initially defined for all $t \in \R \setminus \{ 0 \}$. To imbue the above expression with meaning, we rely on Hypothesis~\ref{hypothesis:current} to view it as the difference of two elements in $\affil(\Alg)$ that is initially defined on the dense domain $\domain_{\mathrm{c}}(H)$. 
		
		To push these properties forward to the limit $t \to 0$ and show that the product $Y_k(t) \, (H - \xi)^{-1} \in \Alg$ lies inside the von Neumann algebra $\Alg$ for all $t \in \R$, we use a version of the Baker-Campbell-Hausdorff formula (see \eg \cite[Section~II.11.B]{schroeck-96}). Here, Hypotheses~\ref{hypothesis:current}~(iii) is crucial because it ensures that the expansion of $Y_k(t)$ in terms of iterated commutators exists on $\domain_{\mathrm{c}}(H)$ and terminates after \emph{finitely} many terms. For our purposes we do not need explicit expressions for the higher-order commutators, it suffices to know that we can write the difference quotient as 
		\begin{align}
			Y_k(t) = - \ii [X_k , H] + \order(t)
			,  
			\label{Kubo_formula:eqn:extension_Y_k}
		\end{align}
		initially defined on the joint core $\domain_{\mathrm{c}}(H)$ where the $H$-bounded term $\order(t)$  involves iterated commutators with $X_k$ where $H$ appears just once. 
		The above expression for $Y_k(t)$ 
		not only shows that the closure of $Y_k(t)$, denoted by the same letter, has a $t$-\emph{in}dependent domain, it also
		allows us to extend $Y_k(t)$ to $t = 0$. 
		Moreover, $Y_k(t) \, (H - \xi)^{-1}$ converges in norm to $J_k \, (H - \xi)^{-1} \in \Alg$, and therefore also $J_k \in \affil(\Alg)$ holds as well (Remark~\ref{framework:remark:algebraic_operations}). 
		Finally, the product $J_k \, \rho$  is also in $\mathfrak{L}^r(\Alg)$: thanks to Lemma~\ref{framework:lem:extension_algebra_unbounded_operators}~(2) and Hypothesis~\ref{hypothesis:state}~(ii) we know that the product
		\begin{align}
			 \left ( J_k \, \frac{1}{H - \xi} \right ) \, (H - \xi) \, \rho = J_k \, \rho
			\label{Kubo_formula:eqn:J_k_rho_insert_resolvent}
		\end{align}
		of $J_k \, (H - \xi)^{-1} \in \Alg$ and $(H - \xi) \, \rho \in \mathfrak{L}^r(\Alg)$ in fact lies in $\mathfrak{L}^r(\Alg)$.
		\item To prove $H \, \eta^{X_k}_t(\rho) \in \mathfrak{L}^r(\Alg)$, we eventually would like to invoke Lemma~\ref{framework:lem:extension_algebra_unbounded_operators}~(2). The difficulty here is to justify the reordering of the product. 
		First of all, adapting the proof of  Lemma~\ref{dynamics:lem:auxilary_difference_dynamics}, we conclude 
		\begin{align}
			\eta^{X_k}_t(H \, \rho) = \e^{+ \ii t X_k} \, H \, \e^{- \ii t X_k} \; \eta^{X_k}_t(\rho) 
			\in \mathfrak{L}^r(\Alg) 
			. 
			\label{dynamics:eqn:X_k_transported_H_rho}
		\end{align}
		Moreover, from \eqref{dynamics:eqn:X_k_transported_H_rho} 
			we can deduce that the following equality 
		\begin{align*}
			\eta^{X_k}_t(\rho) = \frac{1}{\e^{+ \ii t X_k} \, H \, \e^{- \ii t X_k} - \xi} \, \bigl ( \e^{+ \ii t X_k} \, H \, \e^{- \ii t X_k} - \xi \bigr ) \, \eta^{X_k}_t(\rho) 
		\end{align*}
			on the domain $\domain \bigl ( H \overset{\circ}{\cdot} \rho \bigr )$
	which is 		$\mathcal{T}$-dense
			thanks to $\rho \in \rr{D}^{00}_{H,1} \cap \rr{D}^{00}_{H,p}$. Therefore, the operators on the left- and right-hand side are measurable, and invoking Lemma~\ref{framework:lem:extension_algebra_unbounded_operators} once again, this identity extends from $\rr{M}(\Alg)$ to $\mathfrak{L}^r(\Alg)$. 
		Building on top of some of the arguments in the proof of (1) it is straightforward to verify 
		\begin{align*}
			H \, \frac{1}{\e^{+ \ii t X_k} \, H \, \e^{- \ii t X_k} - \xi} 
			&= \e^{+ \ii t X_k} \left ( H \, \frac{1}{H - \xi} + t \, Y_k(-t) \, \frac{1}{H - \xi} \right ) \e^{- \ii t X_k}
			\in \Alg
			\label{Kubo_formula:eqn:product_H_H_Phi_eps_resolvent}
		\end{align*}
		holds on $\Alg$. 
		
		Now we put all the pieces together: the domain 
		$\domain(H \overset{\circ}{\cdot} \rho)$ is $\mathcal{T}$-dense and the compatibility of the $\e^{- \ii t X_k}$ with the algebraic structures  (Hypotheses~\ref{hypothesis:generators}~(i) and (ii)) means the domain of the transported operator $\domain \bigl ( \eta^{X_k}_t(H \, \rho) \bigr )$ is $\mathcal{T}$-dense as well. Therefore, we know
		\begin{align*}
			H \, \eta^{X_k}_t(\rho) = \left ( H \, \frac{1}{\e^{+ \ii t X_k} \, H \, \e^{- \ii t X_k} - \xi} \right ) \; \eta^{X_k}_t \bigl ( (H - \xi) \, \rho \bigr )
		\end{align*}
		to be true on this $\mathcal{T}$-dense domain, and the equality holds in $\rr{M}(\Alg)$ as well. Appealing to Lemma~\ref{framework:lem:extension_algebra_unbounded_operators}~(2) one last time leads us to conclude that the last equality, the claim, holds on $\mathfrak{L}^r(\Alg)$ because we can use \eqref{dynamics:eqn:X_k_transported_H_rho} to pull $H - \xi$ out of the argument of $\eta^{X_k}_t$, and consequently, $H \, \eta^{X_k}_t(\rho) $ can be written as the product of an operator in $\Alg$ and one in $\mathfrak{L}^r(\Alg)$. 
		\medskip
		
		\noindent
		To show that $H \, \eta^{X_k}_t(\rho)$ is continuous in $t = 0$ (in fact for all $t \in \R$), we can equivalently prove 
		\begin{align*}
			\bnorm{H \, \eta^{X_k}_t(\rho) - H \, \rho}_r &= \Bnorm{\eta^{X_k}_{-t} \bigl ( H \, \eta^{X_k}_t(\rho) \bigr ) - \eta^{X_k}_{-t}(H \, \rho)}_r 
			\\
			&= \Bnorm{\bigl ( \e^{- \ii t X_k} \, H \, \e^{+ \ii t X_k} \bigr ) \, \rho - \eta^{X_k}_{-t}(H \, \rho)}_r 
		\end{align*}
		because we can insert $\eta^{X_k}_{-t}$ free of charge without changing the $\mathfrak{L}^r(\Alg)$-norm and apply equation~\eqref{dynamics:eqn:X_k_transported_H_rho}. Evidently, $\eta^{X_k}_{-t}(H \, \rho)$ converges to $H \, \rho$ due to the strong continuity of $\eta^{X_k}_t$ on $\mathfrak{L}^r(\Alg)$. 
		
		Just like in the proof of (1), we can use the Baker-Campbell-Hausdorff formula to write 
		\begin{align*}
			\e^{- \ii t X_k} \, H \, \e^{+ \ii t X_k} &= H - t \, Y_k(t) 
		\end{align*}
		where $Y_k(t)$ is the finite sum of iterated commutators from equation~\eqref{Kubo_formula:eqn:extension_Y_k}. In (1) we have already proven that $Y_k(t) \in \affil(\Alg)$ is affiliated to $\Alg$ and that $Y_k(t) \, (H - \xi)^{-1} \in \Alg$ in fact lies in the von Neumann algebra. Additionally, Hypothesis~\ref{hypothesis:current}~(iv) tells us that $Y_k(t)$ is infinitesimally $H$-bounded which means 
		\begin{align*}
			\bigl ( \e^{- \ii t X_k} \, H \, \e^{+ \ii t X_k} \bigr ) \, \rho &= H \, \rho - t \, Y_k(t) \, \rho 
			\\
			&
			= H \, \rho - \left ( t \, Y_k(t) \, \frac{1}{H - \xi} \right ) \, \bigl ( (H - \xi) \, \rho \bigr ) 
		\end{align*}
		can be initially defined on $\domain \bigl ( H \overset{\circ}{\cdot} \rho \bigr )$. Consequently, because this initial domain is $\mathcal{T}$-dense, the first equation holds also in the measurable sense on $\rr{M}(\Alg)$, and thanks to Lemma~\ref{framework:lem:extension_algebra_unbounded_operators}~(2) we can insert the resolvent and deduce that the equality also holds on $\mathfrak{L}^r(\Alg)$. 
		
		What is more, $t \, Y_k(t) \, (H - \xi)^{-1} \rightarrow 0$ goes to $0$ in norm, and in view of  Theorem~\ref{framework:thm:measurability_products_in_Lp} we conclude 
		\begin{align*}
			\lim_{t \to 0} \Bigl ( \bigl ( \e^{- \ii t X_k} \, H \, \e^{+ \ii t X_k} \bigr ) \, \rho \Bigr ) &= H \, \rho 
			\in \mathfrak{L}^r(\Alg)
			, 
		\end{align*}
		and therefore also $H \, \eta^{X_k}_t(\rho) \rightarrow H \, \rho$. 
		\item The idea of the last statement mimics closely the proof of the continuity from (2), and we will only sketch it: first of all, we need to put Hypothesis~\ref{hypothesis:state}~(iii), $H \, \partial_{X_k}(\rho) \in \mathfrak{L}^r(\Alg)$, to use, when we insert a resolvent 
		\begin{align*}
			\left ( \frac{1}{H - \xi} \right ) \, \bigl ( H \, \partial_{X_k}(\rho) \bigr ) &= \left ( \frac{H}{H - \xi} \right ) \, \partial_{X_k}(\rho)
			\\
			&= \left ( \frac{H}{H - \xi} \right ) \, \left ( \lim_{t \to 0} \frac{\eta^{X_k}_t(\rho) - \rho}{t} \right )
		\end{align*}
		via Lemma~\ref{framework:lem:extension_algebra_unbounded_operators}~(2), and write the derivation as a difference quotient. This difference quotient converges in $\mathfrak{L}^r(\Alg)$ by Hypothesis~\ref{hypothesis:state}~(i), and thus, Theorem~\ref{framework:thm:measurability_products_in_Lp} allows us to pull out the limit, 
		\begin{align*}
			\ldots &= \lim_{t \to 0} \left ( \frac{H}{H - \xi} \right ) \, \left ( \frac{\eta^{X_k}_t(\rho) - \rho}{t} \right )
			, 
		\end{align*}
		and applying Lemma~\ref{framework:lem:extension_algebra_unbounded_operators}~(2) in reverse yields 
		\begin{align*}
			\ldots &= \lim_{t \to 0} \left ( \frac{1}{H - \xi} \right ) \, \left ( H \; \frac{\eta^{X_k}_t(\rho) - \rho}{t} \right )
			. 
		\end{align*}
		Let us focus on the difference quotient: our Hypotheses on $\rho$ make sure that $\domain \bigl ( H \overset{\circ}{\cdot} \eta^{X_k}_t(\rho) \bigr )$ and $\domain(H \overset{\circ}{\cdot} \rho)$ are both $\mathcal{T}$-dense. And since intersections of $\mathcal{T}$-dense domains are again $\mathcal{T}$-dense, we can define 
		\begin{align*}
			H \; \frac{\eta^{X_k}_t(\rho) - \rho}{t} &= \frac{H \, \eta^{X_k}_t(\rho) - \eta^{X_k}_t(H \, \rho)}{t} + \frac{\eta^{X_k}_t(H \, \rho) - H \, \rho}{t} 
		\end{align*}
		on the $\mathcal{T}$-dense domain $\domain \bigl ( H \overset{\circ}{\cdot} \eta^{X_k}_t(\rho) \bigr ) \cap \domain(H \overset{\circ}{\cdot} \rho)$. Therefore, the above holds as an equality in $\rr{M}(\Alg)$. By adding and subtracting $\eta^{X_k}_t(H \, \rho)$, we can express this via $Y_k(t)$ from (1) and a second difference quotient, 
		\begin{align*}
			\ldots &= Y_k(t) \, \eta^{X_k}_t(\rho) + \frac{\eta^{X_k}_t(H \, \rho) - H \, \rho}{t}
			, 
		\end{align*}
		where we regard this as an equality on $\rr{M}(\Alg)$. With the help of Lemma~\ref{framework:lem:extension_algebra_unbounded_operators}~(2) we once more insert a resolvent, and exploit that we have already shown $Y_k(t) \, (H - \xi)^{-1} \in \Alg$ in (1), 
		\begin{align*}
			\ldots &= \left ( Y_k(t) \, \frac{1}{H - \xi} \right ) \, \bigl ( (H - \xi) \, \eta^{X_k}_t(\rho) \bigr ) + \frac{\eta^{X_k}_t(H \, \rho) - H \, \rho}{t}
			, 
		\end{align*}
		concluding that this equality holds in $\mathfrak{L}^r(\Alg)$ in the process. The first factor $Y_k(t) \, (H - \xi)^{-1}$ converges in $\Alg$ to $Y_k(0) \, (H - \xi)^{-1} = J_k \, (H - \xi)^{-1}$, and as we have shown $t \mapsto H \, \eta^{X_k}_t(\rho)$ to be continuous at $t = 0$, we conclude that the first term converges in $\mathfrak{L}^r(\Alg)$. Moreover, we can apply Lemma~\ref{framework:lem:extension_algebra_unbounded_operators}~(2) to take out the regularizing resolvent as $H \, \rho \in \mathfrak{L}^r(\Alg)$ and $J_k \in \affil(\Alg)$. By Hypothesis~\ref{hypothesis:state}~(ii) we also know the second limit exists in $\mathfrak{L}^r(\Alg)$. In conclusion, we have shown 
		\begin{align*}
			H \, \partial_{X_k}(\rho) &= \lim_{t \to 0} \left ( \left ( Y_k(t) \, \frac{1}{H - \xi} \right ) \, \bigl ( (H - \xi) \, \eta^{X_k}_t(\rho) \bigr ) + \frac{\eta^{X_k}_t(H \, \rho) - H \, \rho}{t} \right ) 
			\\
			&= J_k \, \rho + \partial_{X_k}(H \, \rho) 
			\in \mathfrak{L}^r(\Alg)
			. 
		\end{align*}
		This finishes the proof. 
	\end{enumerate}
\end{proof}
%

\subsection{Existence of $\rho_{\mathrm{full}}(t)$ and its expansion in $\pmb{\Phi}$} 
\label{Kubo_formula:comparing_evolutions:difference_evolved_states}
All of the conditions on $\rho$ — save for $H \, \partial_{X_k}(\rho) \in \mathfrak{L}^r(\Alg)$, $r = 1 , p$ — play a role when proving the existence of $\rho_{\mathrm{full}}(t)$. We can already see why the equilibrium condition $\alpha^0_t(\rho) = \rho$ is important: for times in the distant past $H_{\Phi,\eps}(t) \approx H$ and consequently, $\alpha^{\Phi,\eps}_{(t,s)}(\rho) \approx \alpha^0_{t-s}(\rho) = \rho$ hold in some sense. If $\rho$ were not an equilibrium state, then 
\begin{align*}
	\lim_{s \to -\infty} \alpha^0_s(\rho) &= \lim_{s \to -\infty} \bigl ( \e^{- \ii s H} \rho \, \e^{+ \ii s H} \bigr ) 
\end{align*}
cannot exist because the phase factors do not cancel one another. Therefore the equilibrium condition is necessary. The technical conditions arise because we want to prove the existence of $\alpha^{\Phi,\eps}_{(t,s)}$ as an isometric automorphism on $\mathfrak{L}^r(\Alg)$ where $r = 1 , p$. 
\begin{proof}[Theorem~\ref{main_results:thm:comparison_dynamics}]
	\begin{enumerate}
		\item Let $r$ stand for either $1$ or the regularity degree $p$ of $\rho$. Due to Hypothesis~\ref{hypothesis:state}~(ii) we know \emph{a priori} that $\rho \in \rr{D}^{00}_{H,r}$ holds. In conjunction with Lemma~\ref{dynamics:lem:auxilary_difference_dynamics} this ensures that $H_{\Phi,\eps}(s) \, \rho_{\mathrm{int}}(s) \in \mathfrak{L}^r(\Alg)$, and hence we can invoke Proposition~\ref{dynamics:prop:properties_perturbed_dynamics}~(2) to differentiate $\alpha^{\Phi,\eps}_{(t,s)} \bigl (\rho_{\mathrm{int}}(s) \bigr )$. The chain rule leads to
		\begin{align*}
			\frac{\dd}{\dd s} \alpha^{\Phi,\eps}_{(t,s)} \bigl ( \rho_{\mathrm{int}}(s) \bigr ) = 
			\alpha^{\Phi,\eps}_{(t,s)} \left ( \ii \bigl [ H_{\Phi,\eps}(s) \, , \, \rho_{\mathrm{int}}(s)\bigr ]_{\ddagger} \right ) + 
			\alpha^{\Phi,\eps}_{(t,s)} \Bigl ( \tfrac{\dd}{\dd s}\gamma^{\Phi,\eps}_s(\rho) \Bigr ) 
			. 
		\end{align*}
		The generalized commutator is well-defined, and from Lemma~\ref{dynamics:eqn:evolution_H_rho} and the equilibrium condition \eqref{main_results:eqn:selfadjointness_product_H_rho} we deduce that 
		\begin{align}
			\bigl [ H_{\Phi,\eps}(s) \, , \, \rho_{\mathrm{int}}(s) \bigr ]_{\ddagger} = \gamma^{\Phi,\eps}_s \bigl ( H \, \rho - (H \, \rho)^* \bigr ) = 0
			\label{Kubo_formula:eqn:evaluation_generalized_commutator_H_per_rho_int}
		\end{align}
		vanishes. Hypothesis~\ref{hypothesis:state}~(i), $\rho \in \Alg^+ \cap \mathfrak{W}^{1,1}(\Alg) \cap \mathfrak{W}^{1,p}(\Alg)$, allows us to apply Proposition~\ref{dynamics:prop:interaction_dynamics}~(3), and we get
		\begin{align*}
			\frac{\dd}{\dd s} \alpha^{\Phi,\eps}_{(t,s)} \bigl ( \rho_{\mathrm{int}}(s) \bigr ) = 
			\alpha^{\Phi,\eps}_{(t,s)} \left ( \switch(s) \, \sum_{k = 1}^d \Phi_k \, f_k(s) \, \partial_{X_k} \bigl ( \rho_{\mathrm{int}}(s) \bigr ) \right ) 
			. 
		\end{align*}
		We can integrate up this expression from $s$ to $t$, and obtain 
		\begin{align}
			\rho_{\mathrm{int}}(t) - \alpha^{\Phi,\eps}_{(t,s)} \bigl ( \rho_{\mathrm{int}}(s) \bigr )  
			&= \sum_{k = 1}^d \Phi_k \; \int_s^t \dd \tau \; \switch(\tau) \, f_k(\tau) \;
			\alpha^{\Phi,\eps}_{(t,\tau)} \bigl ( \partial_{X_k} \bigl ( \rho_{\mathrm{int}}(\tau) \bigr ) \bigr ) 
			.
			\label{Kubo_formula:eqn:difference_rho_full_rho_int_Bochner_integral}
		\end{align}
		Therefore, to recover equation \eqref{main_results:eqn:comparison_rho_full_rho_int} we merely need to take the limit $s \to -\infty$ of the last equation. More specifically, we need to prove the existence of the Bochner integral on the right-hand side. Estimating the $\mathfrak{L}^r(\Alg)$ norm by taking the limits $s \to -\infty$ and $t \to +\infty$ yields 
		\begin{align*}
			&\abs{\int_s^t \dd \tau \; \switch(\tau) \, f_k(\tau) \;
			\alpha^{\Phi,\eps}_{(t,\tau)} \bigl ( \partial_{X_k} \bigl ( \rho_{\mathrm{int}}(\tau) \bigr ) \bigr )} 
			\leqslant \\
			&\qquad \qquad
			\leqslant \bnorm{\partial_{X_k}(\rho)}_r \; \int_{-\infty}^{+\infty} \dd \tau \; \switch(\tau) \, \babs{f_k(\tau)}
			< \infty 
			. 
		\end{align*}
		There are several ingredients here: to start off, not just $\rho$ but also all its derivatives are in $\mathfrak{L}^r(\Alg)$ by Hypothesis~\ref{hypothesis:state}~(i). Secondly, the derivatives with respect to the $X_k$ commute with $\gamma^{\Phi,\eps}_t$ in $\mathfrak{L}^r(\Alg)$, $\partial_{X_k} \bigl ( \rho_{\mathrm{int}}(t) \bigr ) = \gamma^{\Phi,\eps}_t \bigl ( \partial_{X_k}(\rho) \bigr )$, because the generators $X_1 , \ldots , X_d$ commute amongst one another. Lastly, both $\alpha^{\Phi,\eps}_{(t,\tau)}$ and $\gamma^{\Phi,\eps}_{\tau}$ define isometries on $\mathfrak{L}^r(\Alg)$. Therefore the integrability conditions imposed on the $f_k$ (Hypothesis~\ref{hypothesis:perturbation}) imply the Bochner integral~\eqref{Kubo_formula:eqn:difference_rho_full_rho_int_Bochner_integral} exists \emph{uniformly} in $s , t \in \R$. 
		
		The last step is to show that the left-hand side in fact approaches $\rho_{\mathrm{int}}(t) - \rho_{\mathrm{full}}(t)$ as $s \to -\infty$, \ie 
		\begin{align*}
			\lim_{s \to -\infty} \alpha^{\Phi,\eps}_{(t,s)} \bigl ( \rho_{\mathrm{int}}(s) \bigr ) 
			 = \rho_{\mathrm{full}}(t) 
			 \in \mathfrak{L}^r(\Alg) 
			 . 
		\end{align*}
		But this again follows from $\alpha^{\Phi,\eps}_{(t,s)} \in \mathrm{Iso} \bigl ( \mathfrak{L}^r(\Alg) \bigr )$ and $\slim_{s \to -\infty} \gamma^{\Phi,\eps}_s = \id_{\mathfrak{L}^r}$ (Proposition~\ref{dynamics:prop:interaction_dynamics}~(1)), as these two facts immediately imply 
		\begin{align*}
			\Bnorm{\alpha^{\Phi,\eps}_{(t,s)} \bigl ( \rho_{\mathrm{int}}(s) \bigr ) - \alpha^{\Phi,\eps}_{(t,s)}(\rho)}_r = \bnorm{\gamma^{\Phi,\eps}_s(\rho) - \rho}_r 
			\xrightarrow{s \to -\infty} 0 
			. 
		\end{align*}
		Thus, not only have we proven the existence of $\rho_{\mathrm{full}}(t)$, we have also established equation~\eqref{main_results:eqn:comparison_rho_full_rho_int}. 
		\item Again, let $r = 1 , p$. Taking the limit $t \to -\infty$ of equation~\eqref{main_results:eqn:comparison_rho_full_rho_int} in $\mathfrak{L}^r(\Alg)$ in conjunction with $\lim_{t \to -\infty} \rho_{\mathrm{int}}(t) = \rho$ yields that also $\rho_{\mathrm{full}}(t)$ satisfies the initial value condition. 
		
		Next, we will show that $\rho_{\mathrm{full}}$ is \emph{a} solution of the Cauchy problem \eqref{main_results:eqn:full_dynamical_problem}; as usual, proving uniqueness is the last step. Let us first \emph{formally} differentiate equation~\eqref{main_results:eqn:comparison_rho_full_rho_int}, 
		\begin{align*}
			\frac{\dd \rho_{\mathrm{full}}}{\dd t}(t) &= \frac{\dd \rho_{\mathrm{int}}}{\dd t}(t) - \frac{\dd }{\dd t} \left ( \sum_{k = 1}^d \Phi_k\; \int_{-\infty}^t \dd \tau \; \switch(\tau) \; f_k(\tau) \;
			\alpha^{\Phi,\eps}_{(t,\tau)} \bigl ( \partial_{X_k} \bigl ( \rho_{\mathrm{int}}(\tau) \bigr ) \bigr ) \right ) 
			\\
			&= \frac{\dd \rho_{\mathrm{int}}}{\dd t}(t) - \sum_{k = 1}^d \Phi_k \; \switch(t) \, f_k(t) \; \partial_{X_k} \bigl ( \rho_{\mathrm{int}}(t) \bigr ) 
			+ \\
			&\qquad -
			\sum_{k = 1}^d \Phi_k \; \int_{-\infty}^t \dd \tau \; \switch(\tau) \, f_k(\tau) \; \frac{\dd }{\dd t} \alpha^{\Phi,\eps}_{(t,\tau)} \bigl ( \partial_{X_k} \bigl ( \rho_{\mathrm{int}}(\tau) \bigr ) \bigr ) 
			. 
		\end{align*}
		We will consider this derivative in $\mathfrak{L}^r(\Alg)$. First of all, the first two terms cancel exactly (Proposition~\ref{dynamics:prop:interaction_dynamics}~(3)), and the expression simplifies to 
		\begin{align*}
			\frac{\dd \rho_{\mathrm{full}}}{\dd t}(t) &= - \sum_{k = 1}^d \Phi_k \; \int_{-\infty}^t \dd \tau \; \switch(\tau) \, f_k(\tau) \; \frac{\dd }{\dd t}\alpha^{\Phi,\eps}_{(t,\tau)} \bigl ( \partial_{X_k} \bigl ( \rho_{\mathrm{int}}(\tau) \bigr ) \bigr )
			. 
		\end{align*}
		Therefore the crucial question is whether this integral exists in $\mathfrak{L}^r(\Alg)$: It hinges on the existence of the derivative 
		\begin{align}
			\frac{\dd }{\dd t} \alpha^{\Phi,\eps}_{(t,\tau)} \bigl ( \partial_{X_k} \bigl ( \rho_{\mathrm{int}}(\tau) \bigr ) \bigr ) \in \mathfrak{L}^r(\Alg)
			. 
		\end{align}
		Thanks to Lemmas~\ref{dynamics:lem:auxilary_difference_dynamics} and \ref{Kubo_formula:lem:Lp_regularity_states}~(3) we know
		\begin{align*}
			H_{\Phi,\eps}(\tau) \, \partial_{X_k} \bigl ( \rho_{\mathrm{int}}(\tau) \bigr ) = \gamma_{\tau}^{\Phi,\eps} \bigl (H \, \partial_{X_k} (\rho) \bigr ) \in \mathfrak{L}^r(\Alg)
		\end{align*}
		holds, and thus, $\partial_{X_k} \bigl ( \rho_{\mathrm{int}}(\tau) \bigr ) \in \rr{D}^{00}_{H_{\Phi,\eps}(\tau),r}$ lies in the set on which Proposition~\ref{dynamics:prop:properties_perturbed_dynamics}~(3) gives an explicit expression of the time derivative as the integral of a generalized commutator, 
		\begin{align*}
			\frac{\dd \rho_{\mathrm{full}}}{\dd t}(t) &= \ii \sum_{k = 1}^d \Phi_k \; \int_{-\infty}^t \dd \tau \; \switch(\tau) \, f_k(\tau) \; 
			\Bigl [ H_{\Phi,\eps}(t) \, , \, \alpha^{\Phi,\eps}_{(t,\tau)} \bigl ( \partial_{X_k} \bigl ( \rho_{\mathrm{int}}(\tau) \bigr ) \bigr ) \Bigr ]_{\ddagger}
			. 
		\end{align*}
		The right-hand side is a Bochner integral in $\mathfrak{L}^r(\Alg)$ and its existence is guaranteed by the integrability condition~\eqref{main_results:eqn:integrability_condition} along with the bound
		\begin{align}
			&\norm{\Bigl [ H_{\Phi,\eps}(t) \, , \, \alpha^{\Phi,\eps}_{(t,\tau)} \bigl ( \partial_{X_k} \bigl ( \rho_{\mathrm{int}}(\tau) \bigr ) \bigr ) \Bigr ]_{\ddagger}}_r 
			\leqslant \notag \\
			&\qquad \qquad 
			\leqslant 2 \bnorm{M(t,\tau)} \; \Bnorm{\bigl ( H_{\Phi,\eps}(\tau) - \xi \bigr ) \; \partial_{X_k} \bigl ( \rho_{\mathrm{int}}(\tau) \bigr )}_r
			\notag \\
			&\qquad \qquad 
			\leqslant 2K_t^{\Phi} \; \bnorm{(H - \xi) \; \partial_{X_k}(\rho)}_r
			.
			\label{Kubo_formula:eqn:bound_generalized_commutator}
		\end{align}
		The first inequality follows from \eqref{dynamics:eqn:norm_bound_commutator_H_evolved_observable} and in the second we used the bound \eqref{dynamics:eqn:sup_s_norm_bound_M_ts} as well as the isometry of $\gamma_{\tau}^{\Phi,\eps}$ on $\mathfrak{L}^r(\Alg)$. Exploiting the linearity (\cf Corollary~\ref{framework:cor:integral_linearity_product}), the fact that $H_{\Phi,\eps}(t)$ and $\rho_{\mathrm{int}}(t)$ commute in the generalized sense, equation~\eqref{Kubo_formula:eqn:evaluation_generalized_commutator_H_per_rho_int}, and \eqref{main_results:eqn:comparison_rho_full_rho_int} yield 
		\begin{align*}
			\frac{\dd \rho_{\mathrm{full}}}{\dd t}(t) &= + \ii \left [ H_{\Phi,\eps}(t) \; , \; \left ( \sum_{k = 1}^d \Phi_k \; \int_{-\infty}^t \dd \tau \; f_k(\tau) \;
			\alpha^{\Phi,\eps}_{(t,\tau)} \bigl ( \partial_{X_k} \bigl ( \rho_{\mathrm{int}}(\tau) \bigr ) \bigr ) \right ) \right ]_{\ddagger}
			\\
			&= - \ii \Bigl [ H_{\Phi,\eps}(t) \, , \, \bigl ( \rho_{\mathrm{int}}(t) + \pmb{\Phi} \cdot \mathbf{K}^{\Phi,\eps}(t) \bigr ) \Bigr ]_{\ddagger}
			\\
			&= - \ii \bigl [ H_{\Phi,\eps}(t) \, , \, \rho_{\mathrm{full}}(t) \bigr ]_{\ddagger}
			. 
		\end{align*}
		To prove the uniqueness of solutions of \eqref{main_results:eqn:full_dynamical_problem} it is enough to show that if $\beta(t) \in \mathfrak{D}^{00}_{H,r} \cap \Alg$ is a solution such that $\lim_{t \to -\infty} \beta(t) = 0$ then $\beta(t) = 0$ for all $t \in \R$. 
		Set $\beta_s(t) := \alpha^{\Phi,\eps}_{(s,t)} \bigl (\beta(t) \bigr )$ and consider an arbitrary $A \in \mathfrak{D}^{00}_{H,q}$ with $q = r \, (r - 1)^{-1}$. Using the invariance of the trace $\mathcal{T}$ under the dynamics $\alpha^{\Phi,\eps}_{(t,s)}$ and adapting some of the arguments in the proof of Lemma~\ref{dynamics:lem:auxilary_difference_dynamics}, we push $\alpha^{\Phi,\eps}_{(t,s)}$ to the other side, 
		\begin{align*}
			\mathcal{T} \bigl ( A \, \beta_s(t) \bigr ) = \mathcal{T} \Bigl ( A \; \alpha^{\Phi,\eps}_{(s,t)} \bigl ( \beta(t) \bigr ) \Bigr ) 
			= \mathcal{T} \Bigl ( \alpha^{\Phi,\eps}_{(t,s)}(A) \; \beta(t) \Bigr ) 
			. 
		\end{align*}
		Taking the time derivative and exploiting that the generator is “anti-selfadjoint” in the sense that \eqref{dynamics:eqn:anti_adjoint} holds, we conclude that the two terms cancel, 
		\begin{align*}
			&\ii \frac{\dd}{\dd t} \mathcal{T} \bigl ( A \, \beta_s(t) \bigr ) =
			\\
			&\qquad 
			= \mathcal{T} \Bigl ( \bigl [ H_{\Phi,\eps}(t) \, , \, \alpha^{\Phi,\eps}_{(t,s)}(A) \bigr ]_{\ddagger} \; \beta(s) \Bigr ) + \mathcal{T} \Bigl ( \alpha^{\Phi,\eps}_{(t,s)}(A) \; \bigl [ H_{\Phi,\eps}(t) \, , \, \beta(s) \bigr ]_{\ddagger} \Bigr ) 
			= 0 
			. 
		\end{align*}
		Because the derivative of the trace vanishes on the dense set $\rr{D}^{00}_{H,q} \ni A$, we in fact can elevate this to $\beta_s(t) = \mathrm{const.}$ for all $t , s \in \R$ with the help of Lemma~\ref{framework:lem:weak_zero_vector_zero}. Hence, $\beta(s) = \beta_s(s) = \beta_s(t) \in \mathfrak{L}^r(\Alg)$ holds as well. Taking the (trivial) limit $s \to -\infty$ of $\beta(t) = \alpha^{\Phi,\eps}_{(t,s)} \bigl ( \beta(s) \bigr )$ yields $\beta(t) = 0$. Therefore, $\rho_{\mathrm{full}}(t)$ is \emph{the} unique solution of \eqref{main_results:eqn:full_dynamical_problem}. 
	\end{enumerate}
\end{proof}
Now that we know $\rho_{\mathrm{full}}(t)$ exists, let us collect some properties of $\rho_{\mathrm{full}}(t)$ and $\rho_{\mathrm{int}}(t)$, especially with regards to their behavior in the far past where $t \to -\infty$. From the proof of Theorem~\ref{main_results:thm:comparison_dynamics} we have learned the following: 
\begin{corollary}\label{Kubo_formula:lem:limit_two_dynamics}
	Suppose Hypotheses~\ref{hypothesis:trace}–\ref{hypothesis:state} are satisfied, and let $r = 1 , p$. Then the following holds: 
	\begin{enumerate}
		\item For all $t \in \R$ the evolved states $\rho_{\mathrm{full}}(t)$ and $\rho_{\mathrm{int}}(t)$ 
		\begin{align*}
			\bnorm{\rho_{\mathrm{full}}(t)}_r = \snorm{\rho}_r = \bnorm{\rho_{\mathrm{int}}(t)}_r 
		\end{align*}
		is independent of the time. Moreover, if $p \geqslant 2$ and $\rho$ is a projection (\ie $\rho^2 = \rho$) the same holds also for $\rho_{\mathrm{full}}(t)$ and $\rho_{\mathrm{int}}(t)$ for all $t \in \R$.
		\item The maps $t \mapsto \rho_{\mathrm{full}}(t)$ and $t \mapsto \rho_{\mathrm{int}}(t)$ are both continuous with respect to the topology of $\mathfrak{L}^r(\Alg)$, and the limits
		\begin{align*}
			\lim_{t \to -\infty}\rho_{\mathrm{full}}(t) = \rho = \lim_{t \to -\infty}\rho_{\mathrm{int}}(t)
		\end{align*}
		exist in $\mathfrak{L}^r(\Alg)$. 
		\item The fully evolved state satisfies 
		\begin{align*}
			\lim_{s \to -\infty} \alpha^{\Phi,\eps}_{(t,s)} \bigl ( \rho_{\mathrm{int}}(s) \bigr ) = \rho_{\mathrm{full}}(t)
			= \lim_{s \to -\infty} \alpha^{\Phi,\eps}_{(t,s)} \bigl ( \rho_{\mathrm{full}}(s) \bigr )
			\in \mathfrak{L}^r(\Alg)
			, 
			&&
			r = 1 , p
			. 
		\end{align*}
	\end{enumerate}
\end{corollary}
\begin{remark}
	The condition $p \geqslant 2$ in (1) is due to the fact that we cannot guarantee that $\rho_{\mathrm{full}}(t)$ is a projection otherwise (\cf Lemma~\ref{framework:lem:project_limit}). 
\end{remark}
\begin{remark}
	The statements of Theorem~\ref{main_results:thm:comparison_dynamics} and its proof simplify significantly in the case of a finite initial time $t_0 > -\infty$ (\eg when the $f_k$ are supported in $[t_0 , +\infty)$). In this case the definition of the full dynamics reduces to 
	\begin{align*}
		\rho_{\mathrm{full}}(t) := \alpha^{\Phi,\eps}_{(t,t_0)}(\rho) 
	\end{align*}
	due to the continuity of the map $\alpha^{\Phi,\eps}_{(t,s)}$ and one has $\rho_{\mathrm{full}}(t) = \rho = \rho_{\mathrm{int}}(t)$ for all $t \leqslant t_0$.
\end{remark}
%

\section{The Kubo formula for the conductivity} 
\label{Kubo_formula:Kubo_formula}
The second step consists of plugging in the difference $\rho_{\mathrm{full}}(t) - \rho_{\mathrm{int}}(t)$, equation~\eqref{main_results:eqn:comparison_rho_full_rho_int}, into the macroscopic net current $\mathscr{J}^{\Phi,\eps}[J,\rho](t)$ to obtain a “Taylor expansion” in $\pmb{\Phi}$, and then derive it with respect to the $\Phi_k$ to obtain the conductivity tensor. At the end of the day this line of reasoning yields the celebrated Kubo formula~\eqref{main_results:eqn:Kubo_formula} (Theorem~\ref{main_results:thm:Kubo_formula}), and in the special case where $\rho$ is a pure state, the Kubo-Str\v{e}da formula~\eqref{main_results:eqn:Kubo_Streda_formula} (Corollary~\ref{main_results:cor:conductivity_tensor_special_cases}). We will connect the dots to finish the proofs of these two statements next.

\subsection{The macroscopic net current and the conductivity tensor} 
\label{Kubo_formula:Kubo_formula:net_current}
The macroscopic net current 
\begin{align*}
	\mathscr{J}^{\Phi,\eps}[J,\rho](t) 
	= \mathcal{T} \Bigl ( J_{\Phi,\eps}(t) \, \bigl ( \rho_{\mathrm{full}}(t) - \rho_{\mathrm{int}}(t) \bigr ) \Bigr )
	= \mathcal{T} \bigl ( J_{\Phi,\eps}(t) \, \rho_{\mathrm{full}}(t) \bigr ) - \mathcal{T}(J \, \rho)
\end{align*}
associated with the current-type observable $J$ and state $\rho$ is the difference between the stationary current and the current driven by the perturbation which is switched on adiabatically at rate $\eps > 0$ between the initial time $-\infty$ and the final time $t$. While the main focus lies with the current operators $J_{\Phi,\eps}(t)$ from Hypothesis~\ref{hypothesis:current}, our arguments are in fact more general than that, and at least in this section do not rely on any specific form for $J$. We have condensed the necessary properties into Definition~\ref{main_results:defn:current} where we have introduced the notion of “current-type observable”. Strictly speaking, we still have to deliver a proof that the $J_{\kappa}$ \emph{are} current-type observables, something we will remedy at the end of this section. 

The definition of the macroscopic net current~\eqref{main_results:eqn:net_density_current_1} involves the difference $\rho_{\mathrm{full}}(t) - \rho_{\mathrm{int}}(t) = \pmb{\Phi} \cdot \mathbf{K}^{\Phi,\eps}(t)$, and so plugging equation~\eqref{main_results:eqn:comparison_rho_full_rho_int} yields 
\begin{align}
	\mathscr{J}^{\Phi,\eps}[J,\rho](t) &= \pmb{\Phi} \cdot \mathcal{T} \Bigl ( J_{\Phi,\eps}(t) \; \mathbf{K}^{\Phi,\eps}(t) \Bigr ) 
	\label{Kubo_formula:eqn:macroscopic_current_explicit}
\end{align}
where the components of $\mathbf{K}^{\Phi,\eps}(t)$ are defined in equation~\eqref{main_results:eqn:components_K_Phi_eps}. Formally, we can pull the current operator $J_{\Phi,\eps}(t)$ and the trace into the integral, so that $\mathscr{J}^{\Phi,\eps}[J,\rho](t) = \sum_{k = 1}^d \Phi_j \; \mathscr{K}^{\Phi,\eps}_j[J,\rho](t)$ is given in terms of 
\begin{align}
	\mathscr{K}^{\Phi,\eps}_j[J,\rho](t) := \int_{-\infty}^t \dd \tau \; \switch(\tau) \; f_j(\tau) \; \mathcal{T} \Bigl ( J_{\Phi,\eps}(t) \; \alpha^{\Phi,\eps}_{(t,\tau)} \bigl ( \partial_{X_j} \bigl ( \rho_{\mathrm{int}}(\tau) \bigr ) \bigr ) \Bigr )
	.
	\label{Kubo_formula:eqn:macroscopic_current_explicit}
\end{align}
These manipulations will be justified in the next 
\begin{lemma}\label{Kubo_formula:lem:existence_integral_net_current}
	Suppose Hypotheses~\ref{hypothesis:trace}–\ref{hypothesis:state} hold true, and let $J$ be a current-type observable in the sense of Definition~\ref{main_results:defn:current}. Then we have: 
	\begin{enumerate}
		\item The components $\mathscr{K}^{\Phi,\eps}_j[J,\rho](t)$, $j = 1 , \ldots , d$, exist, are given in \eqref{Kubo_formula:eqn:macroscopic_current_explicit}, and depend continuously on the perturbation field $\pmb{\Phi}$. 
		\item The net current $\mathscr{J}^{\Phi,\eps}[J,\rho](t)$ also exists, is given by \eqref{Kubo_formula:eqn:macroscopic_current_explicit}, and depends continuously on the perturbation field $\pmb{\Phi}$. 
		\item The net current vanishes in the limit 
		\begin{align*}
			\lim_{\Phi \to 0} \mathscr{J}^{\Phi,\eps}[J,\rho](t) = \mathscr{J}^{0,\eps}[J,\rho](t) 
			= 0
			&&
			\forall t \in \R 
			. 
		\end{align*}
	\end{enumerate}
\end{lemma}
\begin{remark}
	Neither the differentiability of $H \, \rho$ in $\mathfrak{L}^1(\Alg)$ nor regularity conditions of $\rho$ and $H \, \rho$ for $p > 1$ are necessary. 
\end{remark}
\begin{proof}
	\begin{enumerate}
		\item First, we need to justify the existence of all objects: The existence of the full dynamics $\rho_{\mathrm{full}}(t) = \rho_{\mathrm{int}}(t) + \pmb{\Phi} \cdot \mathbf{K}^{\Phi,\eps}(t)$ in $\mathfrak{L}^1(\Alg)$ (Theorem~\ref{main_results:thm:comparison_dynamics}) in particular implies the existence of each of the $K^{\Phi,\eps}_j[J,\rho](t) \in \mathfrak{L}^1(\Alg)$. Furthermore, thanks to the assumptions contained in Definition~\ref{main_results:defn:current} we can write 
		\begin{align*}
			J_{\Phi,\eps}(t) \, \rho_{\mathrm{int}}(t) &= \gamma^{\Phi,\eps}_t (J \, \rho)
			\in \mathfrak{L}^1(\Alg)
		\end{align*}
		as the interaction evolution of $J \, \rho$ is integrable by assumption, and $J_{\Phi,\eps}(t) \, \rho_{\mathrm{full}}(t) \in \mathfrak{L}^1(\Alg)$ combined with $\domain(H) \subseteq \domain(J)$ allows us to repeat the proof of Lemma~\ref{dynamics:lem:auxilary_difference_dynamics} \emph{mutatis mutandis}. Putting all of these arguments together, we conclude that the $J_{\Phi,\eps}(t) \, K^{\Phi,\eps}_j[J,\rho](t)$ exist in $\mathfrak{L}^1(\Alg)$ for all $t \in \R$. 
		
		The existence of this product in $\mathfrak{L}^1(\Alg)$ in particular implies that it is measurable, and the $\mathcal{T}$-dense domain of the integral $J_{\Phi,\eps}(t) \, K^{\Phi,\eps}_j[J,\rho](t)$ is necessarily contained in the domain of the integrand 
		\begin{align*}
			\switch(\tau) \; f_j(\tau) \; J_{\Phi,\eps}(t) \; \alpha^{\Phi,\eps}_{(t,\tau)} \bigl ( \partial_{X_j} \bigl ( \rho_{\mathrm{int}}(\tau) \bigr ) \bigr ) \in \rr{M}(\Alg)
			, 
		\end{align*}
		which is therefore $\mathcal{T}$-dense. Consequently, we have shown its measurability. 
		
		While questions of measurability are a subtle points, they are an essential ingredient in the next step where we want to apply Corollary~\ref{framework:cor:integral_linearity_product}. And conveniently, these arguments also prove the existence of the integral: We will use the norm estimate~\eqref{Kubo_formula:eqn:bound_generalized_commutator} and apply it to the two operators in the integrand, 
		\begin{align}
			&\norm{J_{\Phi,\eps}(t) \; \alpha^{\Phi,\eps}_{(t,\tau)} \bigl ( \partial_{X_j} \bigl ( \rho_{\mathrm{int}}(\tau) \bigr ) \bigr )}_1 
			\leqslant \notag \\
			&\qquad 
			\leqslant \norm{J_{\Phi,\eps}(t) \, \frac{1}{H_{\Phi,\eps}(t) - \xi}} \; \Bnorm{\bigl ( H_{\Phi,\eps}(t) - \xi \bigr ) \; \alpha^{\Phi,\eps}_{(t,\tau)} \bigl ( \partial_{X_j} \bigl ( \rho_{\mathrm{int}}(\tau) \bigr ) \bigr )}_1
			\notag \\
			&\qquad 
			\leqslant K_t \; \norm{J \, \frac{1}{H - \xi}} \; \bnorm{(H - \xi) \, \partial_{X_j}(\rho)}_1
			, 
			\label{Kubo_formula:eqn:norm_estimate_J_alpha_rho}
		\end{align}
		where $K_t = \max_{\sabs{\Phi} \leqslant 1} K_t^{\Phi}$ and $K_t^{\Phi}$ given by equation~\eqref{dynamics:eqn:sup_s_norm_bound_M_ts} is the constant which appears in the proof of Proposition~\ref{dynamics:prop:M_limit_Phi}. Therefore, 
		\begin{align*}
			J_{\Phi,\eps}(t) \; \alpha^{\Phi,\eps}_{(t,\tau)} \bigl ( \partial_{X_j} \bigl ( \rho_{\mathrm{int}}(\tau) \bigr ) \bigr ) 
			\in \mathfrak{L}^1(\Alg)
			. 
		\end{align*}
		is not just measurable, but in fact lies in $\mathfrak{L}^1(\Alg)$ for all $\tau , t \in \R$. 
		
		Once we combine this estimate with the integrability condition~\eqref{main_results:eqn:integrability_condition} of Hypothesis~\ref{hypothesis:perturbation}, we deduce the existence of the Bochner integral
		\begin{align*}
			\int_{-\infty}^t \dd \tau \; \switch(\tau) \; f_j(\tau) \; H_{\Phi,\eps}(t) \; \alpha^{\Phi,\eps}_{(t,\tau)} \bigl ( \partial_{X_j} \bigl ( \rho_{\mathrm{int}}(\tau) \bigr ) \bigr ) \in \mathfrak{L}^1(\Alg)
			. 
		\end{align*}
		Thus, noting $\domain(H) \subseteq \domain(J)$ we have checked all assumptions of Corollary~\ref{framework:cor:integral_linearity_product} with whose help we now can push the current inside the integral, 
		\begin{align}
			J_{\Phi,\eps}(t) \, K^{\Phi,\eps}_j(t) &= \int_{-\infty}^t \dd \tau \; \switch(\tau) \; f_j(\tau) \; J_{\Phi,\eps}(t) \; \alpha^{\Phi,\eps}_{(t,\tau)} \bigl ( \partial_{X_j} \bigl ( \rho_{\mathrm{int}}(\tau) \bigr ) \bigr ) \in \mathfrak{L}^1(\Alg)
			. 
			\label{Kubo_formula:eqn:J_K_exists_L1}
		\end{align}
		The continuity of the trace with respect to the topology of $\mathfrak{L}^1(\Alg)$ coupled with the inequality
		\begin{align*}
			\abs{\mathcal{T} \Bigl ( J_{\Phi,\eps}(t) \; \alpha^{\Phi,\eps}_{(t,\tau)} \bigl ( \partial_{X_j} \bigl ( \rho_{\mathrm{int}}(\tau) \bigr ) \bigr ) \Bigr )} 
			\leqslant \norm{J_{\Phi,\eps}(t) \; \alpha^{\Phi,\eps}_{(t,\tau)} \bigl ( \partial_{X_j} \bigl ( \rho_{\mathrm{int}}(\tau) \bigr ) \bigr )}_1
		\end{align*}
		ensures that we can move the trace $\mathcal{T}$ inside the integral in equation~\eqref{Kubo_formula:eqn:macroscopic_current_explicit}: we can approximate the Bochner integral arbitrarily well by a finite sum of simple operator-valued functions, and exchange trace and limit. 
		
		To show the continuity of the limit, we refactor the trace 
		\begin{align*}
			\lim_{\Phi \to 0} \mathcal{T} \left ( J_{\Phi,\eps}(t) \, \alpha^{\Phi,\eps}_{(t,\tau)} \bigl ( \partial_{X_j} \bigl ( \rho_{\mathrm{int}}(\tau) \bigl ) \bigr ) \right ) = 
			\lim_{\Phi\to 0} \mathcal{T} \left ( A_{\Phi} \, B_{\Phi} \right )
		\end{align*}
		into the product of two operators, namely 
		\begin{align*}
			A_{\Phi} := U_{\Phi,\eps}(\tau,t)\;
			\gamma^{\Phi,\eps}_t \left ( J \, \frac{1}{H - \xi} \right ) \; M_{\Phi,\eps}(t,\tau) 
			\in \Alg 
		\end{align*}
		which contains $M_{\Phi,\eps}(t,\tau)$ from equation~\eqref{dynamics:eqn:definition_M_ts}, and 
		\begin{align*}
			B_{\Phi} := \gamma^{\Phi,\eps}_{\tau} \bigl ( (H - \xi) \, \partial_{X_j}(\rho) \bigr ) 
			\in \mathfrak{L}^1(\Alg)
			. 
		\end{align*}
		The sequence $A_{\Phi}$ is the product of three equibounded and SOT convergent sequences in $\Alg$ when $\Phi \to 0$. According to Corollary~\ref{dynamics:cor:Phi_0_limit_perturbed_dynamics} the full propagator $U_{\Phi,\eps}(\tau,t) \to U_0(\tau-t)$ converges strongly to the unperturbed one, $\gamma^{\Phi,\eps}_t \rightarrow \id_{\mathfrak{L}^1(\Alg)}$ strongly (Proposition~\ref{dynamics:prop:interaction_dynamics}), and $M_{\Phi,\eps}(t,\tau) \to U_0(t-\tau)$ in view of Proposition~\ref{dynamics:prop:M_limit_Phi}. 
		
		The product is continuous in the SOT with respect to equibounded sequences (Lemma~\ref{framework:lem:strong_convergence_trace_product}), in this case 
		\begin{align*}
			\slim_{\Phi \to 0} A_{\Phi} = U_0(\tau-t) \, \left ( J \, \frac{1}{H - \xi} \right ) \, U_0(t-\tau)
			. 
		\end{align*}
		for the first factor and, invoking Proposition~\ref{dynamics:prop:interaction_dynamics} a second time, 
		\begin{align*}
			\norm{\cdot}_1-\lim_{\Phi \to 0} B_{\Phi} = (H - \xi) \, \partial_{X_j}(\rho)
		\end{align*}
		for the second. In conclusion, the application of Lemma~\ref{framework:lem:strong_convergence_trace_product} and the continuity of the trace with respect to the topology of $\mathfrak{L}^1(\Alg)$ ensure that
		\begin{align*}
			\lim_{\Phi \to 0} \mathcal{T} &\left ( 
			J_{\Phi,\eps}(t) \; \alpha^{\Phi,\eps}_{(t,\tau)} \bigl ( \partial_{X_j} \bigl ( \rho_{\mathrm{int}}(\tau) \bigl ) \bigr ) \right ) 
			= \\
			&= \mathcal{T} \left ( U_0(\tau-t) \left ( J \, \frac{1}{H - \xi} \right ) \, U_0(t-\tau) \, (H - \xi) \, \partial_{X_j}(\rho) \right ) 
			\\
			&= \mathcal{T} \left ( \left ( J \, \frac{1}{H - \xi} \right ) \, (H - \xi) \, U_0(t-\tau) \, \partial_{X_j}(\rho) \, U_0(\tau-t) \right ) 
			\\
			&= \mathcal{T} \left ( \left ( J \, \frac{1}{H - \xi} \right ) (H - \xi) \; \alpha^0_{t-\tau} \bigl ( \partial_{X_j}(\rho) \bigr ) \right ) 
			 \\
			&= \mathcal{T} \Bigl ( J \, \alpha^0_{t-\tau} \bigl ( \partial_{X_j}(\rho) \bigr ) \Bigr )
		\end{align*}
		where we used the cyclic property of the trace, the commutativity between $H$ and the unperturbed evolution $\alpha^0_t$ (\cf Proposition~\ref{dynamics:prop:density_core_Liouvillian}) and Lemma~\ref{framework:lem:extension_algebra_unbounded_operators}~(2). This concludes the proof.
		\item This follows immediately from (1) and $\mathscr{J}^{\Phi,\eps}[J,\rho](t) = \sum_{k = 1}^d \Phi_k \; \mathscr{K}^{\Phi,\eps}_k[J,\rho](t)$. 
		\item The vanishing of the net current follows from the continuity in $\pmb{\Phi}$ proven in (1) and Proposition~\ref{dynamics:prop:continuity_perturbed_dynamics_field_strength}. 
	\end{enumerate}
\end{proof}
\begin{remark}[Absence of net current]\label{Kubo_formula:remark:equilibrium_conditions_current}
	There are two cases in which the system is in equilibrium. The first is when $t \to -\infty$. In fact the continuity condition~(ii) of Definition~\ref{main_results:defn:current} ensures
	\begin{align*}
		\lim_{t \to -\infty} \mathscr{J}^{\Phi,\eps}[J,\rho](t) = \lim_{t \to -\infty} \mathcal{T} \bigl (J_{\Phi,\eps}(t) \, \rho_{\mathrm{full}}(t) \bigr ) - \mathcal{T}(J \, \rho) = 0
		. 
	\end{align*}
	For finite initial times $t_0 > -\infty$ we in fact have the stronger condition
	\begin{align*}
		\mathscr{J}^{\Phi,\eps}[J,\rho](t) = \mathcal{T} \bigl ( J_{\Phi,\eps}(t_0) \, \rho_{\mathrm{full}}(t_0) \bigr ) - \mathcal{T}(J \, \rho) = 0 
		&&
		\forall t \leqslant t_0
	\end{align*}
	due to $J_{\Phi,\eps}(t) = J_{\Phi,\eps}(t_0) = J$ and $\rho_{\mathrm{full}}(t) = \rho_{\mathrm{full}}(t_0) = \rho$ for all $t \leqslant t_0$. In many situations one expects also a \emph{null net current} at the equilibrium, namely
	\begin{align}
		\lim_{t \to -\infty} \mathcal{T} \bigl ( J_{\Phi,\eps}(t) \, \rho_{\mathrm{full}}(t) \bigr ) = \mathcal{T}(J \, \rho) = 0
		.
		\label{Kubo_formula:eqn:no_go_equation}
	\end{align}
	Other than Definition~\ref{main_results:defn:current}~(i)–(ii), the validity of condition~\eqref{Kubo_formula:eqn:no_go_equation} depends also on a lot of factors like the nature of $H$, the generators $\mathbf{X}$ and the initial state $\rho$. The vanishing of the current expectation value at equilibrium is usually referred to as a \emph{no go theorem} \cite[Section~1.1]{Bellissard_Schulz_Baldes:quantum_transport_aperiodic_media:1998}, and no go theorems have been proven under various circumstances for discrete quantum systems \cite[Proposition~3]{Bellissard_van_Elst_Schulz_Baldes:noncommutative_geometry_quantum_Hall_effect:1994} and in the continuum (\cf \cite[Lemma~5.7]{Bouclet_Germinet_Klein_Schenker:linear_response_theory_magnetic_Schroedinger_operators_disorder:2005} and \cite[Proposition~5]{Kellendonk_Schulz-Baldes:quantization_edge_currents:2004}). The second equilibrium condition concerns the vanishing of the perturbation field $\pmb{\Phi}$. Clearly, in the limit of vanishing perturbation at fixed $\eps > 0$ the net current vanishes,
	\begin{align*}
		\mathscr{J}^{\Phi=0,\eps}[J,\rho](t) = \mathcal{T} \bigl ( J_{\Phi,\eps}(t) \, \rho_{\mathrm{full}}(t) \bigr ) \Big \vert_{\Phi = 0} - \mathcal{T}(J\rho) = 0
		, 
	\end{align*}
	since $G_{\Phi=0,\eps}(t) = \id$ implies $J_{\Phi=0,\eps}(t) = J$, $H_{\Phi=0,\eps} = H$ and $\rho_{\mathrm{full}}(t) = \rho$ as $\rho$ is an equilibrium state. The continuity of 
	\begin{align*}
		\lim_{\Phi \to 0} \mathscr{J}^{\Phi,\eps}[J,\rho](t) = \mathscr{J}^{\Phi=0,\eps}[J,\rho](t) = 0
	\end{align*}
	at $\pmb{\Phi} = 0$ requires some additional assumptions, \eg those in  Lemma~\ref{Kubo_formula:lem:existence_integral_net_current}.
\end{remark}
Much of what we do works for operators other than the $J_{\kappa}$ defined in Hypothesis~\ref{hypothesis:current}. Instead, the proofs rely on some of the properties these $J_{\kappa}$ have, prompting us to define the notion of current operators in Definition~\ref{main_results:defn:current}. However, to close the circle, we have yet to verify that the currents generated by perturbing via the $X_j$ satisfy these conditions. 
\begin{proposition}\label{Kubo_formula:prop:J_k_is_of_current_type}
	Suppose Hypotheses~\ref{hypothesis:trace}–\ref{hypothesis:state} hold true. Then the operators 
	\begin{align*}
		J_{\kappa,\Phi,\eps}(t) := (-1)^{\abs{\kappa}} \, \mathrm{ad}_{\mathbf{X}}^{\kappa} \bigl ( H_{\Phi,\eps}(t) \bigr ) 
		, 
		&&
		\kappa \in \N_0^d
		, \; 
		\abs{\kappa} \geqslant 1
		, 
	\end{align*}
	from Hypothesis~\ref{hypothesis:current} are current-type operators in the sense of Definition~\ref{main_results:defn:current}. 
\end{proposition}
\begin{proof}
	Pick an arbitrary operator $J_{\kappa}$ with $\kappa \in \N_0^d$, $\abs{\kappa} \geqslant 1$. Let us start with point (i). On the joint core $\domain_{\mathrm{c}}(H)$
	\begin{align*}
		J_{\kappa,\Phi,\eps}(t) &= (-1)^{\abs{\kappa}} \, \mathrm{ad}_{\mathbf{X}}^{\kappa} \bigl ( H_{\Phi,\eps}(t) \bigr ) 
		= G_{\Phi,\eps}(t) \, (-1)^{\abs{\kappa}} \, \mathrm{ad}_{\mathbf{X}}^{\kappa}(H) \, G_{\Phi,\eps}(t)^*
		\\
		&= G_{\Phi,\eps}(t) \, J_{\kappa} \, G_{\Phi,\eps}(t)^* 
	\end{align*}
	holds, because of the latter's invariance under the $X_j$ (Hypothesis~\ref{hypothesis:current}~(i)) and the fact that the $X_j$ mutually commute amongst one another. The compatibility of the $X_j$ (which enter the definition of $G_{\Phi,\eps}(t)$) with the algebraic structure (Hypothesis~\ref{hypothesis:generators}~(i) and (ii)), combined with the affiliation of $J_{\kappa}$ to $\Alg$ (by generalizing Lemma~\ref{Kubo_formula:lem:Lp_regularity_states}~(i) in a straightforward way) yields $J_{\kappa,\Phi,\eps}(t) \in \affil(\Alg)$ as well. 
	
	Now on to (iii): Due to $J_{\kappa}$ being infinitesimally $H$-bounded, Hypothesis~\ref{hypothesis:current}~(iv), the domain of $J_{\kappa}$ contains $\domain(H)$. Hence, $J_{\kappa} \, (H - \xi)^{-1}$ is affiliated to $\Alg$ and bounded, and therefore an element of $\Alg$ itself. Reading the first equality initially on $\Hil$, 
	\begin{align}
		J_{\kappa,\Phi,\eps}(t) \, \frac{1}{H_{\Phi,\eps} - \xi} &= G_{\Phi,\eps}(t) \; \left ( J_{\kappa} \, \frac{1}{H - \xi} \right ) \; G_{\Phi,\eps}(t) 
		\notag \\
		&= \gamma^{\Phi,\eps}_t \left ( J_{\kappa} \, \frac{1}{H - \xi} \right )
		, 
		\label{Kubo_formula:eqn:interaction_evolution_J_kappa_resolvent}
	\end{align}
	and making use of $J_{\kappa} \, (H - \xi)^{-1}$ and the fact that conjugating with $G_{\Phi,\eps}(t)$ preserves the algebraic structures of $\Alg$, we conclude that we can in fact interpret them as equalities on $\Alg$ itself. This shows (iii). 
	
	Lastly, point (ii), follows from inserting a resolvent in the product and using \eqref{Kubo_formula:eqn:interaction_evolution_J_kappa_resolvent}, 
	\begin{align}
		J_{\kappa,\Phi,\eps}(t) \, \rho_{\mathrm{full}}(t) &= \left ( J_{\kappa,\Phi,\eps}(t) \frac{1}{H_{\Phi,\eps}(t) - \xi} \right ) \, \Bigl ( \bigl ( H_{\Phi,\eps}(t) - \xi \bigr ) \, \rho_{\mathrm{full}}(t) \Bigr )
		\notag \\
		&= \gamma^{\Phi,\eps}_t \left ( J_{\kappa} \, \frac{1}{H - \xi} \right ) \, \Bigl ( \bigl ( H_{\Phi,\eps}(t) - \xi \bigr ) \, \rho_{\mathrm{full}}(t) \Bigr )
		\label{Kubo_formula:eqn:J_rho_product_decomposition}
	\end{align}
	and invoking Lemma~\ref{framework:lem:extension_algebra_unbounded_operators}~(2). We have already verified that $J_{\kappa,\Phi,\eps}(t) \, \bigl ( H_{\Phi,\eps}(t) - \xi \bigr )^{-1} \in \Alg$ is in the algebra, that $J_{\kappa,\Phi,\eps}(t) \in \affil(\Alg)$ is affiliated, so the remaining step is proving $\bigl ( H_{\Phi,\eps}(t) - \xi \bigr ) \, \rho_{\mathrm{full}}(t) \in \mathfrak{L}^1(\Alg)$. It turns out we have already done all of the hard work: in view of $\rho_{\mathrm{full}}(t) = \rho_{\mathrm{int}}(t) + \pmb{\Phi} \cdot \mathbf{K}^{\Phi,\eps}(t)$ (Theorem~\ref{main_results:thm:comparison_dynamics}) and $H_{\Phi,\eps}(t) \, \rho_{\mathrm{int}}(t) = \gamma^{\Phi,\eps}_t(H \, \rho) \in \mathfrak{L}^1(\Alg)$ (Lemma~\ref{dynamics:lem:auxilary_difference_dynamics}) it suffices to consider $H_{\Phi,\eps}(t) \, K^{\Phi,\eps}_j(t)$ and prove its integrability. But this has already been established in the proof of Lemma~\ref{Kubo_formula:lem:existence_integral_net_current} via \eqref{Kubo_formula:eqn:J_K_exists_L1}. 
	
	The fact that $\gamma^{\Phi,\eps}_t \to \id_{\mathfrak{L}^1(\Alg)}$ is equibounded and converges in the SOT to the identity as $t \to -\infty$ (Proposition~\ref{dynamics:prop:interaction_dynamics}~(i)) means the first factor does what we want. To control the second, we again use Theorem~\ref{main_results:thm:comparison_dynamics}: the first term $H_{\Phi,\eps}(t) \, \rho_{\mathrm{int}}(t) = \gamma^{\Phi,\eps}_t(H \, \rho) \in \mathfrak{L}^1(\Alg)$ once more defines an equibounded sequence, converging to $H \, \rho$ — the term we are looking for. 
	
	Therefore, if we can show that the last term which is a sum over 
	\begin{align}
		\bigl ( H_{\Phi,\eps}(t) - \xi \bigr ) \; &\int_{-\infty}^t \dd \tau \, \switch(\tau) \, f_j(\tau) \, \alpha^{\Phi,\eps}_{(t,\tau)} \bigl ( \partial_{X_j} \bigl ( \rho_{\mathrm{int}}(\tau) \bigr ) \bigr ) 
		= \notag \\
		&= \int_{-\infty}^t \dd \tau \, \switch(\tau) \, f_j(\tau) \, \bigl ( H_{\Phi,\eps}(t) - \xi \bigr ) \; \alpha^{\Phi,\eps}_{(t,\tau)} \bigl ( \partial_{X_j} \bigl ( \rho_{\mathrm{int}}(\tau) \bigr ) \bigr ) 
		\label{Kubo_formula:eqn:integral_H_K}
	\end{align}
	vanishes, we will have shown $\norm{\cdot}_1-\lim_{t \to -\infty} J_{\kappa,\Phi,\eps}(t) \, \rho_{\mathrm{full}} = J_{\kappa} \, \rho$. First, we need to justify pulling $H_{\Phi,\eps}(t)$ inside the integral with the help of Corollary~\ref{framework:cor:integral_linearity_product}. Fortunately, we can invoke this Corollary as soon as we can prove the existence of the second integral. This is again something we have already proven earlier via estimate~\eqref{Kubo_formula:eqn:bound_generalized_commutator} and our integrability conditions~\eqref{main_results:eqn:integrability_condition} imposed in Hypothesis~\ref{hypothesis:perturbation}. Moreover, the constant $K_t^{\Phi}$ defined in \eqref{dynamics:eqn:sup_s_norm_bound_M_ts} which enters this estimate is non-decreasing in $t$ and bounded from below by $1$. Thus, the integral~\eqref{Kubo_formula:eqn:integral_H_K} exists in $\mathfrak{L}^1(\Alg)$, and the estimate also allows us to invoke Dominated Convergence to conclude that it vanishes as $t \to -\infty$. This finishes the proof. 
\end{proof}
%

\subsection{Proof of the Kubo formula} 
\label{Kubo_formula:Kubo_formula:proof}
Because the expansion of $\mathscr{J}^{\Phi,\eps}[J,\rho](t)$ in $\Phi$ vanishes to leading order, 
\begin{align*}
	\mathscr{J}^{\Phi,\eps}[J,\rho](t) &= \sum_{k = 1}^d \Phi_k \; \mathcal{T} \bigl ( J_{\Phi,\eps}(t) \; K^{\Phi,\eps}_k(t) \bigr ) 
	= \sum_{k = 1}^d \Phi_k \, \sigma^{\eps}_k[J,\rho](t) + \order(\Phi^2)
	,
\end{align*}
it is suggestive to identify the coefficients of the conductivity coefficients (\cf Definition~\ref{main_results:defn:conductivity_tensor}) with 
\begin{align*}
	\sigma^{\eps}_k[J,\rho](t) = \mathscr{K}^{\Phi,\eps}_k[J,\rho](t) \Big \vert_{\Phi = 0} &= \mathcal{T} \bigl ( J_{\Phi,\eps}(t) \; K^{\Phi,\eps}_k(t) \bigr ) \Big \vert_{\Phi = 0}
	. 
\end{align*}
We will prove now that this is indeed the case. 
\begin{proof}[Theorem~\ref{main_results:thm:Kubo_formula}]
	The main ingredients of the proof are that the net current can be expressed as $\mathscr{J}^{\Phi,\eps}[J,\rho](t) = \sum_{k = 1}^d \Phi_k \; \mathscr{K}^{\Phi,\eps}_k[J,\rho](t)$ and the continuity of $\mathscr{K}^{\Phi,\eps}_k[J,\rho](t)$ in $\pmb{\Phi}$ (Lemma~\ref{Kubo_formula:lem:existence_integral_net_current}~(1)). Here, Hypothesis~\ref{hypothesis:trace}–\ref{hypothesis:state} were used. Let us use the shorthand $\Phi_k$ for $(0,\ldots,\Phi_k,\ldots,0)$. Then writing the $\partial_{\Phi_k}$ derivative as a difference quotient, the conductivity tensor 
	\begin{align*}
		\sigma^{\eps}_k[J,\rho](t) &= \lim_{\Phi_k \to 0} \frac{\mathscr{J}^{\Phi_k}[J,\rho](t) - \mathscr{J}^{0,\eps}[J,\rho](t)}{\Phi_k}
		\\
		&= \lim_{\Phi_k \to 0} \mathscr{K}^{\Phi_k,\eps}_k[J,\rho](t)
		= \mathscr{K}^{0,\eps}_k[J,\rho](t)
	\end{align*}
	reduces to $\mathscr{K}^{\Phi,\eps}_k[J,\rho](t)$ evaluated at $\Phi = 0$. 
	
	The continuity of the trace in the topology of $\mathfrak{L}^1(\Alg)$ allows us to pull the limit into the trace and exploit that $\alpha^{\Phi,\eps}_{(t,\tau)} \to \alpha^0_{t-\tau}$ and $\gamma^{\Phi,\eps}_t \to \id_{\mathfrak{L}^1}$ converge strongly as $\Phi_k \to 0$, thereby obtaining the Kubo formula, 
	\begin{align*}
		\sigma^{\eps}_k[J,\rho](t) &= - \int_{-\infty}^t \dd \tau \; \switch(\tau) \; f_k(\tau) \; \mathcal{T} \Bigl ( \lim_{\Phi_k \to 0} \Bigl ( J_{\Phi_k,\eps}(t) \; \alpha^{\Phi_k,\eps}_{(t,\tau)} \bigl ( \partial_{X_k} \bigl ( \rho_{\mathrm{int}}(\tau) \bigr ) \bigr ) \Bigr ) \Bigr )
		\\
		&= - \int_{-\infty}^t \dd \tau \; \switch(\tau) \; f_k(\tau) \; \mathcal{T} \bigl ( J \; \alpha^0_{t-\tau} \bigl ( \partial_{X_k}(\rho) \bigr ) \bigr )
		. 
	\end{align*}
\end{proof}
To find more explicit expressions for special choices of $f_k$, it is helpful to rewrite the conductivity tensor 
\begin{align}
	\sigma^{\eps}_k[J,\rho](t) &= \widetilde{\sigma}^{\eps}_k[J,\rho](t) + \delta^{\eps}_k[J,\rho](t)
	\notag \\
	:& \negmedspace= - \int_{-\infty}^t \dd \tau \; \e^{\eps \tau} \; f_k(\tau) \; \mathcal{T} \bigl ( J \; \alpha^0_{t-\tau} \bigl ( \partial_{X_k}(\rho) \bigr ) \bigr ) 
	+ \notag \\
	&\qquad 
	- \int_0^t \dd \tau \; \bigl ( 1 - \e^{\eps \tau} \bigr ) \; f_k(\tau) \; \mathcal{T} \bigl ( J \; \alpha^0_{t-\tau} \bigl ( \partial_{X_k}(\rho) \bigr ) \bigr )
	\label{Kubo_formula:eqn:splitting_conductivity_tensor}
\end{align}
for non-negative times by replacing $\switch(t) = 1$ for $t \geq 0$ with $\e^{\eps t}$. The quantity $\widetilde{\sigma}^{\eps}_k[J,\rho](t)$ provides an alternative way to compute the conductivity in the adiabatic limit because the remainder $\delta^{\eps}_k[J,\rho](t)$ vanishes as $\eps \to 0^+$ (\cf Lemma~\ref{Kubo_formula:lem:residual_term_vanishes_in_adiabatic_limit}). 
\begin{proof}[Corollary~\ref{main_results:cor:conductivity_tensor_special_cases}]
	\begin{enumerate}
		\item So let $f_k = 1$. The idea for this proof is to exploit equation~\eqref{framework:eqn:resolvent_equation_integral} which connects the Laplace transform of $\alpha^0_t \bigl ( \partial_{X_k}(\rho) \bigr )$ with the resolvent of $\mathscr{L}_H^{(1)} = - \mathscr{L}_{-H}^{(1)}$, which eliminates the integral over $\tau$. After a change of variables to $\tau \mapsto t - \tau$ in the integral and inserting $\id = (H - \xi) \, (H - \xi)^{-1}$, we obtain 
		\begin{align*}
			\widetilde{\sigma}^{\eps}_k[J,\rho](t) &= - \int_0^{+\infty} \dd \tau \; \e^{\eps (t - \tau)} \; \mathcal{T} \bigl ( J \; \alpha^0_{\tau} \bigl ( \partial_{X_k}(\rho) \bigr ) \bigr ) 
			\\
			&= - \int_0^{+\infty} \dd \tau \; \e^{\eps (t - \tau)} \; \mathcal{T} \left ( J \, \frac{1}{H - \xi} \; \alpha^0_{\tau} \bigl ( (H - \xi) \, \partial_{X_k}(\rho) \bigr ) \right ) 
			\\
			&= - \e^{\eps t} \; \mathcal{T} \left ( \left ( J \, \frac{1}{H - \xi} \right ) \; \frac{1}{\mathscr{L}^{(1)}_H + \eps} \bigl ( (H - \xi) \, \partial_{X_k}(\rho) \bigr ) \right ) 
			. 
		\end{align*}
		To make sure we have dotted all the i's and and crossed all the t's, we remark that we have used three technical facts: (1) $\mathfrak{L}^r(\Alg)$ possesses an extended left module structure (\cf Lemma~\ref{framework:lem:extension_algebra_unbounded_operators}). (2) $\mathscr{L}_H^{(1)} \bigl ( g(H) \, A \bigr ) = g(H) \, \mathscr{L}_H^{(1)}(A)$ holds for any bounded function $g$ of $H$ and any $A \in \mathfrak{L}^1(\Alg)$. And (3) the condition $\domain(H) \subseteq \domain(J)$ ensures that the two quantities inside the trace $\mathcal{T}$ agree on a $\mathcal{T}$-dense domain and so as elements of $\mathfrak{L}^1(\Alg)$.
		\item This is just the result of a straight-forward computation, adapting the arguments from (1), and using the Fubini-Tonelli theorem. 
	\end{enumerate}
\end{proof}
%

\section{The adiabatic limit of the conductivity tensor} 
\label{Kubo_formula:adiabatic_limit}
The last step in making LRT rigorous is an adiabatic limit where the time scale of the intrinsic, microscopic dynamics is infinitely fast compared to the speed at which we ramp up the perturbation. As we have explained in the previous subsection, we can replace the conductivity tensor with a different quantity, $\widetilde{\sigma}^{\eps}_k[J,\rho](t)$ defined via \eqref{Kubo_formula:eqn:splitting_conductivity_tensor}, that is more amenable to an adiabatic limit because we can use \eqref{framework:eqn:resolvent_equation_integral} to integrate over time and obtain a projection. The difference between the first principles conductivity tensor and $\widetilde{\sigma}^{\eps}_k[J,\rho](t)$ vanishes in the adiabatic limit. 
\begin{lemma}\label{Kubo_formula:lem:residual_term_vanishes_in_adiabatic_limit}
	Suppose Hypotheses~\ref{hypothesis:trace}–\ref{hypothesis:state} hold true. Then for all $t \geq 0$ the remainder term 
	\begin{align*}
		\lim_{\eps \to 0^+} \delta^{\eps}_k[J,\rho](t) = 0 
	\end{align*}
	vanishes in the adiabatic limit, and therefore 
	\begin{align*}
		\lim_{\eps \to 0^+} \sigma^{\eps}_k[J,\rho](t) = \lim_{\eps \to 0^+} \widetilde{\sigma}^{\eps}_k[J,\rho](t) 
	\end{align*}
	holds, provided these limits exist. 
\end{lemma}
\begin{proof}
	This follows from estimating $\bnorm{J \; \alpha^{0}_{\tau} \bigl ( \partial_{X_j}(\rho) \bigr )}_1$ just like in the proof of equation~\eqref{Kubo_formula:eqn:norm_estimate_J_alpha_rho}, using that the $f_k \in C(\R)$ are bounded on bounded subsets of $\R$, and Dominated Convergence. 
\end{proof}
Let us proceed to prove the adiabatic limit of the Kubo formula: 
\begin{proof}[Theorem~\ref{main_results:thm:adiabatic_limit_Kubo_formula}]
	Due to the assumption $\partial_{X_k}(\rho) \in \mathfrak{L}^p(\Alg)$ (Hypothesis~\ref{hypothesis:state}~(i)), we can replace $\mathscr{L}^{(1)}_H$ by $\mathscr{L}^{(p)}_H$ in equation~\eqref{main_results:eqn:kubo_formula_f_eq_1}, and combine that with 
	\begin{align*}
		\mathcal{T} \left ( J \; \frac{1}{\mathscr{L}_H^{(p)}+\eps}(R) \right ) = \mathcal{T} \left ( \frac{1}{-\mathscr{L}_H^{(q)} + \eps}(J) \; R \right ) 
		, 
	\end{align*}
	because the trace $\mathcal{T}$ establishes a duality relation between $\mathfrak{L}^q(\Alg)$ and $\mathfrak{L}^p(\Alg)$ and $\mathfrak{L}^{(q)}_{-H} = - \mathfrak{L}^{(q)}_H$ is adjoint of $\mathfrak{L}^{(p)}_H$. Applying the above to the special case $J = \mathscr{L}^{(q)}_H(Q_J)$ and $R = \partial_{X_k}(\rho)$ yields 
	\begin{align}
		\widetilde{\sigma}^{\eps}_k[J,\rho](t) = + \e^{\eps t} \;
		\mathcal{T} \left ( \frac{\mathscr{L}_H^{(q)}}{\mathscr{L}_H^{(q)} - \eps}(Q_J) \; \partial_{X_k}(\rho) \right ) 
 		.
		\label{Kubo_formula:eqn:Kubo_formula_pre_projection}
	\end{align}
	Now to the adiabatic limit: Exploiting the Hölder inequality~\eqref{framework:eqn:trace_Lp_Lq_estimate}, we obtain equation~\eqref{main_results:eqn:adiabatic_limit_conductivity_tensor} after pulling the limit $\eps \to 0^+$ into the argument of $\mathcal{T}$ and
	invoking Lemma~\ref{dynamics:lem:limit_projection_pneq2} (or Lemma~\ref{dynamics:lem:limit_projection_p=2} if $p = 2$). 
	
	For the special case $p = q = 2$, the duality relation can be rephrased in terms of the scalar product $\sdscpro{A}{B}_{\mathfrak{L}^2} = \mathcal{T}(A^* \, B)$ and rewriting equation~\eqref{main_results:eqn:adiabatic_limit_conductivity_tensor} in this fashion yields \eqref{main_results:eqn:adiabatic_limit_L2_conductivity_tensor}. 
\end{proof}
The Kubo-Str\v{e}da formula is a special case of the Kubo formula we have just proven where 
\begin{enumerate}
	\item the current is density current $J_k = \ii [ H , X_k ]$ defined in terms of the $X_k$, 
	\item the state $\rho = P = \chi_{\Delta}(H)$ is a spectral projection of $H$ associated to a bounded energy region $\Delta$, and lastly, 
	\item the regularity index of the projection is $p = 2$. 
\end{enumerate}
Put another way, the main objects of interests are the \emph{conductivity coefficients} 
\begin{align}
	\sigma_{jk} :& \negmedspace= \lim_{\eps \to 0^+} \widetilde{\sigma}^{\eps}_j[J_k , P](t)
	= \lim_{\eps \to 0^+}  \int_0^{\infty} \dd \tau \; \e^{\eps (t - \tau)} \; \mathcal{T} \bigl ( J_k \, \alpha^0_{\tau} \bigl ( \partial_{X_j}(P) \bigr ) \bigr )
	\notag \\
	&
	=: \lim_{\eps \to 0^+} \int_0^{\infty} \dd \tau \; \e^{\eps (t - \tau)} \; \Sigma_{jk}(\tau)
	. 
	\label{Kubo_formula:eqn:definition_Kubo_Streda_coefficients}
\end{align}
We have already provided a proof that these limits exist and are independent of time. Specializing the Kubo formula opens the door for a host of clever manipulations such as shuffling (iterated) commutators around. 
\begin{proof}[Theorem~\ref{main_results:thm:Kubo_Streda_formula}]
	Under the Hypotheses we have already shown the existence of the $\sigma_{jk}$, and that they are given by  equation~\eqref{Kubo_formula:eqn:definition_Kubo_Streda_coefficients}. Even though the computations are straightforward, some of the technical arguments which are necessary to dot the i's and cross the t's are unfortunately somewhat tedious, because we will have to consider operators as elements of different algebras. 
	
	Furthermore, to make the commutator manipulations rigorous, we will have to regularize the commutators by introducing a smoothened energy cutoff: Consider a family $f_n \in C^{\infty}_{\mathrm{c}}(\R)$ of non-negative, smooth and compactly supported functions with the properties $f_n(\lambda) \leq 1$ and $f_n(\lambda) = 1$ for all $\lambda \in [-n,+n]$. By functional calculus we deduce $f_n(H) \in \Alg^+$, $\bnorm{f_n(H)} \leqslant 1$, and $\slim_{n \to \infty} f_n(H) = \id$ holds. Moreover, the standard trick of inserting $\id = (H - \xi)^{-1} \, (H - \xi)$ in conjunction with Lemma~\ref{Kubo_formula:lem:Lp_regularity_states} allows us to deduce that also 
	\begin{align*}
		J_k \, f_n(H) &= J_k \, \frac{1}{H - \xi} \; (H - \xi) \, f_n(H) 
		\in \Alg
	\end{align*}
	is the product of two algebra elements, and therefore itself an element of $\Alg$. 
	
	Now we replace the affiliated operator $J_k \in \affil(\Alg)$ with the algebra element $J_k \, f_n(H) \in \Alg$ in the $\Sigma_{jk}(\tau)$ defined in equation~\eqref{Kubo_formula:eqn:definition_Kubo_Streda_coefficients}, 
	\begin{align*}
		\Sigma_{jk,n}(\tau) := \mathcal{T} \Bigl ( J_k \, f_n(H) \; \alpha^0_{\tau} \bigl ( \partial_{X_j}(P) \bigr ) \Bigr )
		, 
	\end{align*}
	and then eventually the limit $n \to \infty$ \emph{after} making the necessary manipulations. That taking the limit $\lim_{n \to \infty} \Sigma_{jk,n}(\tau) = \Sigma_{jk}(\tau)$ again eliminates the regularization follows from the convergence of traces of products proven in Lemma~\ref{framework:lem:strong_convergence_trace_product}. 
	
	The purpose of the regularization procedure is that it allows us to consider usual commutators in $\Alg$ or $\mathfrak{L}^1(\Alg)$, denoted with $[ \, \cdot \, , \, \cdot \, ]$ and $[ \, \cdot \, , \, \cdot \, ]_{(1)}$, respectively, because $\partial_{X_k}(P)$ is an element of the trace ideal $\mathfrak{L}^1(\Alg)$ by Hypothesis~\ref{hypothesis:state}~(i). Consequently, we can rewrite the $\Sigma_{j k , n}$ exploiting the cyclicity of the trace and Lemma~\ref{framework:lem:double_commutator_identity} for $p = 1$, 
	\begin{align*}
		\Sigma_{jk,n}(\tau) &= \mathcal{T} \Bigl ( \alpha^{0}_{-\tau} \bigl ( J_k \, f_n(H) \bigr ) \; \partial_{X_j}(P) \Bigr ) 
		\\
		&= \mathcal{T} \left ( \alpha^{0}_{-\tau} \bigl ( J_k \, f_n(H) \bigr ) \; \Bigl [ \, P \, , \, \bigl [ P \, , \, \partial_{X_j}(P) \bigr ]_{(1)} \Bigr ]_{(1)} \right ) 
		, 
	\end{align*}
	the same commutator identity~\eqref{framework:eqn:trace_product_with_commutator_switch} for $p = \infty$, 
	\begin{align*}
		\ldots &= \mathcal{T} \Bigl ( \big[\alpha^{0}_{-\tau} \bigl ( J_k \, f_n(H) \bigr ) \, , \, P \, \bigr ] \; \bigl [ \, P \, , \, \partial_{X_j}(P) \bigr ]_{(1)} \Bigr ) 
		\\
		&= \mathcal{T} \left ( \Bigl [ \bigl [ \alpha^{0}_{-\tau} \bigl ( J_k \, f_n(H) \bigr ) \, , \, P \bigr ] \, , \, P \Bigr ] \; \partial_{X_j}(P) \right )
		, 
	\end{align*}
	and the equilibrium condition $\alpha^0_{\tau}(P) = P$, 
	\begin{align}
		\ldots &= \mathcal{T} \left ( 
		\alpha^{0}_{-\tau} \left ( \Bigl [ \, P \, , \, \bigl [ P \, , \, J_k \, f_n(H) \bigr ] \Bigr ] \right ) \, \partial_{X_j}(P) \right ) 
		\notag \\
		&= \mathcal{T} \left ( 
		\Bigl [ \, P \, , \, \bigl [ P \, , \, J_k \, f_n(H) \bigr ] \Bigr ] \, \alpha^{0}_{\tau} \bigl ( \partial_{X_j}(P) \bigr ) \right ) 
		. 
		\label{Kubo_formula:eqn:Streda_formula_via_auxiliary_limit}
	\end{align}
	The next step is to pull $f_n(H)$ out of the commutator via the identity
	\begin{align}
		\bigl [ P \, , \, J_k \, f_n(H) \bigr ] = - \bigl [ H \, , \, \partial_{X_k}(P) \bigr ]_{\ddagger}\; f_n(H)
		\label{Kubo_formula:eqn:Streda_formula_commutator_identity}
	\end{align}
	where $[\, \cdot \, , \, \cdot \, ]_{\ddagger}$ is the generalized commutator from Definition~\ref{framework:defn:generalized_commutators}: The equilibrium condition implies that $P$ commutes with $H$ and $(H \, P)^* = H \, P$ in view of the fact that both $H$ and $P$ are selfadjoint. Since the derivations $\partial_{X_k}$ are selfadjoint we conclude from Lemma~\ref{Kubo_formula:lem:Lp_regularity_states}~(3) that
	\begin{align*}
		\bigl [ H \, , \, \partial_{X_k}(P) \bigr ]_{\ddagger} \; f_n(H) &= \bigl ( J_k \, P - (J_k \, P)^* \bigr ) \; f_n(H)\\
		&= \bigl ( J_k \, P \, f_n(H) \bigr ) - \bigl ( f_n(H) \, J_k \, P \bigr )^*
		\\
		&= \bigl ( J_k \, f_n(H) \bigr ) \, P - P \, \bigl ( f_n(H) \, J_k \bigr )^*
		\\
		&= \bigl ( J_k \, f_n(H) \bigr ) \, P - P \, \bigl ( J_k \, f_n(H) \bigr ) 
	\end{align*}
	which is exactly the relation \eqref{Kubo_formula:eqn:Streda_formula_commutator_identity}. Note that these arguments make sense in $\mathfrak{L}^1(\Alg)$ \emph{and} $\mathfrak{L}^2(\Alg)$, and depending on our needs, we will replace $[ \, \cdot \, , \, \cdot \, ]_{(1)}$ with $[ \, \cdot \, , \, \cdot \, ]_{(2)}$. 
	
	It is now that we will need the $p = 2$ regularity of $\rho$ and $H \, \rho$, and regard commutators in $\mathfrak{L}^2(\Alg)$. We plug this commutator identity into \eqref{Kubo_formula:eqn:Streda_formula_via_auxiliary_limit}, use once again that $P$ and $f_n(H)$ commute to compute the limit 
	\begin{align*}
		\Sigma_{jk} &= - \lim_{n \to \infty} \; \mathcal{T} \left ( 
		\Bigl [ \, P \, , \, \bigl [ H \, , \, \partial_{X_k}(P) \bigr ]_{\ddagger} \Bigr ]_{(2)} \, f_n(H) \, \alpha^{0}_{\tau} \bigl ( \partial_{X_j}(P) \bigr ) \right ) 
		\\
		&= - \mathcal{T} \left ( 
		\Bigl [ \, P \, , \, \bigl [ H \, , \, \partial_{X_k}(P) \bigr ]_{\ddagger} \Bigr ]_{(2)} \, \alpha^{0}_{\tau} \bigl ( \partial_{X_j}(P) \bigr ) \right )
		. 
	\end{align*}
	To obtain the commutator of $P$ with its derivative, we pull out $H$ via Lemma~\ref{framework:lem:extension_algebra_unbounded_operators}~(2), 
	\begin{align*}
		P \, \bigl ( H \, \partial_{X_k}(P) \bigr ) &= \bigl ( P \, H \bigr ) \, \partial_{X_k}(P) 
		= H \, \bigl ( P \, \partial_{X_k}(P) \bigr ) 
		\in \mathfrak{L}^2(\Alg) 
		, 
	\end{align*}
	and an analogous equality for the adjoint; we placed the brackets for emphasis, \eg $P$ and $P \, H = H \, P$ are elements in $\Alg$ because $P$ is a spectral projection associated to a bounded region of the spectrum whereas $\partial_{X_k}(P) , H \, \partial_{X_k}(P) \in \mathfrak{L}^2(\Alg)$ by Hypothesis~\ref{hypothesis:state}. 
	
	As $\mathfrak{L}^2(\Alg)$ is self-dual this computation allows us to switch the order of generalized and $\mathfrak{L}^2(\Alg)$ commutators, and express the double commutator in terms of the $2$-Liouvillian,  
	\begin{align*}
		\Bigl [ \, P \, , \, \bigl [ H \, , \, \partial_{X_j}(P) \bigr ]_{\ddagger} \Bigr ]_{(2)} = \Bigl [ \, H \, , \, \bigl [ P \, , \, \partial_{X_j}(P) \bigr ]_{(2)} \Bigr ]_{\ddagger}
		= \ii \, \mathscr{L}^{(2)}_H \Bigl ( \bigl [ P \, , \, \partial_{X_k}(P) \bigr ]_{(2)} \Bigr ) 
		. 
	\end{align*}
	Therefore, we have reduced the form of the conductivity tensor to \eqref{Kubo_formula:eqn:Kubo_formula_pre_projection}, and we can repeat the arguments following \eqref{Kubo_formula:eqn:Kubo_formula_pre_projection} onwards to conclude 
	\begin{align*}
		\sigma_{jk} = \ii \Bdscpro{\mathscr{P}_H^{(2)\perp} \Bigl ( \bigl [ P \, , \, \partial_{X_k}(P) \bigr ]_{(2)} \Bigr )^* \; }{ \, \partial_{X_j}(P)}_{\mathfrak{L}^2}
		. 
	\end{align*}
	It remains to show that $\bigl [ P \, , \, \partial_{X_k}(P) \bigr ]_{(2)}$ is orthogonal to the kernel of $\mathscr{P}_H^{(2)}$, because then we can omit the projection $\mathscr{P}_H^{(2)\perp}$. Fortunately, this is easy to see: $A \in \ker \bigl ( \mathscr{L}^{(2)}_H \bigr )$ implies $[A \, , \, f(H)]_{(2)} = 0$ for $f \in L^{\infty}(\R)$. In particular this holds true for $P = \chi_{\Delta}(H)$ where $\Delta$ is the spectral region of interest, and with that piece of information it is easy to see that such $A$ are necessarily orthogonal, 
	\begin{align*}
		\Bdscpro{A}{\bigl [ P \, , \, \partial_{X_k}(P) \bigr ]_{(2)}}_{\mathfrak{L}^2} = \Bdscpro{[A,P]_{(2)} \;}{\, \partial_{X_k}(P)}_{\mathfrak{L}^2} 
		= 0
	\end{align*}
	where in the first equality we used the formula~\eqref{framework:eqn:trace_product_with_commutator_switch}. Hence, we can leave out the projection, and obtain the Kubo-Str\v{e}da formula, 
	\begin{align*}
		\sigma_{kj} &= - \ii \, \Bdscpro{\bigl [P \, , \, \partial_{X_k}(P) \bigr ]_{(2)}^* \;}{\, \partial_{X_j}(P) }_{\mathfrak{L}^2} 
		\notag \\
		&= + \ii \, \Bdscpro{\bigl [ \, P \, , \, \partial_{X_k}(P) \bigr ]_{(2)} \;}{\, \partial_{X_j}(P) }_{\mathfrak{L}^2}
		. 
	\end{align*}
	Shifting the commutators between the arguments with the help of equation~\eqref{framework:eqn:trace_product_with_commutator_switch} yields the alternative form for the conductivity coefficients, 
	\begin{align*}
		\sigma_{kj} &= - \ii \, \mathcal{T} \Bigl ( \bigl [P \, , \, \partial_{X_k}(P) \bigr ]_{(2)} \; \partial_{X_j}(P) \Bigr )
		\\
		&= - \ii \, \mathcal{T} \Bigl ( P \, \bigl [ \partial_{X_k}(P) \, , \, \partial_{X_j}(P) \bigr ]_{(1)} \Bigr ) 
		. 
	\end{align*}
	This concludes the proof. 
\end{proof}
%
%
%
\chapter{Applications} 
\label{applications}
The framework fleshed out in Chapter~\ref{unified} applies directly to the two most common cases, namely $\mathbb{G} = \Z^d$ and $\R^d$, which describe discrete (tight-binding) models and continuum systems. One of the main motivations to write this book has been to develop a framework for LRT that applies directly to continuum models. However, so far we still owe it to the reader to point to a single \emph{new} application which goes beyond existing results in the literature.

\section{Linear response theory for periodic and random light conductors} 
\label{applications:Maxwell}
About 10 years ago Raghu and Haldane \cite{Raghu_Haldane:quantum_Hall_effect_photonic_crystals:2008} proposed an analog of the Quantum Hall Effect in periodic \emph{light} conductors. Their seminal work kickstarted the search for topological effects in classical and bosonic waves, giving birth to several highly active sub fields in the process. Not only have topological effects been experimentally confirmed in electromagnetic waves \cite{Wang_et_al:unidirectional_backscattering_photonic_crystal:2009}, but also in coupled mechanical oscillators \cite{Suesstrunk_Huber:observation_topological_mechanical_edge_modes:2015,Suesstrunk_Huber:classification_mechanical_metamaterials:2016}, periodic waveguide arrays \cite{Rechtsman_Zeuner_et_al:photonic_topological_insulators:2013} and acoustic waveguides \cite{Xiao_et_al:geometric_phase_acoustic_systems:2015}.

Giving a derivation of such topological effects, starting from the corresponding fundamental equations, is an open problem, both, from the vantage point of theoretical and mathematical physics. Due to the similarity to the Quantum Hall Effect, and the role LRT has played in providing a first principles explanation for the quantization of the transverse conductivity, we think establishing LRT for periodic and random light conductors will similarly yield insights into the inner workings of topological effects for light. Implementing it in its entirety, though, is beyond the scope of this Chapter and we postpone it to a future work \cite{DeNittis_Lein:LRT_light:2017}. First, we need to connect Maxwell's equations in matter that describe the propagation of classical electromagnetic waves in media to the Schrödinger equation.

\subsection{Schrödinger formalism of electromagnetism} 
\label{applications:Maxwell:Schrodinger_formalism}
What those classical waves have in common mathematically, is that their dynamical equation can be recast in the form of a Schrödinger equation
\begin{align}
	\ii \partial_t \Psi(t) = M \Psi(t)
	,
	&&
	\Psi(0) = \Phi \in \Hil
	,
	\label{applications:eqn:Maxwell_Schroedinger_equation}
\end{align}
where the \emph{Maxwell-type operator} is of the form $M = W^{-1} \, D$ and selfadjoint on a closed subspace of $L^2_W(\R^d,\C^N)$ space (so that $\Psi$ satisfies \eg transversality conditions). The weight $W$ satisfies point~(ii) of Definition~\ref{unified:defn:random_weights}, meaning $W \in L^{\infty}(\R^d) \otimes \mathrm{Mat}_{\C}(N)$ takes values in the hermitian matrices, is bounded and has a bounded inverse. A second requirement is that \eqref{applications:eqn:Maxwell_Schroedinger_equation} supports real solutions: in contrast to complex quantum wave functions classical fields such as electromagnetic or acoustic waves are real, and this has to be reflected in the dynamical equation and the space $\Hil$ it acts on. A detailed discussion for electromagnetic waves can be found in \cite[Section~2.2–2.3]{DeNittis_Lein:ray_optics_photonic_crystals:2014}.

Let us now be specific and consider the case of Maxwell's equations in media, neglecting randomness for the moment. The fundamental equations are
\begin{subequations}
	\label{applications:eqn:Maxwell}
	\begin{align}
		\left (
		\begin{matrix}
			\eps & \; 0 \\
			0 & \; \mu \\
		\end{matrix}
		\right )
		\, \frac{\dd}{\dd t} \left (
		\begin{matrix}
			\mathbf{E} \\
			\mathbf{H} \\
		\end{matrix}
		\right ) &= \left (
		\begin{matrix}
			+ \nabla \times \mathbf{H} \\
			- \nabla \times \mathbf{E} \\
		\end{matrix}
		\right )
		,
		&& \mbox{(dynamical eqns.)}
		\label{applications:eqn:dynamical_Maxwell}
		\\
		\left (
		\begin{matrix}
			\nabla \cdot \eps \mathbf{E} \\
			\nabla \cdot \mu \mathbf{H} \\
		\end{matrix}
		\right ) &= 0
		,
		&& \mbox{(no sources eqns.)}
		\label{applications:eqn:source_Maxwell}
	\end{align}
\end{subequations}
where the material weights
\begin{align*}
	W = \left (
	\begin{matrix}
		\eps & \; 0 \\
		0 & \; \mu \\
	\end{matrix}
	\right ) \in L^{\infty}(\R^3) \otimes \mathrm{Mat}_{\R}(6)
\end{align*}
take values in the hermitian $6 \times 6$ matrices, composed of the $3 \times 3$ blocks $\eps$, the electric permittivity, and $\mu$, the magnetic permeability. To avoid (manageable) complications we furthermore suppose that the entries of $W$ are all \emph{real}, \ie the medium is non-gyrotropic; All subsequent arguments can be adapted to the gyrotropic case where $W$ has a non-zero imaginary part \cite[Section~2.2–2.3]{DeNittis_Lein:ray_optics_photonic_crystals:2014}.

To derive \eqref{applications:eqn:Maxwell_Schroedinger_equation} from the dynamical Maxwell equations, we multiply both sides of \eqref{applications:eqn:dynamical_Maxwell} with the bounded multiplication operator $\ii \, W^{-1}$, and obtain a Schrödinger-type equation where the Maxwell operator
\begin{align}
	M = W^{-1} \, \mathrm{Rot}
	:= \left (
	\begin{matrix}
		\eps^{-1} & 0 \\
		0 & \mu^{-1} \\
	\end{matrix}
	\right ) \left (
	\begin{matrix}
		0 & + \ii \nabla^{\times} \\
		- \ii \nabla^{\times} & 0 \\
	\end{matrix}
	\right )
	\label{applications:eqn:definition_Maxwell_operator}
\end{align}
plays the role of the Hamiltonian. Physical initial conditions $(\mathbf{E},\mathbf{H})$ are real and must be transversal, \ie satisfy the divergence-free condition~\eqref{applications:eqn:source_Maxwell}, a property which is preserved under the dynamics. It is straightforward to prove that $M$, endowed with the domain of the free Maxwell operator $\mathrm{Rot}$, defines a selfadjoint operator on the \emph{weighted} $L^2$-space $\Hil_W := L^2_W(\R^3,\C^6)$. The latter is defined as $\Hil_\ast :=L^2(\R^3,\C^6)$, seen as a \emph{Banach} space, endowed with the scalar product $\sscpro{\Phi}{\Psi}_W := \scpro{\Phi}{W \, \Psi}_{L^2(\R^3,\C^6)}$. While the link between Maxwell's equations in vacuum and the Schrödinger equation was known since the inception of modern quantum mechanics itself (we found the first systematic explanation in Wigner's seminal paper on representations of the Poincaré group \cite[pp.~151 and 198]{Wigner:representations_Lorentz_group:1939}, although Wigner attributes this insight to Dirac). The first mathematical treatise for Maxwell's equations \emph{in matter} is due to Birman and Solomyak \cite{Birman_Solomyak:L2_theory_Maxwell_operator:1987}. We have explored this further in a number of recent papers \cite{DeNittis_Lein:adiabatic_periodic_Maxwell_PsiDO:2013,DeNittis_Lein:sapt_photonic_crystals:2013,DeNittis_Lein:symmetries_Maxwell:2014,DeNittis_Lein:ray_optics_photonic_crystals:2014}, and gained a better understanding of how the reality of electromagnetic waves is represented in the Schrödinger formalism — a fact that will become important in the discussion of Hypothesis~\ref{hypothesis:current}.

One of our insights in \cite[Section~3.3.2]{DeNittis_Lein:ray_optics_photonic_crystals:2014} was that the expectation value of the current operator
\begin{align*}
	J_k = \ii \, [M , X_k]
\end{align*}
can be linked to averages of the Poynting vector $\mathcal{P}(x) = \mathbf{E}(x) \times \mathbf{H}(x)$ via
\begin{align}
	\bscpro{(\mathbf{E},\mathbf{H}) \,}{\, J_k (\mathbf{E},\mathbf{H})}_W = \int_{\R^3} \dd x \; \mathcal{P}_k(x)
	,
	\label{applications:eqn:relation_expectation_value_current_average_Poynting}
\end{align}
(here we are using the reality of the fields $\mathbf{E}$ and $\mathbf{H}$) and therefore the same operators enter the analysis as in the quantum case. While that fact in isolation may seem curious, there is a deeper reason for this: Observables in quantum mechanics, typically selfadjoint operators affiliated to a von Neumann algebra, are \emph{conceptually different} from electromagnetic observables, \emph{functionals of the fields} \cite[Section~3]{DeNittis_Lein:ray_optics_photonic_crystals:2014}. Quadratic functionals can then be frequently written as an expectation value (see also \cite[Section~3.3]{Bliokh_Bekshaev_Nori:dual_electromagnetism:2013}). That distinction may seem pedantic, but is crucial for a physically meaningful interpretation of mathematical statements.

Mathematically, the structure of Maxwell-type operators is peculiar, because perturbations are \emph{multiplicative} rather than additive, \ie
\begin{align*}
	\mathbf{M}_S := S^{-2} \, M
\end{align*}
which acts on a \emph{different} Hilbert space that is weighted by $S^2 \, W$ instead of $W$. Here, we of course assume that $S^2 \, W$ satisfies the same assumptions as $W$ itself, and that $S$ is a bounded multiplication operator with a bounded inverse in its own right. Previously, operators of the form
\begin{align}
	S(x) = \left (
	\begin{matrix}
		\tau_{\eps}(x)^{-1} \, \id_{\C^3} & 0 \\
		0 & \tau_{\mu}(x)^{-1} \, \id_{\C^3} \\
	\end{matrix}
	\right )
	\label{applications:eqn:standard_PhC_perturbation}
\end{align}
for some suitable functions $\tau_{\eps} , \tau_{\mu} : \R^3 \longrightarrow \R$ have been considered in the literature \cite{Raghu_Haldane:quantum_Hall_effect_photonic_crystals:2008,Onoda_Murakami_Nagaosa:geometrics_optical_wave-packets:2006,DeNittis_Lein:sapt_photonic_crystals:2013}, including Haldane's original paper.

Conjugating $\mathbf{M}_S$ with $S$ yields
\begin{align}
	M_S := S \; \mathbf{M}_S \, S^{-1}
	= S^{-1} M \, S^{-1}
	,
	\label{applications:eqn:other_representation_perturbed_Maxwell}
\end{align}
an operator which is selfadjoint on the \emph{unperturbed} Hilbert space $L^2_W(\R^3,\C^6)$ \cite[Section~2.2]{DeNittis_Lein:adiabatic_periodic_Maxwell_PsiDO:2013}, because $S$, seen as a map between $L^2_{S^2 W}(\R^3,\C^6)$ and $L^2_W(\R^3,\C^6)$, is unitary. Note, thought, that this is subtly different from \eqref{main_results:eqn:definition_isospectrally_perturbed_hamiltonian} since $S^{-1}$ appears to the left \emph{and} to the right of $M$. Similar considerations hold for many other classical wave equations.

A second class of perturbations $M_G := G \, M \, G^{-1}$ which \emph{does} fit our LRT scheme is obtained from a time-dependent unitary of the form
\begin{align}
	G(t) := \e^{- \ii \sum_{j = 1}^3 \Phi_j(t) \, X_j}
	,
	\label{applications:eqn:perturbation_unitary_Maxwell}
\end{align}
defined in terms of the positions operators $X_1$, $X_2$ and $X_3$, which at time $t$ imparts electromagnetic waves with momentum $\bigl ( \Phi_1(t) , \Phi_2(t) , \Phi_3(t) \bigr )$.

\subsection{Random media} 
\label{applications:Maxwell:random_media}
Randomly distributed ensembles of media are treated within the framework developed in Chapter~\ref{unified}: suppose we are given a probability space $(\Omega , \mathbb{P})$ encoding the randomness, an ergodic $\R^3$-action $\tau$ on $\Omega$, and $\widehat{W} = \{ W_{\omega} \}_{\omega \in \Omega}$ a field of weights in the sense of Definition~\ref{unified:defn:random_weights} that collectively describe the statistical variations of $\eps$ and $\mu$; Examples of such random Maxwell operators have been studied previously by Figotin and Klein in the late 1990s \cite{Figotin_Klein:localization_classical_waves_I:1996,Figotin_Klein:localization_classical_waves_II:1997}. The relevant Hilbert space that treats all configurations simultaneously, weighted by their probabilities, is the direct integral
\begin{align*}
	\widehat{\Hil} := \int_{\Omega}^{\oplus} \dd \mathbb{P}(\omega) \, \Hil_{\omega}
\end{align*}
where $\Hil_{\omega} := L^2_{W_{\omega}}(\R^3,\C^6)$ is the $L^2$-space weighted by $W_{\omega}$. The collection of fiber Maxwell operators $M_{\omega} := W_{\omega}^{-1} \, \mathrm{Rot}$ is denoted with
\begin{align*}
	\widehat{M} := \int_{\Omega}^{\oplus} \dd \mathbb{P}(\omega) \, M_{\omega}
	,
\end{align*}
endowed with the constant fiber domain $\domain(\widehat{M}) = \int_{\Omega}^{\oplus} \dd \mathbb{P}(\omega) \, \domain(\mathrm{Rot})$. Spatial translations of electromagnetic fields
\begin{align*}
	\bigl ( \widehat{U}_y \widehat{\varphi} \bigr )_{\tau_y(\omega)}(x) = \varphi_{\omega}(x - y)
\end{align*}
necessarily map real fields onto real fields, and therefore the $2$-cocycle twist $\Theta = \id$ is absent. Due to all of this, we conclude that Proposition~\ref{unified:prop:affiliation} applies, and $\widehat{M}$ is affiliated to the von Neumann algebra
\begin{align*}
	\Alg(\Omega , \mathbb{P} , \R^3 , \id) = \mathrm{Span}_{\R^3} \bigl \{ \widehat{U}_g \bigr \} ' \cap \mathrm{Rand}(\widehat{\Hil})
	.
\end{align*}
This algebra is endowed with the trace per unit volume $\mathcal{T}_{\mathbb{P}}$, for which an explicit formula is given in Proposition~\ref{unified:prop:trace_per_unit_volume} and in equation~\eqref{unified:eqn:trace_per_unit_volume_as_infinite_volume_limit}. Therefore, Hypotheses~\ref{hypothesis:trace} and \ref{hypothesis:Hamiltonian} are satisfied. Moreover, the usual position operators $X_j$, $X_2$ and $X_3$ are compatible with $\mathcal{T}_{\mathbb{P}}$ (Hypothesis~\ref{hypothesis:generators}), something we have proven in Propostion~\ref{unified:prop:compatibility_generators}. Hypothesis~\ref{hypothesis:perturbation} specifies generic properties the adiabatic switching between the unperturbed and the perturbed configuration must have, but does not place any additional technical restrictions in our situation.

Now on to the assumptions we impose on the current observables (Hypothesis~\ref{hypothesis:current}): As pointed out in the discussion of \eqref{applications:eqn:relation_expectation_value_current_average_Poynting} the current observable is the expectation value of the current operator $\widehat{J}_k := \bigl \{ J_{k , \omega} \bigr \}_{\omega \in \Omega}$ with
\begin{align*}
	J_{k , \omega} = + \ii \, \bigl [ M_{\omega} , X_k \bigr ]
	.
\end{align*}
Note that the $J_{k,\omega}$ are \emph{bounded} selfadjoint multiplication operators, and because $\mathrm{Rot}$ is a first-order differential operator, higher-order commutators necessarily vanish identically. A suitable localizing domain $\domain_{\mathrm{c}}$ consists of $L^2$-functions with compact support (\ie outside of some compact they are zero almost everywhere with respect to the Lebesgue measure). The joint core, the intersection $\domain_{\mathrm{c}}(M_{\omega}) = \domain_{\mathrm{c}} \cap \domain(\mathrm{Rot})$ is in fact independent of $\omega$ because $W_{\omega}^{-1}$, as a bounded multiplication operator, leaves $\domain_{\mathrm{c}}$ invariant. Therefore, we have shown (i)–(iv). The only assumption left is the existence of a spectral gap. This seems non-trivial, because the spectrum of $M_{\omega}$ is not bounded from below and there need not be any gaps. In fact, the symmetry relation $C \, M_{\omega} \, C = - M_{\omega}$, where $C$ is the complex conjugation, ensures the spectrum is point symmetric. However, it turns out we can obtain an equivalent description by restricting ourselves to spectral subspace where $\omega > 0$ whose dynamics are governed by $\widetilde{M}_{\omega} := M_{\omega} \, 1_{(0,\infty)}(M_{\omega}) \vert_{\omega > 0}$ \cite[Section~2.3]{DeNittis_Lein:ray_optics_photonic_crystals:2014}. That is because all electromagnetic fields are \emph{real}, and therefore the positive and negative frequency contributions $\Psi_{\pm} := 1_{(0,\infty)}(\pm M_{\omega}) \, (\mathbf{E},\mathbf{H})$ are necessarily related by complex conjugation, $\Psi_- = \overline{\Psi_+}$. Evidently, $\widetilde{M}_{\omega}$ is bounded from below, and condition~(v) is generically satisfied.

From a physical point of view, only Hypothesis~\ref{hypothesis:state} remains. Here, the differences between quantum mechanics and classical fields become important. For fermionic quantum systems, the typical assumption is that the state is characterized by a Fermi energy $E_{\mathrm{F}}$ and a temperature $T$. Finite temperature effects are typically seen as perturbations of the zero temperature Fermi projection
\begin{align*}
	\rho = 1_{(-\infty,E_{\mathrm{F}}]}(H)
	,
\end{align*}
where all states below $E_{\mathrm{F}}$ are filled and the rest are completely unoccupied. The standing physical assumption here is that either the Fermi energy either lies in a spectral gap or in a region of Anderson localization.

Albeit mathematically perfectly well-defined, for electromagnetic waves such states do not make much physical sense, because excitations are typically peaked around some $k_0$ and $\omega_0$ in $k$- and in frequency space, respectively. The simplest such example is a laser beam impinging on the surface of a photonic crystal: the frequency $\omega_0$ is that of the laser light while the direction with respect to the surface normal determines $k_0$. Even though antenna may excite light of several frequencies (including higher harmonics) and be omnidirectional, it is still technically unfeasible to populate all states in a too wide frequency range $[a,b]$, \ie excite the state $\rho = 1_{[a,b]}(M)$. Nevertheless, at least in periodic systems, states $\rho \simeq \int_{\mathbb{B}}^{\oplus} \dd k \, \rho(k)$ with well-defined frequency and momentum can be locally written in terms of spectral projections $\rho(k) = f(k) \; 1_{[a(k),b(k)]} \bigl ( M(k) \bigr )$.

In conclusion, Hypotheses~\ref{hypothesis:trace}–\ref{hypothesis:generators} and \ref{hypothesis:current} are naturally satisfied for Maxwell's equations for as long as the material weights satisfy the assumptions enumerated in Definition~\ref{unified:defn:random_weights} — which also happen to be necessary and sufficient conditions for being able to rewrite Maxwell's equations in Schrödinger form. Moreover, Hypothesis~\ref{hypothesis:perturbation} imposes reasonable restrictions on the experimental realization of how to adiabatically switch on the perturbation that drives the current. Only Hypothesis~\ref{hypothesis:state} is not automatically satisfied. Then all of our main results from Chapter~\ref{main_results} apply directly, \eg in the adiabatic limit $\eps \to 0$ the conductivity coefficients are given by the Kubo-Str\v{e}da formula.

\subsection{Open questions} 
\label{applications:Maxwell:open_questions}
Here, the reader may suspect that with a little bit more effort we could have specified a setting and constructed states for which our LRT scheme applies to Maxwell's equations, and conclude this book with a novel application. While we would have liked that, too, we believe that a more in-depth treatment of the subject matter is necessary. Specifically, we would like to make sure we have a sound line of argumentation which imbues the mathematical results with physically meaning. Here are a few issues we would like to clarify in \cite{DeNittis_Lein:LRT_light:2017} and subsequent works:
\begin{enumerate}
	\item \textbf{The physical nature of perturbations via \eqref{applications:eqn:perturbation_unitary_Maxwell}:}
	The operator $G_{\Phi,\eps}(t)$ describes time-dependent boosts applied to the wave's momentum. This does \emph{not} correspond to a moving medium (which is of interest to the physics community) because in that situation the generators are the ordinary derivatives $- \ii \partial_k$. Therefore we need to understand whether perturbations of the form~\eqref{applications:eqn:perturbation_unitary_Maxwell} are physical, and if they are, how one might realize them experimentally.
	\item \textbf{Assumptions on the state (Hypothesis~\ref{hypothesis:state}):}
	This is mainly a technical challenge where we start with a periodic system, construct states focussed around some $(\omega_0,k_0)$ as described above, and try to verify the assumptions enumerated in Hypothesis~\ref{hypothesis:state}. We reckon it will be helpful that the $J_k$ are \emph{bounded} selfadjoint operators, and perhaps even elements of the von Neumann algebra itself (although this has to be checked). We reckon these arguments will straightforwardly extend to the case of weak disorder.
	\item \textbf{Adaptation of our LRT scheme to include sources:}
	This is another instance where classical wave equations and quantum theory differ: in many experiments (\eg those by Wang et al \cite{Wang_et_al:unidirectional_backscattering_photonic_crystal:2009}) an antenna is inserted into the photonic crystal which then excites states inside of it. Mathematically, that means we need to add a source term to Maxwell's equations~\eqref{applications:eqn:Maxwell}. Since also the solution to the \emph{in}homogeneous equation is explicitly known, there is hope to generalize Theorems~\ref{dynamics:thm:Duhamel_formula} and \ref{main_results:thm:comparison_dynamics}. Once this is understood, we expect to be able to rigorously justify formulas for the conductivity coefficients both, for finite $\eps > 0$ and in the adiabatic limit via the strategy outlined in the beginning of Chapter~\ref{Kubo_formula} (Theorems~\ref{main_results:thm:Kubo_formula}, \ref{main_results:thm:adiabatic_limit_Kubo_formula} and \ref{main_results:thm:Kubo_Streda_formula}).
	\item \textbf{Extension to perturbations of the form \eqref{applications:eqn:other_representation_perturbed_Maxwell}:}
	A second, much, much more ambitious goal would be to go beyond perturbations that are conjugations by unitaries compatible with the von Neumann algebra. At least for photonic crystals (where $W$ is periodic) certain simplifying assumptions may lend a helping hand: Suppose the functions $\tau_{\eps}$ and $\tau_{\mu}$ which enter the perturbation~\eqref{applications:eqn:standard_PhC_perturbation} are bounded away from $0$ and $\infty$, and have bounded first-order derivatives. Then we can recast $M_S$ from equation~\eqref{applications:eqn:other_representation_perturbed_Maxwell} into the form
	\begin{align*}
		M_S = A \, M + B
	\end{align*}
	by commuting $M$ and $S^{-1}$ in \eqref{applications:eqn:other_representation_perturbed_Maxwell}. Here, $A$, $A^{-1}$ and $B$ are explicitly computably, \emph{bounded} operators.

	Should the perturbation $S$ be in addition slowly varying and $C^{\infty}_{\mathrm{b}}$-regular, then we have already given a derivation for effective dynamics in \cite{DeNittis_Lein:sapt_photonic_crystals:2013} based \emph{pseudodifferential} arguments via space-adiabatic perturbation theory \cite{PST:sapt:2002,PST:effective_dynamics_Bloch:2003}. We would like to extend this work to include effects of disorder, something that would also be of interest to the quantum theory of solids. Of course, this is a longer-term ambition where we need to find a way to balance the rigidity of algebraic structures with pseudodifferential theory where properties typically hold “up to an error”. We suspect that incorporating pseudodifferential theory (via the analytic-algebraic point of view developed in \eg \cite{Mantoiu_Purice_Richard:twisted_X_products:2004,Lein_Mantoiu_Richard:anisotropic_mag_pseudo:2009,Belmonte_Lein_Mantoiu:mag_twisted_actions:2010}) into our LRT framework could be an avenue forward.
\end{enumerate}
%

\section{Quantum Hall effect in solid state physics} 
\label{applications:quantum_Hall_effect}
We close this book by revisiting the Quantum Hall Effect, one of \emph{the} prototypical applications of LRT to condensed matter physics. The main purpose is to show \emph{that} and \emph{how} our framework subsumes earlier works on the subject, and not give an exhaustive overview; In particular we will explain the link to the abstract framework from Chapter~\ref{unified}. In a nutshell, the Quantum Hall Effect is a phenomenon that occurs in thin films of semiconductors or insulators subjected to strong external magnetic fields which are perpendicular to the sample. Then the transverse conductivity $\sigma_{\perp}(B)$, seen as a function of the magnetic field, takes values in $\tfrac{e}{h^2} \Z$ with long, pronounced plateaux separated by sharp jumps. The Quantum Hall Effect allowed scientists to measure $\nicefrac{e}{h^2}$, the ratio of the fundamental charge and the square of the Planck constant, with unprecedented accuracy. There are a number of models of varying complexity, but at the end of the day the goal is to justify
\begin{align}
	\sigma_{\perp}(B) := \sigma^{12}[P_{\mathrm{F}}] = \ii \, \mathcal{T} \Bigl ( P_{\mathrm{F}} \, \bigl [ [ X_1 , P_{\mathrm{F}} ] \, , \, [ X_2 , P_{\mathrm{F}}] \bigr ] \Bigr )
	\label{applications:eqn:quantum_Hall_transverse_conductivity}
\end{align}
via the Kubo-Str\v{e}da formula for the \emph{Fermi projection} $P_{\mathrm{F}} := \chi_{(-\infty,E_{\mathrm{F}}]}(H)$. Here, the central assumption is that the Fermi energy $E_{\mathrm{F}}$ lies either in a spectral gap or in a region of Anderson localization, so that the longitudinal conductivity vanishes. Note that $H$, and therefore the Fermi projection $P_{\mathrm{F}}$ depend on the vector potential associated to the magnetic field $B$, but that thanks to gauge covariance of these operators combined with the cyclicity of the trace, the transverse conductivity is actually a function of $B$. A separate argument shows that $\sigma_{\perp}(B)$ is necessarily $\nicefrac{e}{h^2}$ times an integer by connecting the above to a topological quantity, the Chern number; This was Thouless' ingenious insight \cite{Thouless_Kohmoto_Nightingale_Den_Nijs:quantized_hall_conductance:1982} for which he was awarded the Nobel Prize in Physics in 2016.

\subsection{Continuum models} 
\label{applications:quantum_Hall_effect:continuum}
For understanding the Quantum Hall Effect, it turns out we may neglect spin, and the relevant Hilbert space is just $\Hil_{\ast} = L^2(\R^2)$. The magnetic field enters indirectly in the representation of $\R^2$ on $\Hil_{\ast}$,
\begin{align*}
	\bigl ( S^A_y \varphi \bigr )(x) := \e^{- \ii \int_{[x,x-y]} A} \, \varphi(x-y)
\end{align*}
via the phase factor that is the exponential of the magnetic circulation
\begin{align*}
	\int_{[x,x-y]} A := \int_0^1 \dd s \, y \cdot A(x - s y)
	.
\end{align*}
A choice of vector potential $A \in L^2_{\mathrm{loc}}(\R^2,\R^2)$ is necessary to represent the magnetic field $B = \dd A := \partial_1 A_2 - \partial_2 A_1$; For the purpose of the Quantum Hall Effect, $B$ is constant. The magnetic field does not uniquely determine \emph{gauge} $A$, in fact, if $A$ and $A' = A + \nabla \chi$ differ by a gradient, then they give rise to the same magnetic field $\dd A = B = \dd A'$. A common choice is the \emph{transversal} or Coulomb gauge $A = \tfrac{B}{2} (-x_2 , x_1)$.

All of this can be rephrased in the more abstract language of cohomology; A nice discussion can be found in \cite[Section~2.3]{Mantoiu_Purice_Richard:twisted_X_products:2004}. Here, $\Lambda^A(y) \equiv \Lambda^A(\hat{x};y) := \e^{- \ii \int_{[\hat{x},\hat{x}-y]} A}$, seen as a multiplication operator in $x$, is a $1$-cochain associated to the $2$-cocycle twist
\begin{align*}
	\Theta^B(y,z) := \bigl ( \delta^1(\Lambda^A) \bigr )(y,z)
	:= \Lambda^A(\hat{x};y) \, \Lambda^A(\hat{x}-y;z) \, \Lambda^A(\hat{x};y+z)^{-1}
	,
\end{align*}
a multiplication operator that has all the properties listed in Definition~\ref{unified:defn:2_cocycle}. As the notation suggests, $\Theta^B$ depends only on the magnetic field rather than the particular choice of vector potential. This quantity has a geometric interpretation: with the help of Stoke's Theorem and the relation $B = \dd A$, it turns out that it is the exponential of the magnetic flux through the triangle with corners $x$, $x - y$ and $x - y - z$. But because the magnetic field is constant, the magnetic flux only depends on the size of the flux triangle which is independent of its base point $x$. Therefore, $\Theta^B(y,z)$ takes values in $\mathbb{U}(1)$ as stipulated in the definition.

This phase factor also appears as the phase factor when composing translations, writing $1 = \Lambda^A(y+z)^{-1} \, \Lambda^A(y+z)$ immediately yields
\begin{align*}
	S^A_y \, S^A_z = \Theta^B(y,z) \; S^A_{y+z}
	.
\end{align*}
The above ideas generalize to variable magnetic fields where $\Theta^B(y,z)$ is a multiplication operator and $y \mapsto S^A_y$ a \emph{generalized} projective representation (see \eg \cite{Mantoiu_Purice:magnetic_Weyl_calculus:2004,Mantoiu_Purice_Richard:twisted_X_products:2004}).

The most commonly studied model is the Landau Hamiltonian
\begin{align*}
	H^A_{\omega} := (- \ii \nabla - A_{\omega})^2 + V_{\omega}
\end{align*}
with random potentials. The randomness is modeled via the probability space $(\Omega,\mathbb{P})$ on which we are given an ergodic $\R^2$-action $\tau$; Furthermore, we assume that all other technical conditions enumerated in Definition~\ref{unified:defn:topological_dynamical_system} hold. We will always assume that the potentials are covariant random variables, \ie $\mathbb{P}$- and Lebesgue-almost everywhere they satisfy $V_{\omega}(x-y) = V_{\tau_y(\omega)}(x)$ and $A_{\omega}(x-y) = A_{\tau_y(\omega)}(x)$.

Then $H^A_{\omega}$ satisfies \emph{two} different covariance conditions, namely \emph{gauge covariance},
\begin{align}
	H^{A + \nabla \chi}_{\omega} = \e^{- \ii \chi} \, H^A_{\omega} \, \e^{+ \ii \chi}
	,
	\label{applications:eqn:gauge_covariance}
\end{align}
where $\chi : \R^2 \longrightarrow \R$ is seen as a multiplication operator on $\Hil_{\ast}$, and \emph{covariance with respect to magnetic translations},
\begin{align}
	S^A_y \, H^A_{\omega} \, {S^A_y}^{-1} &= H^A_{\tau_y(\omega)}
	.
\end{align}
Dispensing with technical questions about domains for the moment, we see that the Landau Hamiltonian satisfies both of these covariance conditions. Gauge covariance~\eqref{applications:eqn:gauge_covariance} guarantees that expressions such as the right-hand side of \eqref{applications:eqn:quantum_Hall_transverse_conductivity} indeed depend on $B$ rather than the particular choice of $A$.

The continuous case has been treated rigorously in \cite{Elgart_Schlein:Kubo_for_Landau:2004} in the absence of disorder and in \cite{Bouclet_Germinet_Klein_Schenker:linear_response_theory_magnetic_Schroedinger_operators_disorder:2005,klein-lenoble-muller-07,Dombrowski_Germinet:linear_response_theory_non_commutative_integration:2008} with disorder. To ensure the Hamiltonian is selfadjoint, we impose Leinfelder-Simader conditions on the potentials, and hence, $H^A_{\omega}$ is essentially selfadjoint on $C_{\mathrm{c}}^{\infty}(\R^2)$ \cite[Theorem~3]{Leinfelder1981}. The covariant dependence on $\omega$ means we can instead consider all random Schrödinger operators $\widehat{H}^A = \bigl \{ H^A_{\omega} \bigr \}_{\omega \in \Omega}$ simultaneously on the direct integral Hilbert space
\begin{align*}
	\widehat{\Hil} := \int_{\Omega}^{\oplus} \dd \mathbb{P}(\omega) \; L^2(\R^2)
	\simeq L^2(\Omega) \otimes L^2(\R^2)
	,
\end{align*}
which now coincides with the tensor product of Hilbert spaces. Therefore, the principle of affiliation, Proposition~\ref{unified:prop:affiliation}, applies to $\widehat{H}^A$ and the von Neumann algebra
\begin{align*}
	\Alg_{\R^2}^A := \Alg(\Omega,\mathbb{P},\R^2,\Theta^B) := \mathrm{Span}_{\R^2} \bigl \{ \widehat{U}^A_y \bigr \} ' \cap \mathrm{Rand}(\widehat{H})
\end{align*}
where $\widehat{U}_y^A := \bigl \{ S_y^{A_{\omega}} \bigr \}_{\omega \in \Omega}$ is the projective representation of $\R^2$ on $\widehat{\Hil}$ derived from the magnetic translations, seen now as maps between the fibers at $\omega$ and $\tau_y(\omega)$. While the realization of this algebra depends on the choice of $\widehat{A} = \{ A_{\omega} \bigr \}_{\omega \in \Omega}$ rather than $B$, for equivalent vector potentials $\Alg_{\R^2}^A$ and $\Alg_{\R^2}^{A'}$ are isomorphic (by conjugating fiberwise with $\e^{+ \ii \chi_{\omega}}$).

Suppose that the initial equilibrium state is a spectral projection, say the Fermi projection $\widehat{P}_{\mathrm{F}}^A = \bigl \{ P_{\mathrm{F},\omega}^A \bigr \}_{\omega \in \Omega}$ whose Fermi energy $E_{\mathrm{F}}$ is located in a region of dynamical localizaion (\ie the Sobolev condition $\widehat{P}_{\mathrm{F}}^A \in \rr{W}^{1,1} \bigl ( \Alg_{\R^2}^A \bigr ) \cap \rr{W}^{1,2} \bigl ( \Alg_{\R^2}^A \bigr )$ holds). In this situation the Kubo-Str\v{e}da formula for the usual current operators $\widehat{J}_k = \bigl \{ \ii \, \bigl [ X_k , H_{\omega}^A \bigr ] \bigr \}$,
\begin{align}
	\sigma_{\perp}(B) = \sigma^{12}[P_{\mathrm{F}}^A]
	= \ii \int_{\Omega} \dd \mathbb{P}(\omega) \; \mathrm{Tr}_{L^2(\R^2)} \Bigl ( \chi_{[0,1]^2}(x) \, P_{\mathrm{F}}^A \, \Bigl [ \bigl [ X_1 , P_{\mathrm{F}}^A \bigr ] \, , \, \bigl [ X_2 , P_{\mathrm{F}}^A \bigr ] \Bigr ] \Bigr )
	,
	\label{applications:eqn:Kubo_Streda_for_random_Landau}
\end{align}
follows from Theorem~\ref{main_results:thm:Kubo_Streda_formula}, with $\chi_{[0,1]^2}(x) := \chi_{[0,1]}(x_1) \, \chi_{[0,1]}(x_2)$ being the projection onto the cube $[0,1]^2$.

\subsection{Discrete models} 
\label{applications:quantum_Hall_effect:discrete}
Because controlling all the technical minutiae inherent to LRT for continuum models proved challenging, a lot of works (\eg \cite{Bellissard_van_Elst_Schulz_Baldes:noncommutative_geometry_quantum_Hall_effect:1994,Bellissard_Schulz_Baldes:quantum_transport_aperiodic_media:1998,klein-lenoble-muller-07}) are dedicated to simpler, effective tight-binding models that still retain some of the essential features of the continuum model $H^A_{\omega} = (- \ii \nabla - A_{\omega})^2 + V_{\omega}$. The most well-known example is the \emph{Hofstadter model}
\begin{align*}
	H^B := T^B_1 + T^{B,\ast}_1 + T^B_2 + T^{B,\ast}_2
\end{align*}
defined on the Hilbert space $\Hil_{\ast} := \ell^2(\Z^2)$ by means of the generators
\begin{align*}
	\bigl (T^B_1 \varphi \bigr )(n_1,n_2) :& \negmedspace= \e^{+ \ii B n_2}  \, \varphi(n_1-1,n_2)
	,
	\\
	\bigl (T^B_2 \varphi \bigr )(n_1,n_2) :& \negmedspace= \e^{- \ii B n_1} \, \varphi(n_1,n_2-1)
	,
\end{align*}
where $n = (n_1,n_2) \in \Z^2$ is a point on the lattice and $\varphi \in \ell^2(\Z^2)$. The quantity $B \in \R$ represents the magnetic flux per unit cell. For simplicity we will stay in this representation, and garnish symbols with $B$ instead of making the dependence on the choice of vector potential explicit.

The generators $T^B_1$ and $T^B_2$, and therefore the Hamiltonian $H^B$, commute with respect to \emph{magnetic translations}
\begin{align*}
	\bigl ( S^B_1 \varphi \bigr )(n_1,n_2) :& \negmedspace= \e^{+ \ii B n_2} \, \varphi(n_1+1,n_2)
	,
	\\
	\bigl ( S^B_2 \varphi \bigr )(n_1,n_2) :& \negmedspace= \e^{+ \ii B n_1} \, \varphi(n_1,n_2-1)
	.
\end{align*}
A direct computation shows that $\bigl [ T^B_j , S^B_k \bigr ] = 0$ for $j , k = 1 , 2$. The magnetic translations $S^B_1$ and $S^B_2$ are the generators of a projective representation of the group $\Z^2$ given by $(1,0) \mapsto S^B_1$ and $(0,1) \mapsto S^B_2$ and the composition law
\begin{align*}
	S^B_1 \, S^B_2 = \e^{+ \ii 2 B} \, S^B_2 \, S^B_1
\end{align*}
that involves the twisting $2$-cocycle $\Theta^B(n,m) := \e^{\ii B (n_1 m_2 - n_2 m_1)}$ for all $n = (n_1 , n_2)$ and $m = (m_1 , m_2)$ in $\Z^2$.

This time the Hamiltonian $H^B$ itself is an element of the von Neumann algebra
\begin{align*}
	\Alg_{\Z^2}^B := \mathrm{Span}_{\Z^2} \bigl \{ S^B_1 , S^B_2 \bigr \} ' \cap \mathscr{B} \bigl ( \ell^2(\Z^2) \bigr )
	.
\end{align*}
Obviously, we can replace the Hofstadter Hamiltonian with any other selfadjoint element of $\Alg_{\Z^2}^B$ in our arguments. The trace per unit volume $\mathcal{T}_{\Z^2}(A) = \bscpro{\delta_0}{A \, \delta_0}$ is just the expectation value with respect to the vector $\delta_0 = ( \delta_{n_1,0} \, \delta_{n_1,0})_{(n_1,n_2) \in \Z^2}$; Clearly this trace is finite. Thus, we are in the tight-binding setting described in Section~\ref{main_results:tight_binding_simplification} where almost all of the Hypotheses hold automatically. Provided the Fermi projection satisfies the Sobolev condition $P_{\mathrm{F}} \in \rr{W}^{1,1} \bigl ( \Alg_{\Z^2}^B \bigr ) \cap \rr{W}^{1,2} \bigl ( \Alg_{\Z^2}^B \bigr )$, we can again invoke Theorem~\ref{main_results:thm:Kubo_Streda_formula} to obtain
\begin{align}
	\sigma_{\perp}(B) = \sigma^{12}[P_{\mathrm{F}}]
	= \ii \, \bscpro{\delta_0 \,}{\, P_{\mathrm{F}} \, \bigl [ [ X_1 , P_{\mathrm{F}} ] \, , \,  [ X_2 , P_{\mathrm{F}} ] \bigr ]}
	\label{applications:eqn:transverse_conductivity_discrete}
\end{align}
for the transverse conductivity. When the flux $B = \pi \mathbb{Q}$ is rational, it is possible to find an abelian subalgebra of ${\rm Span}_\n{\Z^2} \bigl \{ S^B_1 , S^B_2 \bigr \}$ and then use the Gel'fand-Fourier (Bloch-Floquet) transform in order to recover from
the above equation the celebrated (commutative) Kubo-Chern formula \cite[eqn.~(5)]{Thouless_Kohmoto_Nightingale_Den_Nijs:quantized_hall_conductance:1982}.

Incorporating randomness is straightforward, and all of the above arguments immediately generalize: we need to promote $H^B \rightsquigarrow H^B_{\omega}$ and all other operators such as currents to random covariant operators and replace $\Hil_{\ast} = \ell^2(\Z^2)$ by $\widehat{\Hil} = L^2(\Omega) \otimes \ell^2(\Z^2)$. At the end of the day, compared to \eqref{applications:eqn:transverse_conductivity_discrete} we merely have to add an ensemble average when computing the transverse conductivity in the presence of randomness.
%
%

\backmatter
\printbibliography

\end{document}